\documentclass[preprint,10pt,sort&compress]{elsarticle}
\usepackage[left=1cm,right=2cm,top=1cm,bottom=0.5cm]{geometry}
\usepackage{epsfig, graphics, graphpap}
\usepackage{graphicx,epstopdf}
\usepackage{csquotes}
\usepackage{makeidx}
\usepackage{amssymb}
\usepackage{amsmath, etoolbox,amsthm, amssymb, amsfonts, mathrsfs, mathtools}
\usepackage{commath}
\usepackage{breqn}
\usepackage{stmaryrd}
\usepackage{float,here}
\usepackage{amsthm,bm}
\usepackage{subcaption}
\usepackage{setspace}
\usepackage{multirow,booktabs}
\usepackage{tikz,gensymb}
\usepackage{tikz-3dplot}
\usepackage{pgfplots}
\usepackage{url}
\makeatletter
\newcommand{\mathleft}{\@fleqntrue\@mathmargin\parindent}
\newcommand{\mathcenter}{\@fleqnfalse}
\makeatother

\newtheorem{theorem}{Theorem}[section]

\newtheorem{lemma}[theorem]{Lemma}

\newtheorem*{remark}{Remark}

\begin{document}
	
\begin{frontmatter}

\title{Comprehensive study of forced convection over a heated elliptical cylinder with varying angle of incidences to uniform free stream}

\author[1]{Raghav Singhal$^{\star}$}
\author[2]{Sailen Dutta}
\author[3]{Jiten C. Kalita}
\address[1]{$^{\star}$ Corresponding author email - \textit{raghav2016@iitg.ac.in} :Deparment of Mathematics, Indian Institute of Technology Guwahati, Assam 781039, India}
\address[2]{Email - \textit{sailen.dutta@alumni.iitg.ac.in} :Deparment of Mechanical Engineering, Indian Institute of Technology Guwahati, Assam 781039, India}
\address[3]{Email - \textit{jiten@iitg.ac.in} :Deparment of Mathematics, Indian Institute of Technology Guwahati, Assam 781039, India}


\begin{abstract}
In this paper we carry out a numerical investigation of forced convection heat transfer from a heated elliptical cylinder in a uniform free stream with angle of inclination $\theta^{\degree}$. Numerical simulations were carried out for $10 \leq Re \leq 120$, $0^{\degree} \leq \theta \leq 180^{\degree}$, and $Pr = 0.71$. Results are reported for both steady and unsteady state regime in terms of streamlines, vorticity contours, isotherms, drag and lift coefficients, Strouhal number, and Nusselt number. In the process, we also propose a novel method of computing the Nusselt number by merely gathering flow information  along the normal to the ellipse boundary.  The critical $Re$ at which which flow becomes unsteady, $Re_c$ is reported for all the values of $\theta$ considered and found to be the same for $\theta$ and $180^\degree -\theta$ for $0^\degree \leq \theta \leq 90^\degree$. In the steady regime, the $Re$ at which flow separation occurs progressively decreases as $\theta$ increases. The surface averaged Nusselt number ($Nu_{\text{av}}$) increases with $Re$, whereas the drag force experienced by the cylinder decreases with $Re$. The transient regime is characterized by periodic vortex shedding, which is quantified by the Strouhal number ($St$). Vortex shedding frequency increases with $Re$ and decreases with $\theta$ for a given $Re$. $Nu_{\text{av}}$ also exhibits a time-varying oscillatory behaviour with a time period which is half the time period of vortex shedding. The amplitude of oscillation of $Nu_{\text{av}}$ increases with $\theta$.
\end{abstract}


\end{frontmatter}

\section{Introduction}
Bodies immersed in fluid flow can be characterized as being streamlined or blunt/bluff, depending on its overall shape and structure. A bluff body can be defined as a body that, as a result of its shape, has separated flow over a substantial part of its surface \cite{NAP5870}; any body, which when kept in fluid flow, the fluid does not touch the whole boundary of the object. Roshko \cite{ROSHKO1993} defined a bluff body as one that resulted in a wide extent of separated flow and is associated with significant drag force as well as vortex-shedding. Flow past bluff bodies is commonly found in nature and engineering applications, for instance flow past an airplane, a submarine, an automobile, or wind blowing past a high-rise building. Thus, over the years, massive research efforts have been undertaken to gain a comprehensive understanding of the fluid flow and heat transfer phenomena past bluff bodies of various cross-sectional geometries. Although much effort has been devoted to analyzing the complex flow physics and thermo-fluid transport phenomenon for a variety of cross-sections (circular, rectangular, square, and elliptical), most of the literature deals with circular geometry. A thorough review of this topic can be found in the works of Williamson \cite{williamson1996}, and the books of Zdravkovich \cite{zdravkovich1997flow1,zdravkovich1997flow2}. 

It is well known that, in general, beyond a critical Reynolds number flow around slender cylindrical bodies exhibits periodic vortex shedding as a result of the B\'{e}nard-von K\'{a}rm\'{a}n instability which then leads to alternate vortex structures known as the von K\'{a}rm\'{a}n vortex street. This phenomena is responsible for fluctuating forces on the body that may cause structural vibrations, acoustic noise emissions, and at times, resonance, which would trigger the failure of structures \cite{akde2006}.  Examples of such cylindrical structures in engineering applications include skyscrapers, towering structures, long-spanned bridges, and wires. The frequency associated with the periodic wake, the forces and moment acting on the body, as well as the heat transfer parameters, are a strong function of the body shape and size, Reynolds number of the flow, and the angle of attack \cite{ranjan2008}. Thus, from an engineering point of view, it is crucial to investigate flow around slender bodies with different shapes. 


Over the years plethora of studies, both numerical and experimental, have been undertaken to investigate forced convection heat transfer over a circular cylinder. Notable among the early studies are the ones carried out by Dennis et.al.\cite{dennis1968_forcedconvection}, Apelt and Ledwich \cite{apelt1979}, and Jafroudi and Yang \cite{jafroudi1986}. Subsequent numerical investigations of impact were undertaken by Lange et.al.\cite{lange1998}, Kieft et.al.\cite{kieft2003wake}, Shi et.al. \cite{shi2004heating}, Bharti et.al.\cite{bharti2007numerical}, Sarkar et.al.\cite{sarkar2011unsteady}. More recently, Cao et.al.\cite{cao2021forced} numerically analyzed forced convection heat transfer around a heated circular cylinder in laminar flow regime ($Re = 20$ - $180$, $Pr=0.7$) from the Lagrangian viewpoint. They computed the Lagrangian coherent structures and employed them to study the convection features around the cylinder at different $Re$'s. Among the experimental studies, the works of Dumouchel et.al.\cite{DUMOUCHEL1998}, Wang et.al. \cite{wang2000relationship}, Kieft et.al.\cite{kieft2003wake}, Nakamura and Igarashi \cite{NAKAMURA2004} stand out. 

The most commonly studied geometry after the circular cylinder is that of a square/rectangular cylinder. Thus, several studies - mostly numerical - exist for the forced convection heat transfer phenomena over a square geometry as well. Notable among them are the works of Sharma and Eswaran \cite{sharma_NHT_2004}, Dhiman et.al.\cite{dhiman2005flow}, Ranjan et.al.\cite{ranjan2008}, Sahu et.al. \cite{sahu2009unsteadyeffects}, Sen et.al.\cite{sen2011_square}, Bai and Alam \cite{bai2018dependence}. Other unusual shapes such as a triangular cylinder (De and Dalal \cite{de_dalal2006}), semi-circular cylinder (Chandra and Chhabra \cite{CHANDRA2011}, Chatterjee et.al.\cite{chatterjee2012}, Bhinder et.al.\cite{bhinder2012flow}), cam-shaped cylinder (Chamoli et.al.\cite{chamoli2019effect}), blunt-headed cylinder (Pawar et.al.\cite{PAWAR2020}) have also garnered the attention of researchers in recent years.

Among the various cross sections/shapes of bluff bodies (cirular, rectangular/square, elliptical) the elliptic geometry has been considered the elementary shape of interest for wings, submarines, rotor blades, and missiles \cite{yoon2016bifurcation}. The problem of flow past an elliptical cylinder has received intermittent attention over the years from scientific community. Lugt and Haussling \cite{lugt1974laminar} numerically investigated laminar flow past an elliptic cylinder at $45^{\degree}$ angle of incidence. The solutions were shown to approach steady and quasi-steady states at $Re = 15$ and $Re = 30$ respectively, while a K\'{a}rm\'{a}n vortex
street developed at $Re = 200$. Patel \cite{patel1981flow} studied the development of K\'{a}rm\'{a}n vortex street for flow past an impulsively started elliptic cylinder for $Re = 200$ at different angles of incidence ($\alpha = 0^{\degree}$, $30^{\degree}$, $45^{\degree}$, $90^{\degree}$) and presented semi-analytical solutions in terms of flow characteristics such as surface pressure and vorticity distributions, the transient development of streamlines and equivorticity lines, and drag coefficient. Jackson \cite{jackson1987finite}, while investigating the critical Reynolds number for the onset of vortex shedding for 2D laminar flow past bluff bodies of different shapes reported that, for an elliptic cylinder, the values of the critical $Re$ and the
corresponding Strouhal number decreased as the angle of incidence increased. Park et.al. \cite{PARK1989} studied the effect of angle of incidence on the unsteady laminar flow past an impulsively started, slender elliptic cylinder for $25 \leq Re \leq 600$. They identified five distinct flow regimes - two steady flow regimes which were demarcated by the presence of a steady separation bubble, and three unsteady regimes which were characterized by the frequency and amplitude of the periodic variations of force coefficients. Johnson et.al. \cite{johnson2001flow} investigated the vortex structures behind 2D elliptic cylinders for $30 \leq Re \leq 300$ and aspect ratio ($AR$) in the range $0.01 - 1$. They reported that as the $AR$ is decreased, the shedding behind the elliptic cylinder changes from steady K\'{a}rm\'{a}n vortex shedding to flow with two distinct regions. The first region is situated directly behind the cylinder and contains two rows of vortices rolling up from the cylinder with a region of relatively dead flow in between. The second region is located further downstream consisting of secondary vortices that results from a strong interaction of the two rows of vortices due to a convective instability. Faruquee et.al.\cite{faruquee2007effects} examined the effect of $AR$ on the flow field of an elliptic cylinder for $0.3 \leq AR \leq 1$ at $Re = 40$ with the cylinder placed with the major axis parallel to the free-stream, and reported various wake parameters, drag coefficient, pressure and velocity distributions in terms of $AR$. They also reported a critical $AR$ of $0.34$ below which no vortices form behind the cylinder. Sen et.al\cite{sen2012steady} calculated the laminar separation Reynolds number ($Re_s$) for $Re \leq 40$, $0^{\degree} \leq \alpha \leq 90^{\degree}$, and $AR = 0.2$, $0.5$, $0.8$, and $1$. Paul et.al. \cite{paul2014onset} presented a numerical study on predicting onset of flow separation and vortex shedding in flow past unconfined 2D elliptical cylinders for various $AR$'s and a wide range of Angles of Attack (AOA). They employed a variety of methods to estimate $Re_s$, critical Reynolds number ($Re_{cr}$), and critical Strouhal number ($St_{cr}$), and proposed functional relationships for $Re_{cr}$ and $St_{cr}$ in terms of $AR$ and AOA. Yoon et.al.\cite{yoon2016bifurcation} investigated the flow around an elliptic cylinder for $20 \leq Re \leq 100$, $0^{\degree} \leq \alpha \leq 90^{\degree}$, and $AR = 0.2$. They reported that the Strouhal number decreased as the angle of incidence increased, and the rate of decrease in the values of the Strouhal number was faster when the value of $Re$ increased. While measuring the variation of the stagnation point, they found that it moved downstream along the lower surface of the cylinder as the angle of incidence increased, and the time-averaged stagnation point is strongly dependent on the angle of incidence and weakly dependent on $Re$. Thus, we see that a number of important studies have been carried out to understand the flow phenomena over an unconfined elliptic cylinder. However, there is a distinct lack of comprehensive studies dealing with heat transfer phenomena w.r.t. to flow past an elliptic cylinder. The current work is attempt to address this issue. 

Over the years, it has been observed that the streamfunction - vorticity ($\psi$-$\zeta$) form of the Navier-Stokes (N-S) equations is preferred over the primitive form for the computation of 2D incompressible viscous flows, owing to the absence of the pressure term in the $\psi$-$\zeta$ form. Recently, Singhal and Kalita\cite{singhal2021novel} developed a new Higher Order Compact  Explicit Jump Immersed Interface Method (HEJIIM) for solving two-dimensional elliptic problems with singular source and discontinuous coefficients in the irregular region on Cartesian mesh. This scheme was shown to maintain its compactness on a nine-point stencil at both regular and irregular points unlike the previous IIM approaches. Further, in order to maintain fourth-order accuracy throughout the computational domain, they modified the explicit jump immersed interface strategy of Wiegmann and Bube \cite{wiegmann2000explicit} to treat the jump across the interface. In a subsequent work, Singhal and Kalita \cite{singhal_transient_2022} proposed a new HOC finite difference Immersed Interface Method (IIM) for 2D transient problems involving bluff bodies immersed in incompressible viscous flows on Cartesian mesh, which like its steady counterpart \cite{singhal2021novel} was shown to maintain its compactness on a nine point stencil at both the regular and irregular points. In this paper we have utilized this recent scheme of Singhal and Kalita \cite{singhal_transient_2022} to simulate and analyze forced convection heat transfer over an elliptic cylinder at an angle incidence .

The manuscript is organized as follows: In section \ref{sec:problem-statement} we lay out the problem, governing equations, and the imposed initial and boundary conditions. In section \ref{sec:Nusselt-drag-lift}, we describe the novel procedure developed to calculate Nusselt number, as well the method used by Singhal and Kalita \cite{singhal_transient_2022} to calculate drag and lift coefficients. The solution procedure is outlined in section \ref{sec:solution-method}. Next, we validate our code by simulating steady state forced convection over a circular cylinder and comparing the present results with well established results in the literature. Grid independence study is also carried out in this section (section \ref{sec:code-validation}). We present our results in section \ref{sec:results}. This section is divided into two subsections: \ref{sec:steady} contains results for steady state, and \ref{sec:transient} the results for transient state. Finally, we conclude this article in section \ref{sec:conclusions}.

\section{Problem statement and governing equations} \label{sec:problem-statement}
Consider a heated elliptical cylinder of aspect ratio $AR (= 2/3)$ placed in a uniform free stream (figure \ref{Fig:forced-convection-ellipse}). The fluid flow is two-dimensional, incompressible and laminar with constant properties. Additionally the effect of gravity is neglected. The free stream velocity is $U_0$ and the fluid Prandtl number ($Pr$) is taken to be $0.71$. The surface of the cylinder is maintained at a constant temperature of $T_s$, whereas the free stream has a temperature $T_{\infty}$. It is assumed that the temperature difference $\Delta T (=T_s-T_{\infty})$ has a negligible effect on the fluid properties.

\begin{figure}[htbp]
	\centering
	\includegraphics[width = \textwidth]{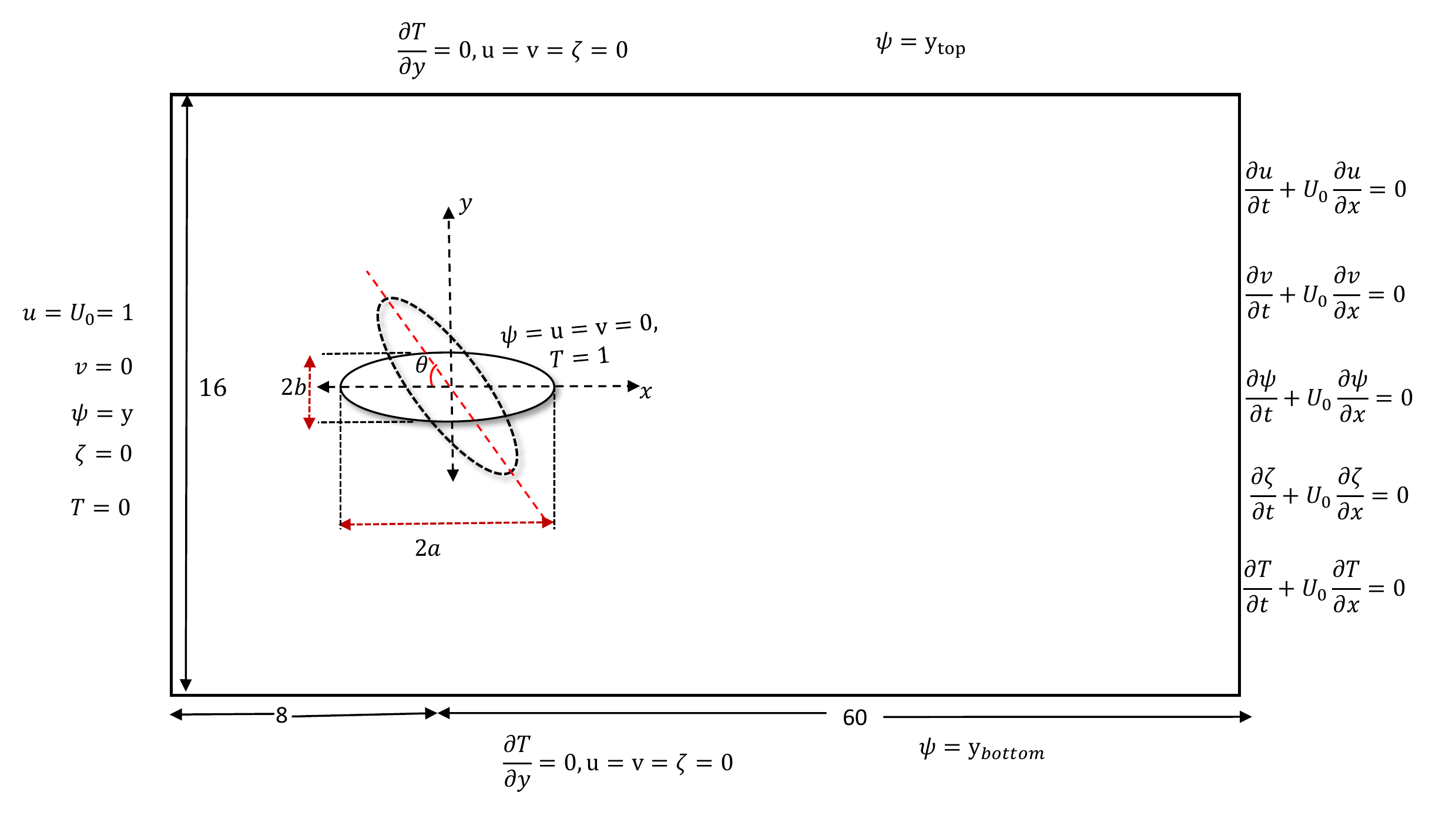}
	\caption{\small{Schematic and boundary conditions for forced convection over an inclined elliptic cylinder.}}\label{Fig:forced-convection-ellipse}
\end{figure}

Under these assumptions the 2D incompressible N-S equations, in streamfunction-vorticity ($\psi$-$\zeta$) form and the energy equation in dimensionless form, are given by

\begin{align}
\nabla^2 \psi & = -\zeta \label{eq:streamfunction} \\
\dfrac{\partial \zeta}{\partial t} + u \dfrac{\partial \zeta}{\partial x} + v \dfrac{\partial \zeta}{\partial y} & = \dfrac{1}{Re} \nabla^2 \zeta \label{eq:vorticity transport} \\
\dfrac{\partial T}{\partial t} + u \dfrac{\partial T}{\partial x} + v \dfrac{\partial T}{\partial y} & = \dfrac{1}{Re Pr} \nabla^2 T \label{eq:temperature}
\end{align}
Here $Re=\dfrac{U_0L}{\nu}$ is the Reynolds number ($U_0$ and $L$ being characteristic velocity and length respectively), and $Pr=\dfrac{\nu}{\alpha}$ is the Prandtl number, where $\nu$ and $\alpha$ are the dynamic viscosity and thermal diffusivity of the fluid respectively. 
The streamfunction ($\psi$) and vorticity ($\zeta$) are defined as follows:
\begin{equation}
	u = \dfrac{\partial \psi}{\partial y}, \quad v = -\dfrac{\partial \psi}{\partial x} \qquad \text{and} \qquad \zeta = \dfrac{\partial v}{\partial x} - \dfrac{\partial u}{\partial y}.
	\label{eq:vel-zeta}
\end{equation}
For the purpose of determining the immersed boundary, we have employed the level set function developed by Sethian and Osher \cite{osher1988fronts}. The level set function for an inclined ellipse is defined as 
\begin{equation}
\phi(x,y, \theta)= \left( \frac{-(x-xc)cos(\theta)+ (y-yc)sin(\theta)}{a} \right)^2 + \left( \frac{(x-xc)sin(\theta)+ (y-yc) cos(\theta)}{b} \right)^2 -1 
\end{equation}
where $a$ and $b$ are major and minor axis, and $(xc, yc)$ is the center of the ellipse, and $\theta$ is an angle which varies from $0 \leq \theta \leq 2 \pi$.  The normal vector is defined as
\begin{equation}
	\mathbf{n} = (n_1, n_2)= \frac{\nabla \phi}{\vert \nabla \phi \vert } = \frac{ \phi_{x} \vec{i} + \phi_{y} \vec{j}}{\sqrt{\phi_{x}^2+\phi_{y}^2}}.
\end{equation} 

\subsection{Initial and boundary conditions}
The following initial and boundary conditions are applied on the non-dimensional variables as follows (figure \ref{Fig:forced-convection-ellipse})
\begin{enumerate}
	\item Initial condition: At time $t=0$, $u=1$, $v=0$, and $T=0$
	
	\item At the inlet of the domain, the fluid flow is uniform with constant temperature i.e., $u=1$, $v=0$, and $T=0$.
	
	\item Convective boundary conditions are applied on the outlet of the domain, i.e.,
	$\dfrac{\partial \Phi}{\partial t} + U_0 \dfrac{\partial \Phi}{\partial x} = 0$, where $\Phi = u$, $v$, $\psi$, $\zeta$, $T$. 
	
	\item Free slip and adiabatic boundary conditions are applied on the top and bottom wall, i.e., $u=v=\zeta=0$, $\dfrac{\partial T}{\partial y}=0$, $\psi=y_T$ at the top boundary and $\psi = y_B$ at the bottom boundary. Here the subscripts $T$ and $B$ denote 'Top' and 'Bottom' respectively.
	
	\item On the surface of the cylinder, no-slip boundary conditions are applied along with constant temperature, i.e., $u=v=\psi=0$ and $T=1$.
\end{enumerate}
\section{Solution methodology}\label{sec:solution-method}
\subsection{Numerical Scheme}
Singhal and Kalita \cite{singhal_transient_2022} have developed an HOC (9,9) scheme for a 2D Parabolic interface problem for the variable $\Phi(x,y,t)$ of the type 
\begin{equation}\label{eq:parabolic_gen}
	\lambda \Phi_t = \nabla. (\beta \nabla \Phi) + \kappa \Phi -f +b\delta\{(x-x^*)(y-y^*)\} \quad \text{in } \Omega \times (0,\infty), \quad (x^*,y^*) \in \Gamma
\end{equation} 
with specified initial and boundary conditions. Here $\Omega$ is an open bounded subset in $\mathbb{R}^2$ and $\bf{x}$ = $(x, y)$ is an interior point in the domain having an interface $\Gamma$ immersed in it, and $(x^*, y^*) \in \Gamma$ is an
interfacial point.

At regular points, the last term in \eqref{eq:parabolic_gen} vanishes, and as such the equation \eqref{eq:parabolic_gen} can be recast in the convection-diffusion-reaction form as

\begin{equation}\label{eq:convection-diffusion-reaction}
	\lambda \Phi_t + \beta_x \Phi_x + \beta_y \Phi_y + \beta \nabla^2 \Phi + \kappa \Phi  = f
\end{equation}

Singhal and Kalita \cite{singhal_transient_2022} used the methodology of Kalita et.al. \cite{kalita_2002} to obtain a high order compact finite difference approximation of equation \eqref{eq:convection-diffusion-reaction} by using  uniform spacings $h$ and $l$ along $x$- and $y$-directions respectively with time step $\Delta t$. The eventual form of the HOC finite difference scheme for equation \eqref{eq:convection-diffusion-reaction} can be written as
\begin{equation}
	\begin{aligned}
		& \lambda \Bigg[1+\dfrac{h^2}{2}\Big(\delta_{xx} + \dfrac{(c-2\beta_x)}{\beta}\delta_x\Big) + \dfrac{l^2}{2}\Big(\delta_{xx} + \dfrac{(d-2\beta_y)}{\beta}\delta_y \Big)\Bigg] (\Phi_{ij}^{n+1}-\Phi_{ij}^n) = \dfrac{\Delta t}{2} (F_{ij}^{n+1}-F_{ij}^n) + \\
		& \dfrac{\Delta t}{2}\Big[A_{ij}\delta^2_x + B_{ij}\delta^2_y + C_{ij}\delta_x + D_{ij}\delta_y + E_{ij}\delta^2_x\delta^2_y + H_{ij}\delta_x\delta^2_y + K_{ij}\delta^2_x\delta_y + L_{ij}\delta_x\delta_y + M_{ij} \Big] (\Phi_{ij}^{n+1}+\Phi_{ij}^n)  
	\end{aligned}
\end{equation}
where $\delta^2_x\delta^2_y, \delta_x, \delta_y, \delta_x\delta_y, \delta_x\delta^2_y, \delta^2_x\delta_y, \delta^2_x\delta^2_y$ are second order accurate central difference operators along $x$- and $y$- directions. The details of the coefficients $A_{ij}, B_{ij}, C_{ij}, D_{ij}, E_{ij}, H_{ij}, K_{ij}, L_{ij}, M_{ij}$ can be found in the work of Singhal and Kalita \cite{singhal_transient_2022}.

\subsection{Solution of system of equations}  
The set of equations that result from discretizing equations \eqref{eq:streamfunction} - \eqref{eq:temperature} can be written in matrix form as
\begin{align}
A_{1} \psi^{n+1} &= f_1 (\zeta^{n}, \tilde{C}_{\psi}^{n}) \label{eq:se1} \\ 
A_{2} \zeta^{n+1} &= f_2 (\zeta^{n}, u^{n+1}, v^{n+1}, Re, \tilde{C}_{\zeta}^{n+1}, \tilde{C}_{\zeta}^{n}) \label{eq:se2} \\ 
A_{3} T^{n+1} &= f_3 (T^{n}, u^{n+1}, v^{n+1}, Re, Pr, \tilde{C}_{T}^{n+1}, \tilde{C}_{T}^{n})  \label{eq:se3}
\end{align}
In above equations, the coefficient matrices $A_1$, $A_2$ and $A_3$ are asymmetric sparse matrices containing a maximum of nine non-zero values on the diagonals in each row. $\tilde{C}_{\psi}^{n}$, $\tilde{C}_{\zeta}^{n}$, $\tilde{C}_{\zeta}^{n+1}$ and $\tilde{C}_{T}^{n}$, $\tilde{C}_{T}^{n+1}$ are the  streamfunction, vorticity and temperature correction vectors respectively at the irregular points corresponding to the $n^{\rm th}$ and $(n+1)^{\rm th}$ time level.  For a grid of size $M \times N$, the matrices $A_1$, $A_2$ and $A_3$ are of order $MN$ and  $\psi^{n+1}$, $\zeta^{n}$, $\zeta^{n+1}$, $T^{n}$, $T^{n+1}$, $u^{n+1}$, $v^{n+1}$, $\tilde{C}_{\psi}^{n}$, $\tilde{C}_{\zeta}^{n}$, $\tilde{C}_{\zeta}^{n+1}$, $\tilde{C}_{T}^{n}$, $\tilde{C}_{T}^{n+1}$ are vectors of length $MN$.  

Note that the discrete values of the velocities at the  $(n+1)^{\rm th}$ time level are contained in the equations \eqref{eq:se2} and \eqref{eq:se3}. However, they are accessible after computing streamfunction from equation \eqref{eq:se1}. The fourth order approximation of the velocities $u$, $v$ are obtained by the method outlined in the work of Kalita et.al.\cite{kalita2001pre}. An outer-inner iteration procedure is used to calculate the solutions to the problems governed by equations \eqref{eq:streamfunction} - \eqref{eq:temperature}. The following steps describe this computational algorithm:
\begin{enumerate}
	\item Initialize $u$, $v$, $\psi$, $\zeta$ and $T$ and apply the appropriate boundary conditions.
	\item Calculate streamfunction jump correction $\tilde{C}_{\psi}$.
	\item Solve equation \eqref{eq:se1} to obtain $\psi$.
	\item Compute $u$ and $v$ by Thomas algorithm from equations \eqref{eq:vel-zeta} \cite{jaiswal2020novel, kalita2008efficient,  kumar2019transformation}.
	\item Calculate vorticity and temperature jump corrections $\tilde{C}_{\zeta}$, $\tilde{C}_{T}$.
	\item Use \eqref{eq:se2} and \eqref{eq:se3} to determine $\zeta$ and $T$.\\
	This comprises an outer iteration.
	\item Once the discrete values of $u$, $v$, $\psi$, $\zeta$ and $T$ are updated, repeat the steps 2-6.
\end{enumerate}

Since $A_1$, $A_2$ and $A_3$ are sparse matrices, solving them requires the use of iterative techniques. Using traditional iterative techniques like Gauss-Seidel is not worthwhile since the coefficient matrices $A_1$, $A_2$, $A_3$ are not diagonally dominant. The inner iterations consist of solving the matrix equations \eqref{eq:se1} - \eqref{eq:se3} at each outer iteration by iterative solvers. The inner iterations are made up of efficient iterative solvers solving equations \eqref{eq:se1}- \eqref{eq:se3} at each time step. In our computations, we employed the Biconjugate gradient stabilized (BiCGStab) \cite{kelley1995iterative} iterative solver, along with Incomplete LU decomposition as a preconditioner, with the help of Lis Library \cite{LIS}. When the residual vectors resulting from equations \eqref{eq:se1}- \eqref{eq:se3} fell below $10^{-9}$, the inner iterations were terminated. We performed all of our calculations on a computer with a 32 GB RAM and an Intel Xeon processor.
\section{Calculation of non-dimensional parameters}\label{sec:Nusselt-drag-lift}
The Nusselt number characterises the rate of heat transfer across the fluid around the heated elliptic cylinder. On the other hand, drag and lift coefficients are dimensionless quantities that is related to the drag and lift generated by a bluff body across the fluid in its neighbourhood. As such they are vital parameters yielding useful information about the heat and fluid flow characteristics for the problem under consideration.  In this section, we describe in brief the procedure for calculating the Nusselt number, and the drag and lift coefficients.
\subsection{Nusselt number}\label{sec:nusselt-calculation}
The quantitative parameter indicating heat transfer, i.e. the local Nusselt number ($Nu$), is defined as
\begin{equation}
	Nu = -\dfrac{\partial T}{\partial n} \label{eq:local_Nusselt}
\end{equation}
where $n$ is the direction normal to the cylinder surface.
\begin{figure}[htbp]
	\centering
	\includegraphics[width = 0.5\textwidth]{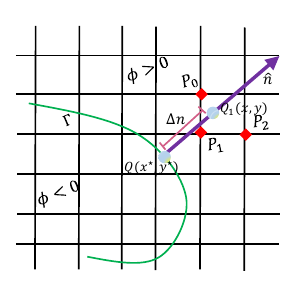}
	\caption{\small{Schematic of Nusselt number computation along the boundary an inclined elliptic cylinder.}}\label{Fig:Schematic for calculating Nusselt number}
\end{figure}

Contrary to the usual approach of resorting to grid-transformation for calculating Nusselt number of bluff bodies, we have calculated it using the following approach, where merely gathering flow information along the normal direction to the boundary of the bluff body suffices. We have divided the interface, i.e., the surface of the cylinder, into $N_P$ number of interfacial points. Now, in order to identify the interfacial points on the interface, we require the polar coordinates of the inclined ellipse, which is obtained as follows: 

Let 
\begin{eqnarray}
	-(x-xc)cos(\theta)+ (y-yc)sin(\theta)  =& a cos(\varphi) \label{polar_1} \\
	(x-xc)sin(\theta)+ (y-yc) cos(\theta)=& b sin(\varphi) \label{polar_2}
\end{eqnarray}
Add both equations \eqref{polar_1} and  \eqref{polar_2} after multiplying   by $cos(\theta)$ in \eqref{polar_1} and  $(-sin(\theta))$ in \eqref{polar_2}, we get
\begin{equation}
	x= xc - a cos(\theta) cos(\varphi) + b sin(\theta) sin(\varphi)
\end{equation}
Similarly, add both equations \eqref{polar_1} and  \eqref{polar_2} after multiplying   by $sin(\theta)$ in \eqref{polar_1} and  $cos(\theta)$ in \eqref{polar_2}, we get
\begin{equation}
	y= yc+ a sin(\theta) cos(\varphi) + b cos(\theta) sin(\varphi)
\end{equation}

Let $Q^\star(x^\star, y^\star)$ be an interfacial point on the bluff body's boundary, and $Q_{1} (x,y)$ be a point in the normal direction of $Q^\star (x^\star, y^\star)$ with $\triangle \mathbf{n}$ being the distance between $Q^\star$ and $Q_1$ along the direction normal to the interface at $Q^\star$ (see figure \ref{Fig:Schematic for calculating Nusselt number}). Then
\begin{equation}
	Q_{1} (x,y)= Q^\star(x^\star, y^\star) +  (n_1, n_2)\triangle \mathbf{n}	\nonumber
\end{equation}
Thus, the local $Nu$ at a point $Q^\star(x^\star, y^\star)$ is given by
	\begin{align}
	\left . Nu \right |_{Q^\star(x^\star, y^\star)}  = & - \left.\dfrac{\partial T}{\partial \mathbf{n}}\right |_{Q^\star(x^\star, y^\star)} \\
	 = & - \dfrac{ T (Q_{1} (x,y)) - T(Q^\star(x^\star, y^\star))}{\triangle \mathbf{n}}
	\end{align}		
Although the approximation of the value of $T$ at the point $Q_1(x,y)$ in the normal direction, it does not have to be a grid point in the computational domain, hence $T (Q_{1} (x,y))$ is unknown.  As such, we compute the value of $T (Q_{1} (x,y))$ using a linear bivariate interpolating polynomial, which is given as follows:

Let $p(a,b)$ be a linear bivariate interpolation polynomial in two variable is defined by\\
\begin{equation}
	p(a,b)=p_{0}+p_{1}a+p_{2}b \label{lbp_1}
\end{equation}
Given three points $P_{0}(a_{0},b_{0})$, $P_{1}(a_{1},b_{1})$, $P_{2}(a_{2},b_{2})$, the Vandermonde matrix on these three nodes is defined as
\begin{center}
	$\mathcal{P}=$
	$\begin{bmatrix}
		1 & a_{0} & b_{0} \\
		1 & a_{1} & b_{1} \\
		1 & a_{2} & b_{2} \\
	\end{bmatrix}$
\end{center}
\begin{lemma}
	Interpolating of $T$ by polynomials $p(a,b)$ on the points $P_{0}$, $P_{1}$ and $P_{2}$ is always possible if and only if $det \mathcal{P} \neq 0$
\end{lemma}
\begin{proof}
	Let us represent the vector of $T$ values at the three points by $\mathcal{F}$ = $[t_{0}, t_{1}, t_{2}]^T$ and define $\mathbb{X}$ = $[p_{0}, p_{1}, p_{2}]^T$. Considering the fact that $p(a,b)$ satisfies $p(a_{i}, b_{i})= t_{i}$ for $i  \in \lbrace 0, 1, 2 \rbrace $ can be expressed as $\mathcal{P}\mathbb{X} = \mathcal{F}$ which provides a solution for an arbitrary $\mathbb{X}$ if and only if $det \mathcal{P} \neq 0$.
\end{proof}
\begin{remark}
	If $P_{0}$, $P_{1}$ and $P_{2}$ are lies on a same line then interpolation by linear polynomials is not possible on these points.
\end{remark}
In the above linear interpolation, the unknown coefficients $p_{0}$, $p_{1}$ and $p_{2}$  are explicitly provided by 
\begin{eqnarray*}
	p_{0}=& (t_{0}a_{1}b_{2} - t_{0}a_{2}b_{1} - t_{1}a_{0}b_{2} + t_{1}a_{2}b_{0} + t_{2}a_{0}b_{1} - t_{2}a_{1}b_{0})/ A\\
	p_{1}=&(t_{0}b_{1} - t_{1}b_{0} - t_{0}b_{2} + t_{2}b_{0} + t_{1}b_{2} - t_{2}b_{1})/A \\
	p_{2}=&  -(t_{0}a_{1} - t_{1}a_{0} - t_{0}a_{2} + t_{2}a_{0} + t_{1}a_{2} - t_{2}a_{1})/A
\end{eqnarray*}
where $A=(a_{0}b_{1} - a_{1}b_{0} - a_{0}b_{2} + a_{2}b_{0} + a_{1}b_{2} - a_{2}b_{1})$.\\

We determined the local Nusselt number at the point $Q_{1} (x,y)$ using the above interpolation formula \eqref{lbp_1} by selecting three nearest grid points. Thus, the local $Nu$ is calculated at $N_P$ points. Note that in our computation of local $Nu$, we have taken $\triangle \mathbf{n} = 0.1$ and $N_P = 201 $. 

The surface averaged Nusselt number is given by 
\begin{equation}
	Nu_{\text{av}} = \dfrac{1}{W}\int_W Nu\,dS \label{eq:Nusselt}
\end{equation}
where $W$ is the surface area of the cylinder. The integral in equation \eqref{eq:Nusselt} is calculated using Simpson's $1/3$ rule.

\subsection{Calculation of Drag and Lift forces}\label{sec:drag-lift-calc}
The drag ($C_D$) and lift ($C_L$) coefficients, which are the non-dimensional form of the drag ($F_D$) and lift ($F_L$) forces, are given by
\begin{equation}
\begin{aligned}
	C_D = -2\int\int_V \dfrac{\partial u}{\partial t} dx dy & +  2 \oint_S \Big(uv + yv\zeta - y\dfrac{\partial v}{\partial t} + \dfrac{1}{Re} y\nabla^2 u\Big) dx \\ &+ 2\oint_S \Bigg[\dfrac{1}{2}(v^2-u^2) - yu\zeta - y\dfrac{\partial v}{\partial t} + \dfrac{1}{Re}\Big(y\nabla^2 v + 2 \dfrac{\partial u}{\partial x} + \dfrac{\partial v}{\partial y} + \dfrac{\partial v}{\partial x} \Big) \Bigg] dy\\
\end{aligned}
\end{equation} \label{eq:drag}

\begin{equation}
	\begin{aligned}
		C_L = -2\int\int_V \dfrac{\partial v}{\partial t} dx dy & + 2\oint_S \Bigg[\dfrac{1}{2}(v^2-u^2) - xv\zeta - x\dfrac{\partial u}{\partial t} + \dfrac{1}{Re}\Big(x\nabla^2 u + \dfrac{\partial u}{\partial y} + \dfrac{\partial v}{\partial x} + 2\dfrac{\partial v}{\partial y} \Big) \Bigg] dx  \\ &+2 \oint_S \Big(-uv + xu\zeta + x\dfrac{\partial u}{\partial t} - \dfrac{1}{Re} x\nabla^2 v\Big) dy\\
	\end{aligned}
\end{equation} \label{eq:lift}
Here $V$ is an arbitrary control volume bounded by a control surface $S$. The expressions given by \eqref{eq:drag} and \eqref{eq:lift} are obtained by utilizing the momentum approach of Noca et.al. \cite{NOCA1999551}, who devised a formula that does not require explicit knowledge of the pressure term. A detailed derivation of the same can be found in the work of Singhal and Kalita \cite{singhal_transient_2022}.

\section{Code Validation and Grid Independence}\label{sec:code-validation}
\subsection{Code validation}
In order to validate our code, firstly we simulate forced convection over a horizontal circular cylinder at low Reynolds numbers. As will be seen shortly, the results from the present computation are an excellent match with well established results in the literature. Note that the computational domain as well as the boundary conditions for this case is the same as shown in figure \ref{Fig:forced-convection-ellipse}. The only difference is that the elliptical cylinder has been replaced by a circular cylinder of characteristic length (diameter) $D=1$. 

\begin{figure}[H]
	\centering
	\includegraphics[width = 0.8\textwidth]{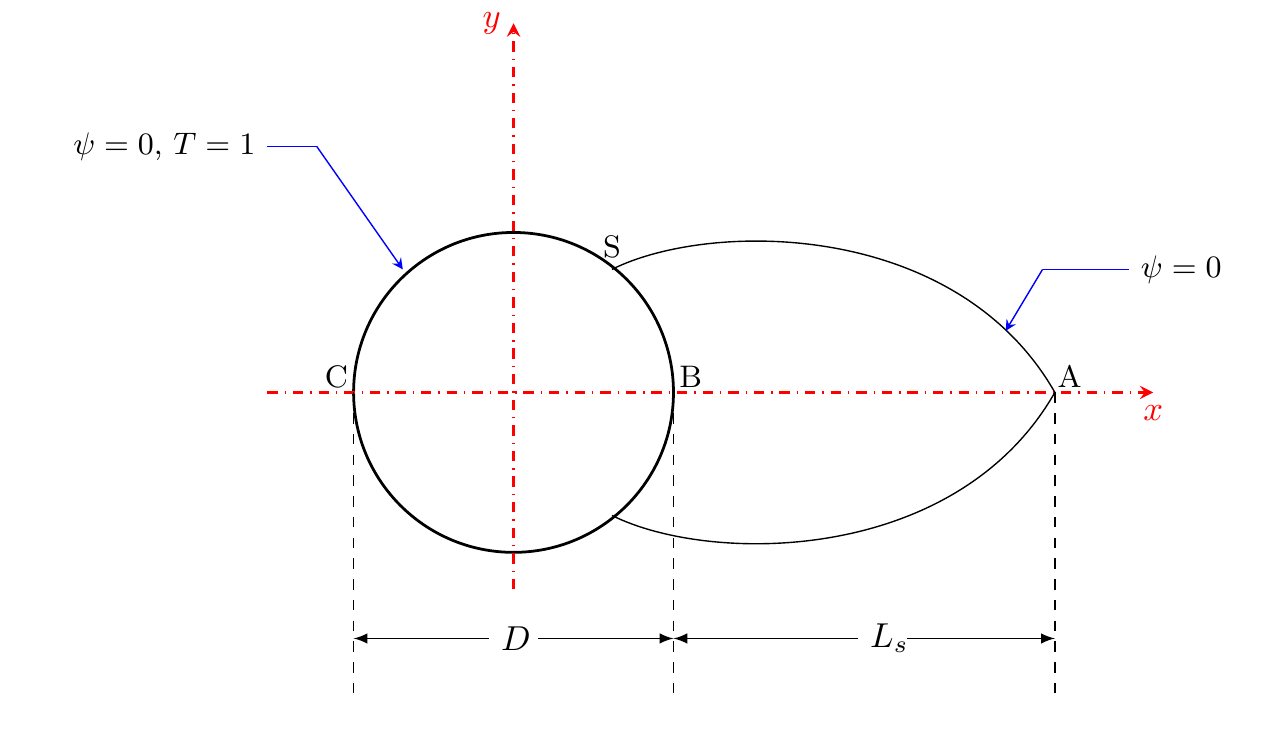}
	\caption{\small{Schematic of wake-bubble geometry for the circular cylinder.}}\label{Fig:fc-circle}
\end{figure}

Figure \ref{Fig:fc-circle} shows the typical wake-bubble geometry of the flow. Points A, B, and C denote wake stagnation point, back stagnation point, and front stagnation point respectively. The eddy length $L_s$ is the distance from the rear of the cylinder to the wake stagnation point. 


\begin{table}[H]
	\centering
	\caption{\small{Comparison of eddy length ($L_s$)}}
	\begin{tabular}{|c|c|c|c|c|}
		\hline
		\multirow{2}[4]{*}{$Re$} & \multicolumn{4}{c|}{$L_s$} \\
		\cline{2-5}          & Present & Biswas and Sarkar \cite{biswas2009} & Takami and Keller \cite{takami1969}& Dennis and Chang \cite{dennis1970} \\
		\hline
		15    & 1.224 & 1.189 & 1.162 & -- \\
		\hline
		20    & 1.831 & 1.865 & 1.844 & 1.88 \\ 
		\hline
		30    & 3.225 & 3.226 & 3.223  & --  \\
		\hline
		35    & 3.859 & 3.793 &  --     & -- \\
		\hline
		40    & 4.455 & 4.424 & 4.650   & 4.69 \\
		\hline
	\end{tabular}%
	\label{tab:circ-comparison-ls}%
\end{table}%


\begin{table}[H]
	\centering
	\caption{\small{Comparison of surface averaged Nusselt number ($Nu_{\text{av}}$)}}
	\begin{tabular}{|c|c|c|c|c|}
		\hline
		\multirow{2}[4]{*}{$Re$} & \multicolumn{4}{c|}{$Nu_{\text{av}}$} \\
		\cline{2-5}          & Present & Biswas and Sarkar \cite{biswas2009} & Jafroudi and Yang \cite{jafroudi1986} & Apelt and Ledwich \cite{apelt1979} \\
		\hline
		15    & 2.2103 & 2.1809 & 2.176 & 2.193 \\
		\hline
		20    & 2.4617 & 2.4483 & 2.433 & -- \\ 
		\hline
		30    & 2.9287 & 2.8877 & 2.850  & --  \\
		\hline
		35    & 3.1281 & 3.0772 &  --     & -- \\
		\hline
		40    & 3.2492 & 3.2351 & 3.2   & 3.255 \\
		\hline
	\end{tabular}%
	\label{tab:circ-comparison-nu}%
\end{table}

\begin{figure}[H]
	\centering
	\begin{subfigure}{0.4\textwidth}
		\includegraphics[width=\linewidth]{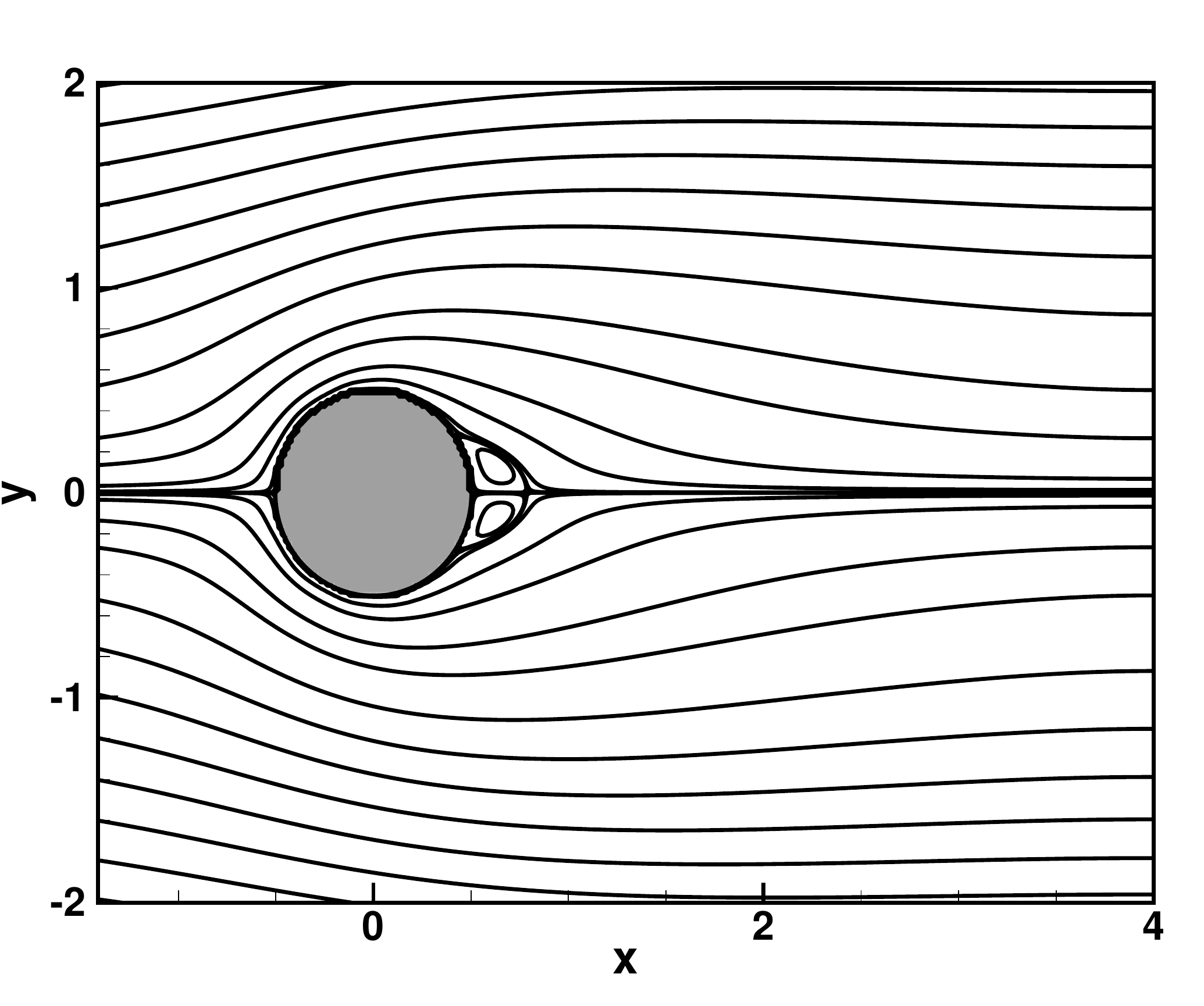} 
		\caption{$Re=10$}
	\end{subfigure} 
\begin{subfigure}{0.4\textwidth}
	\includegraphics[width=\linewidth]{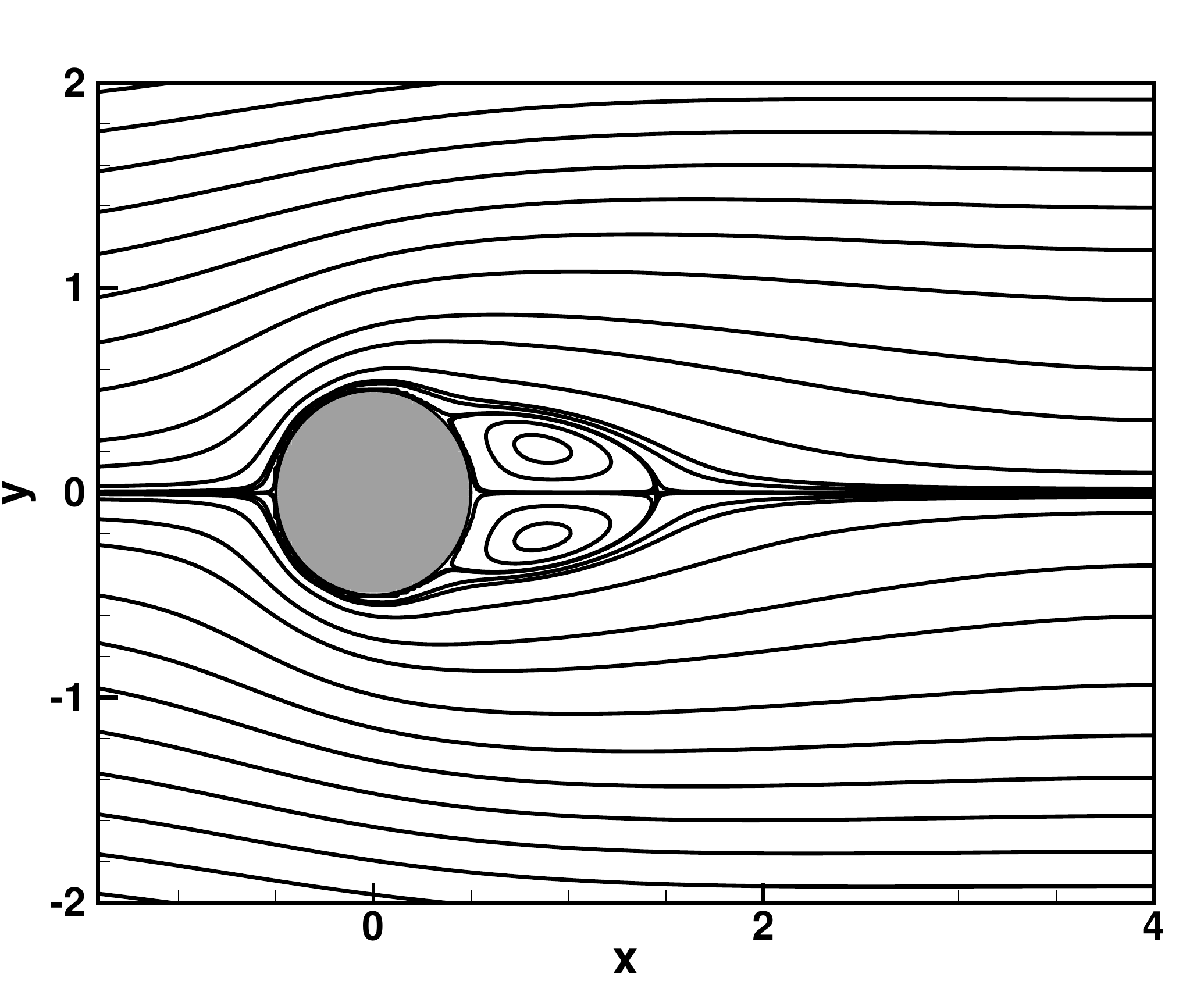} 
	\caption{$Re=20$}
\end{subfigure} 
\begin{subfigure}{0.4\textwidth}
	\includegraphics[width=\linewidth]{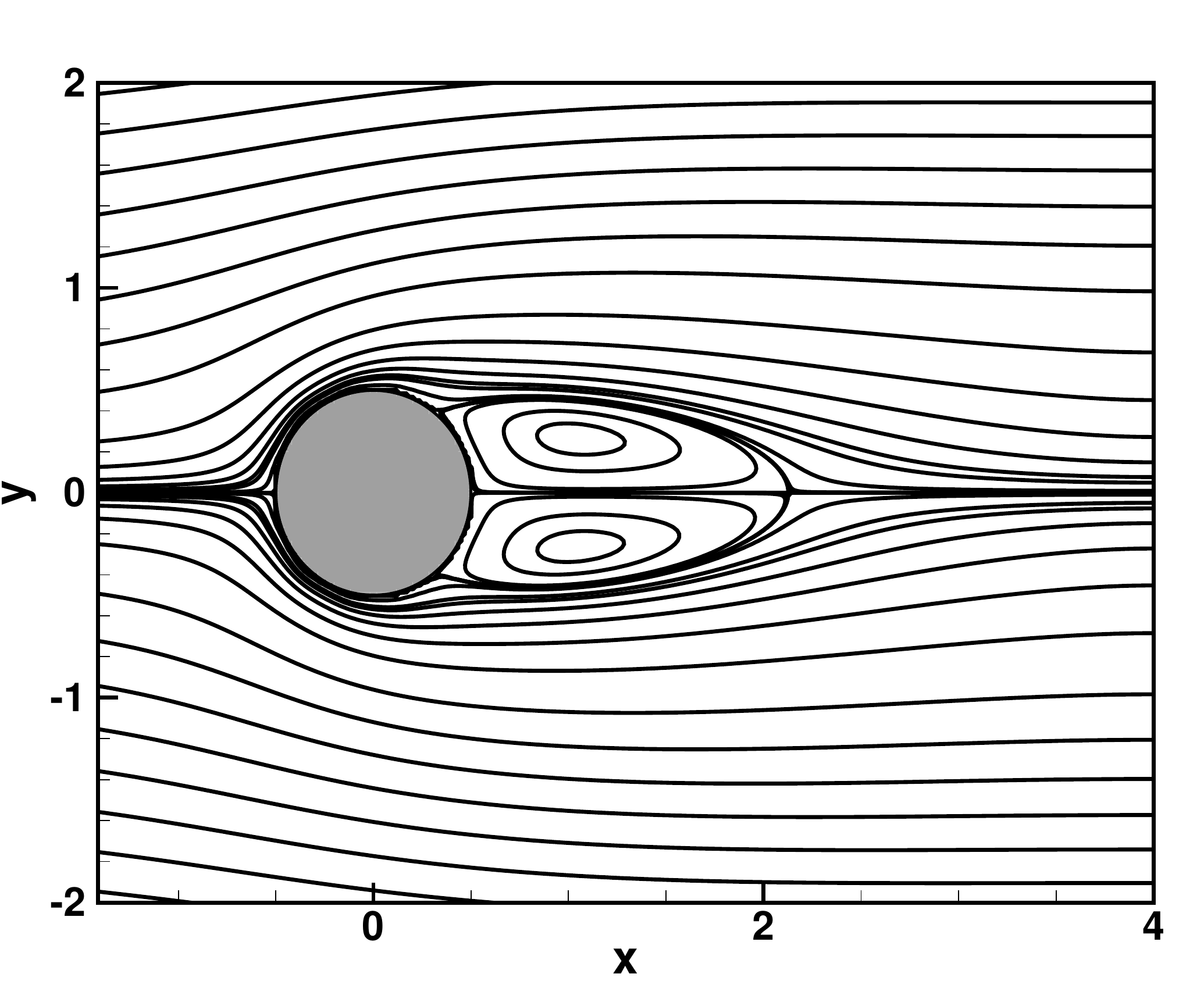} 
	\caption{$Re=30$}
\end{subfigure} 
\begin{subfigure}{0.4\textwidth}
	\includegraphics[width=\linewidth]{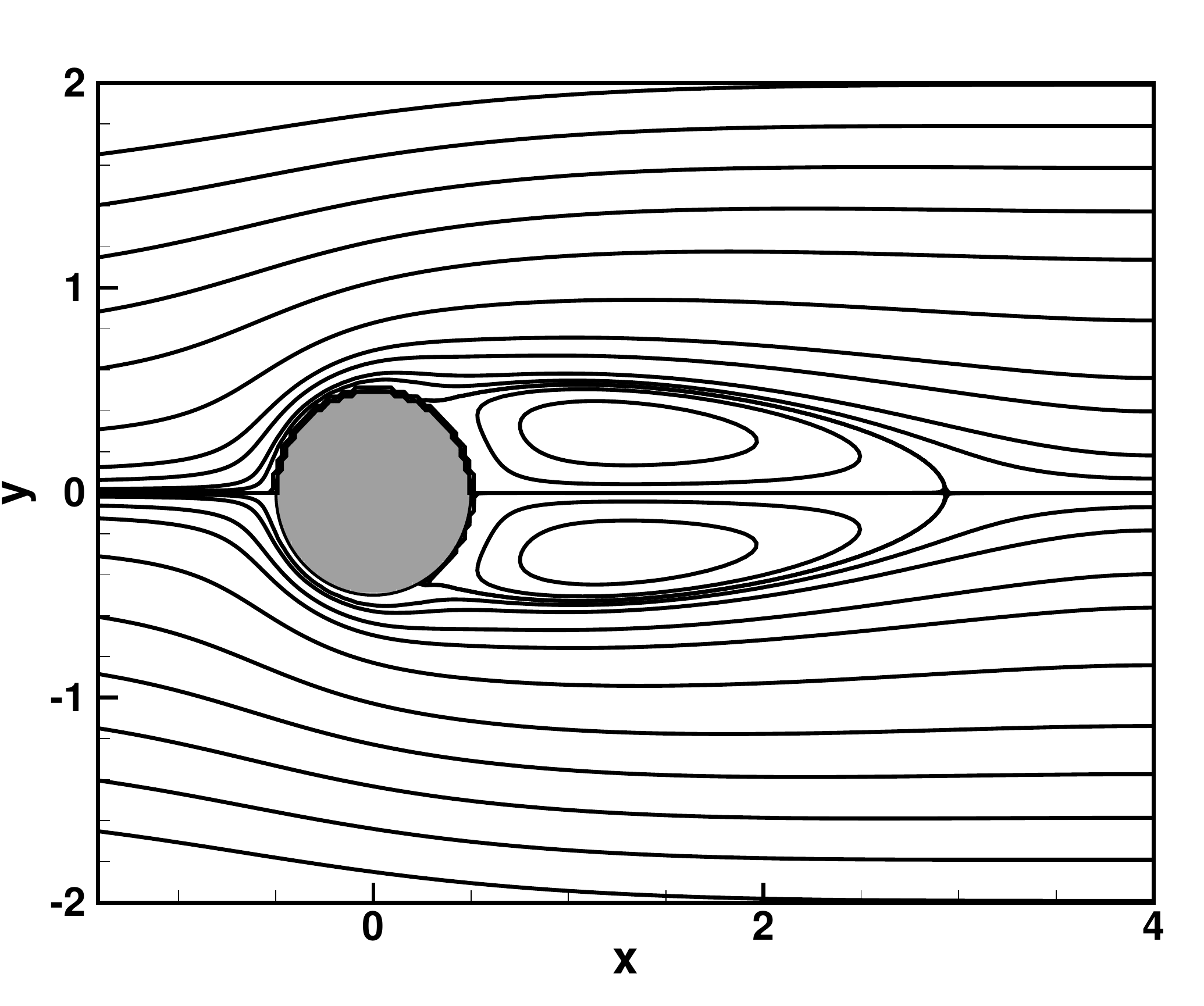} 
	\caption{$Re=40$}
\end{subfigure}
\caption{\small{Steady state streamlines for (a)$Re=10$, (b)$Re=20$, (c)$Re=30$, and (d)$Re=40$.}}
\label{Fig:psi-circ}
\end{figure}

\begin{figure}[H]
	\centering
\begin{subfigure}{0.4\textwidth}
	\includegraphics[width=\linewidth]{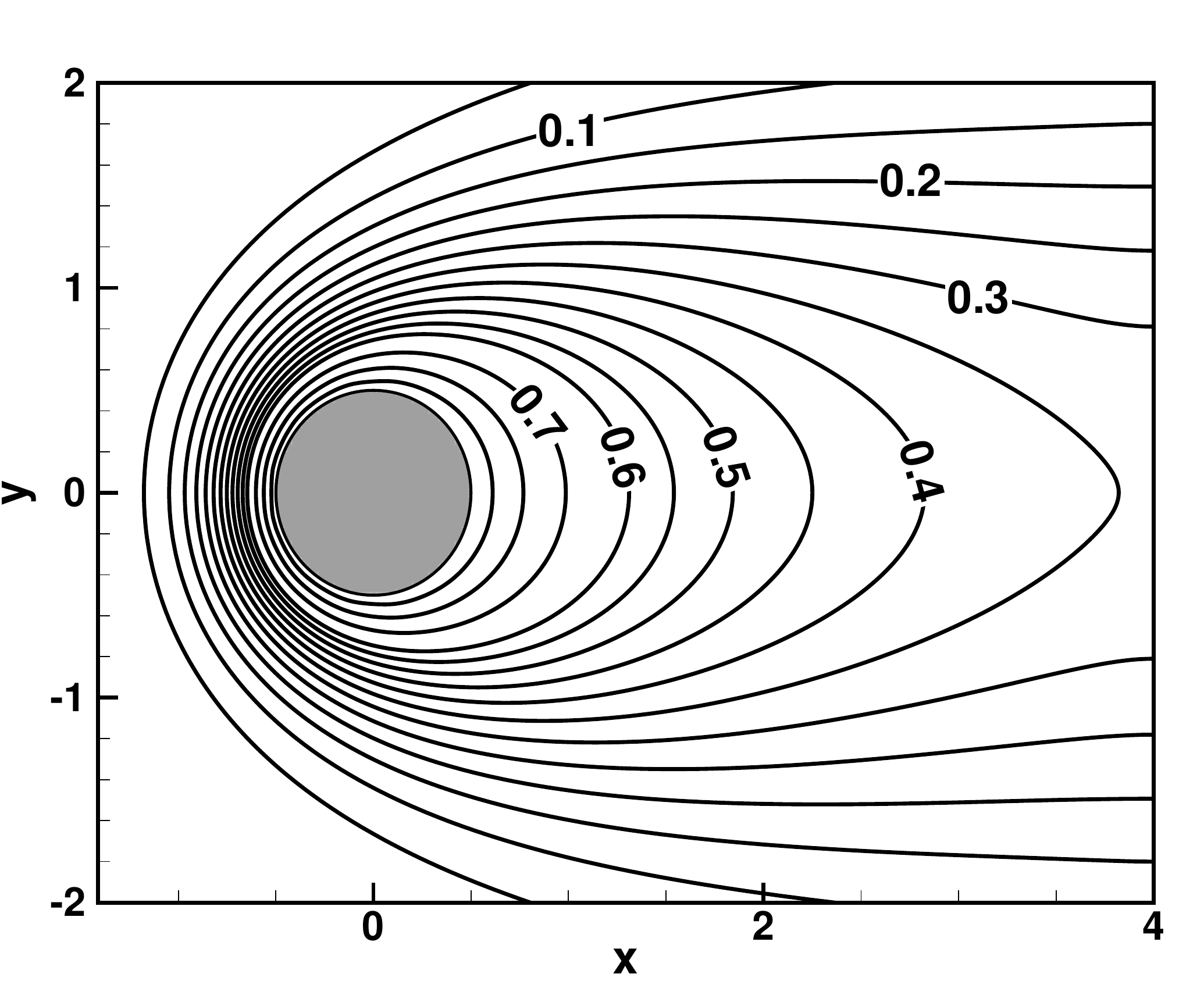} 
	\caption{$Re=10$}
\end{subfigure} 
\begin{subfigure}{0.4\textwidth}
	\includegraphics[width=\linewidth]{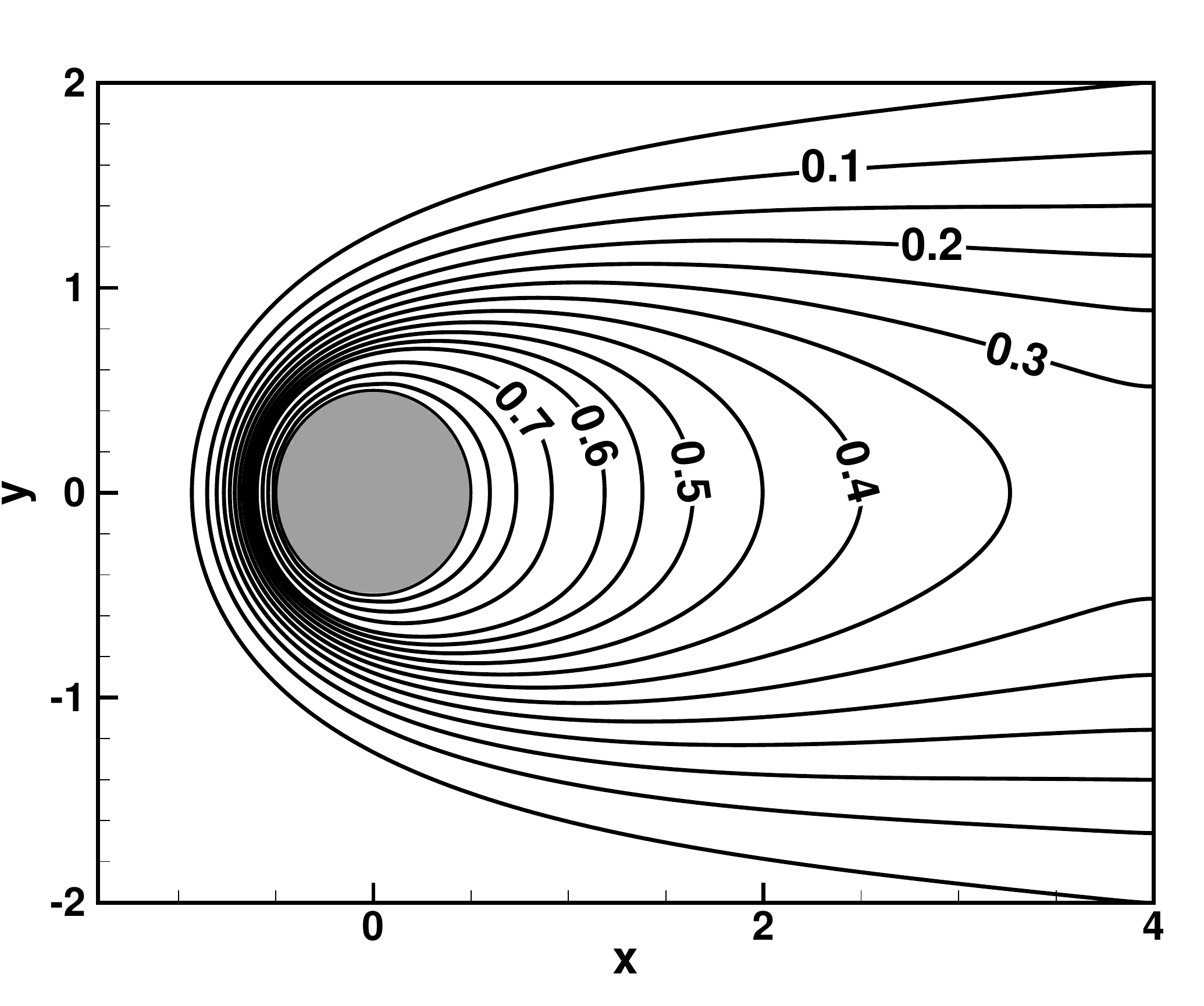} 
	\caption{$Re=20$}
\end{subfigure} 
\begin{subfigure}{0.4\textwidth}
	\includegraphics[width=\linewidth]{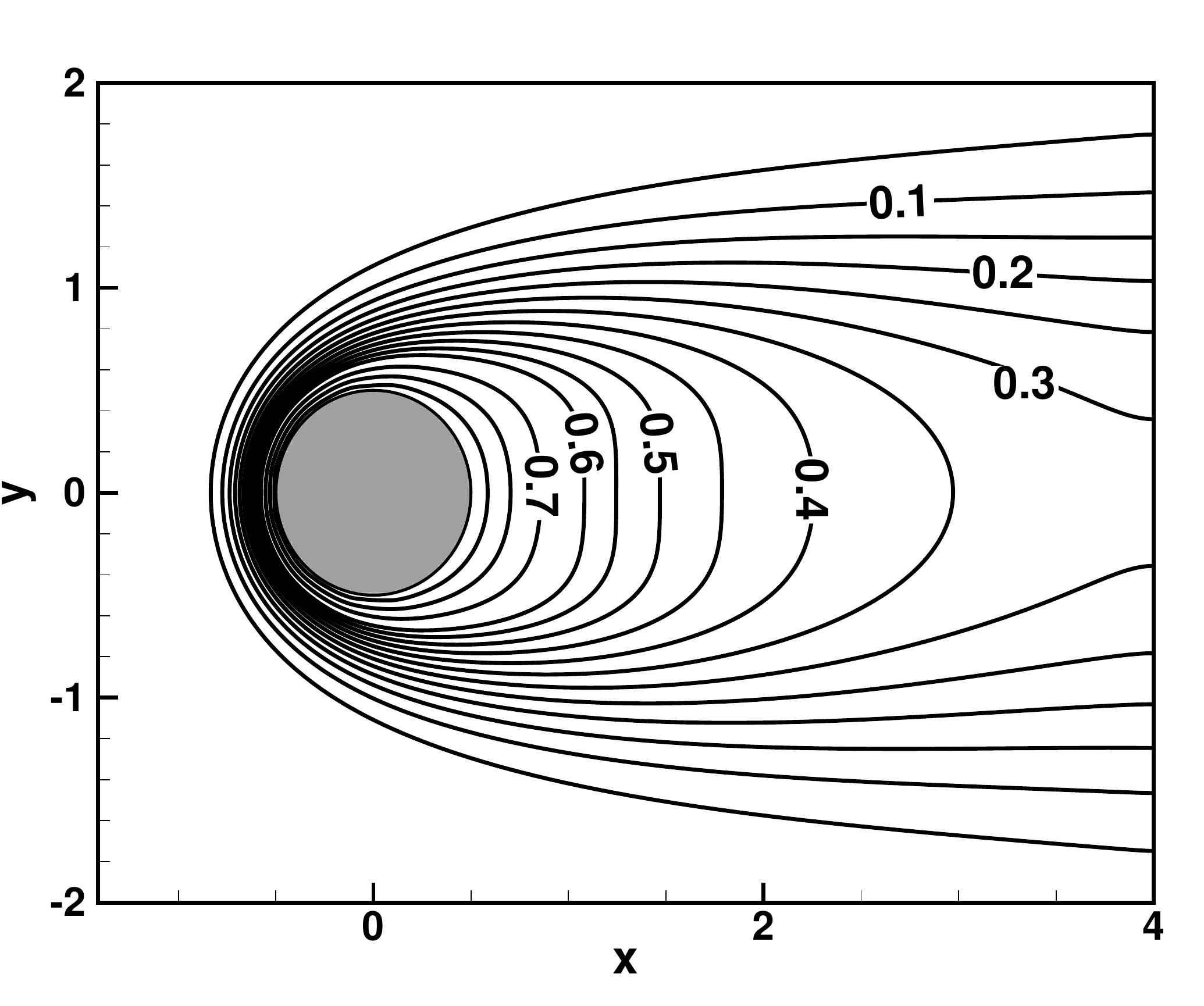} 
	\caption{$Re=30$}
\end{subfigure} 
\begin{subfigure}{0.4\textwidth}
	\includegraphics[width=\linewidth]{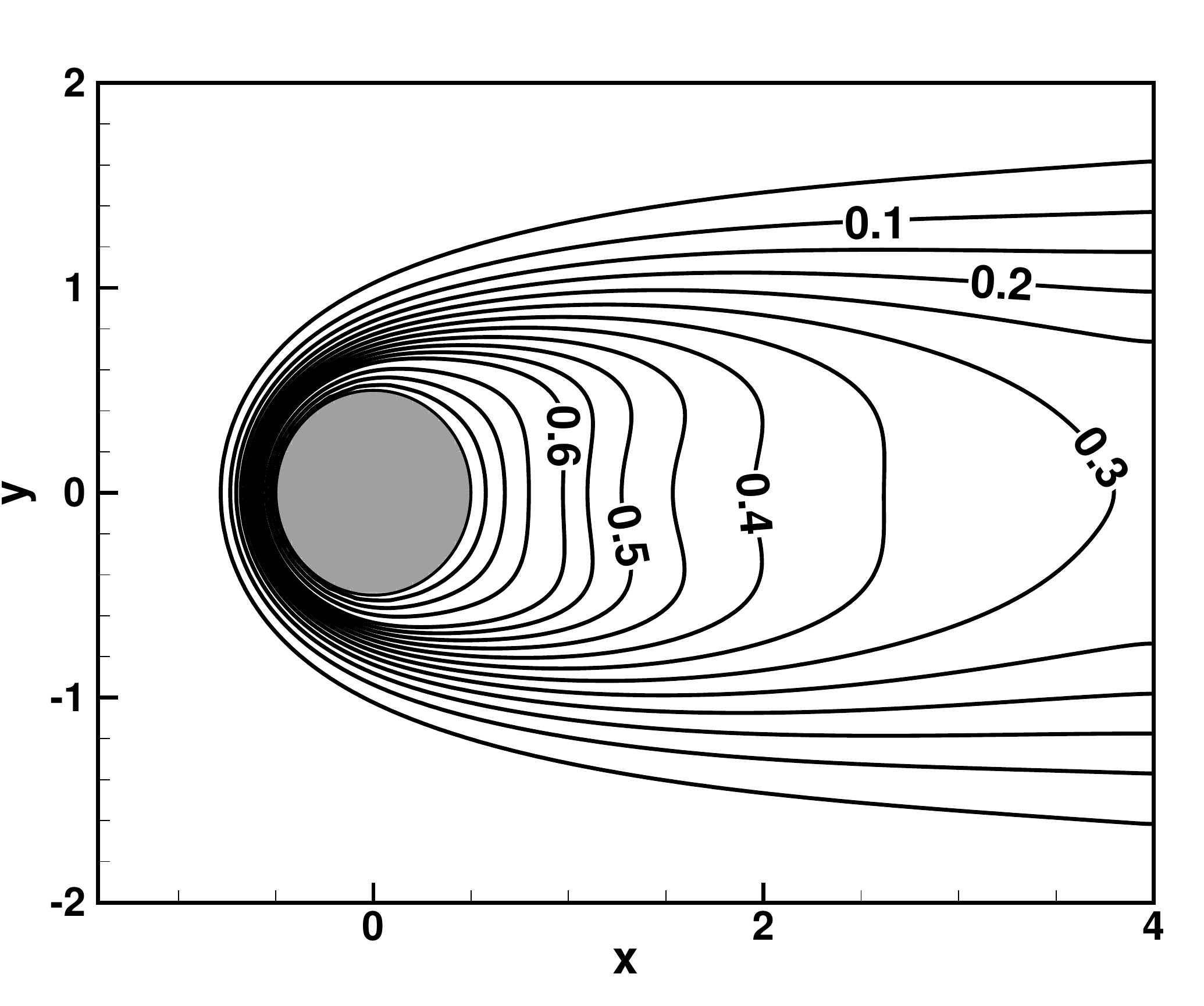} 
	\caption{$Re=40$}
\end{subfigure}
	\caption{\small{Steady state isotherms for (a)$Re=10$, (b)$Re=20$, (c)$Re=30$, and (d)$Re=40$.}.}
	\label{Fig:T-circ}
\end{figure}

For this comparison exercise, simulations are carried out for $Re=15$, $20$, $30$, $35$, and $40$. Previous works (Takami and Keller \cite{takami1969}, Dennis and Chang \cite{dennis1970}, Apelt and Ledwich \cite{apelt1979}, Jafroudi and Yang \cite{jafroudi1986}) have shown that the flow is steady for these values of $Re$'s. In the present case steady-state has been reached through time marching. The values of eddy length ($L_s$), and surface averaged Nusselt number ($Nu_{\text{av}}$) from the present computation have been compared with well established results in tables \ref{tab:circ-comparison-ls}, and \ref{tab:circ-comparison-nu} respectively. One can see that in all the cases, excellent match has been obtained. Figures \ref{Fig:psi-circ} and \ref{Fig:T-circ} (a)-(d) show the streamlines and isotherms for $Re=10 - 40$. One can clearly see from the figures \ref{Fig:psi-circ} (a)-(d)  that the eddy length increases linearly with $Re$. The isotherms are symmetrical about the $x$-axis in the wake region. Figures \ref{Fig:T-circ} (a)-(d) also reveal that the isotherms become steeper with $Re$ in the near wake region. This implies that with an increase in fluid velocity sets a higher temperature gradient resulting in enhanced heat transfer from the cylinder surface. This is evident from the values of $Nu_{\text{av}}$ in table \ref{tab:circ-comparison-nu} as well. The streamlines and isotherms resulting from our computation are very similar to the simulations of \cite{biswas2009}.

\subsection{Grid independence}
In order to establish grid independence of the computed data, we compare the steady state streamlines and isotherms at three different grid sizes for $Re = 40$, and $\theta = 0^{\degree}$.  The three different grid sizes used for this exercise are $319 \times 161$, $463 \times 265$, and $621 \times 353$. As seen from figure \ref{Fig:grid-ind}, the overlapping of contours for streamlines and isotherms at grid sizes $463 \times 265$, and $621 \times 353$ clearly indicate grid independence of the computed data. Thus, all our computations in this work have been carried out on a grid of size $463 \times 265$.
\begin{figure}[H]
	\centering
	\begin{subfigure}{0.5\textwidth}
		\includegraphics[width=\linewidth]{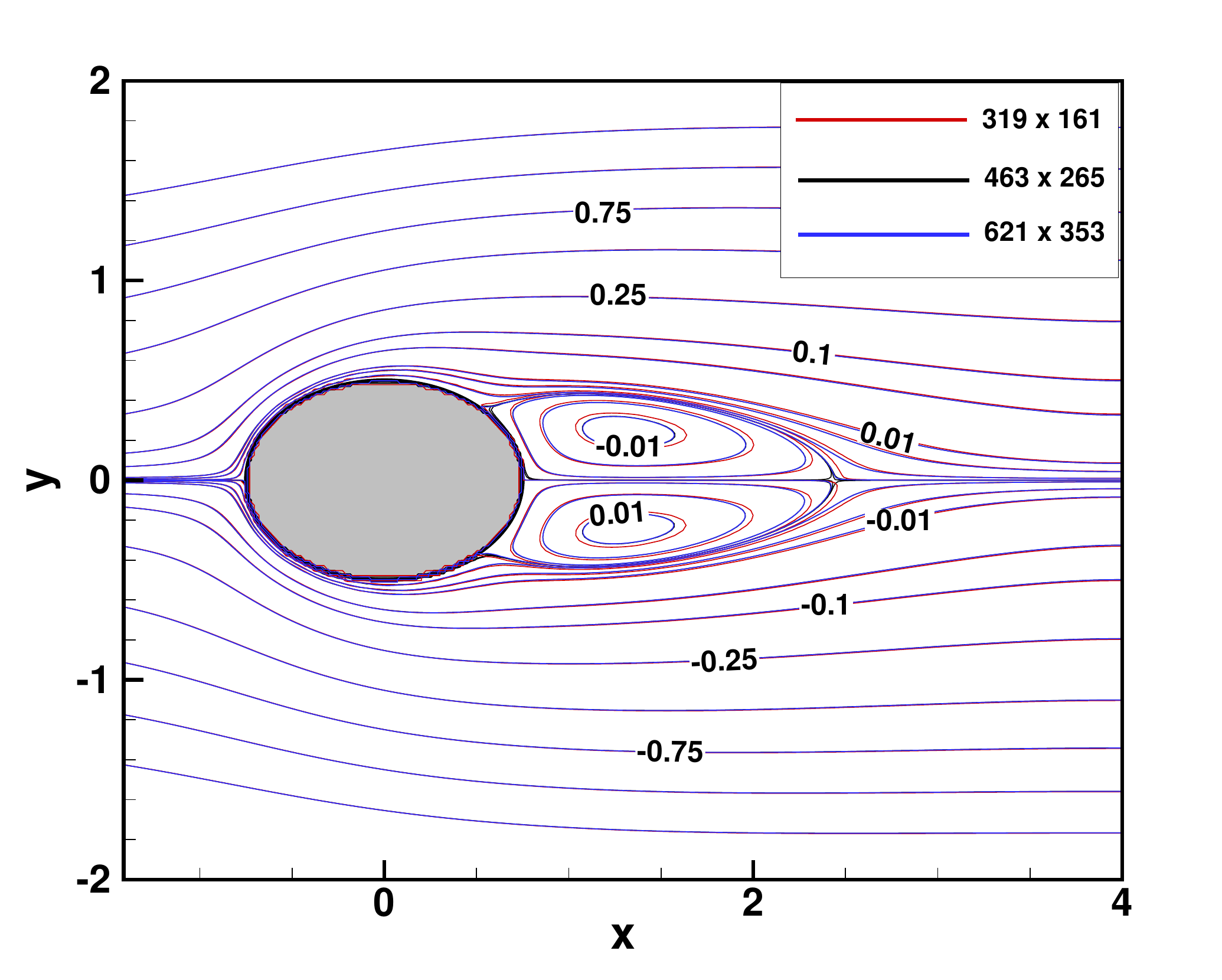} 
		\caption{Streamlines}
	\end{subfigure}\hfil 
	\begin{subfigure}{0.5\textwidth}
		\includegraphics[width=\linewidth]{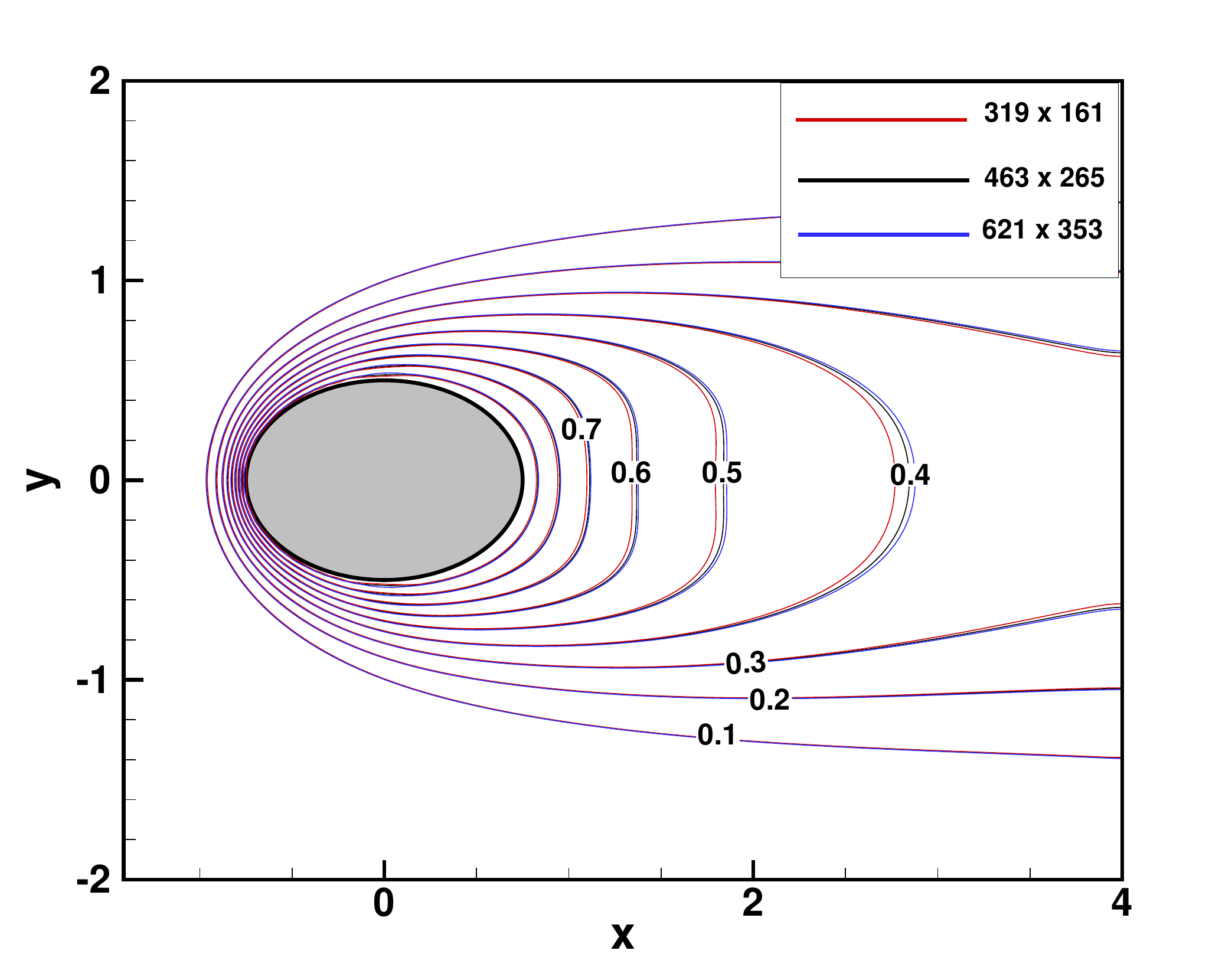} 
		\caption{Isotherms}
	\end{subfigure}\hfil 
	\caption{\small{Steady state (a) streamlines and (b) isotherms for $\theta=0^{\degree}$ and $Re=40$, at three different grids viz. $319 \times 161$, $463 \times 265$, and $621 \times 353$.}}
	\label{Fig:grid-ind}
\end{figure}

\section{Results and discussion}\label{sec:results}
In this section, we document our results from extensive computations that we have carried out and simultaneously, the analysis of the same. For the problem considered in this work, there are two parameters viz. angle of attack ($\theta$) and Reynolds number ($Re$), which are crucial for the study.  Their values have been varied and the subsequent effects on the flow and heat transfer characteristics have been investigated thoroughly. While the angle of attack is varied in increments of $15^{\degree}$ in the range $0^{\degree} \leq \theta < 180^{\degree}$, the Reynolds number is varied in increments of $10$ in the range $10 \leq Re \leq 130$.  Firstly, we present the steady state results, and then the transient ones. 
\subsection{Steady state}\label{sec:steady}
Computations are carried out for $0^{\degree} \leq \theta < 180^0$, and $10 \leq Re \leq Re_{c}$, where $Re_{c}$ denotes the critical $Re$ at which the flow transitions from steady to unsteady state. As the angle of attack changes, the value of $Re_{c}$ also changes.

Figures \ref{Fig:psi-steady-zerodeg} and \ref{Fig:isotherms-steady-zerodeg} show the streamlines and isotherms respectively for $\theta = 0^{\degree}$. $Re_{c}$ for $\theta = 0^{\degree}$ is in the range $59 \leq Re < 60$. For all the $Re$'s considered in this range, the steady recirculation bubble, consisting of two counter-rotating vortices that elongate as $Re$ is increased, remains symmetric about the $x$-axis (figures \ref{Fig:psi-steady-zerodeg} (a)-(f)). The upper vortex rotates in clockwise direction, whereas the lower vortex rotates in counter-clockwise direction. The isotherms are more evenly spread out at $Re = 10$ (figure \ref{Fig:isotherms-steady-zerodeg} (a)) denoting negligible convective heat transfer. As $Re$ increases, the isotherms become more clustered both upstream and downstream of the cylinder, and one can observe the formation of thermal boundary layer on the surface of the cylinder which becomes thinner with increasing $Re$ (figures \ref{Fig:isotherms-steady-zerodeg} (b) - (f)). The thinning of the thermal boundary layer is most prominent near the leading edge of the cylinder. Finally, a slight distortion in the isotherms can be seen when $Re=40$ (figure \ref{Fig:isotherms-steady-zerodeg} (d)), which increases as $Re$ increases (figures \ref{Fig:isotherms-steady-zerodeg} (e) - (f)). Note that the isotherms also appear symmetric about the $x$-axis since the flow is symmetric about the line $y=0$ for $\theta = 0^{\degree}$.
\begin{figure}[H]
	\centering
	\begin{subfigure}{0.3\textwidth}
		\includegraphics[width=\linewidth]{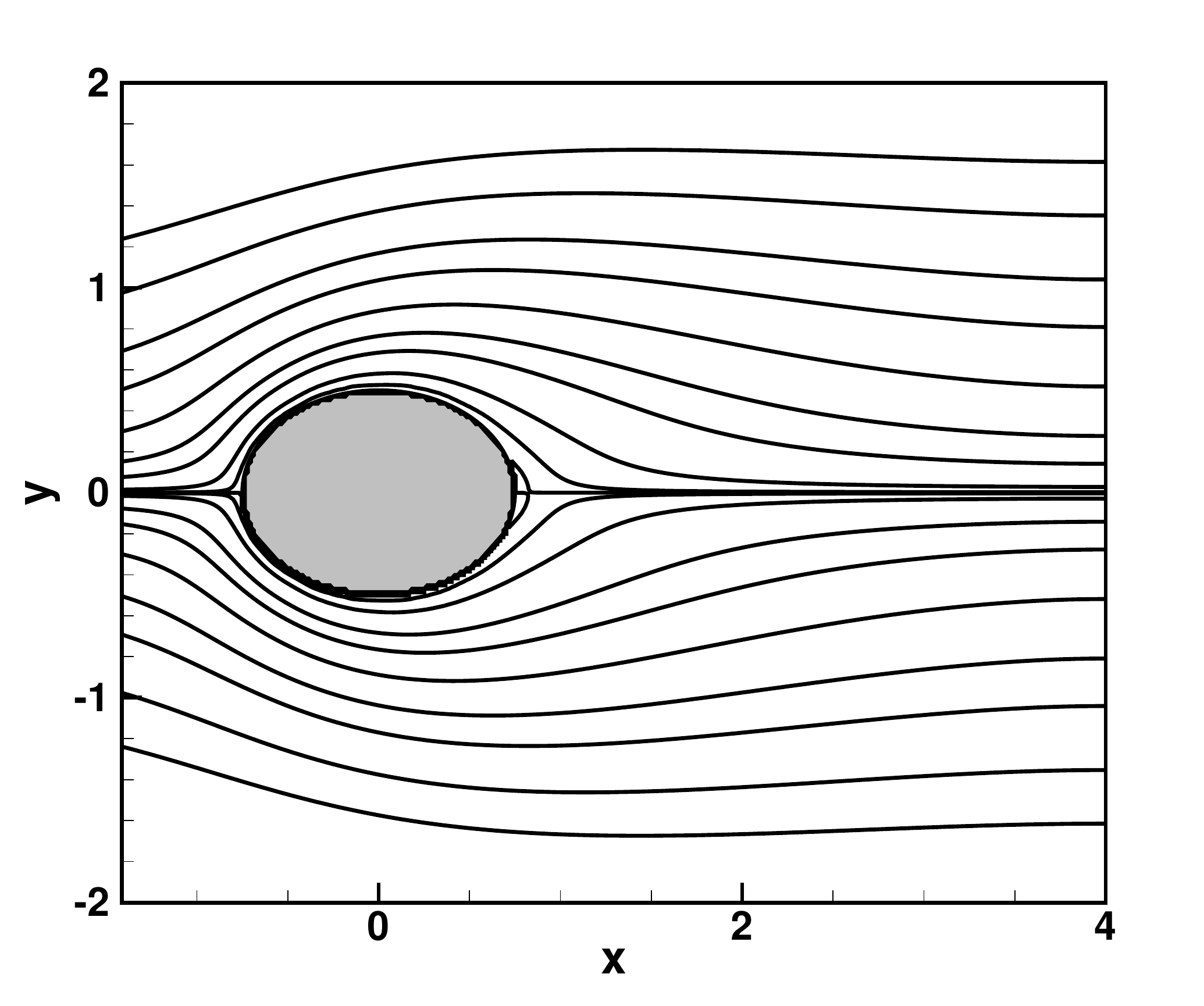} 
		\caption{$Re=10$}
	\end{subfigure}\hfil 
	\begin{subfigure}{0.3\textwidth}
		\includegraphics[width=\linewidth]{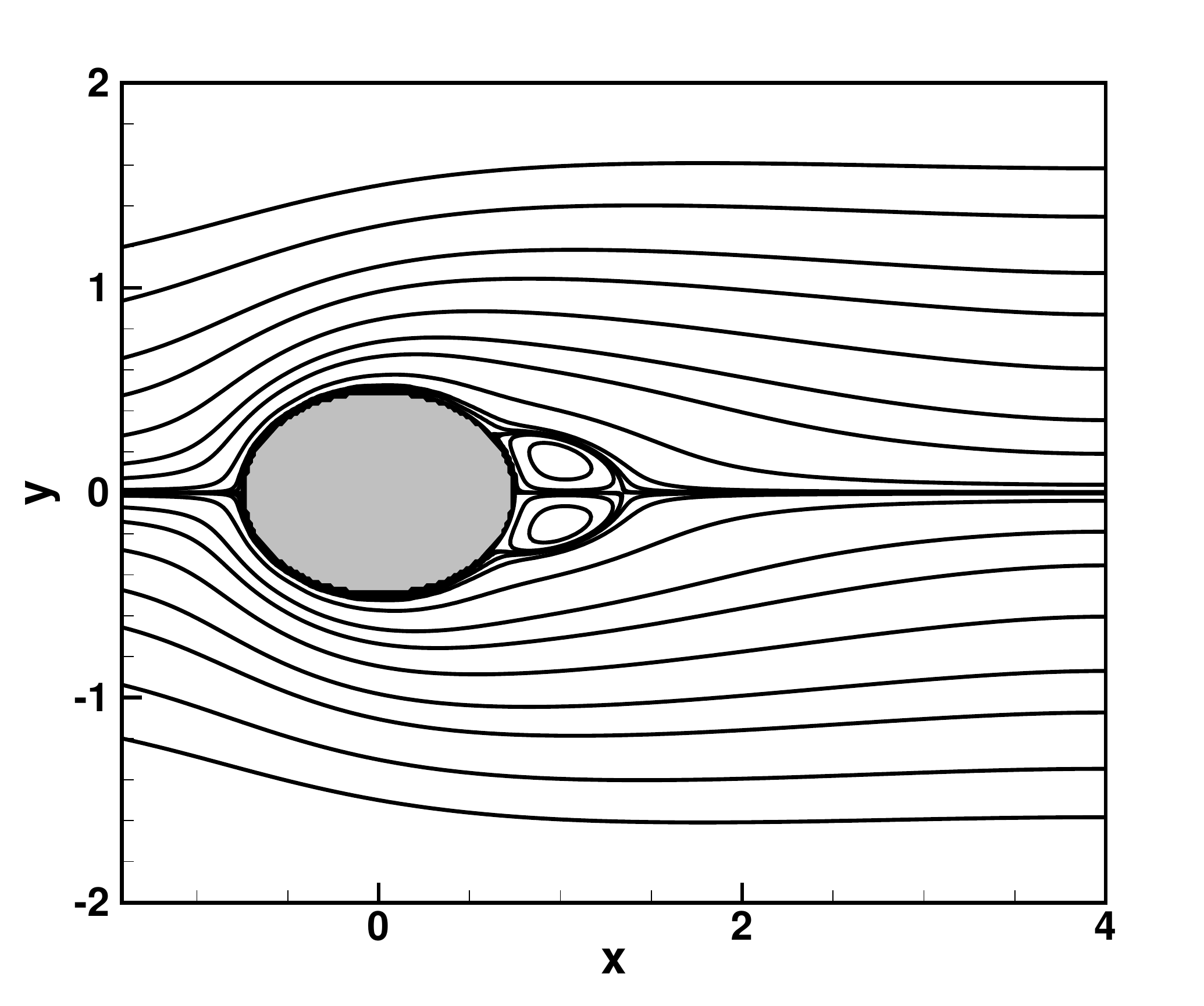} 
		\caption{$Re=20$}
	\end{subfigure}\hfil 
	\begin{subfigure}{0.3\textwidth}
		\includegraphics[width=\linewidth]{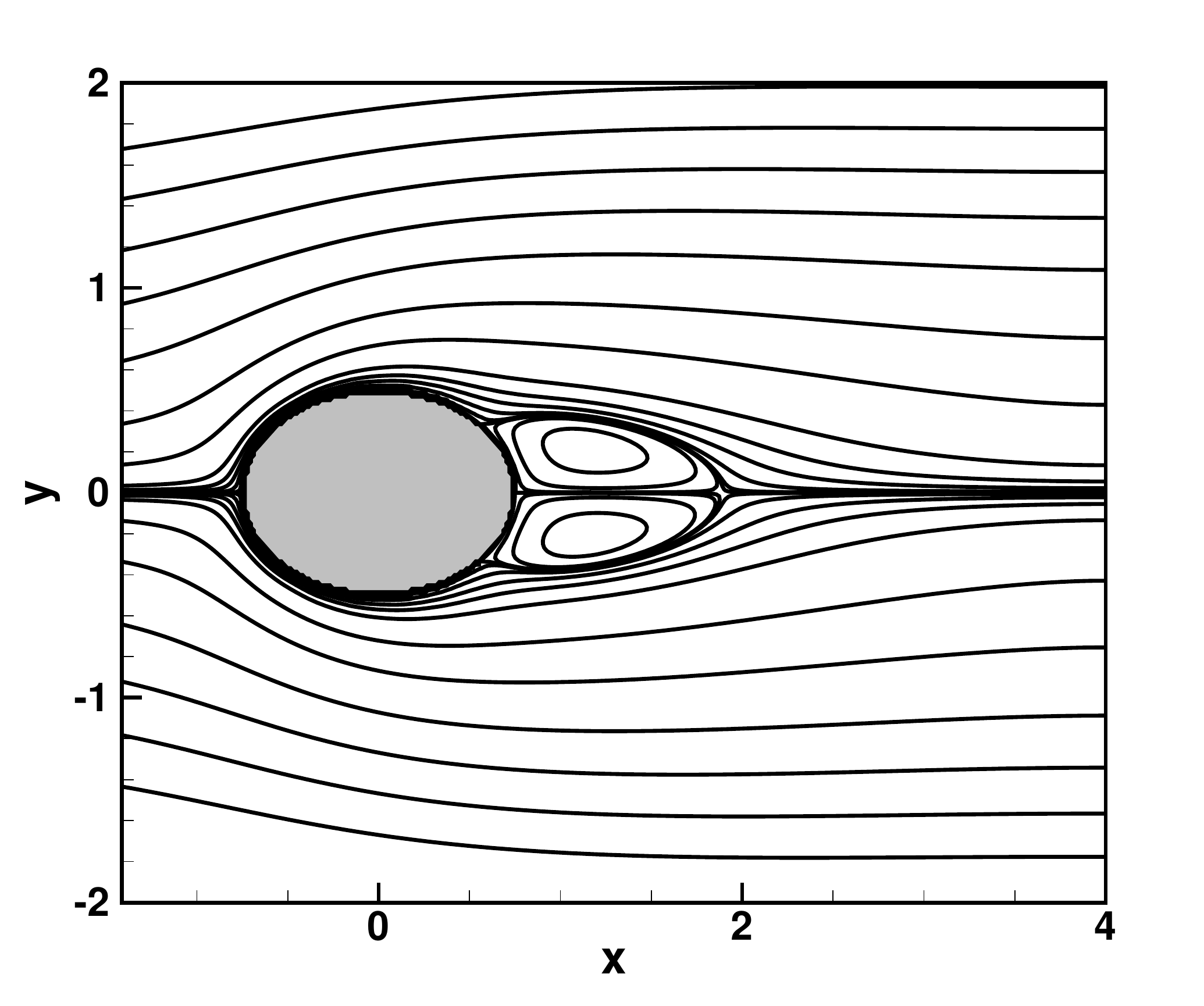} 
		\caption{$Re=30$}
	\end{subfigure}\hfil 
	\begin{subfigure}{0.3\textwidth}
		\includegraphics[width=\linewidth]{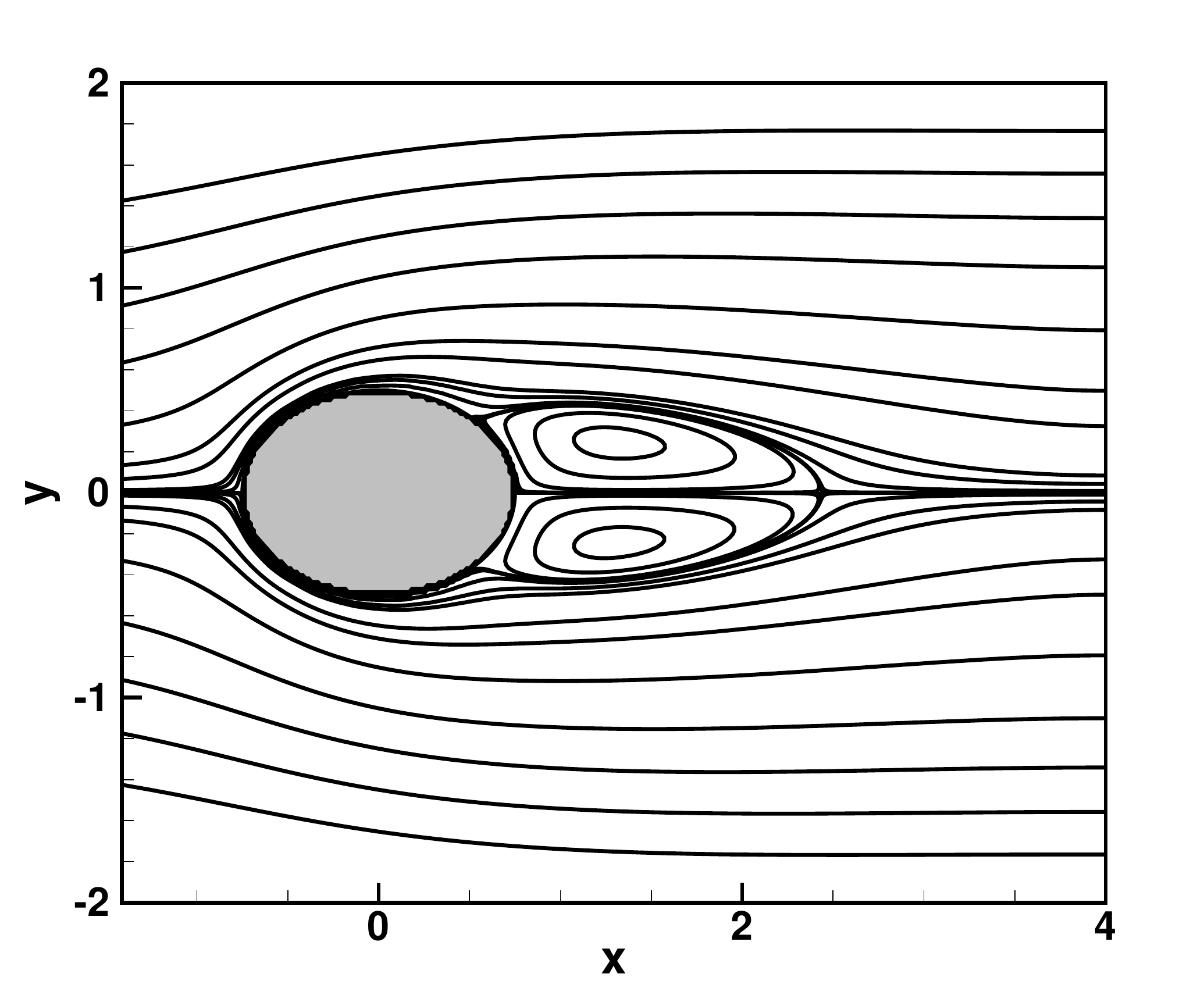} 
		\caption{$Re=40$}
	\end{subfigure}\hfil 
	\begin{subfigure}{0.3\textwidth}
		\includegraphics[width=\linewidth]{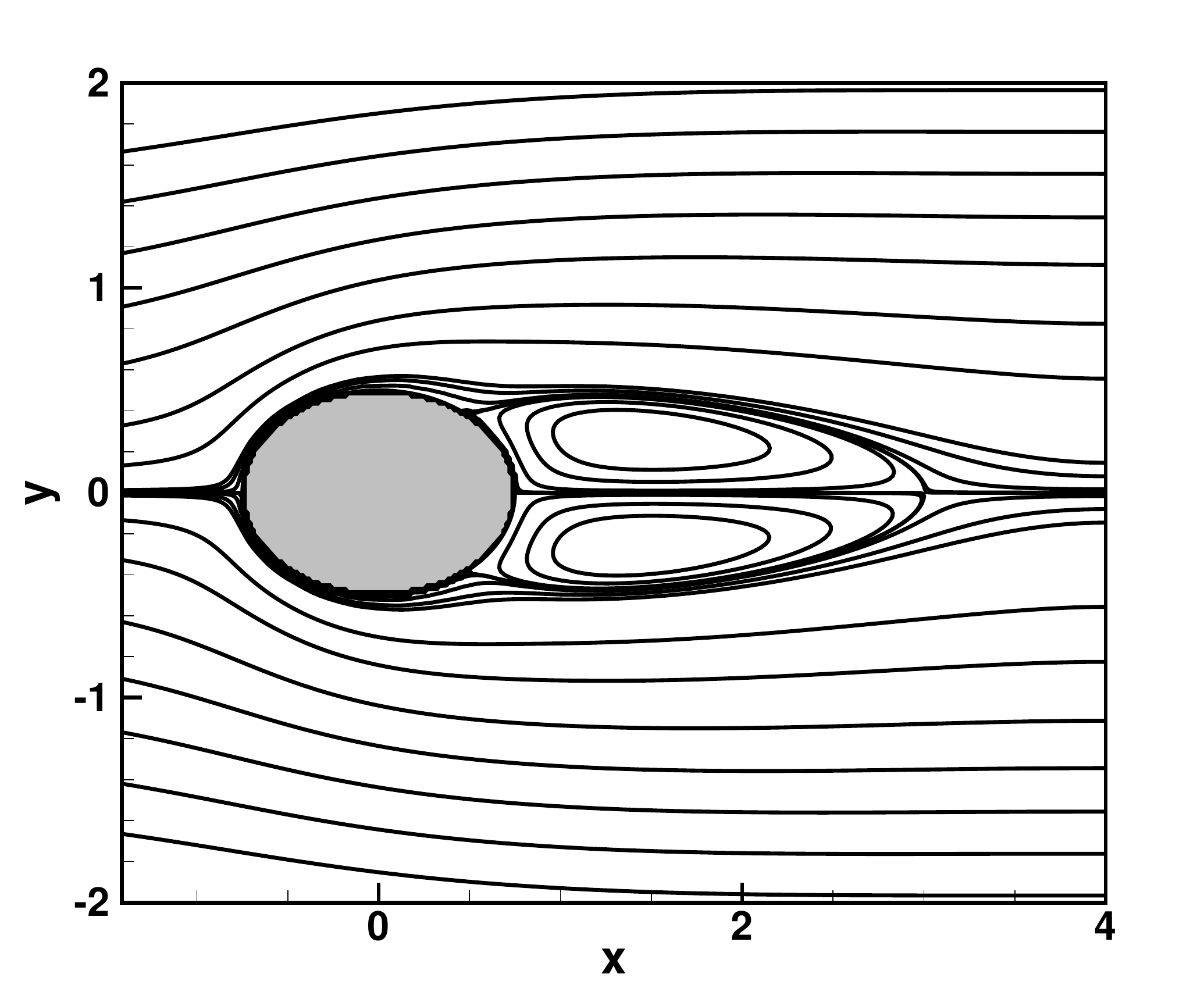} 
		\caption{$Re=50$}		
	\end{subfigure}\hfil 
\begin{subfigure}{0.3\textwidth}
	\includegraphics[width=\linewidth]{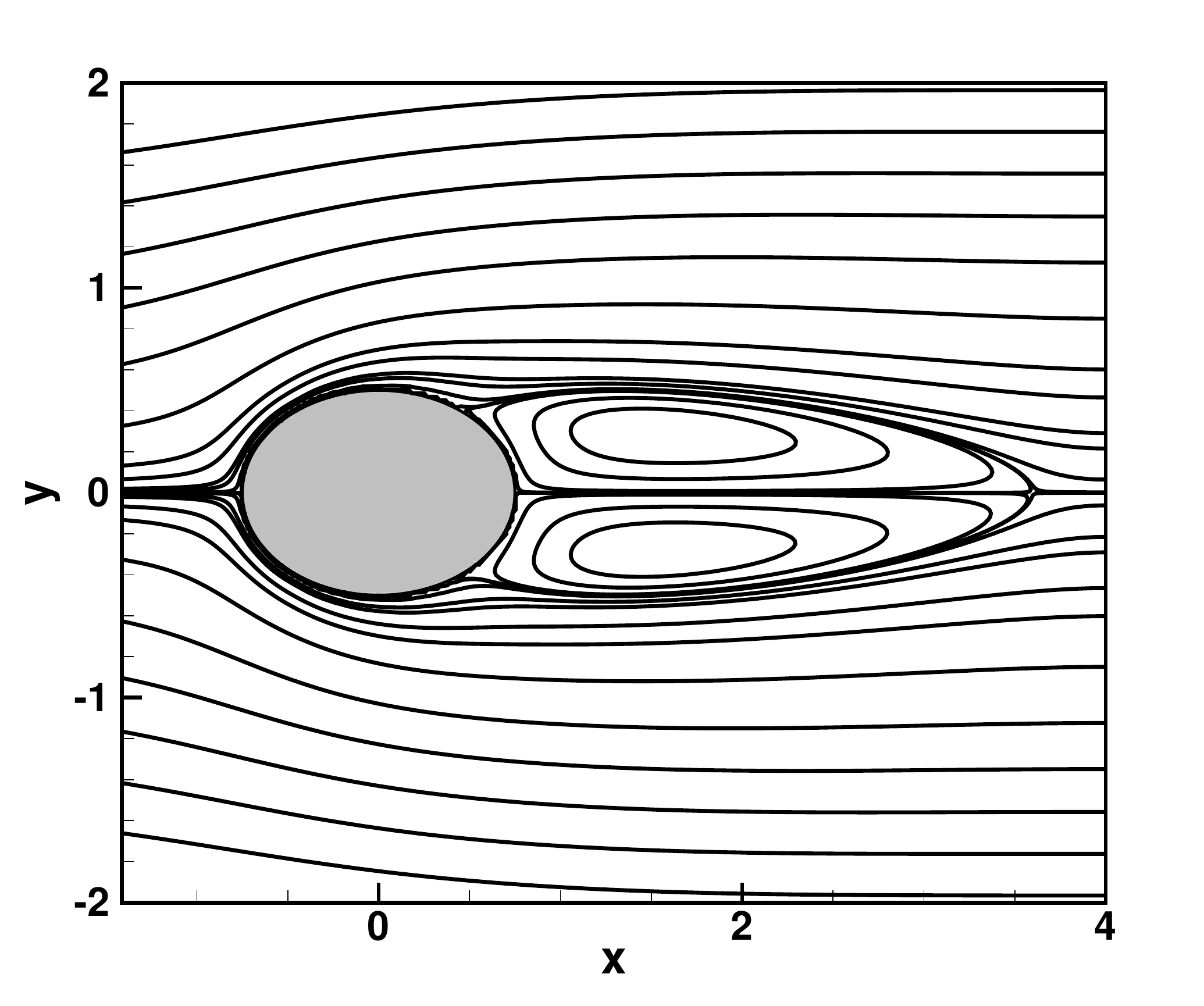} 
	\caption{$Re=59$}		
\end{subfigure}\hfil 
\caption{\small{Steady state streamlines for $\theta=0^{\degree}$ and (a)$Re=10$, (b)$Re=20$, (c)$Re=30$, (d)$Re=40$, and (e)$Re=50$, and (f) $Re=59$.}}
\label{Fig:psi-steady-zerodeg}
\end{figure}

\begin{figure}[H]
	\centering
	\begin{subfigure}{0.3\textwidth}
		\includegraphics[width=\linewidth]{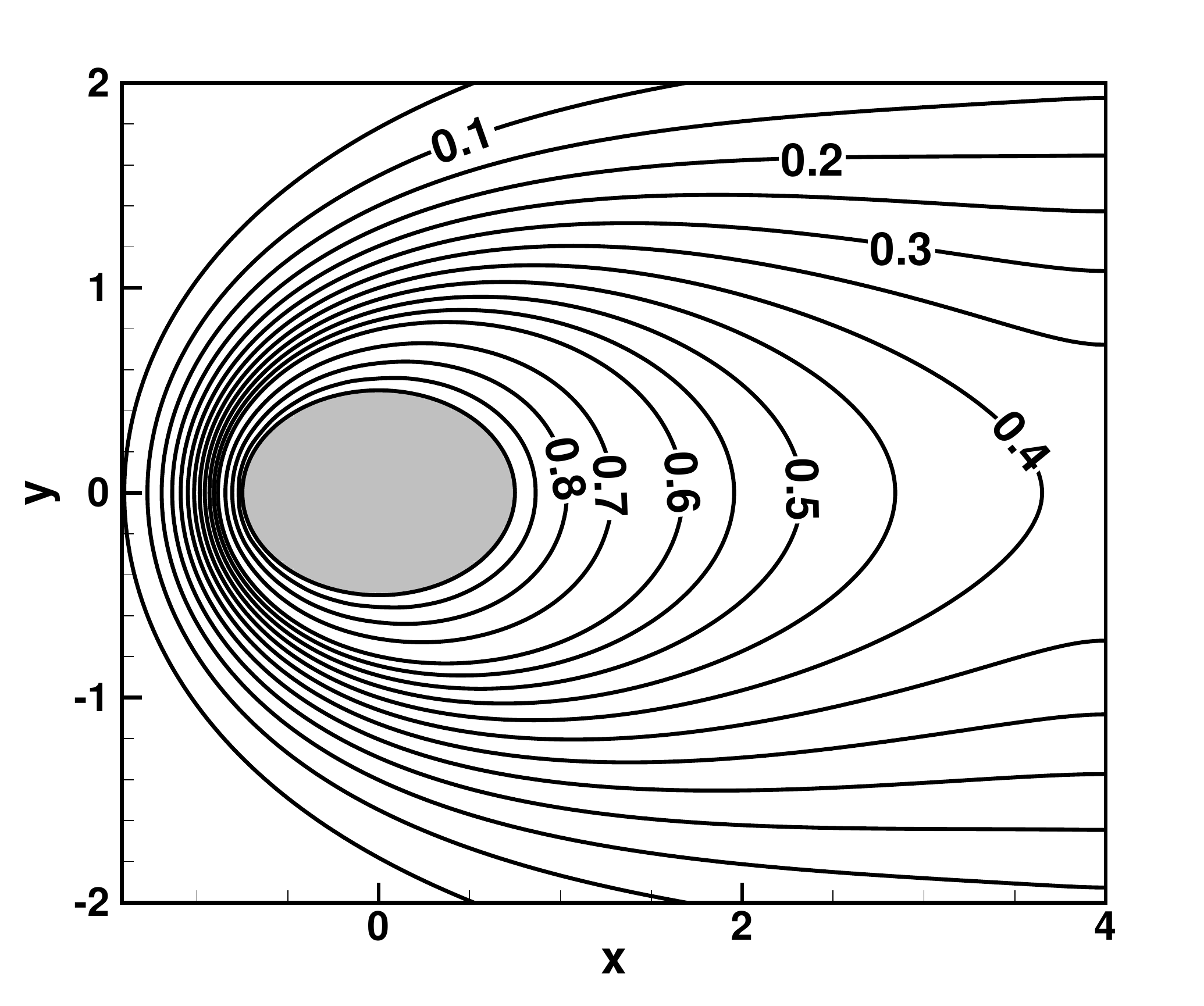} 
		\caption{}
	\end{subfigure}\hfil
	\begin{subfigure}{0.3\textwidth}
		\includegraphics[width=\linewidth]{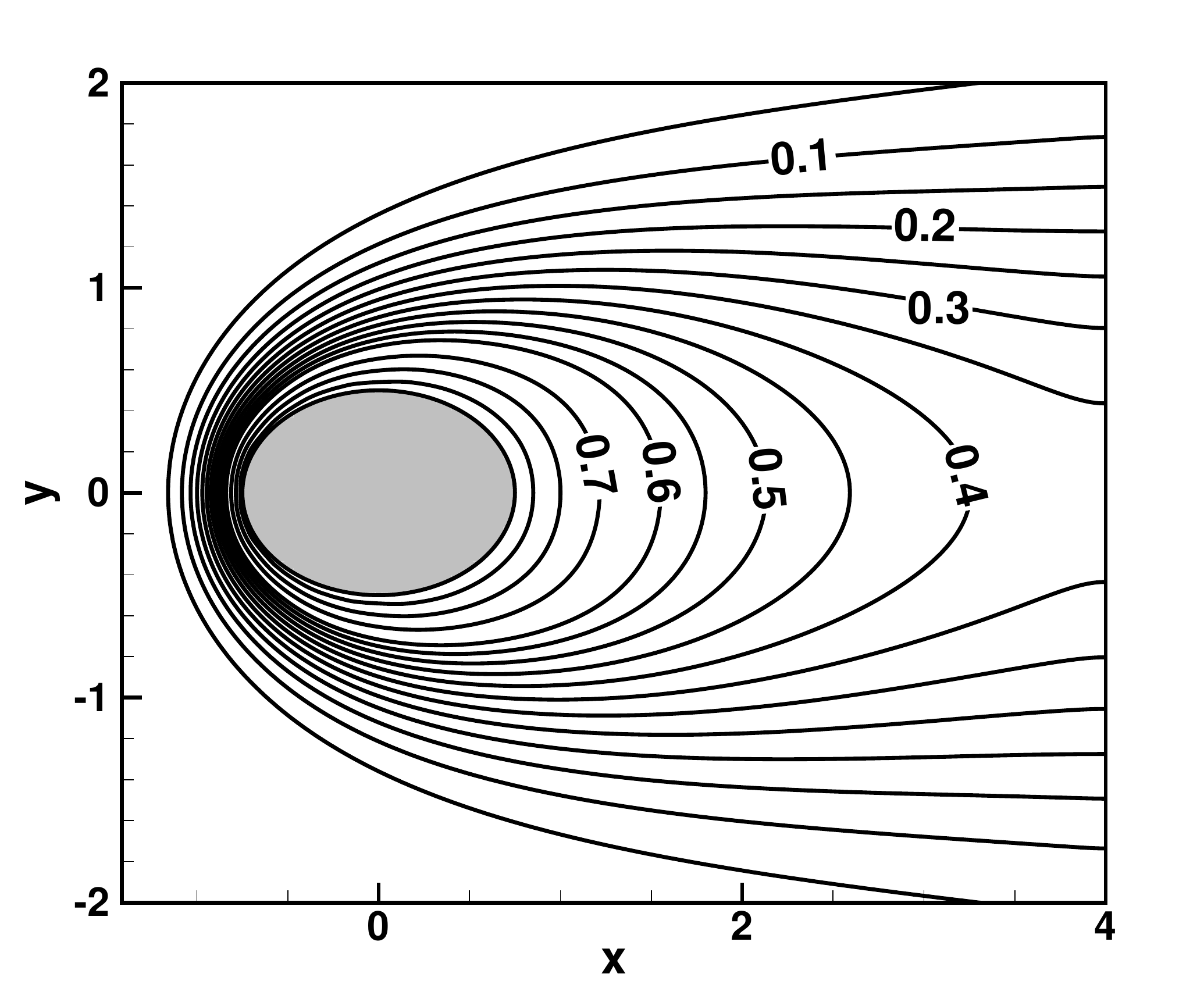} 
		\caption{}
	\end{subfigure}\hfil
	\begin{subfigure}{0.3\textwidth}
		\includegraphics[width=\linewidth]{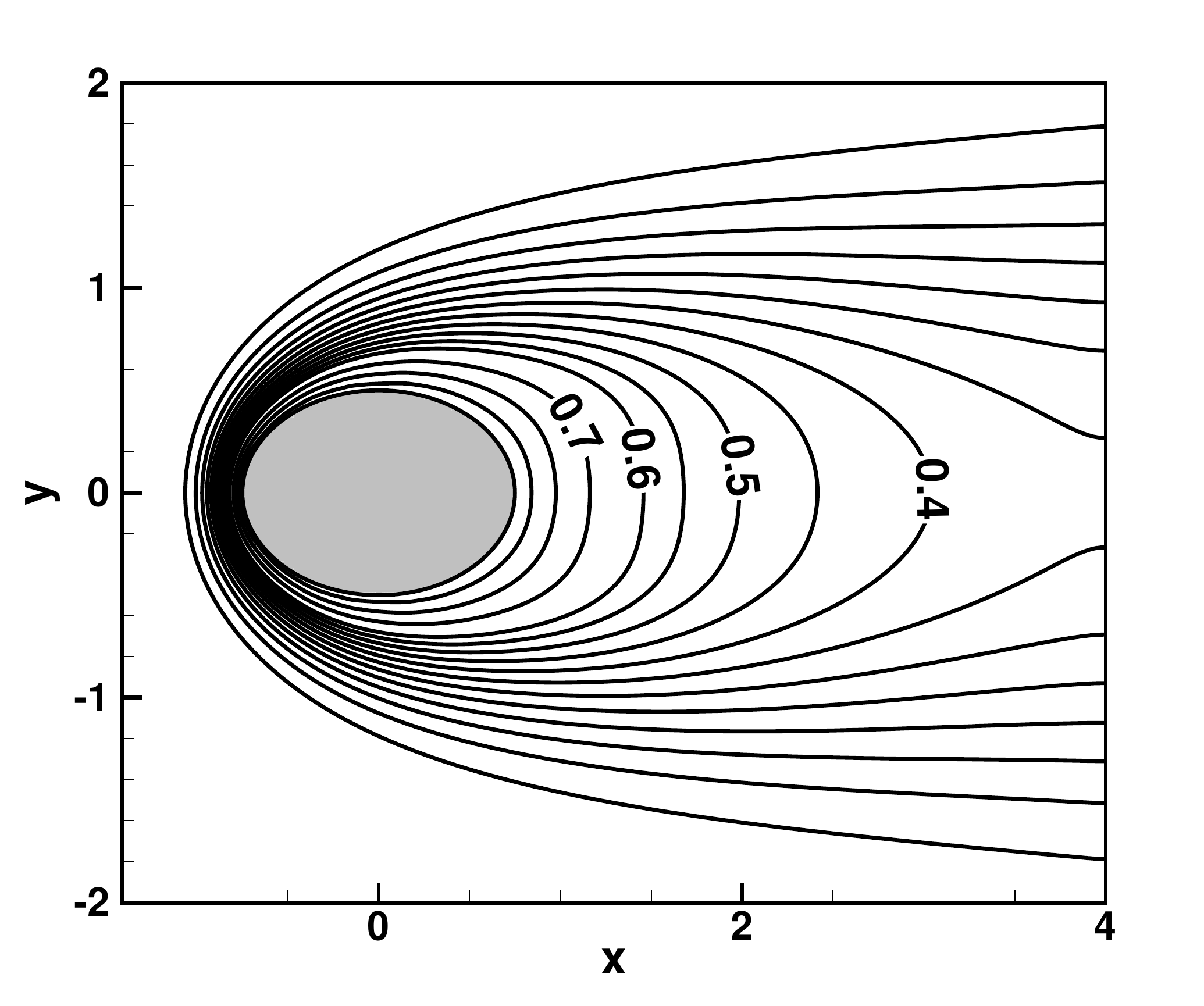} 
		\caption{}
	\end{subfigure}\hfil
	\begin{subfigure}{0.3\textwidth}
		\includegraphics[width=\linewidth]{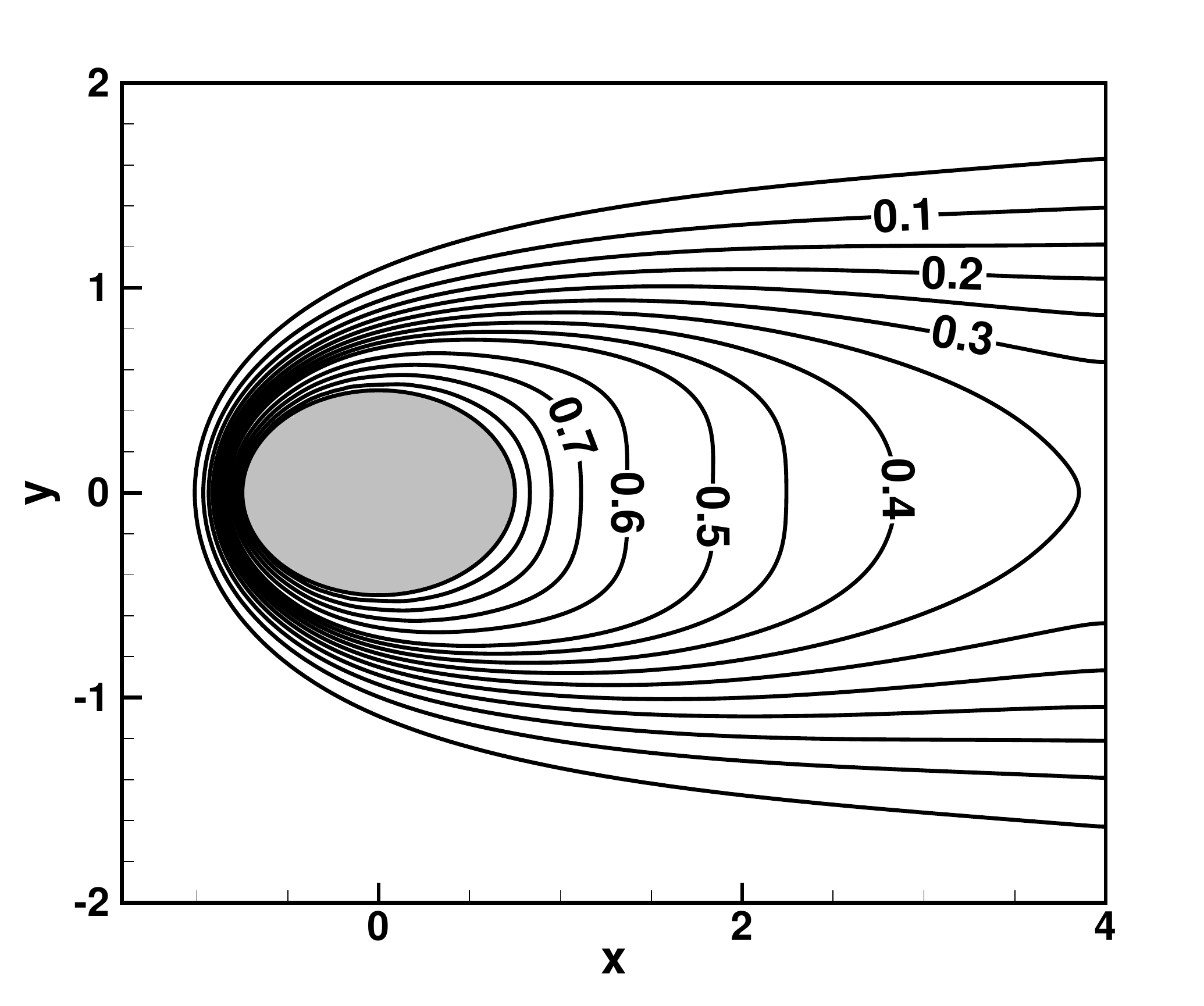} 
		\caption{}
	\end{subfigure}\hfil
	\begin{subfigure}{0.3\textwidth}
		\includegraphics[width=\linewidth]{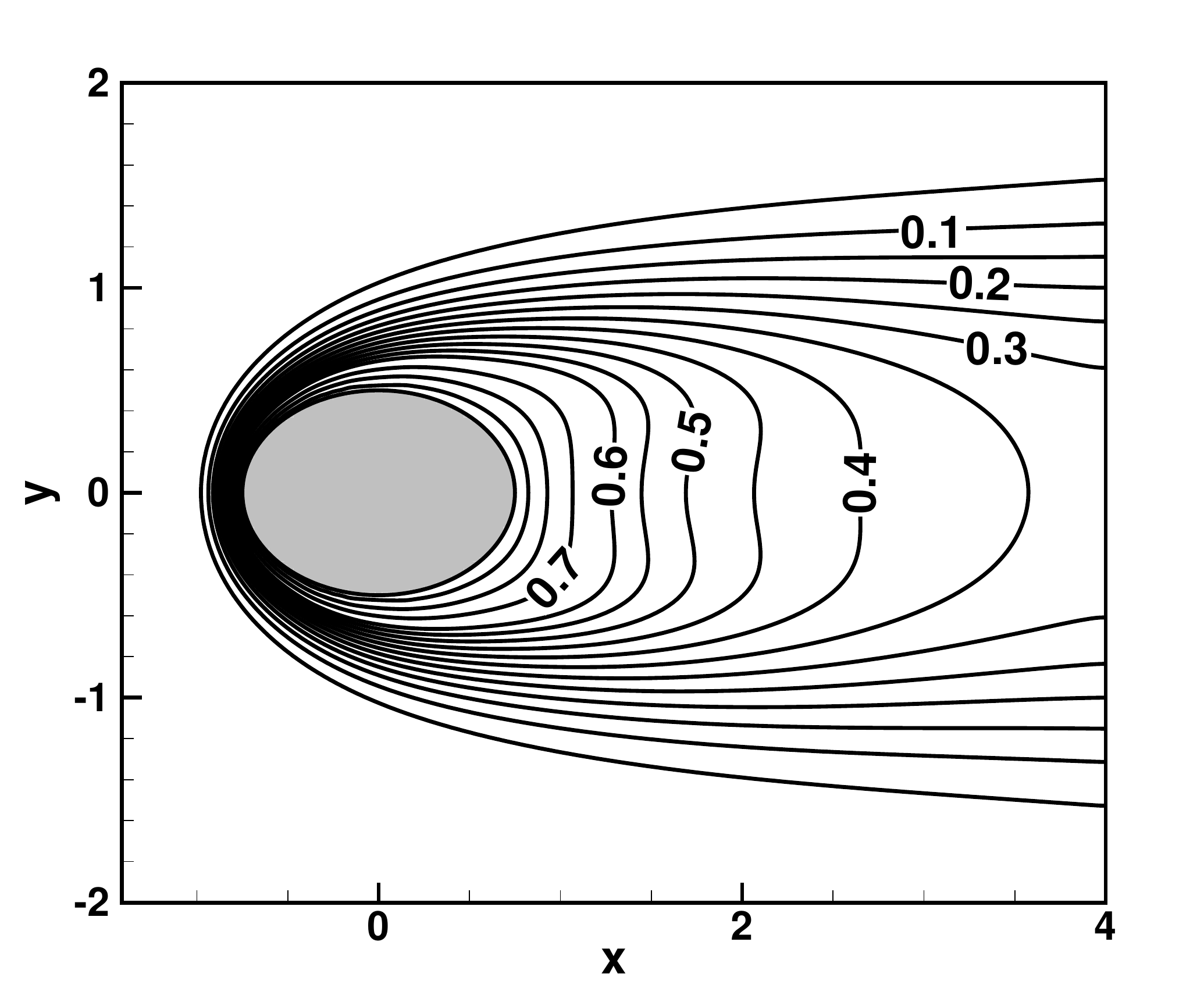} 
		\caption{}
	\end{subfigure}\hfil 
	\begin{subfigure}{0.3\textwidth}
		\includegraphics[width=\linewidth]{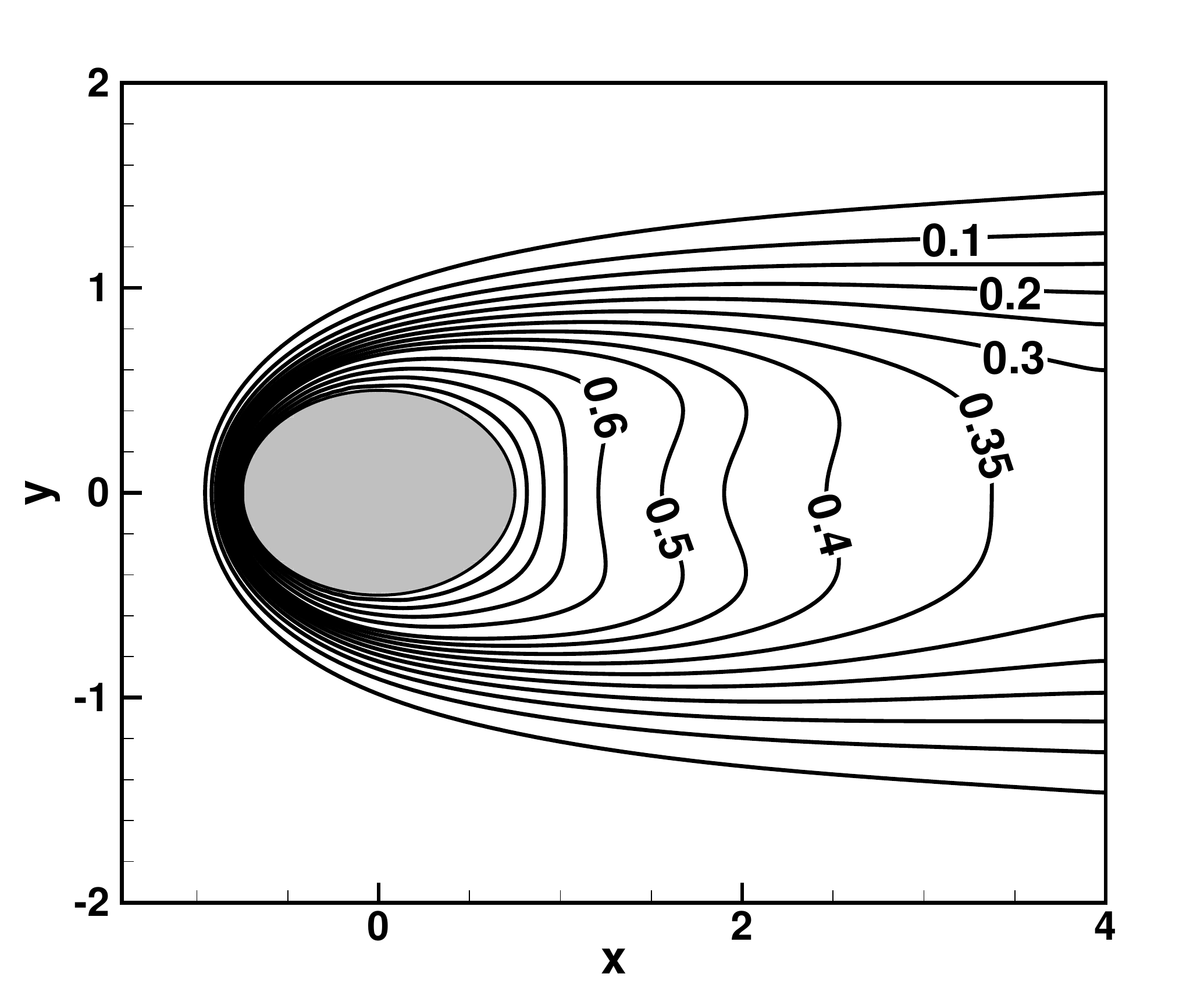} 
		\caption{}
	\end{subfigure}\hfil 
	\caption{\small{Steady state isotherms for $\theta=0^{\degree}$ and (a)$Re=10$, (b)$Re=20$, (c)$Re=30$, (d)$Re=40$, (e)$Re=50$, and (f) $Re=59$.}}
	\label{Fig:isotherms-steady-zerodeg}
\end{figure}

Figures \ref{Fig:psi-steady-15deg} and \ref{Fig:T-steady-15deg} show the streamlines and isotherms respectively for $\theta = 15^{\degree}$. $Re_{c}$ for $\theta = 15^{\degree}$ is in the range $59 \leq Re < 60$. Note that as the cylinder now occupies a position asymmetric to the incoming flow, the flow in the wake of the cylinder also loses its symmetry, which is reflected in the streamlines and isotherms. At $Re=10$ (figure \ref{Fig:psi-steady-15deg} (a)), flow separation does not happen and the tiny recirculation bubble seen for $\theta=0^{\degree}$ (figure \ref{Fig:psi-steady-zerodeg} (a)) vanishes, although a slight bulge in the streamlines can be seen at the rear end of the cylinder. At $Re=20$ (figure \ref{Fig:psi-steady-15deg} (b)), flow separates from the surface of the cylinder and a clockwise rotating recirculation region appears attached on the upper part of the cylinder. A counter-clockwise rotating vortex appears as well on the lower part of the cylinder at $Re=30$ (figure \ref{Fig:psi-steady-15deg} (c)). This vortex, however, remains detached from the cylinder surface. Both vortices grow in size and strength as the $Re$ increases (figures \ref{Fig:psi-steady-15deg} (d) - (f)). Due to the asymmetric nature of the flow w.r.t the cylinder, these vortices are also of unequal strengths and sizes. This asymmetry is reflected in the isotherms as well (figures \ref{Fig:T-steady-15deg} (a) - (f)). A better insight into the nature of heat transfer characteristics can be gleaned from the plot of surface Nusselt number, which is presented in a subsequent section.
\begin{figure}[H]
	\centering
	\begin{subfigure}{0.3\textwidth}
		\includegraphics[width=\linewidth]{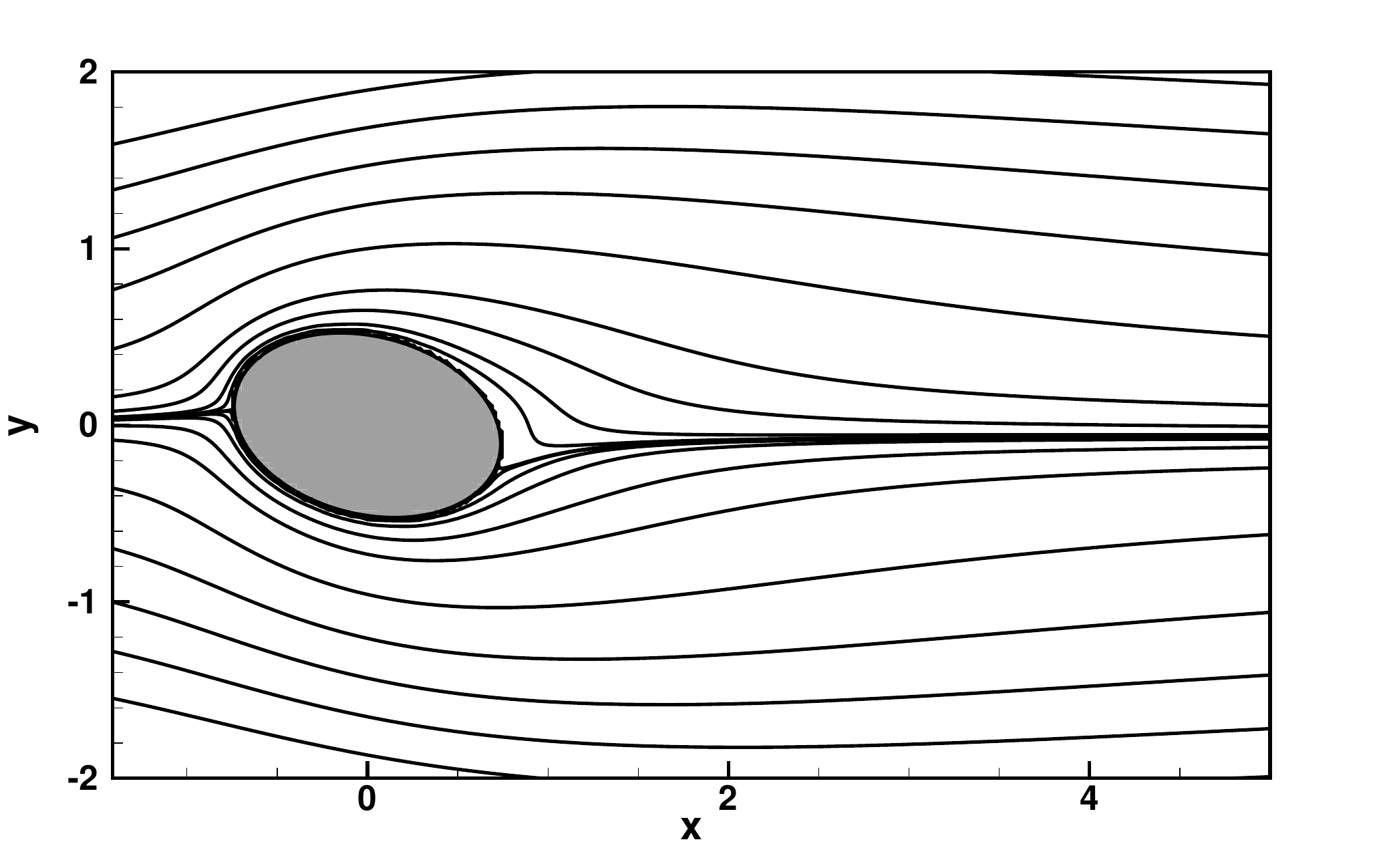} 
		\caption{$Re=10$}
	\end{subfigure}\hfil 
	\begin{subfigure}{0.3\textwidth}
		\includegraphics[width=\linewidth]{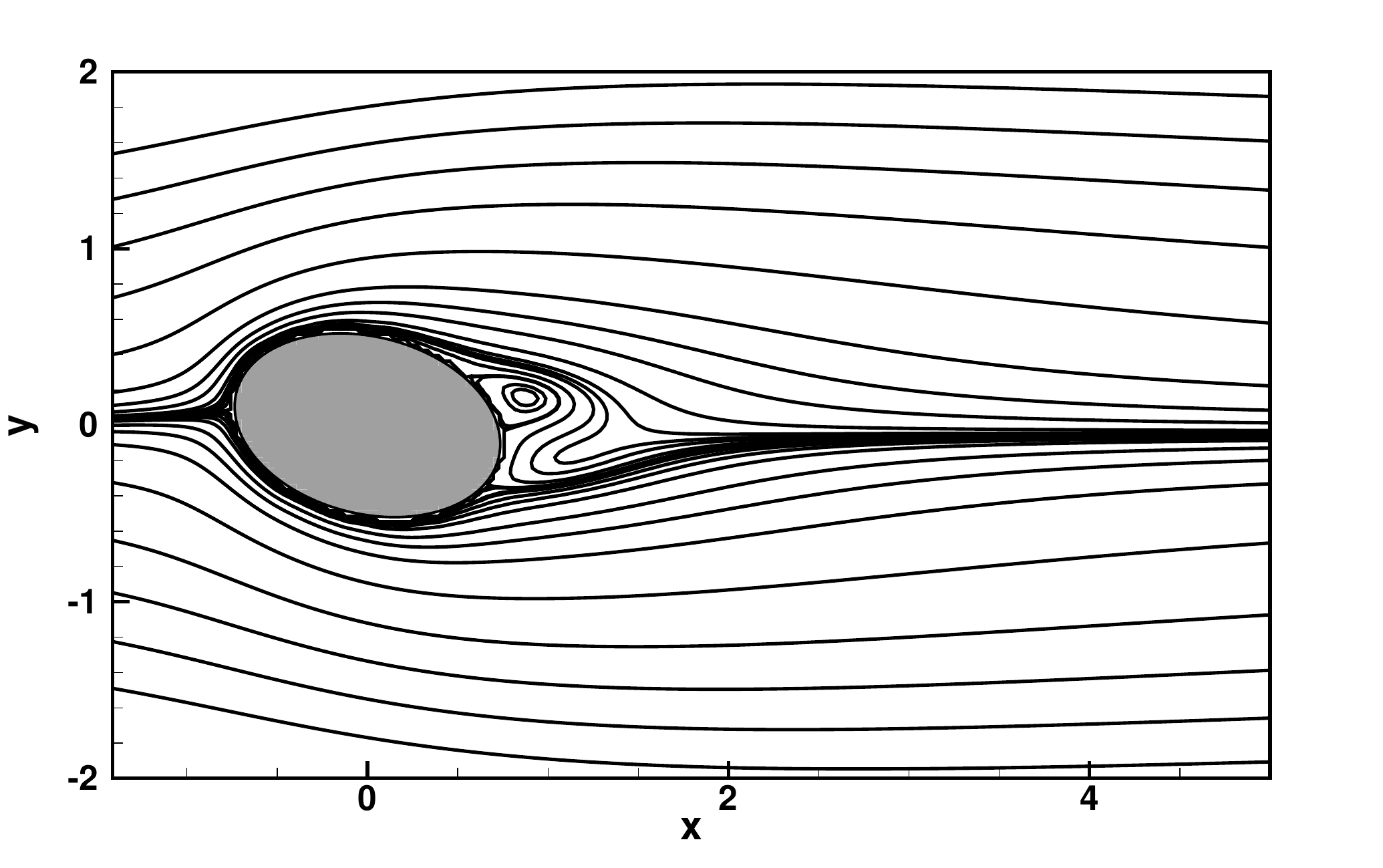} 
		\caption{$Re=20$}
	\end{subfigure}\hfil 
	\begin{subfigure}{0.3\textwidth}
		\includegraphics[width=\linewidth]{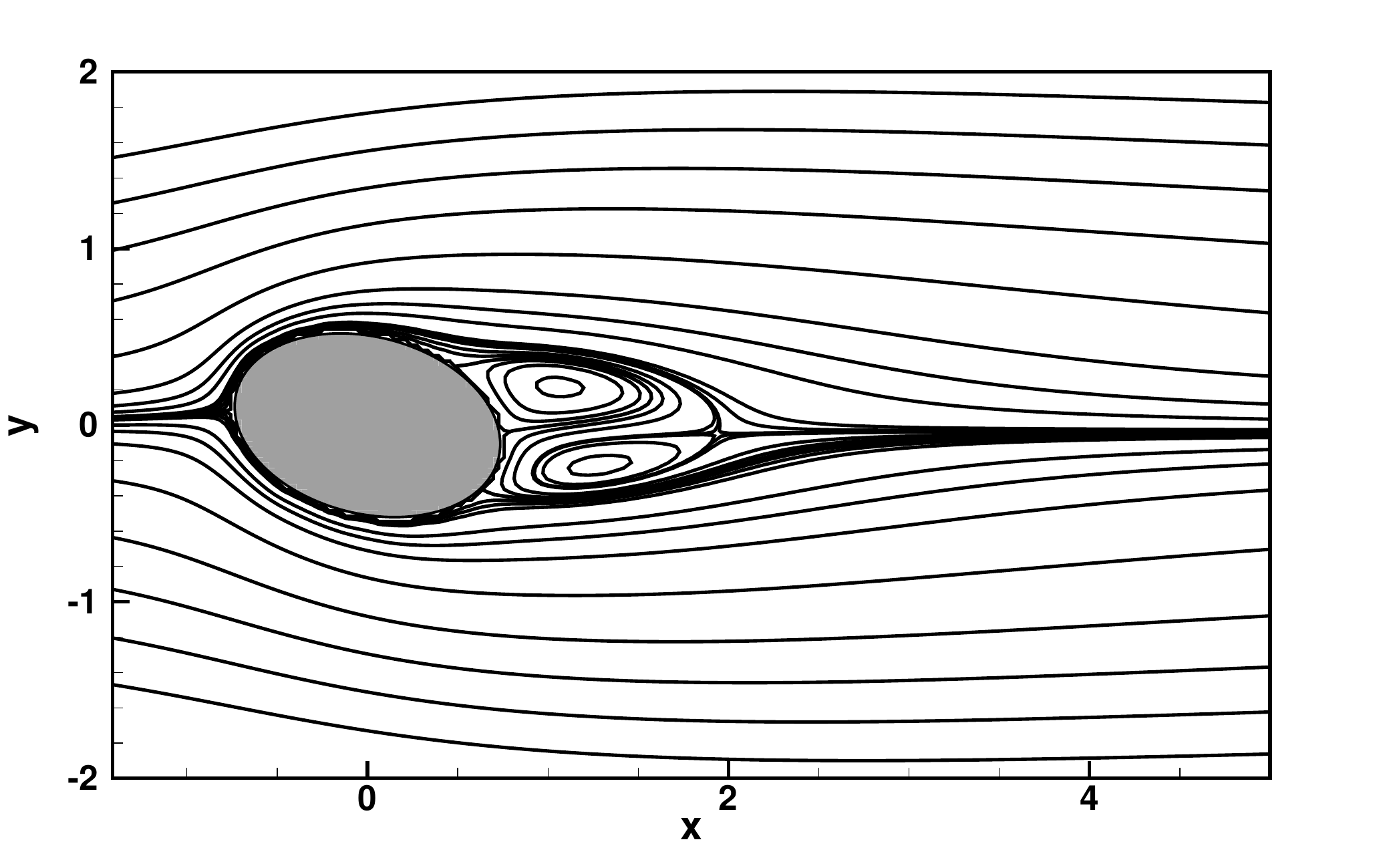} 
		\caption{$Re=30$}
	\end{subfigure}\hfil 
	\begin{subfigure}{0.3\textwidth}
		\includegraphics[width=\linewidth]{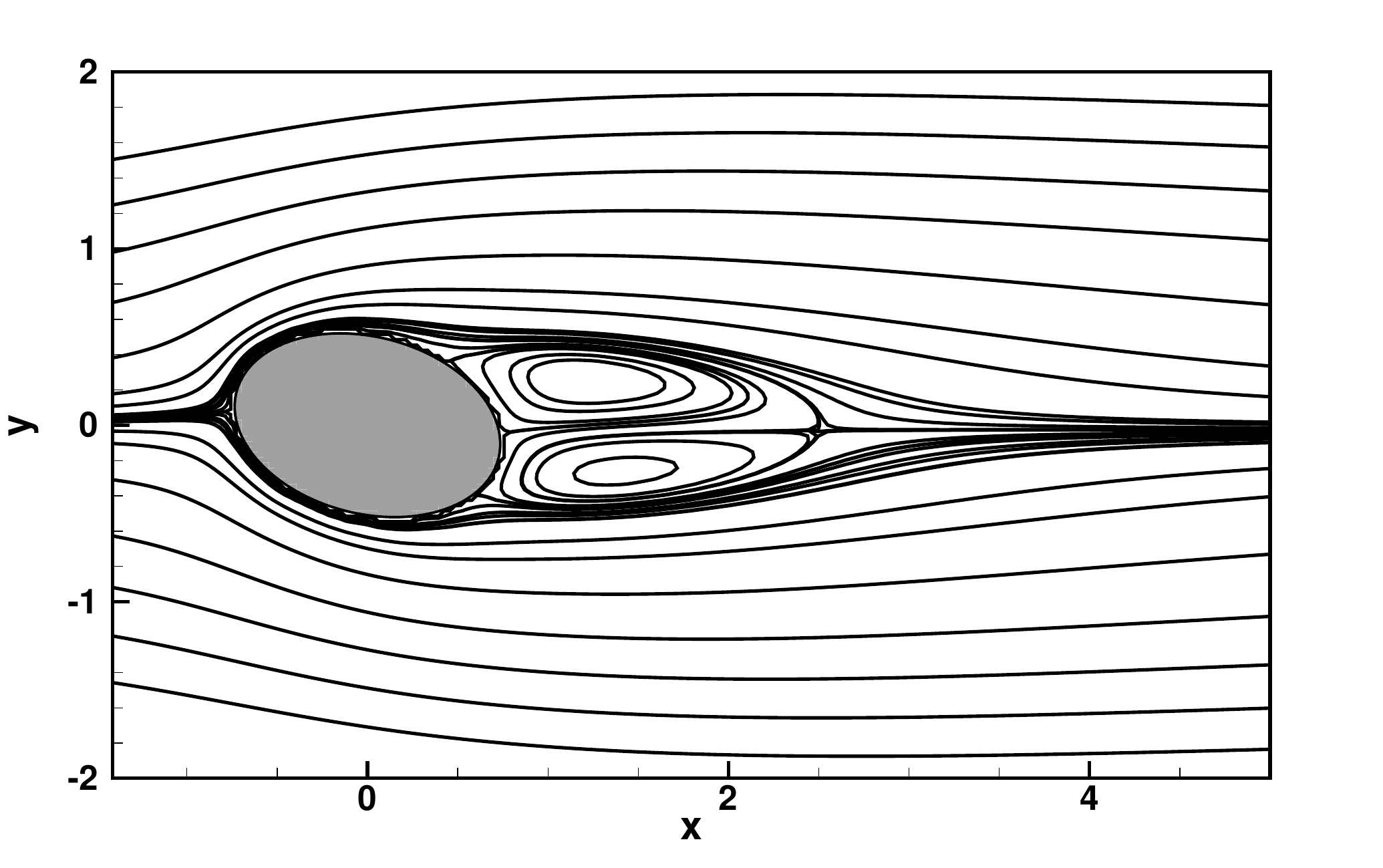} 
		\caption{$Re=40$}
	\end{subfigure}\hfil 
	\begin{subfigure}{0.3\textwidth}
		\includegraphics[width=\linewidth]{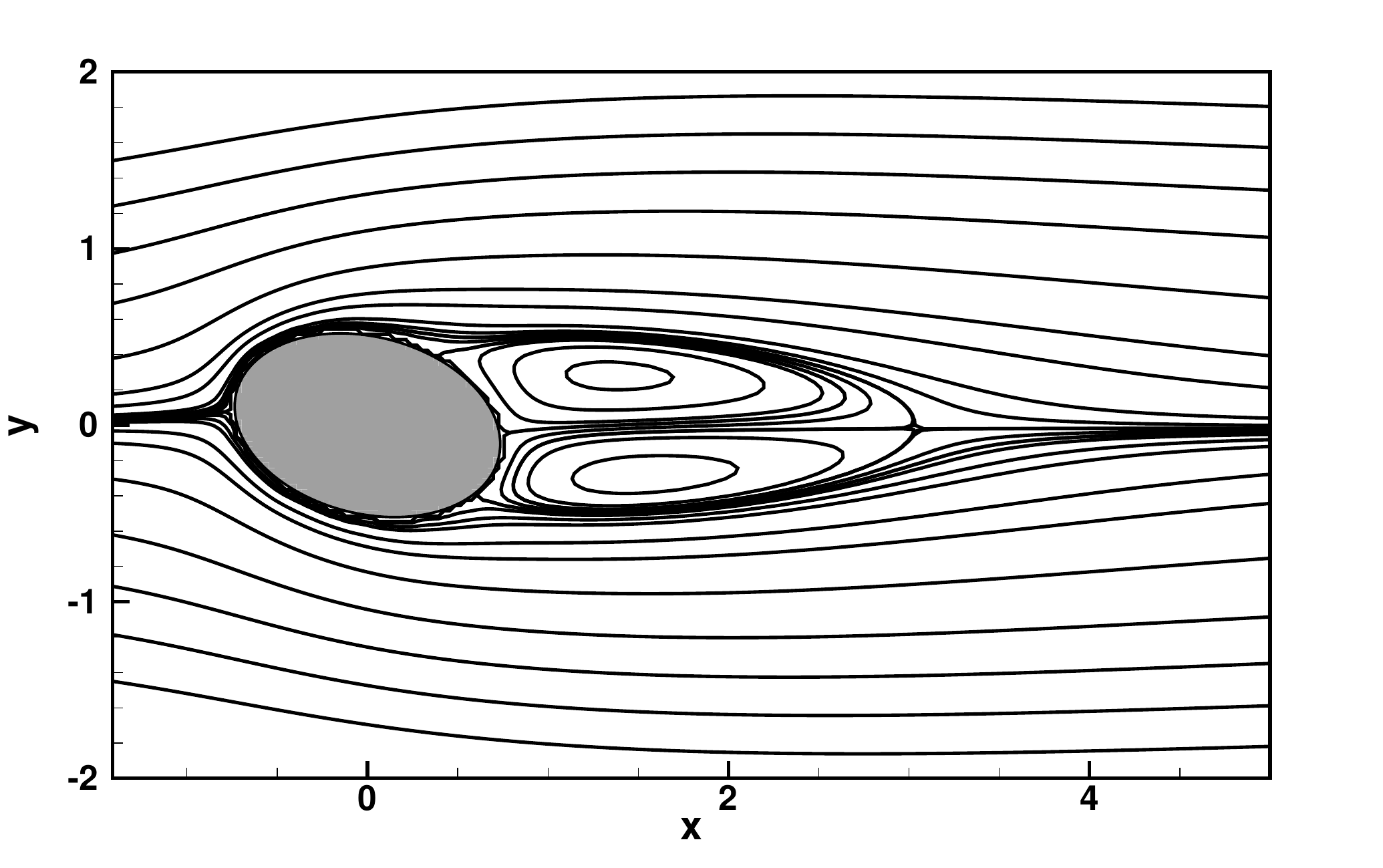} 
		\caption{$Re=50$}		
	\end{subfigure}\hfil 
	\begin{subfigure}{0.3\textwidth}
		\includegraphics[width=\linewidth]{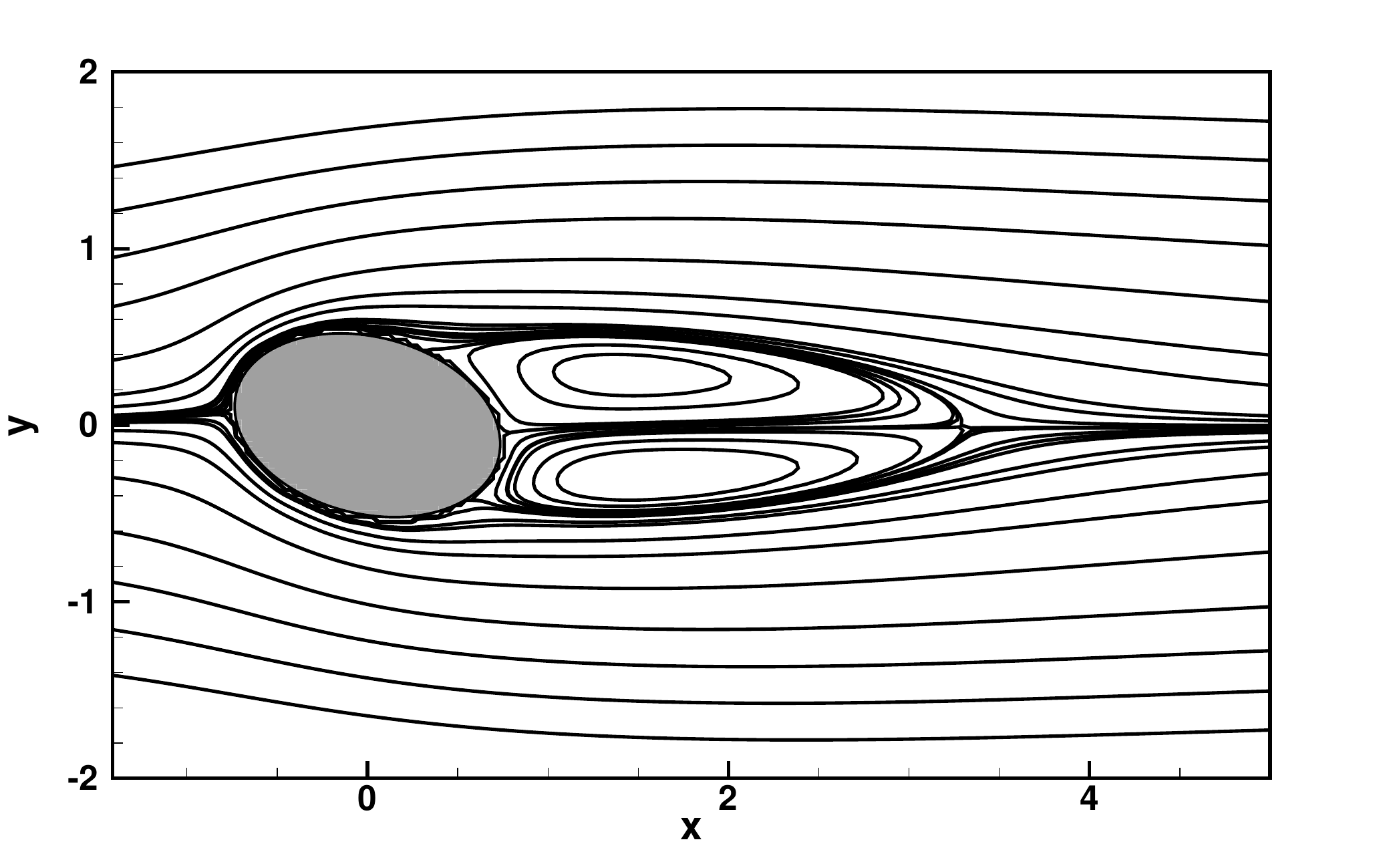} 
		\caption{$Re=59$}		
	\end{subfigure}\hfil 
	\caption{\small{Steady state streamlines for $\theta=15^{\degree}$ and (a)$Re=10$, (b)$Re=20$, (c)$Re=30$, (d)$Re=40$, (e)$Re=50$, and (f) $Re=59$.}}
	\label{Fig:psi-steady-15deg}
\end{figure}

\begin{figure}[H]
	\centering
	\begin{subfigure}{0.3\textwidth}
		\includegraphics[width=\linewidth]{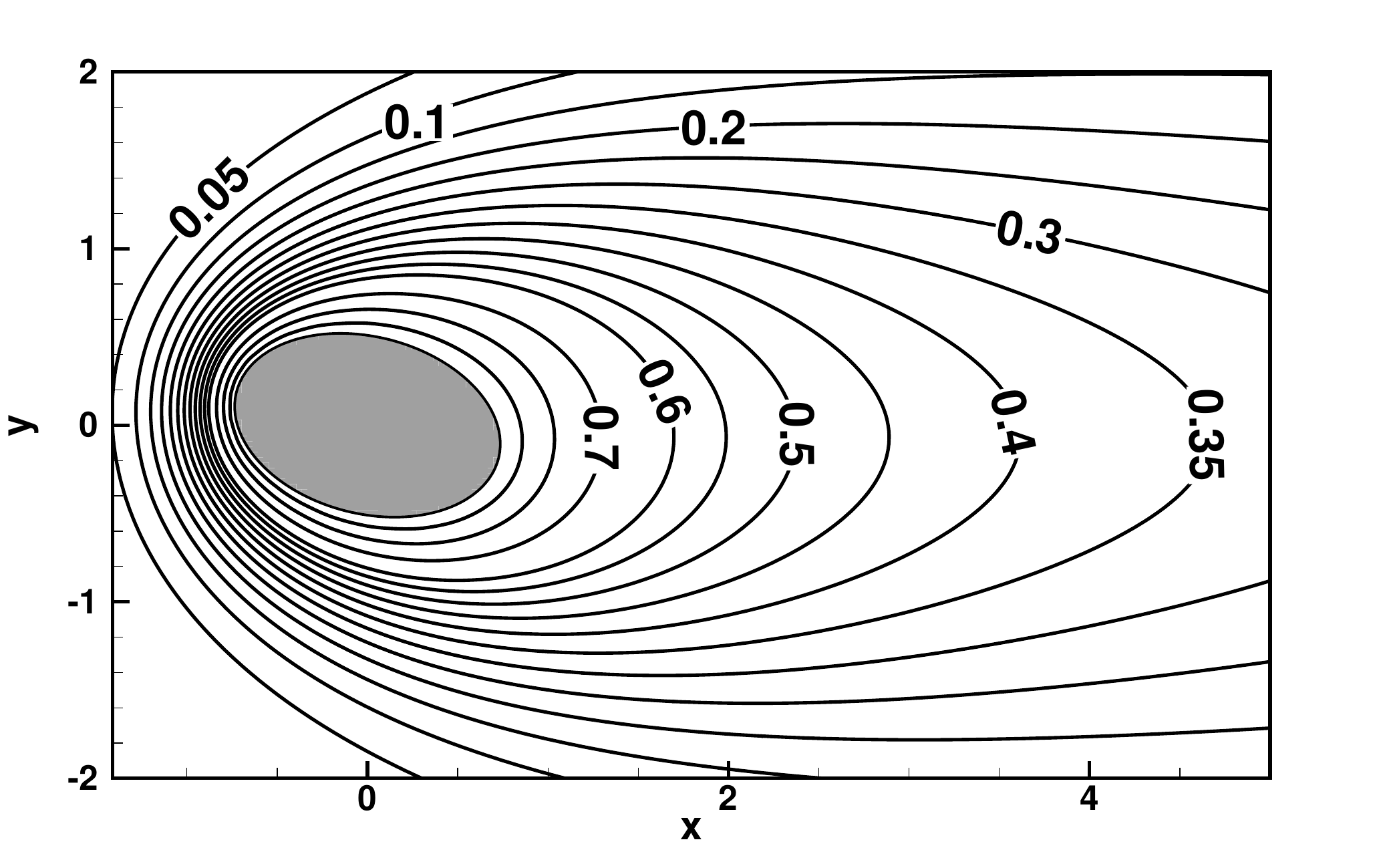} 
		\caption{$Re=10$}
	\end{subfigure}\hfil 
	\begin{subfigure}{0.3\textwidth}
		\includegraphics[width=\linewidth]{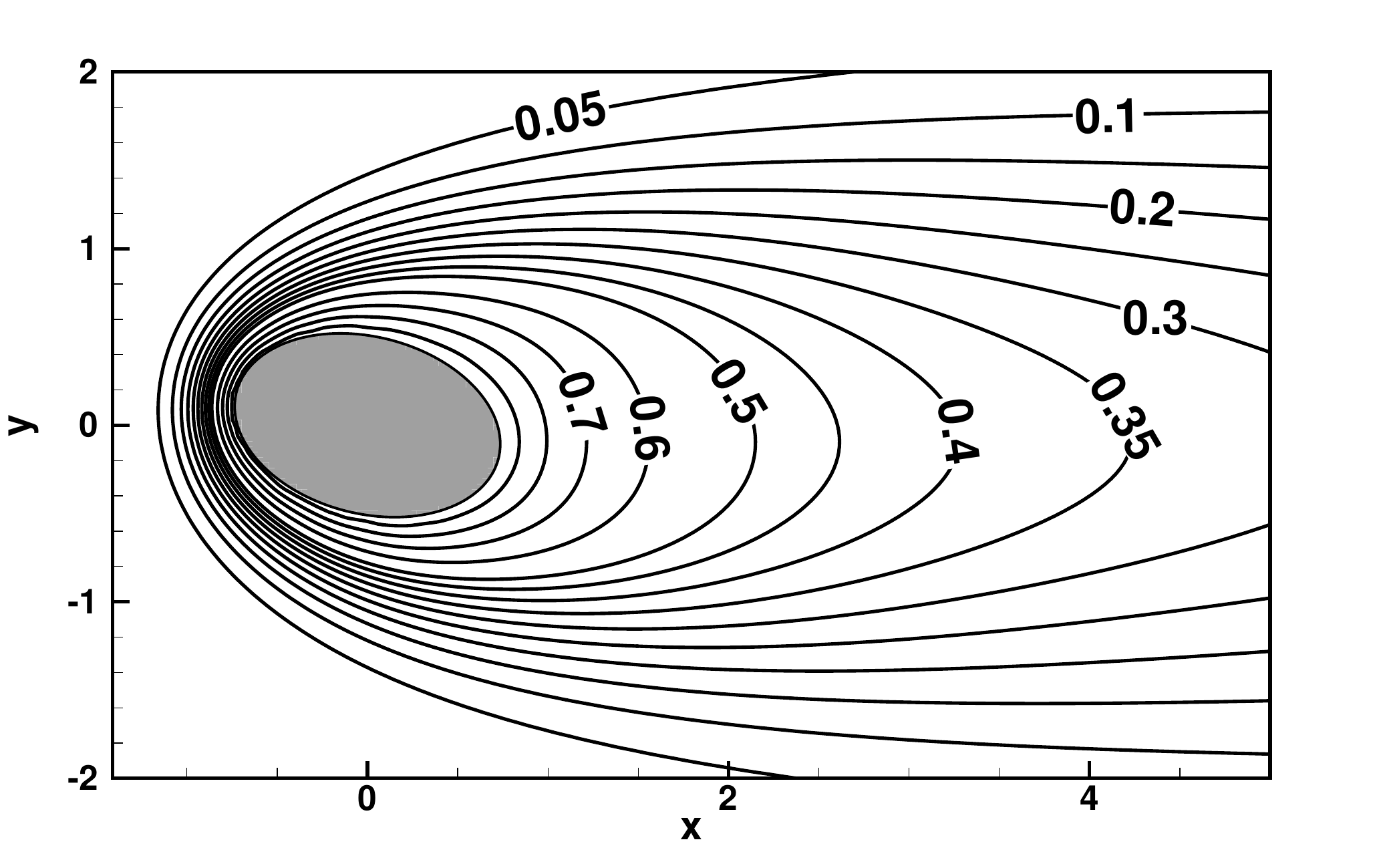} 
		\caption{$Re=20$}
	\end{subfigure}\hfil 
	\begin{subfigure}{0.3\textwidth}
		\includegraphics[width=\linewidth]{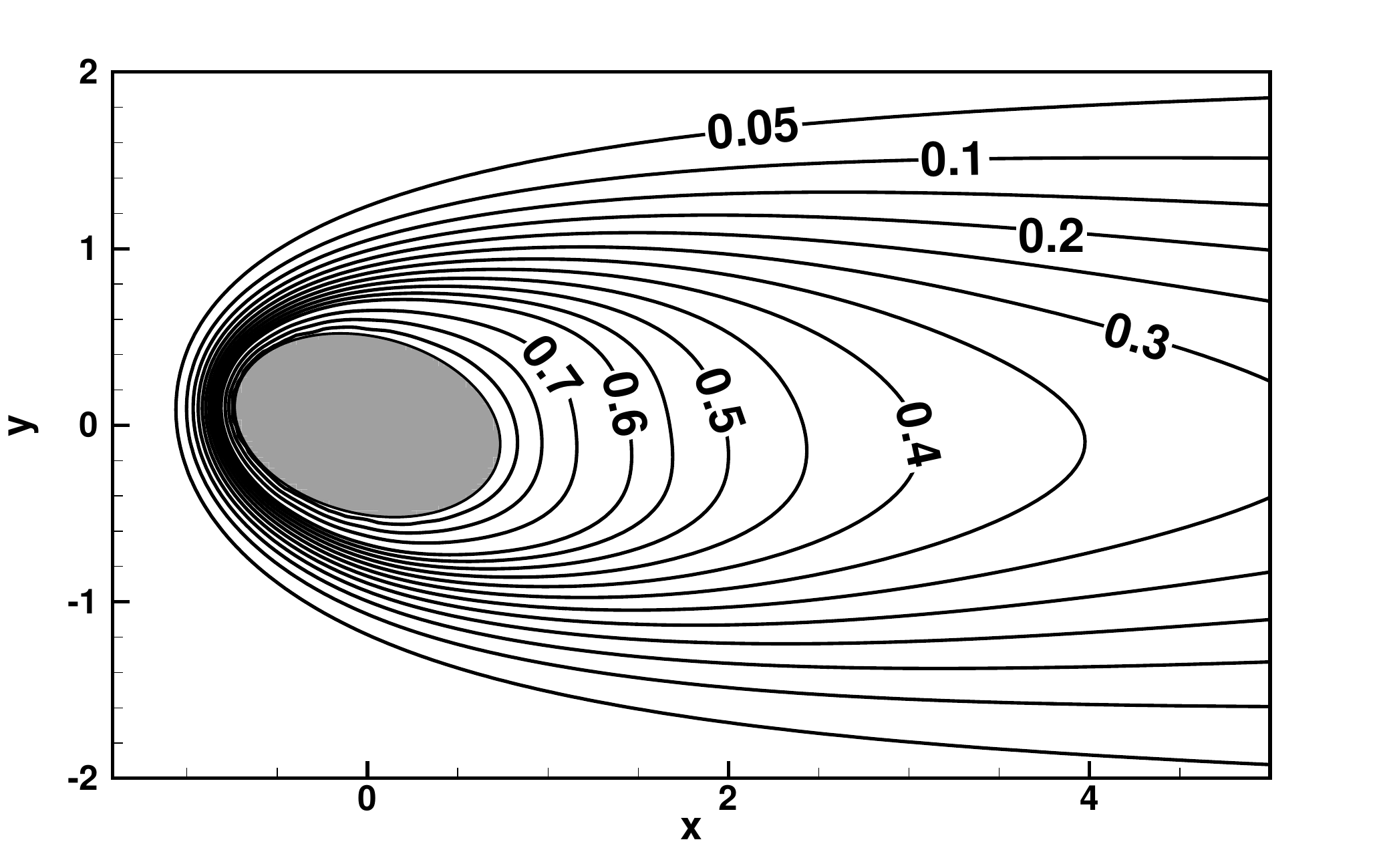} 
		\caption{$Re=30$}
	\end{subfigure}\hfil 
	\begin{subfigure}{0.3\textwidth}
		\includegraphics[width=\linewidth]{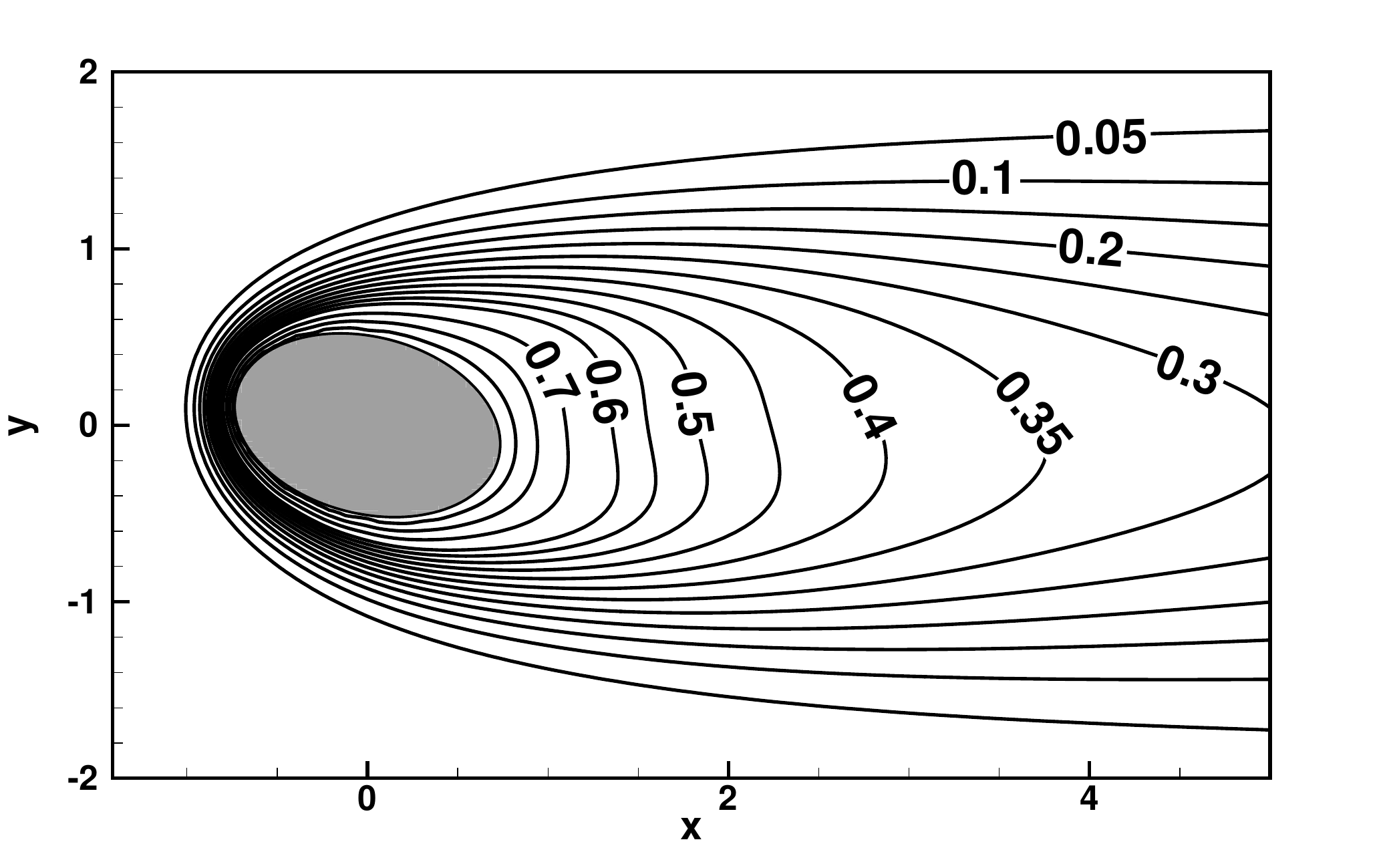} 
		\caption{$Re=40$}
	\end{subfigure}\hfil 
	\begin{subfigure}{0.3\textwidth}
		\includegraphics[width=\linewidth]{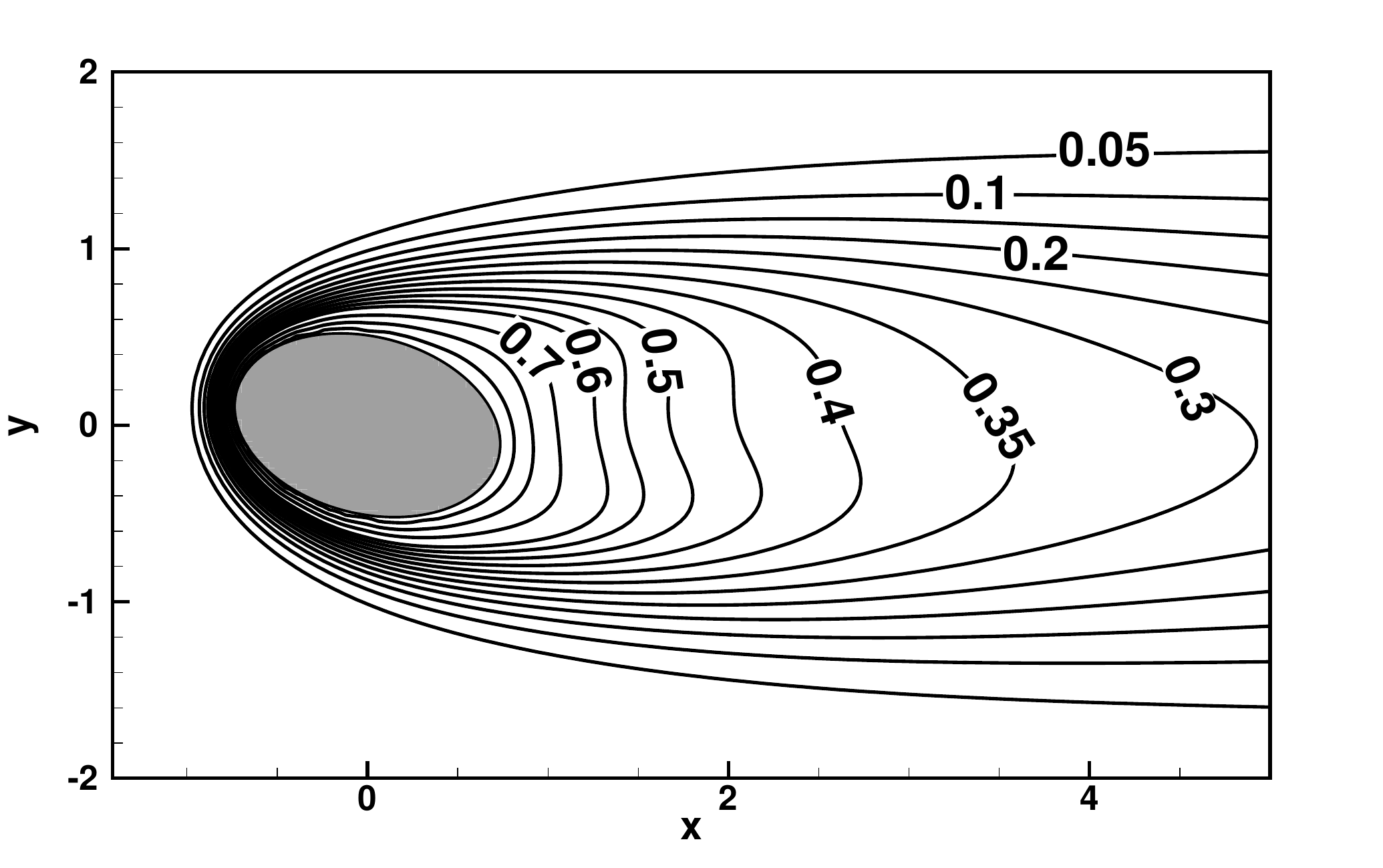} 
		\caption{$Re=50$}		
	\end{subfigure}\hfil 
	\begin{subfigure}{0.3\textwidth}
		\includegraphics[width=\linewidth]{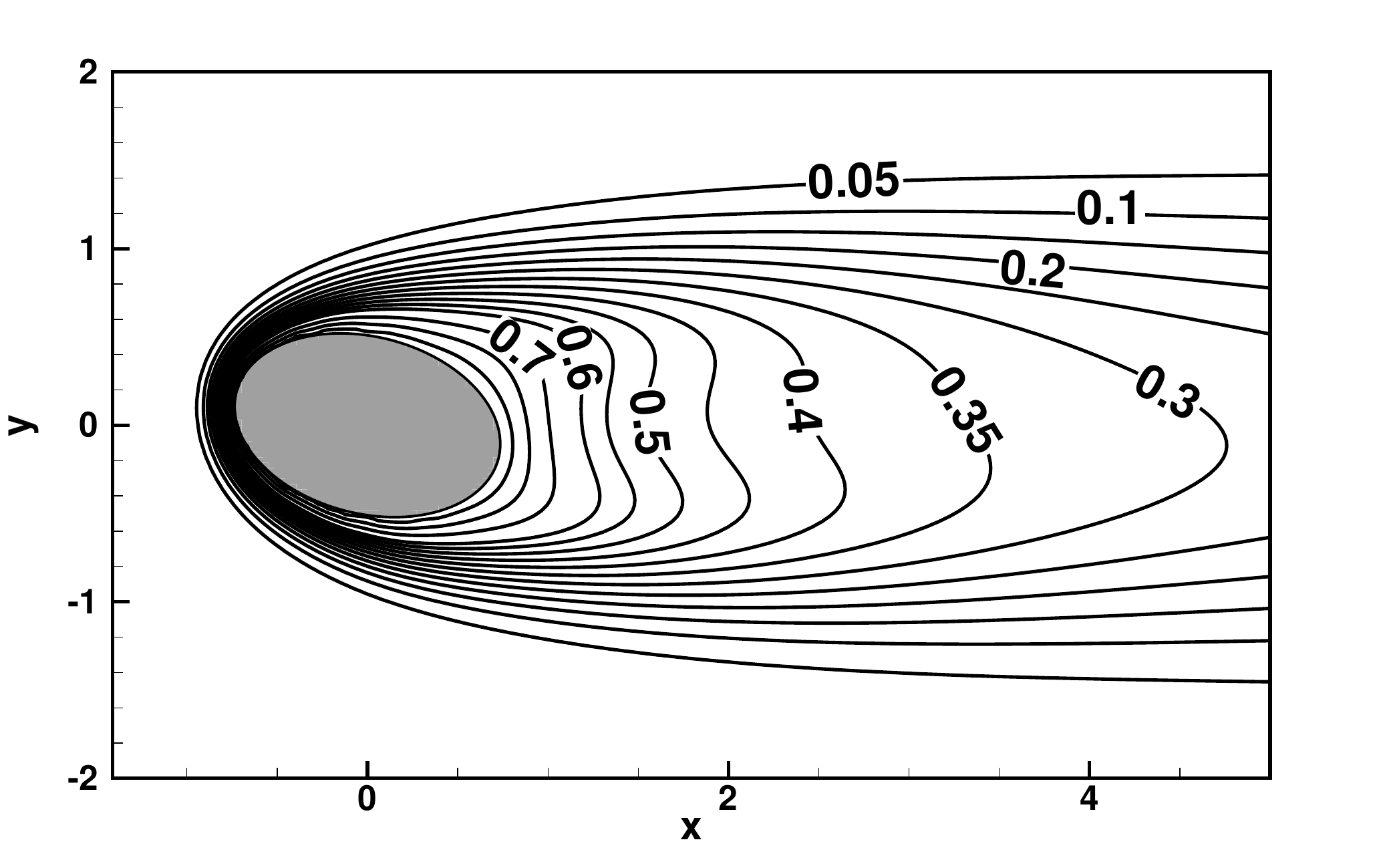} 
		\caption{$Re=59$}		
	\end{subfigure}\hfil 
	\caption{\small{Steady state isotherms for $\theta=15^{\degree}$ and (a)$Re=10$, (b)$Re=20$, (c)$Re=30$, (d)$Re=40$, (e)$Re=50$, and (f) $Re=59$.}}
	\label{Fig:T-steady-15deg}
\end{figure}

Figures \ref{Fig:psi-steady-30deg} and \ref{Fig:T-steady-30deg} show the streamlines and isotherms respectively for $\theta = 30^{\degree}$, where the $Re_{c}$ is in the range $49 \leq Re < 50$. The barely discernible bulge when $\theta=15^{\degree}$ at $Re=10$ (figure \ref{Fig:T-steady-15deg} (a)) is more noticeable when $\theta = 30^{\degree}$ (figure \ref{Fig:psi-steady-30deg} (a)) implying that the flow is on the brink of separating from the cylinder surface. The clockwise rotating vortex attached on the upper surface of the cylinder grows in size at $Re=20$, and a counterclockwise rotating vortex begins to form near the lower surface of the cylinder (figure \ref{Fig:psi-steady-30deg} (b)). Flow pattern for the rest of the $Re$'s follow a similar pattern to the previous configuration. The isotherms also follow a similar pattern, except that the distortions in the isotherms appear at a much lower $Re$  as $\theta$ is increased,  $Re=40$ for this case (figure \ref{Fig:T-steady-30deg} (d)) compared to $Re=50$ for $\theta=15^{\degree}$ (figure \ref{Fig:T-steady-15deg} (e)). This would indicate that the overall heat transfer rate for the same $Re$ is comparably higher (see section \ref{sec:Nusselt-steady}).
\begin{figure}[H]
	\centering
	\begin{subfigure}{0.3\textwidth}
		\includegraphics[width=\linewidth]{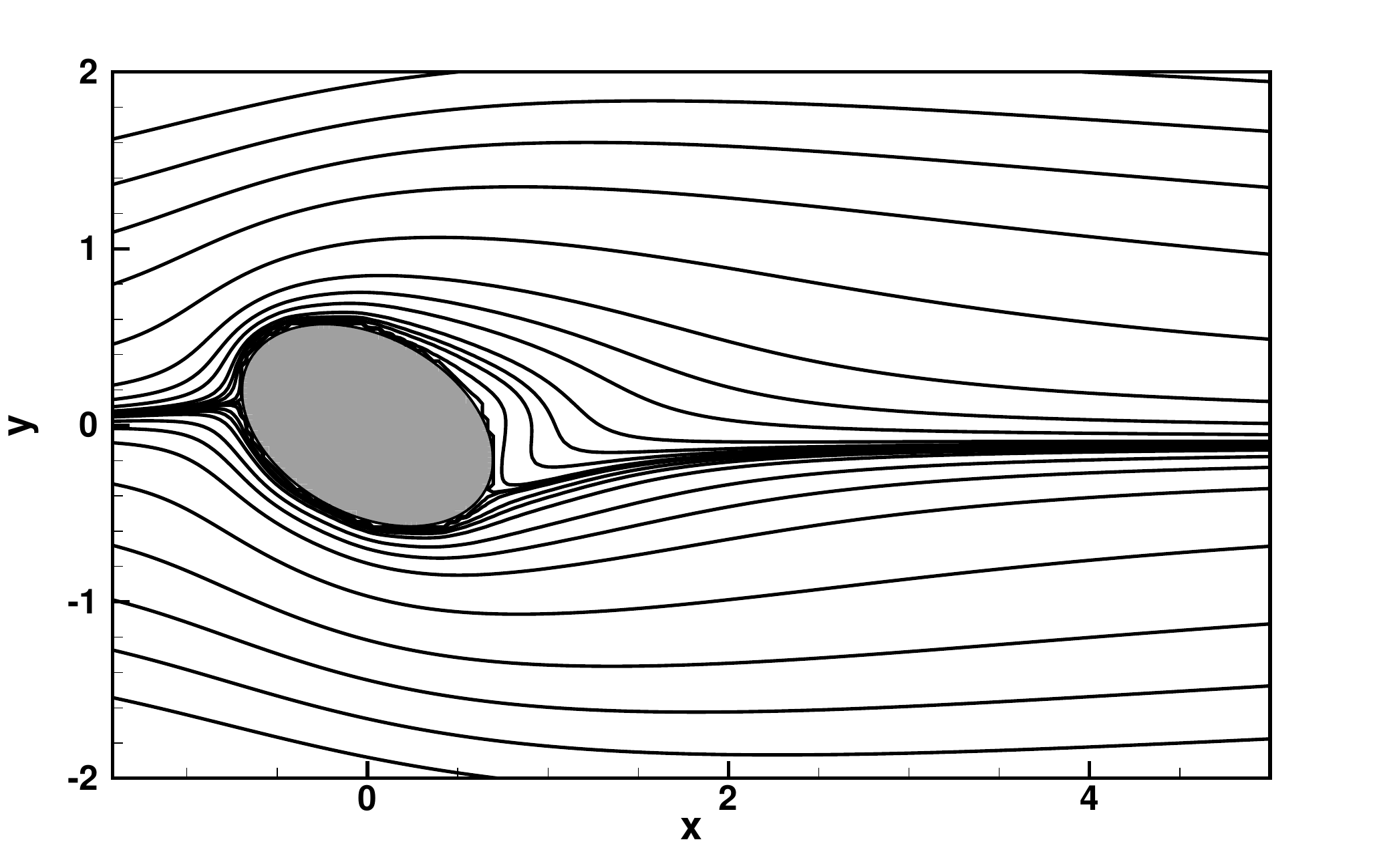} 
		\caption{$Re=10$}
	\end{subfigure}\hfil 
	\begin{subfigure}{0.3\textwidth}
		\includegraphics[width=\linewidth]{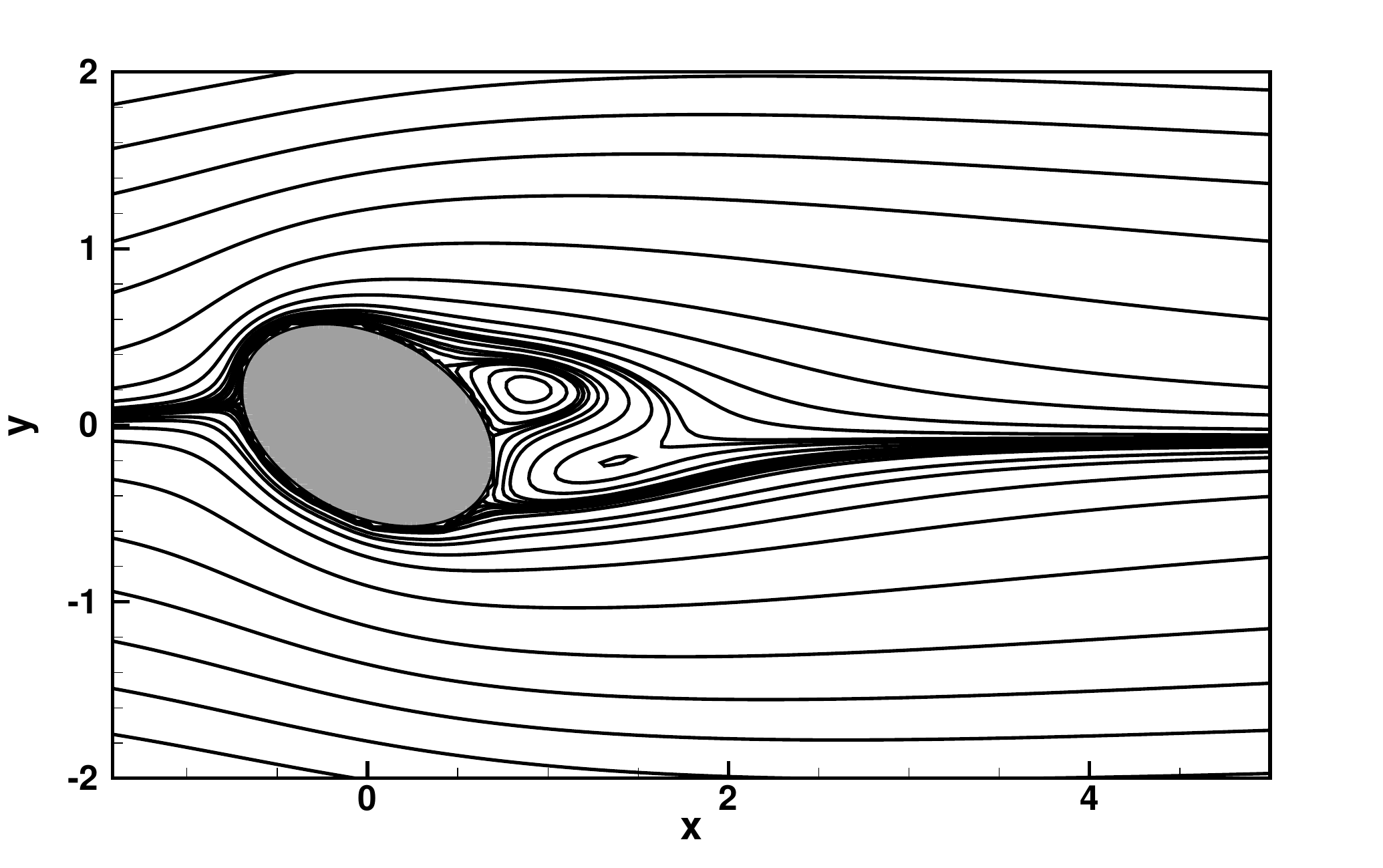} 
		\caption{$Re=20$}
	\end{subfigure}\hfil 
	\begin{subfigure}{0.3\textwidth}
		\includegraphics[width=\linewidth]{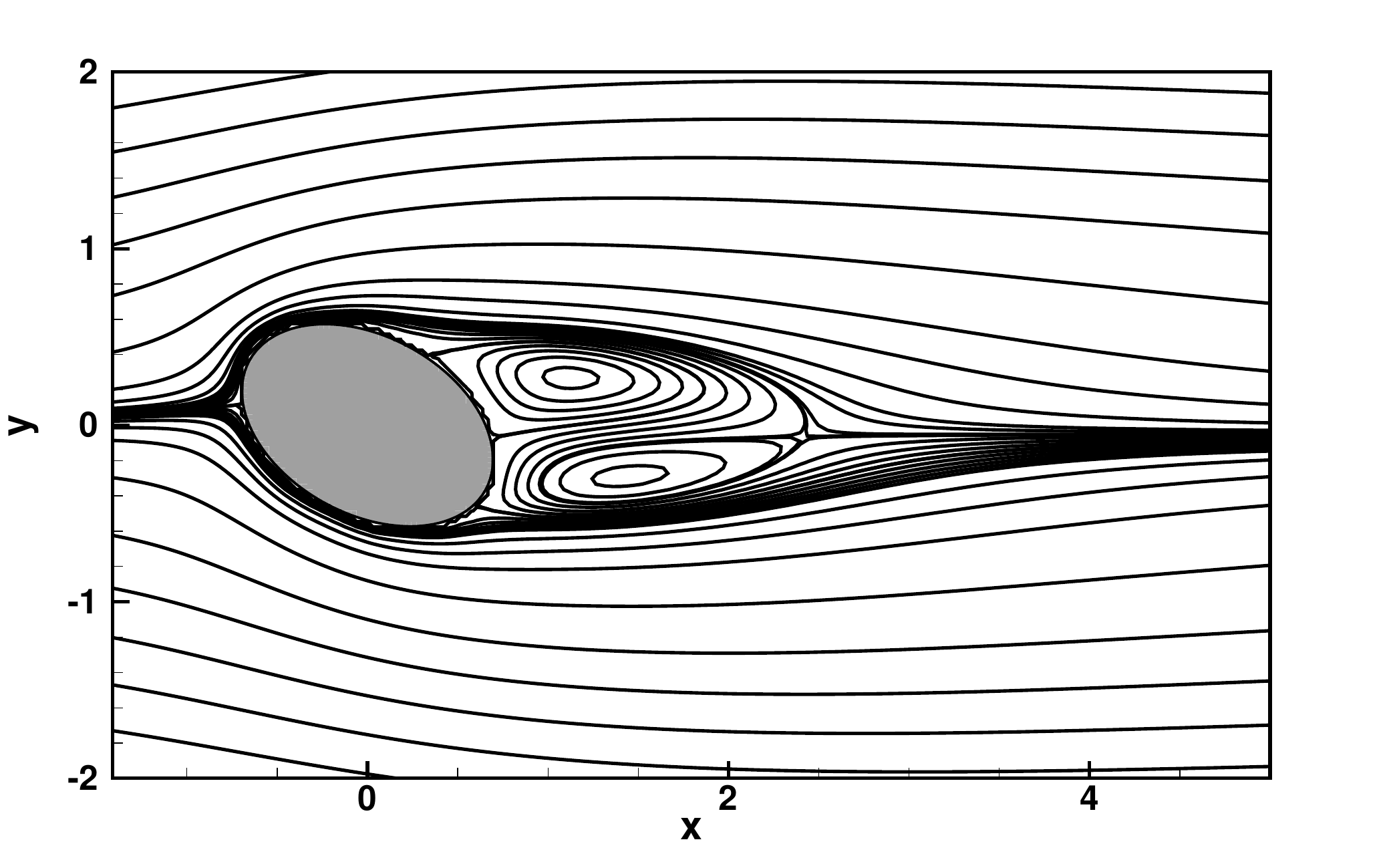} 
		\caption{$Re=30$}
	\end{subfigure}\hfil 
	\begin{subfigure}{0.3\textwidth}
		\includegraphics[width=\linewidth]{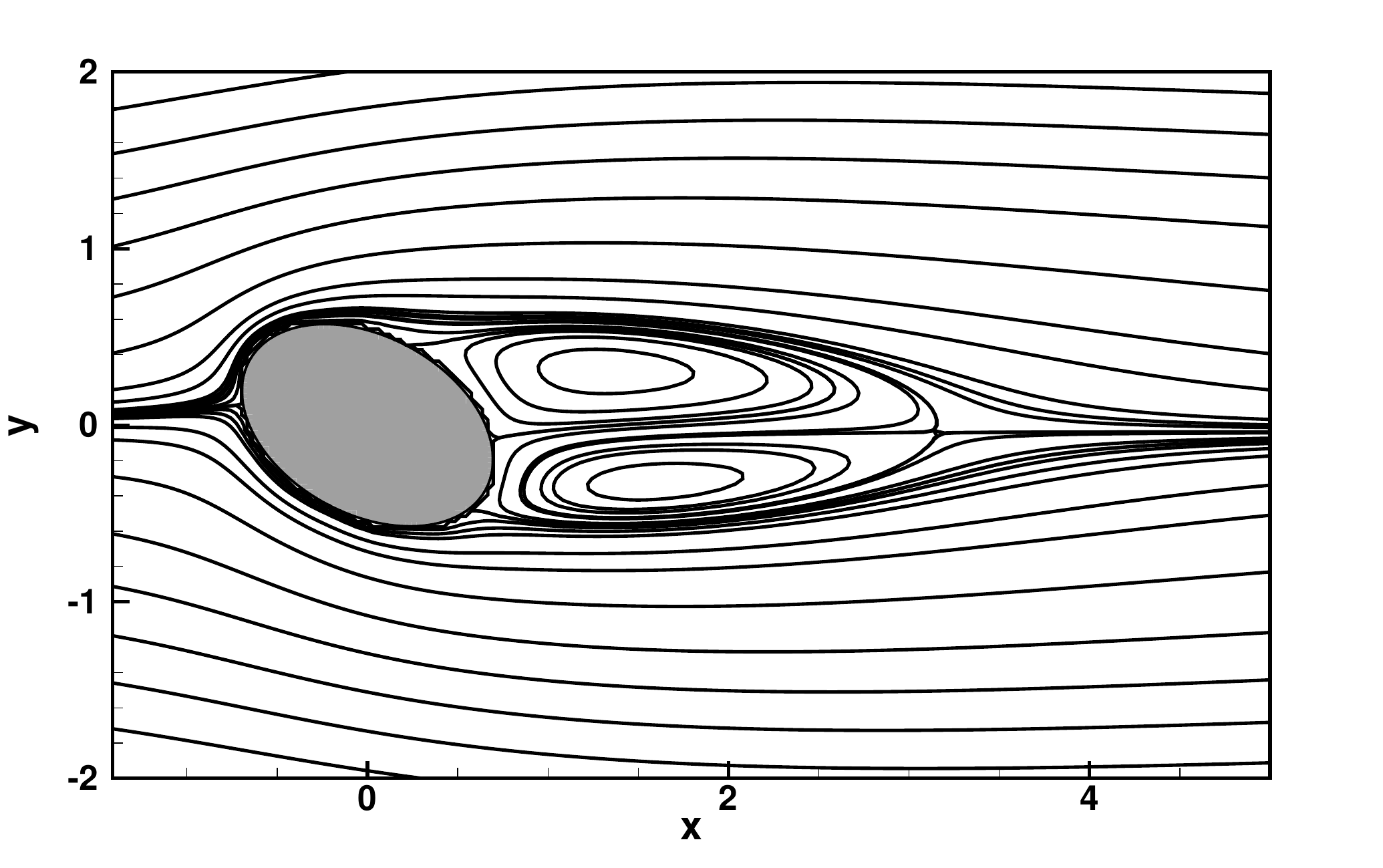} 
		\caption{$Re=40$}
	\end{subfigure}\hfil 
	\begin{subfigure}{0.3\textwidth}
		\includegraphics[width=\linewidth]{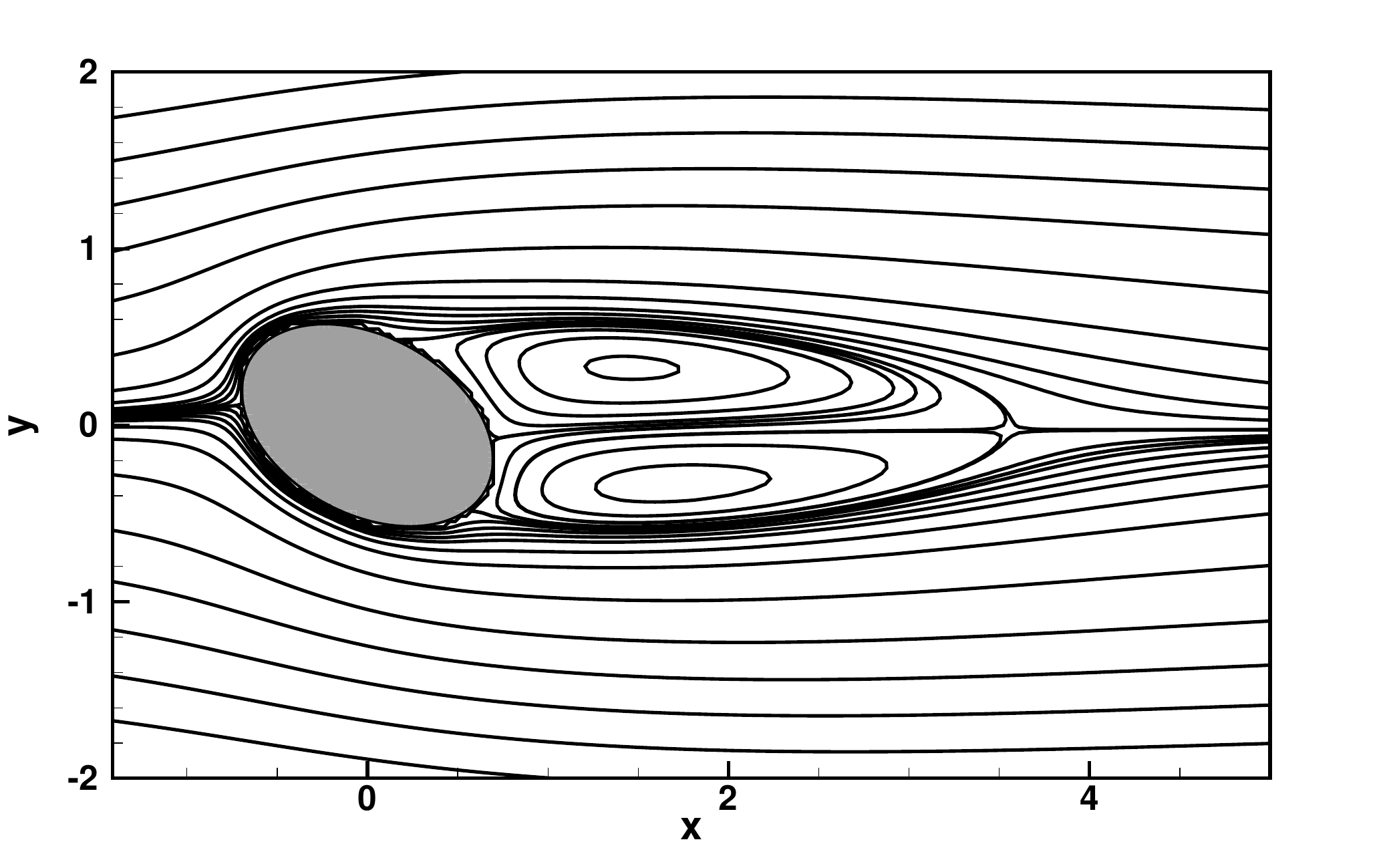} 
		\caption{$Re=49$}		
	\end{subfigure}\hfil 
	\caption{\small{Steady state streamlines for $\theta=30^{\degree}$ and (a)$Re=10$, (b)$Re=20$, (c)$Re=30$, (d)$Re=40$, and (e)$Re=49$.}}
	\label{Fig:psi-steady-30deg}
\end{figure}

\begin{figure}[H]
	\centering
	\begin{subfigure}{0.3\textwidth}
		\includegraphics[width=\linewidth]{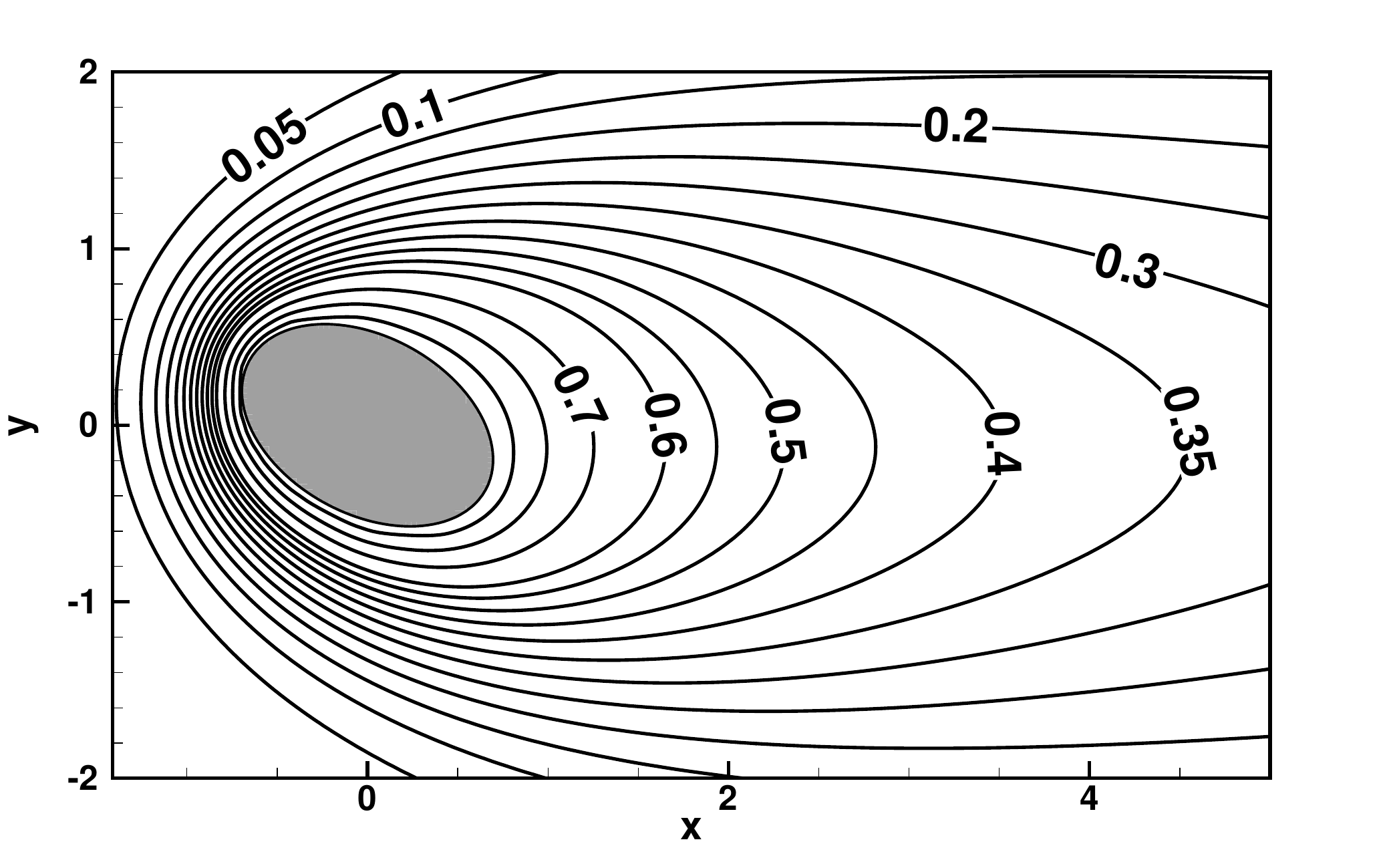} 
		\caption{$Re=10$}
	\end{subfigure}\hfil 
	\begin{subfigure}{0.3\textwidth}
		\includegraphics[width=\linewidth]{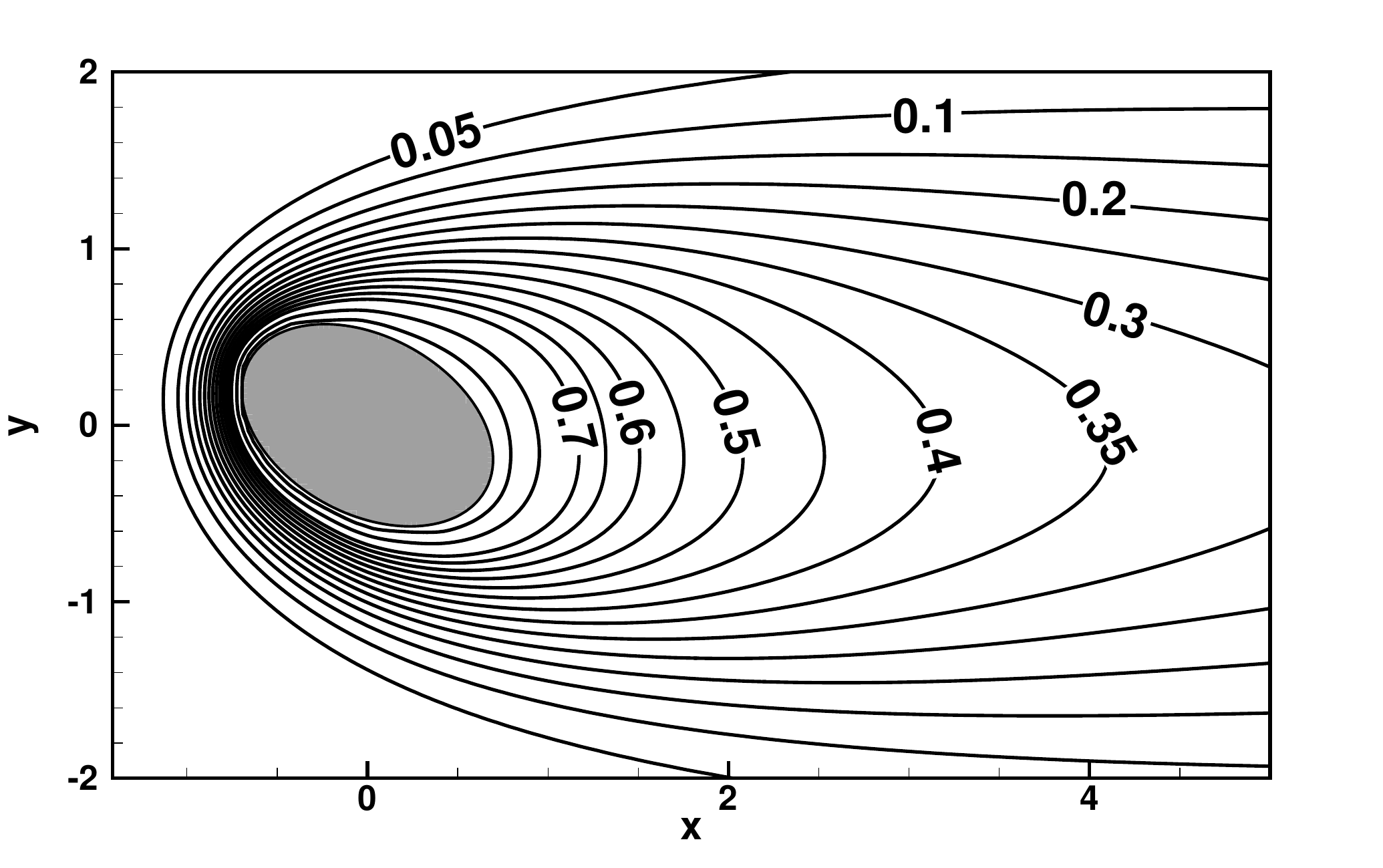} 
		\caption{$Re=20$}
	\end{subfigure}\hfil 
	\begin{subfigure}{0.3\textwidth}
		\includegraphics[width=\linewidth]{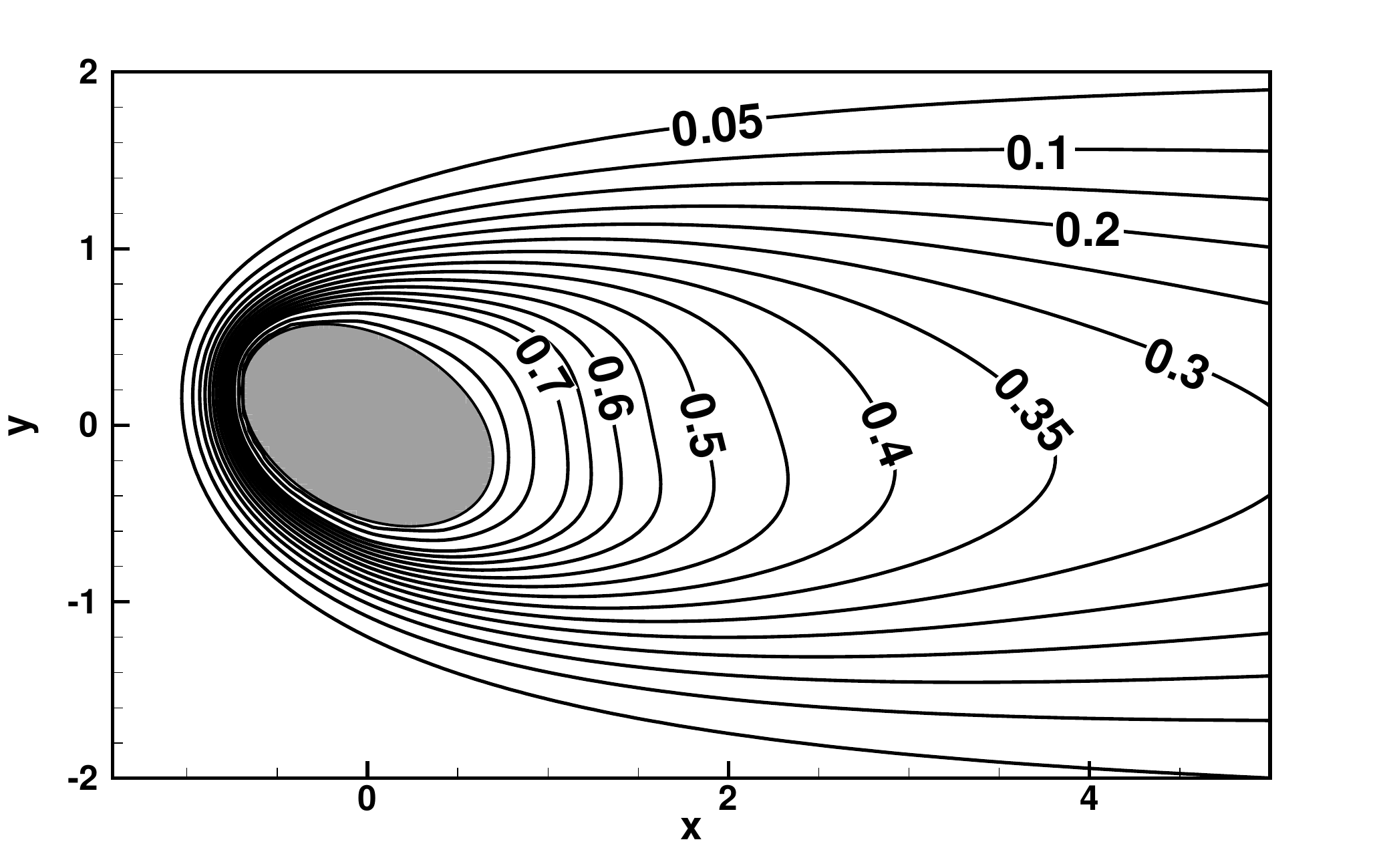} 
		\caption{$Re=30$}
	\end{subfigure}\hfil 
	\begin{subfigure}{0.3\textwidth}
		\includegraphics[width=\linewidth]{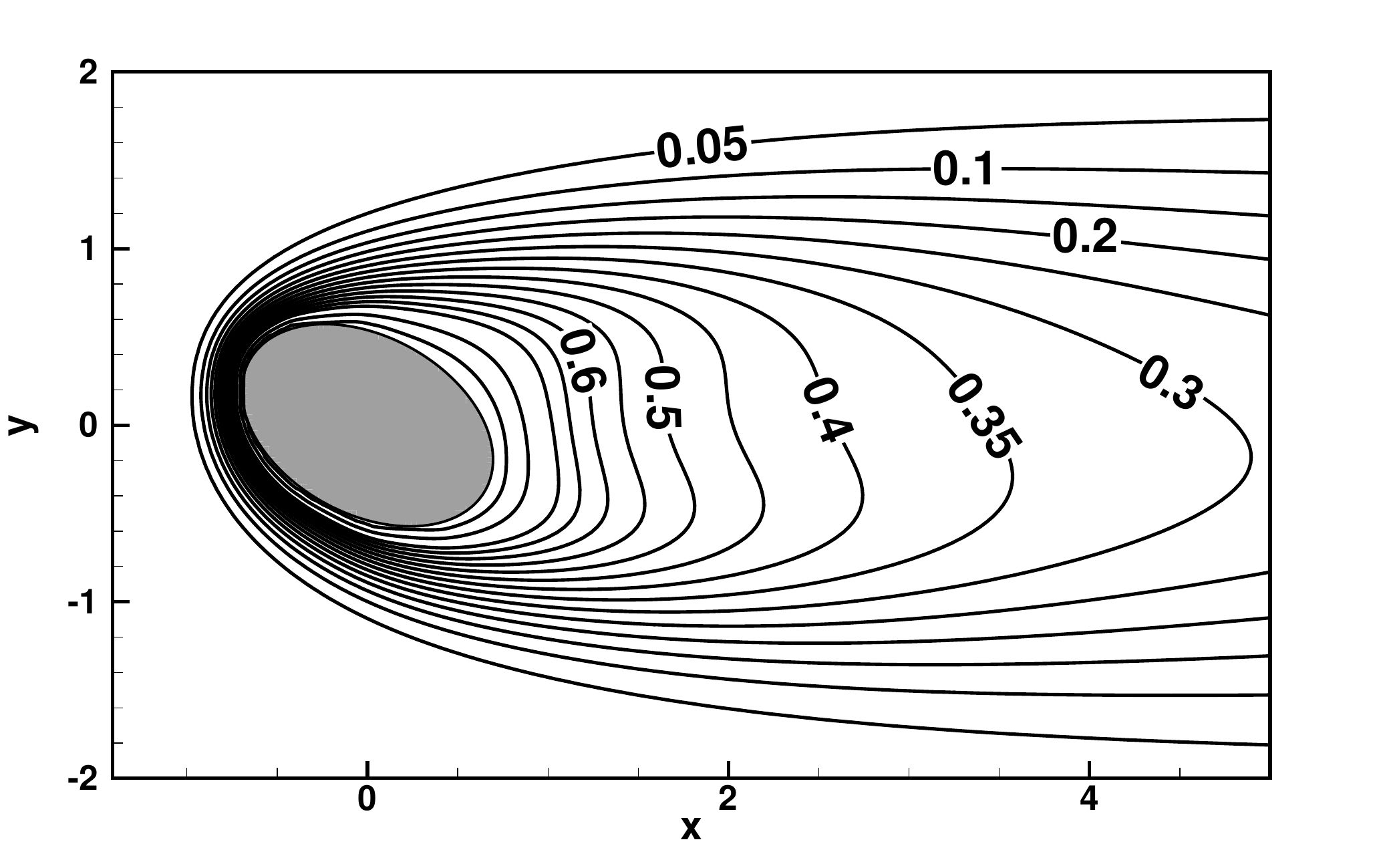} 
		\caption{$Re=40$}
	\end{subfigure}\hfil 
	\begin{subfigure}{0.3\textwidth}
		\includegraphics[width=\linewidth]{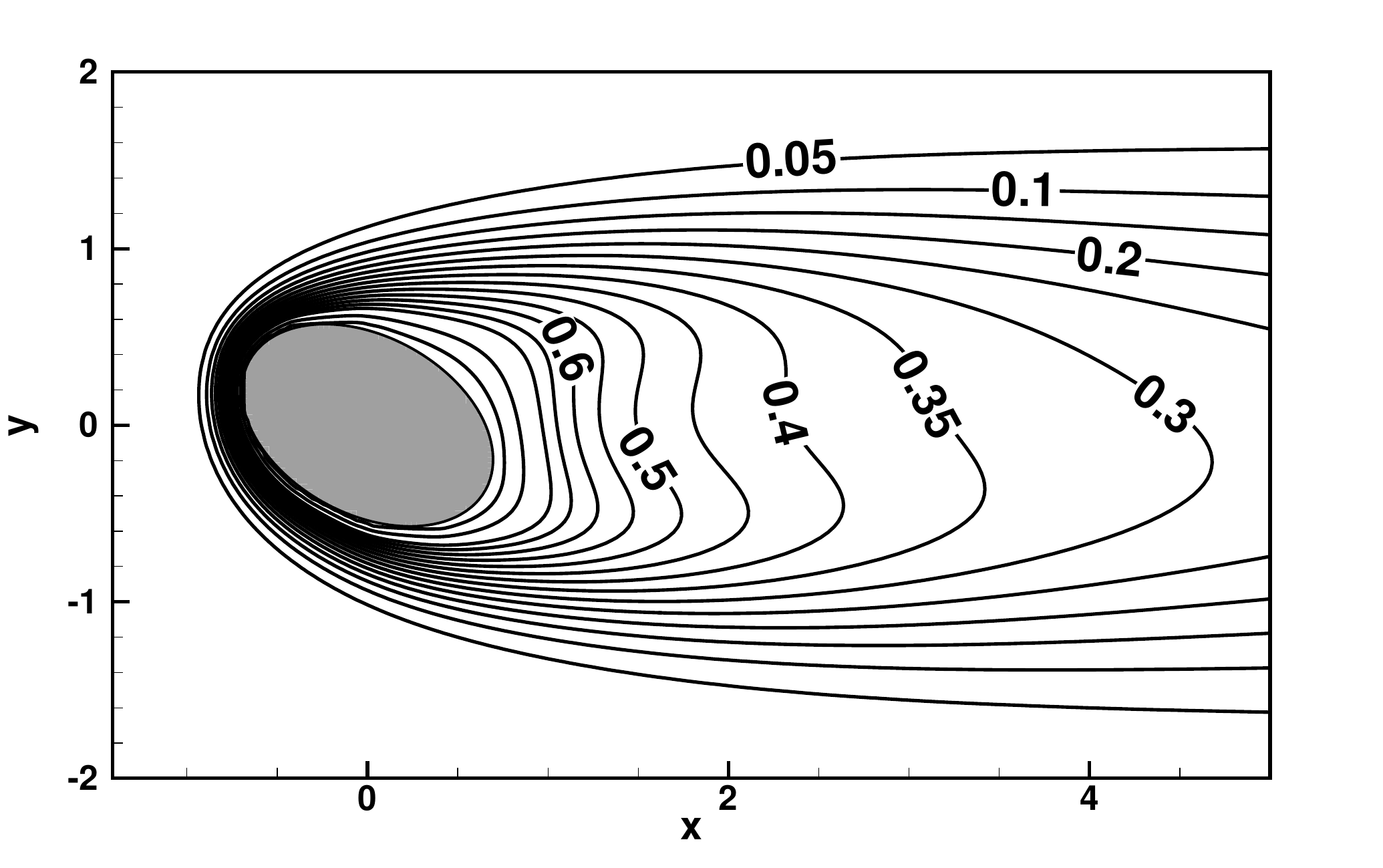} 
		\caption{$Re=49$}		
	\end{subfigure}\hfil 
	\caption{\small{Steady state isotherms for $\theta=30^{\degree}$ and (a)$Re=10$, (b)$Re=20$, (c)$Re=30$, (d)$Re=40$, and (e)$Re=49$.}}
	\label{Fig:T-steady-30deg}
\end{figure}

Steady state streamlines and isotherms for $\theta=45^{\degree}$ are shown in figures \ref{Fig:psi-steady-45deg} and \ref{Fig:T-steady-45deg} respectively where the $Re_{c}$ is in the range $38 \leq Re < 39$. Here, in sharp contrast to the previous two cases, flow separation occurs at $Re = 10$ and we see the appearance of a recirculation region on the upper surface of the cylinder (figure \ref{Fig:psi-steady-45deg} (a)). Also, distortions in the isotherms appear at a lesser $Re$ (figure \ref{Fig:T-steady-45deg} (c)) compared to the previous two cases. The evolution of streamlines follow a similar pattern - the size and strength of the vortices increase with $Re$. However, note that the value of $Re_c$ decreases when $\theta$ is increased.

\begin{figure}[H]
	\centering
	\begin{subfigure}{0.25\textwidth}
		\includegraphics[width=\linewidth]{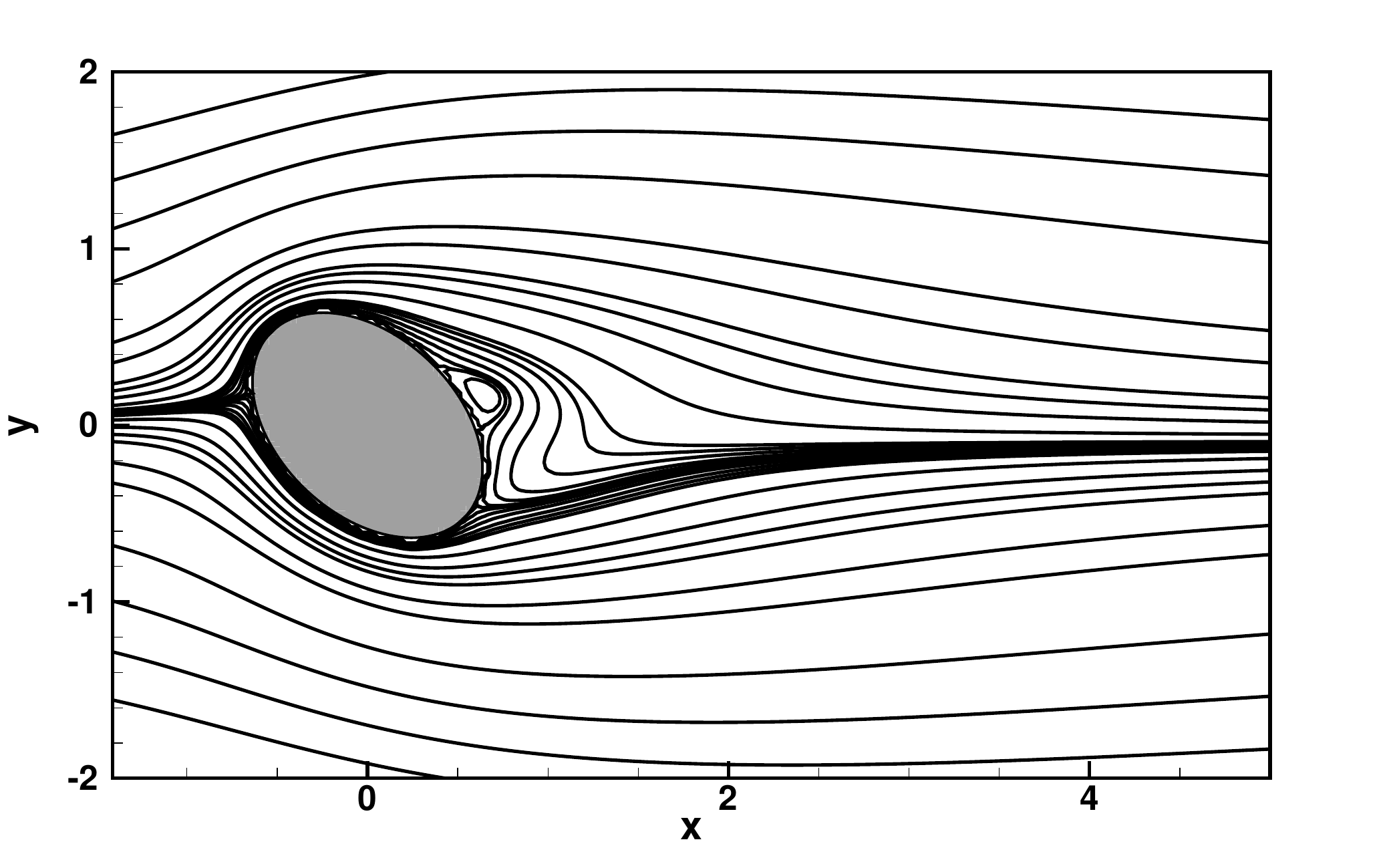} 
		\caption{$Re=10$}
	\end{subfigure}\hfil 
	\begin{subfigure}{0.25\textwidth}
		\includegraphics[width=\linewidth]{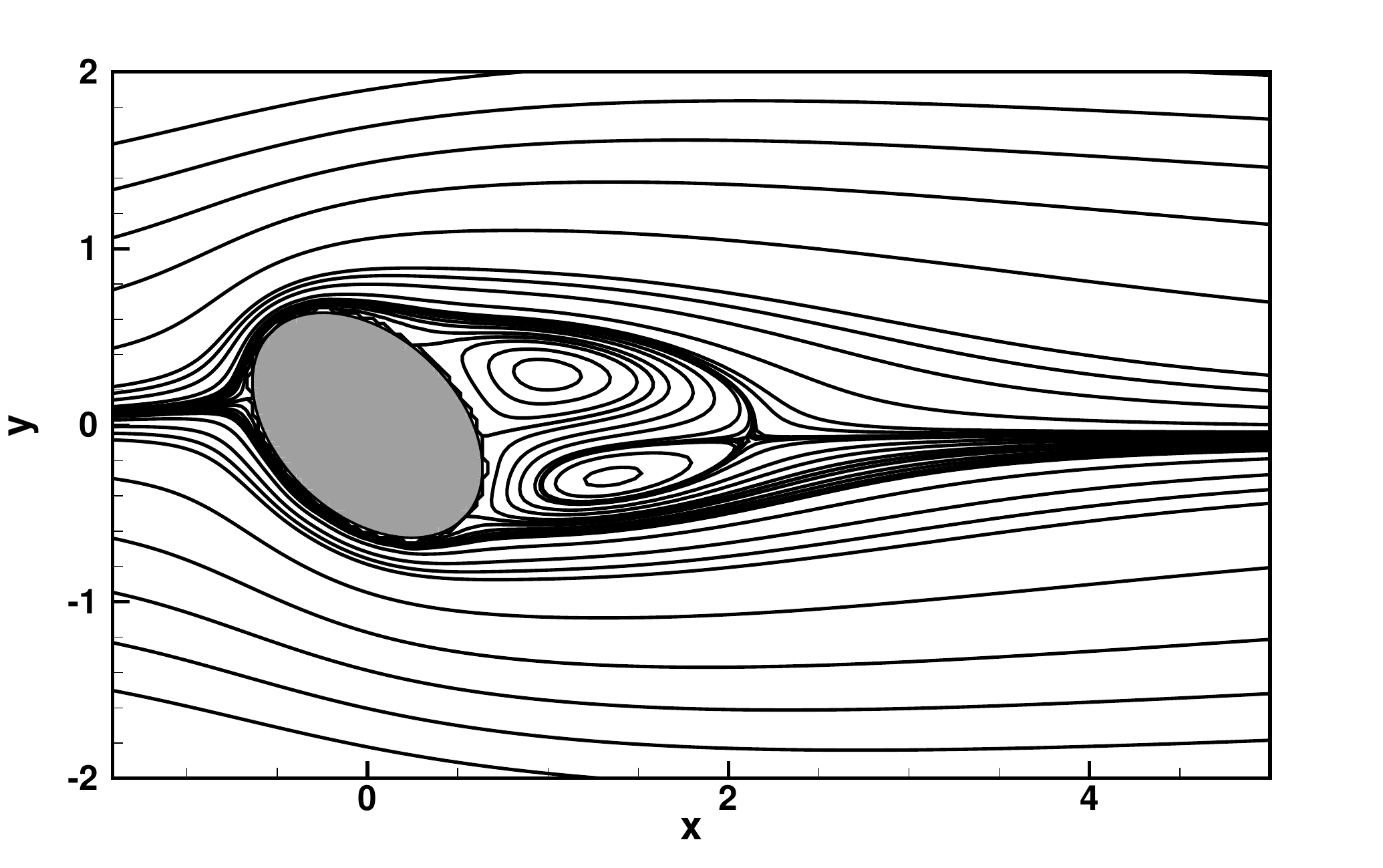} 
		\caption{$Re=20$}
	\end{subfigure}\hfil 
	\begin{subfigure}{0.25\textwidth}
		\includegraphics[width=\linewidth]{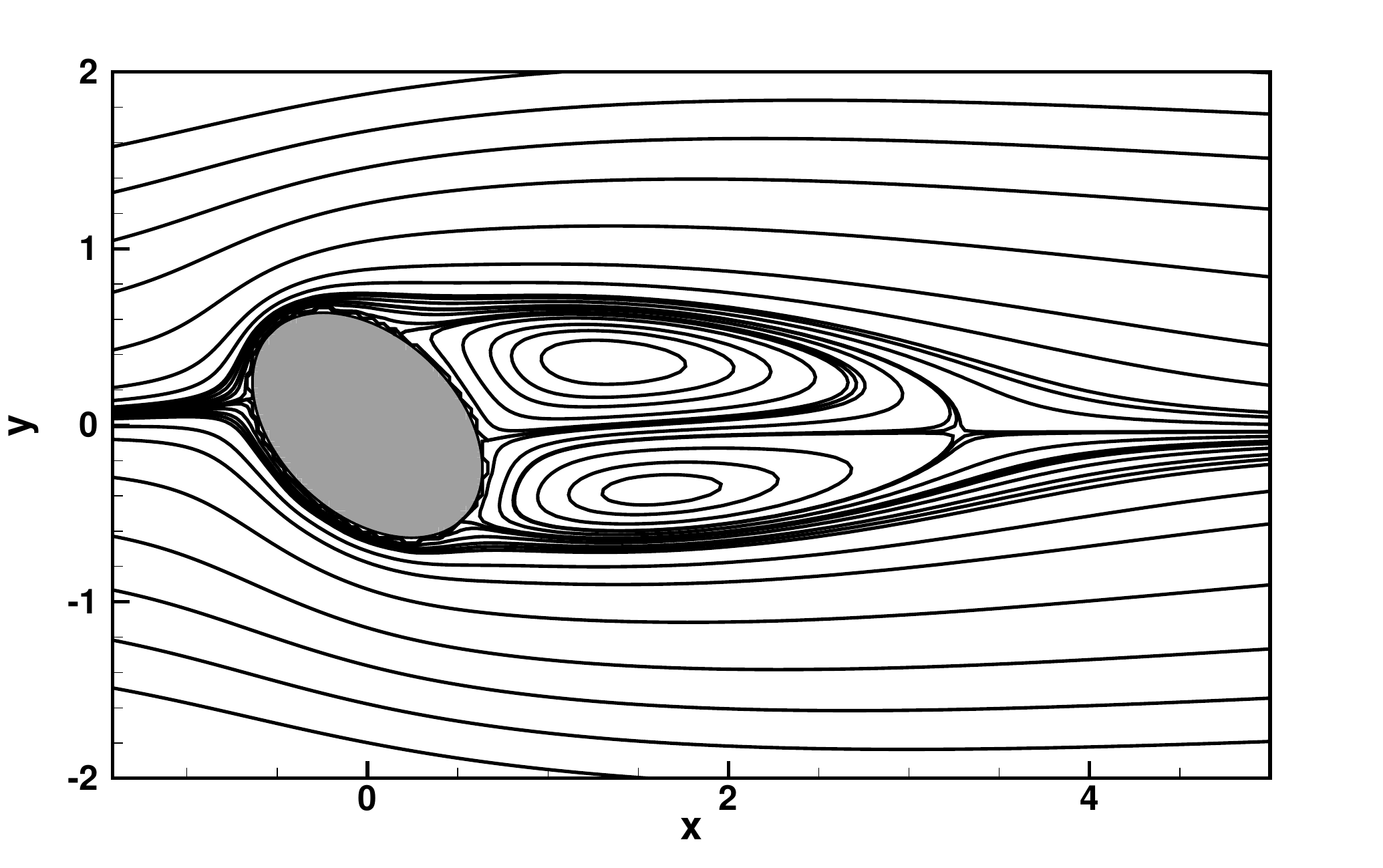} 
		\caption{$Re=30$}
	\end{subfigure}\hfil 
	\begin{subfigure}{0.25\textwidth}
		\includegraphics[width=\linewidth]{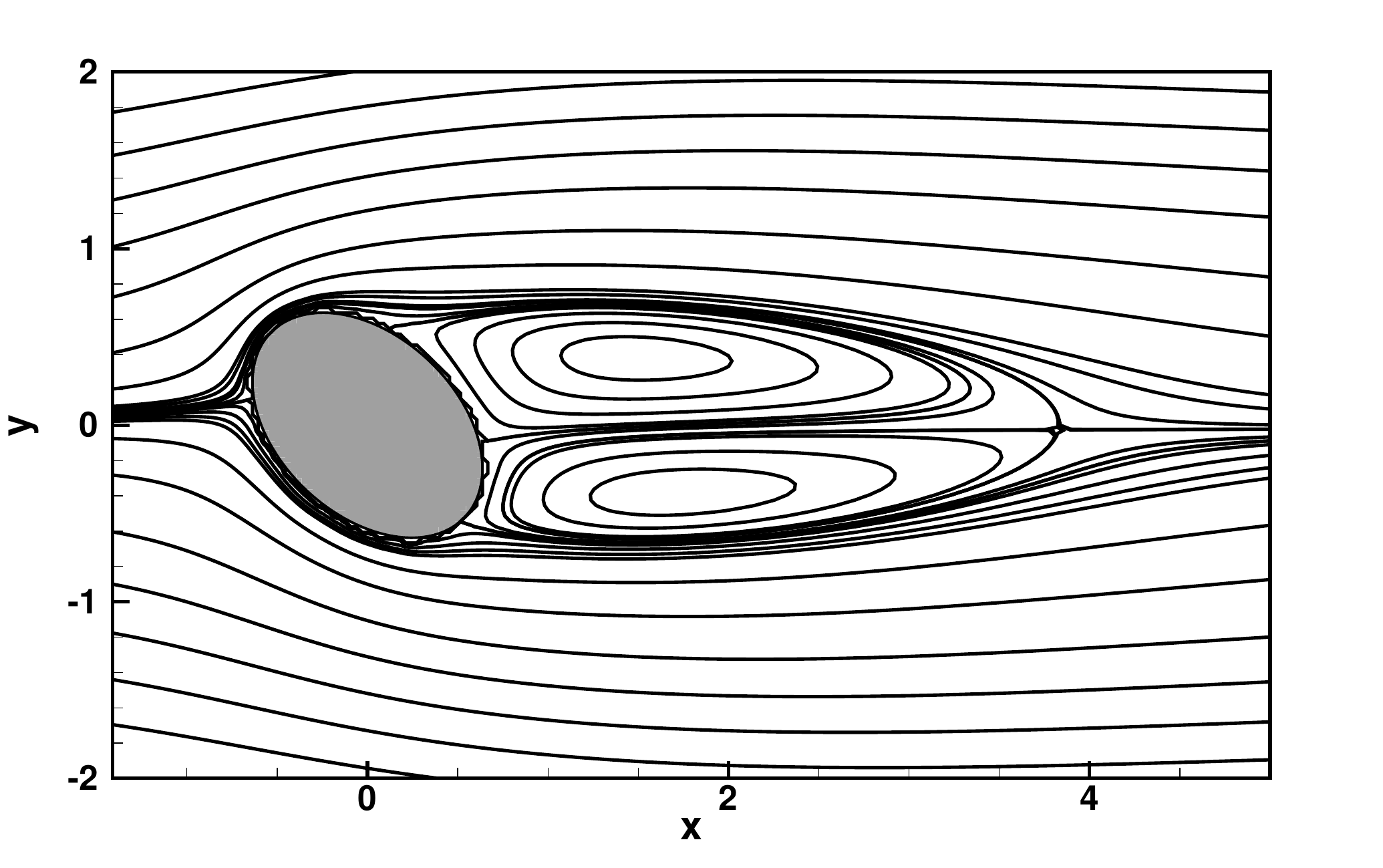} 
		\caption{$Re=38$}
	\end{subfigure}\hfil 
	\caption{\small{Steady state streamlines for $\theta=45^{\degree}$ and (a)$Re=10$, (b)$Re=20$, (c)$Re=30$, and (d)$Re=38$.}}
	\label{Fig:psi-steady-45deg}
\end{figure}

\begin{figure}[H]
	\centering
	\begin{subfigure}{0.25\textwidth}
		\includegraphics[width=\linewidth]{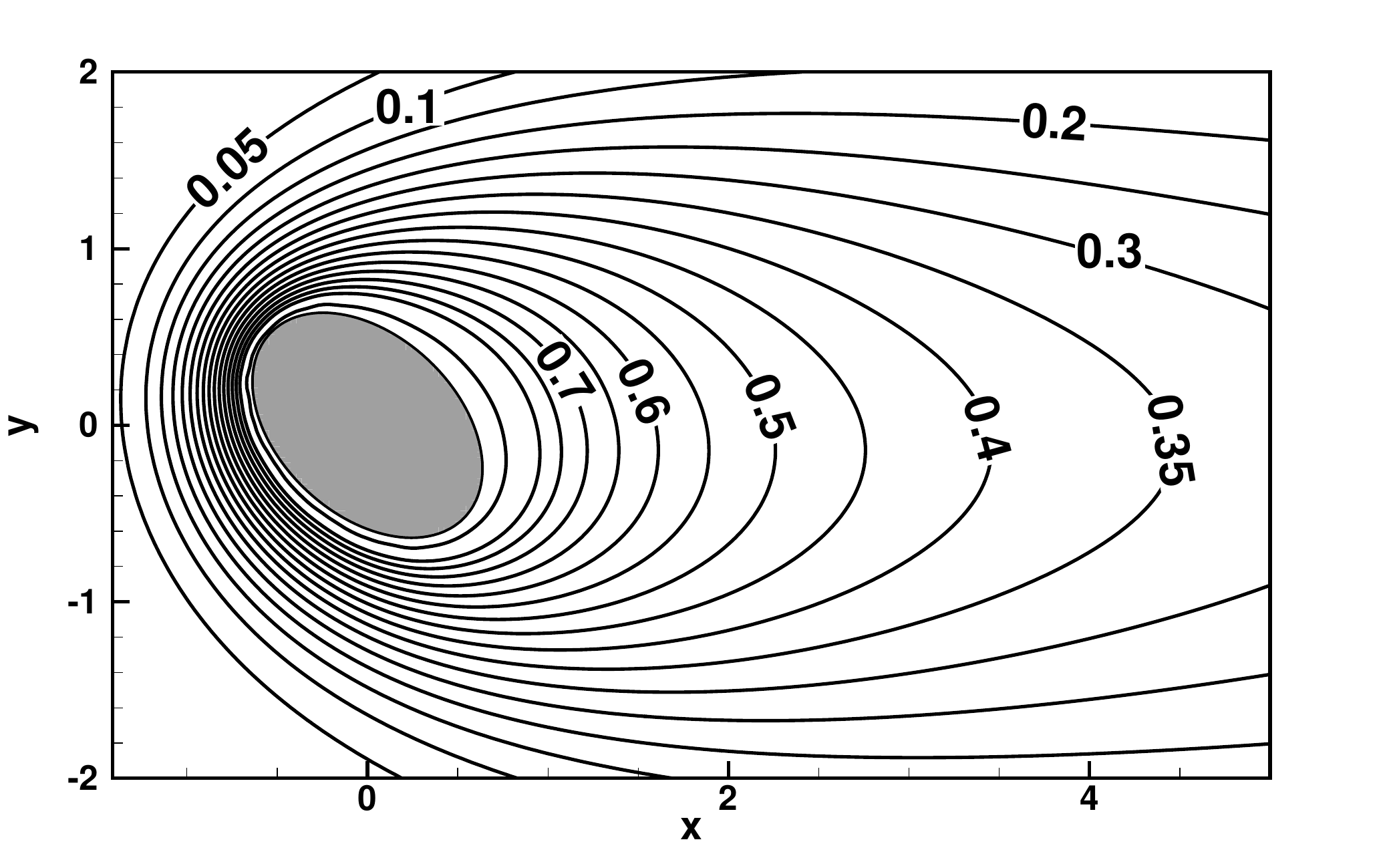} 
		\caption{$Re=10$}
	\end{subfigure}\hfil 
	\begin{subfigure}{0.25\textwidth}
		\includegraphics[width=\linewidth]{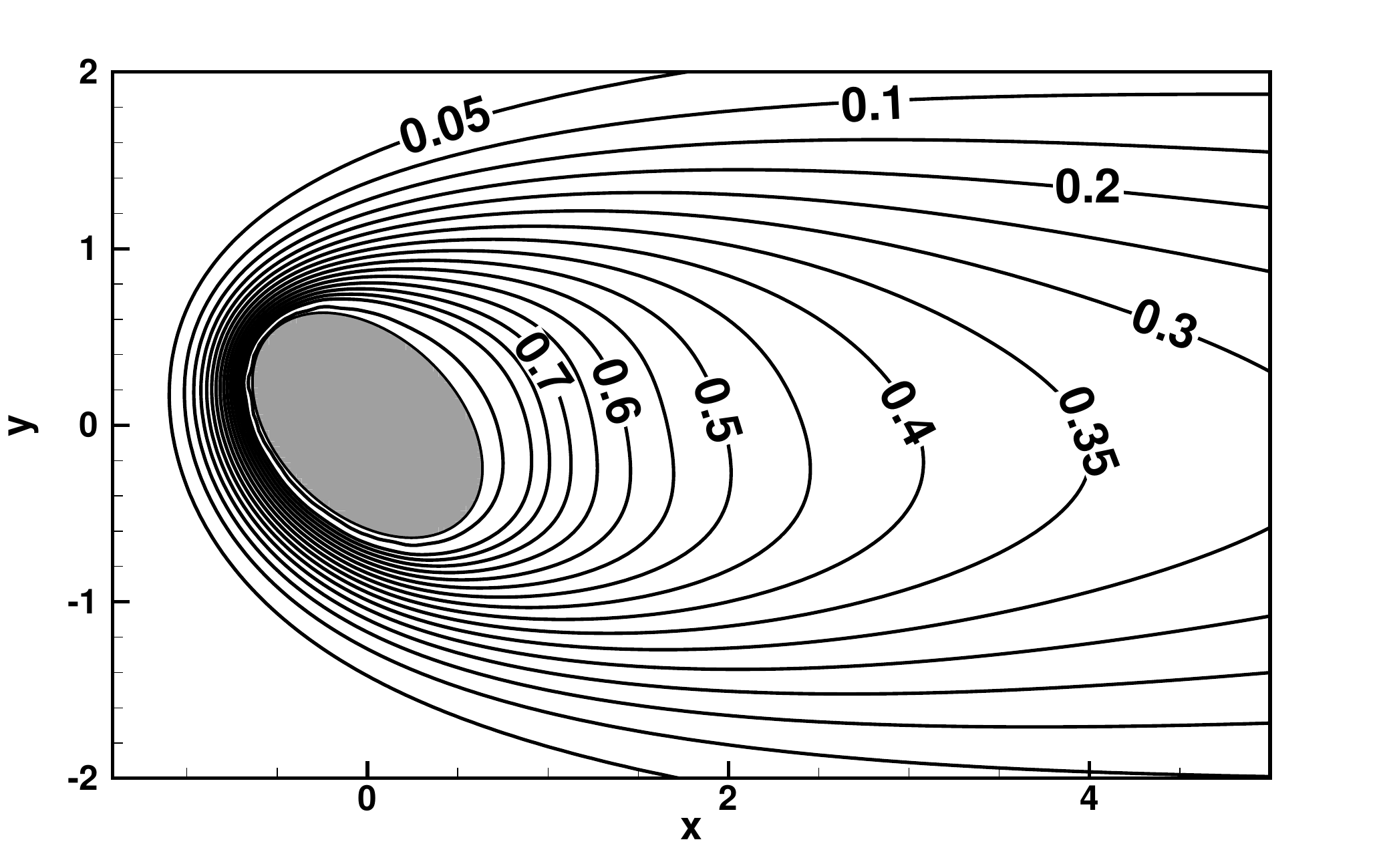} 
		\caption{$Re=20$}
	\end{subfigure}\hfil 
	\begin{subfigure}{0.25\textwidth}
		\includegraphics[width=\linewidth]{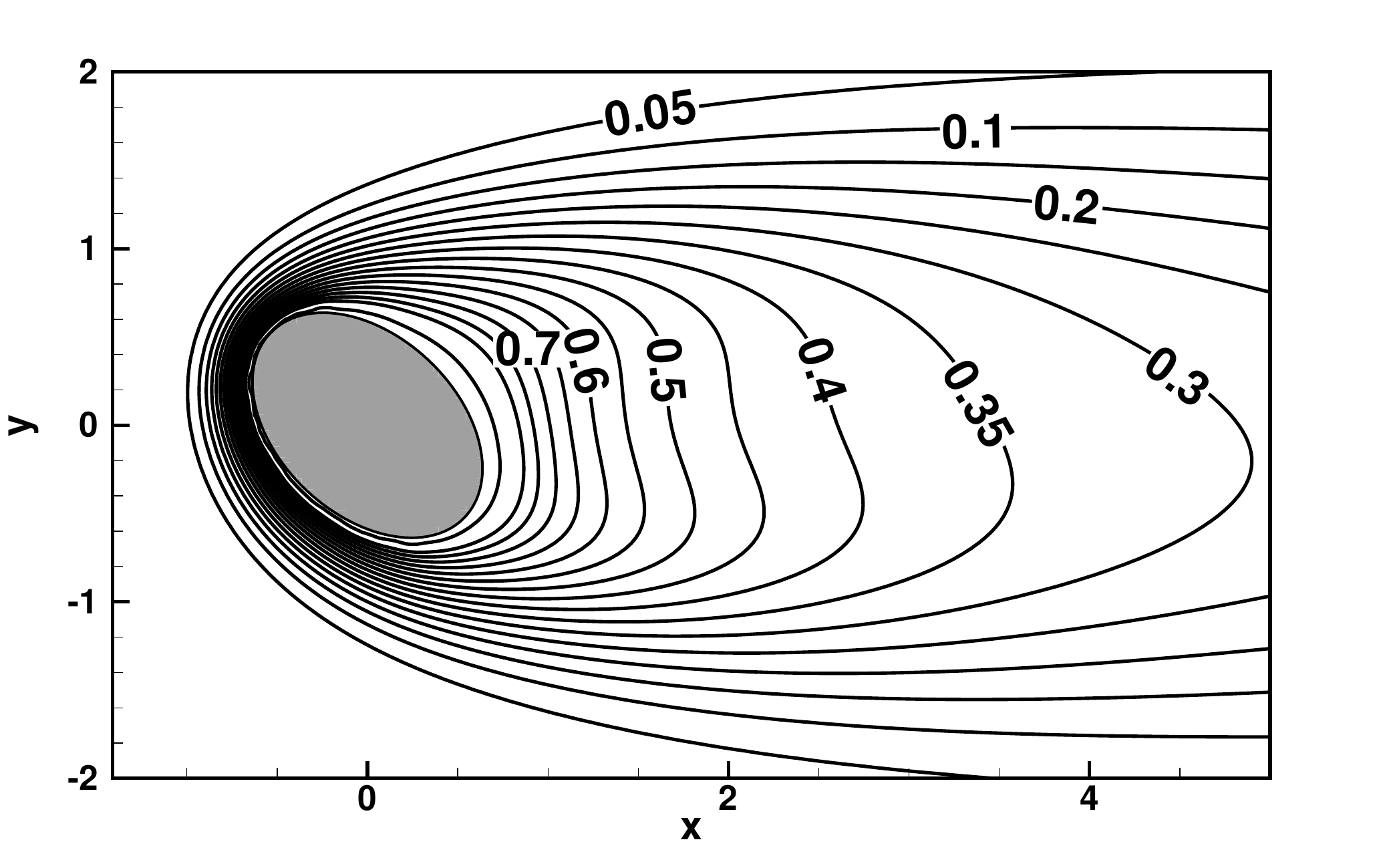} 
		\caption{$Re=30$}
	\end{subfigure}\hfil 
	\begin{subfigure}{0.25\textwidth}
		\includegraphics[width=\linewidth]{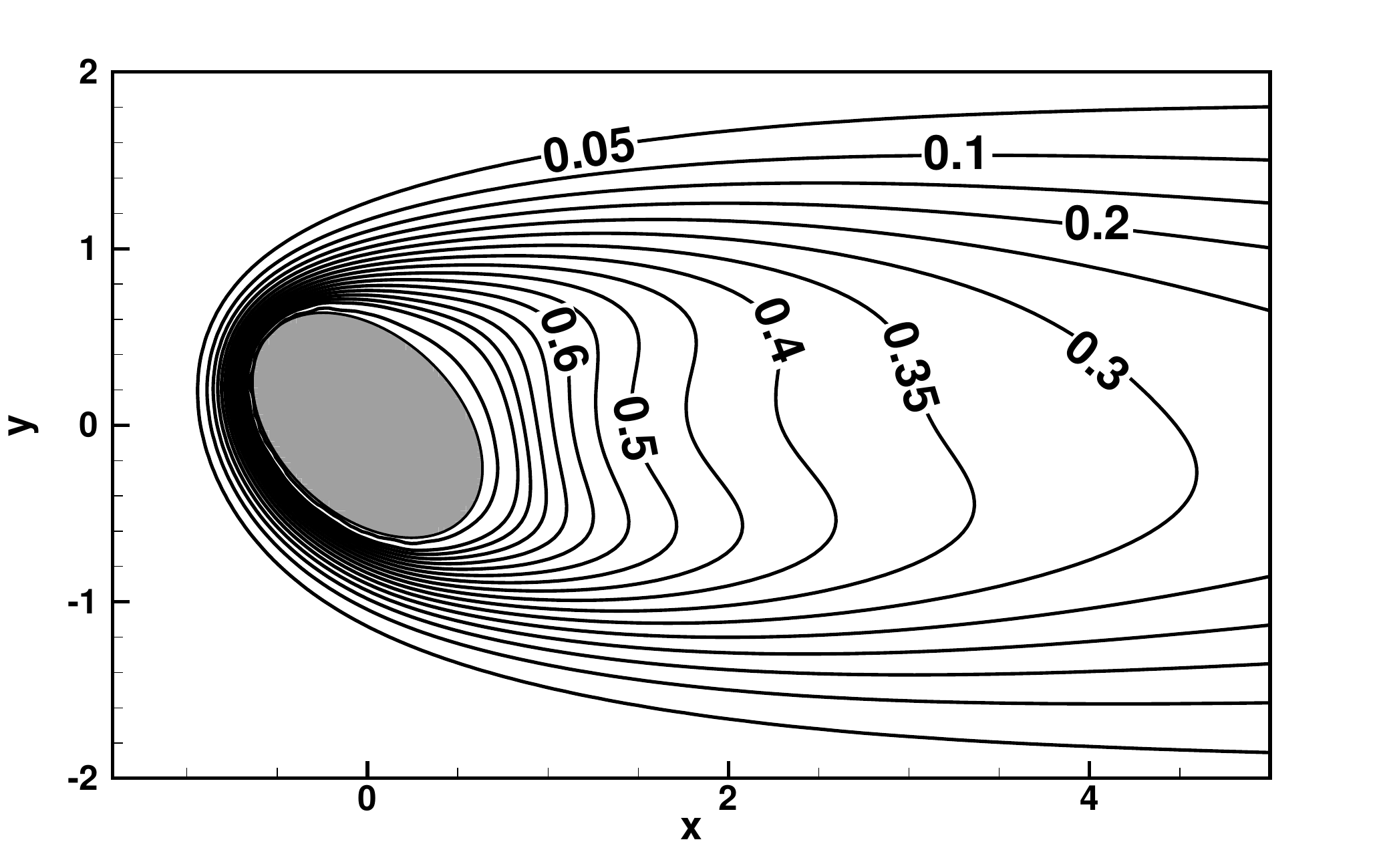} 
		\caption{$Re=38$}
	\end{subfigure}\hfil 
	\caption{\small{Steady state isotherms for $\theta=45^{\degree}$ and (a)$Re=10$, (b)$Re=20$, (c)$Re=30$, and (d)$Re=38$.}}
	\label{Fig:T-steady-45deg}
\end{figure}

Figures \ref{Fig:psi-steady-60deg} and \ref{Fig:T-steady-60deg} show the steady state streamlines and isotherms respectively for $\theta = 60^{\degree}$. Here, the $Re_c$ is in the range $31 \leq Re < 32$. At $Re = 10$ (figure \ref{Fig:psi-steady-60deg} (a)), the recirculation region that formed at $\theta = 45^{\degree}$ increases in size. As seen previously there is a gradual increase in the sizes of the vortices formed on the surfaces of the cylinder as $Re$ is increased, and the value of $Re_c$ also drops to $Re = 31$ at $\theta = 60^{\degree}$. One can also notice that the flow is gradually becoming symmetric as $\theta$ is increased. The distortions in the isotherms appear at a much lower $Re$ (figure \ref{Fig:T-steady-60deg} (b)) than for $\theta = 45^{\degree}$.
\begin{figure}[H]
	\centering
	\begin{subfigure}{0.25\textwidth}
		\includegraphics[width=\linewidth]{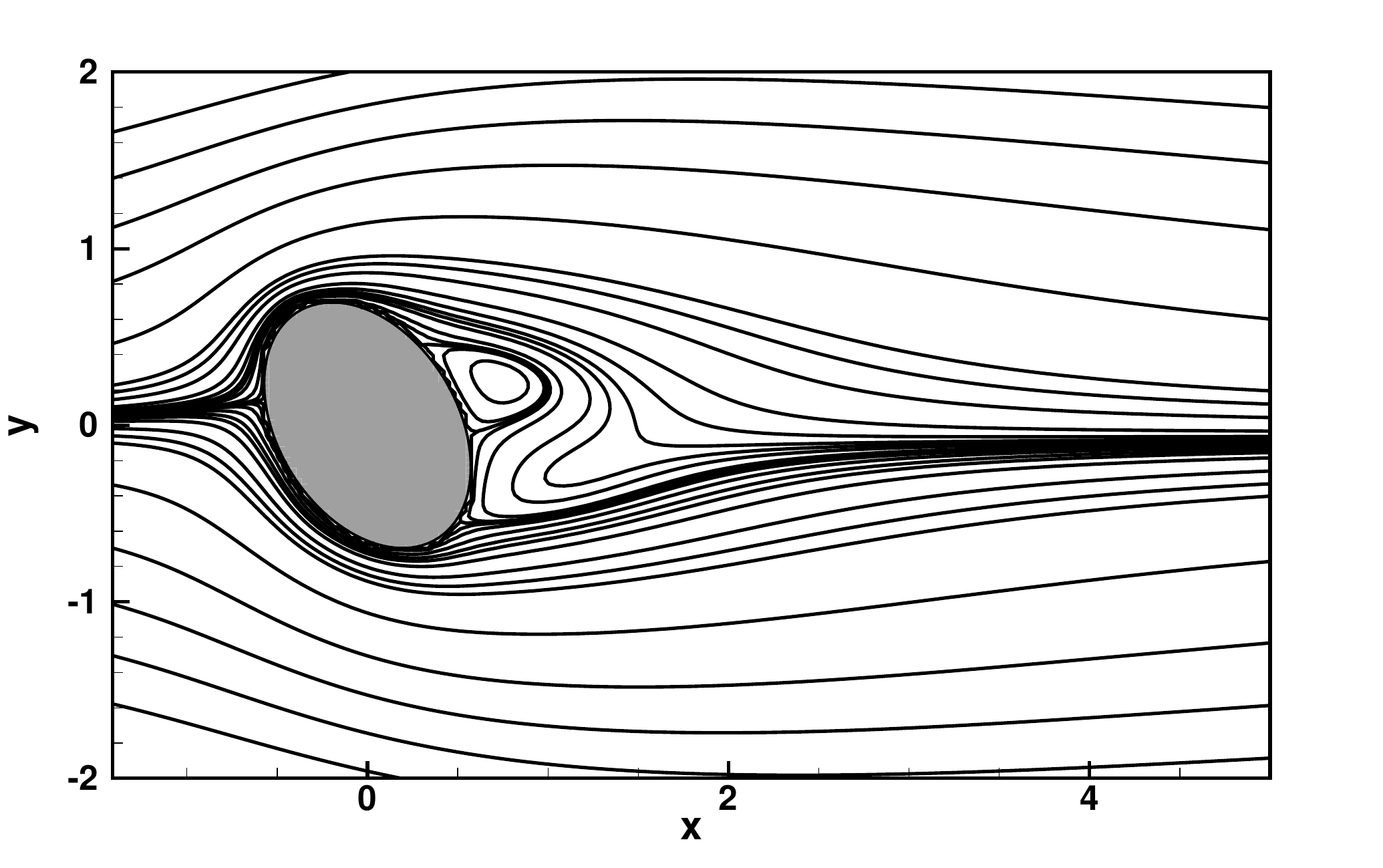} 
		\caption{$Re=10$}
	\end{subfigure}\hfil 
	\begin{subfigure}{0.25\textwidth}
		\includegraphics[width=\linewidth]{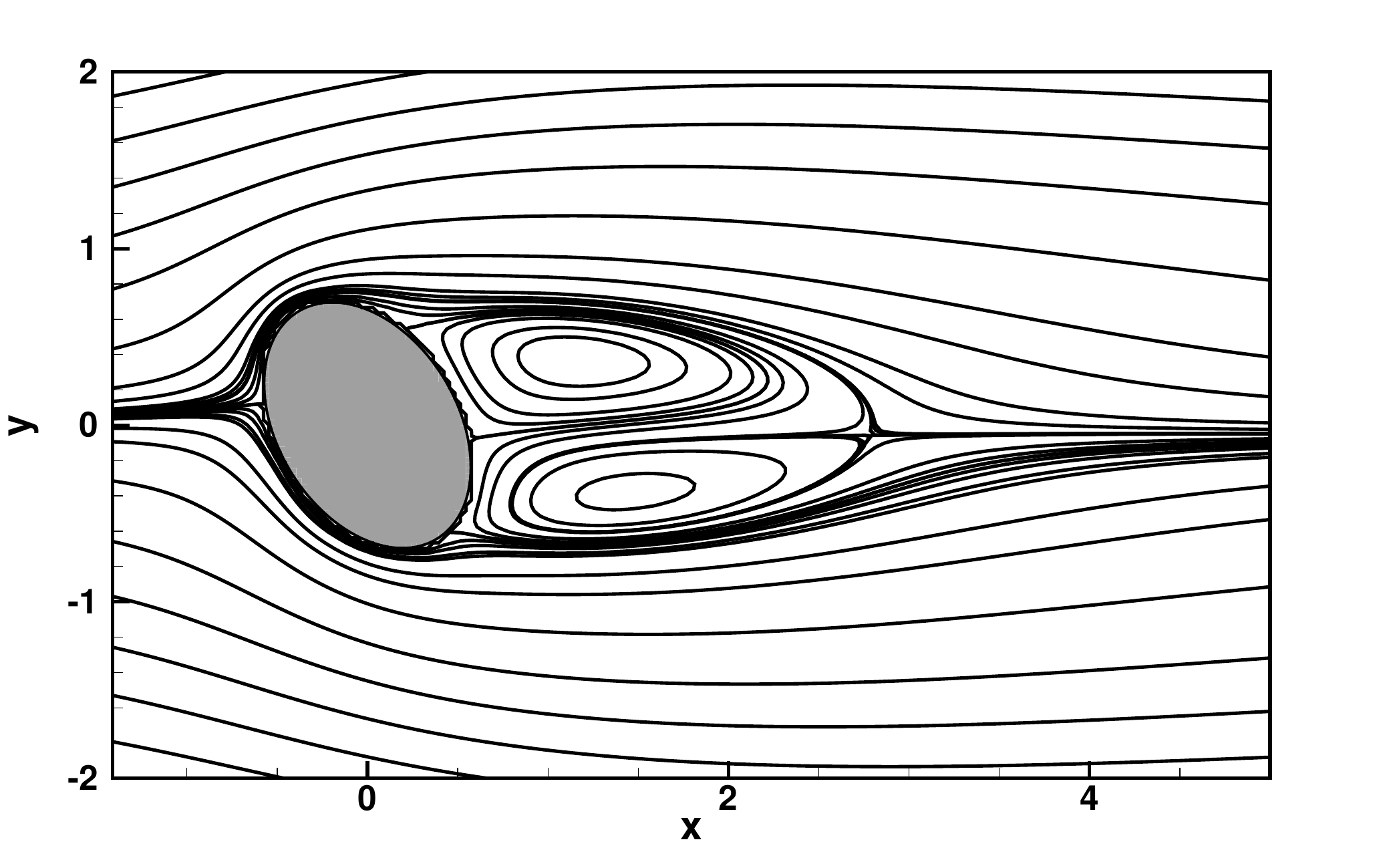} 
		\caption{$Re=20$}
	\end{subfigure}\hfil 
	\begin{subfigure}{0.25\textwidth}
		\includegraphics[width=\linewidth]{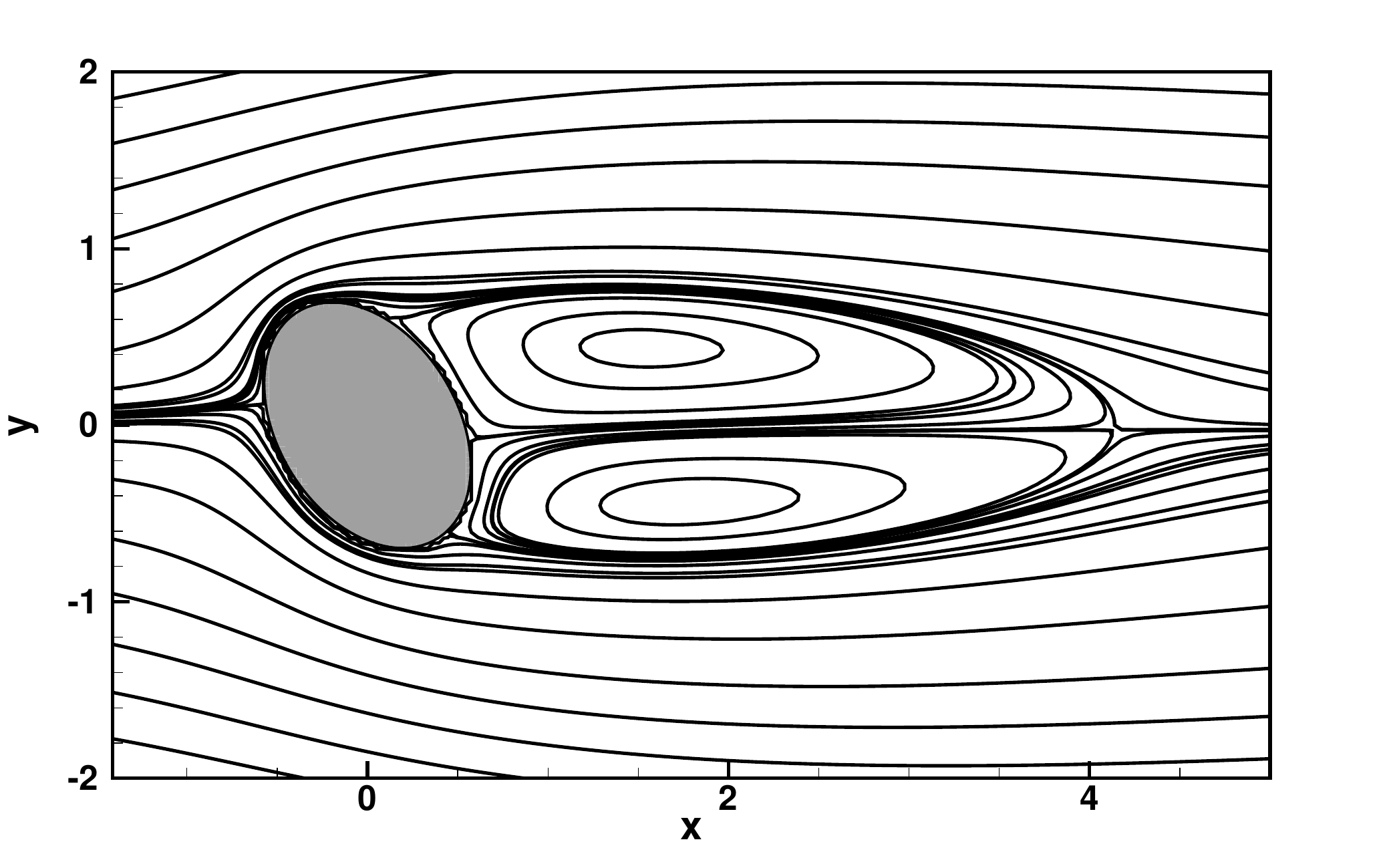} 
		\caption{$Re=30$}
	\end{subfigure}\hfil 
	\begin{subfigure}{0.25\textwidth}
		\includegraphics[width=\linewidth]{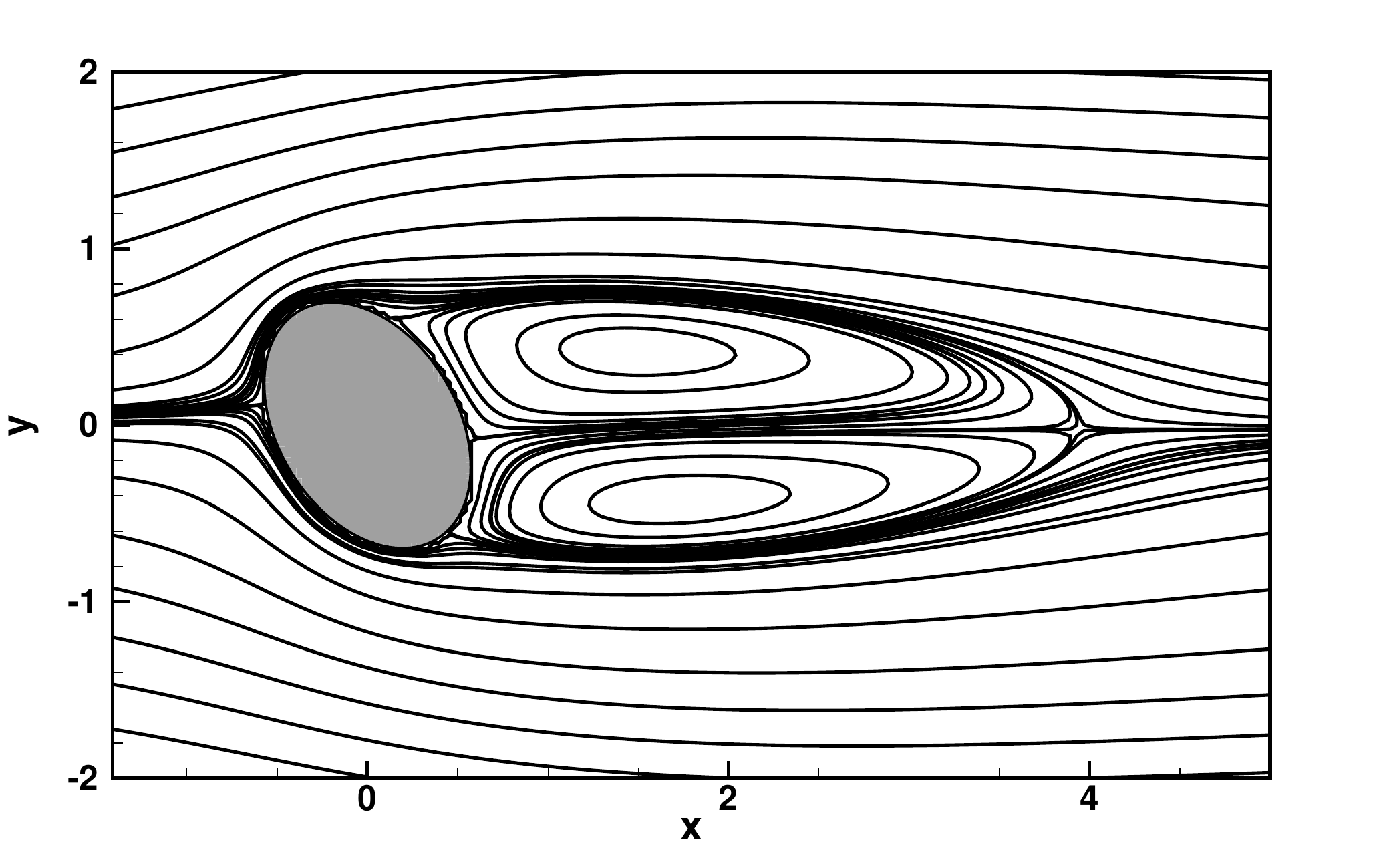} 
		\caption{$Re=31$}
	\end{subfigure}\hfil 
	\caption{\small{Steady state streamlines for $\theta=60^{\degree}$ and (a)$Re=10$, (b)$Re=20$, (c)$Re=30$, and (d)$Re=31$.}}
	\label{Fig:psi-steady-60deg}
\end{figure}

\begin{figure}[H]
	\centering
	\begin{subfigure}{0.25\textwidth}
		\includegraphics[width=\linewidth]{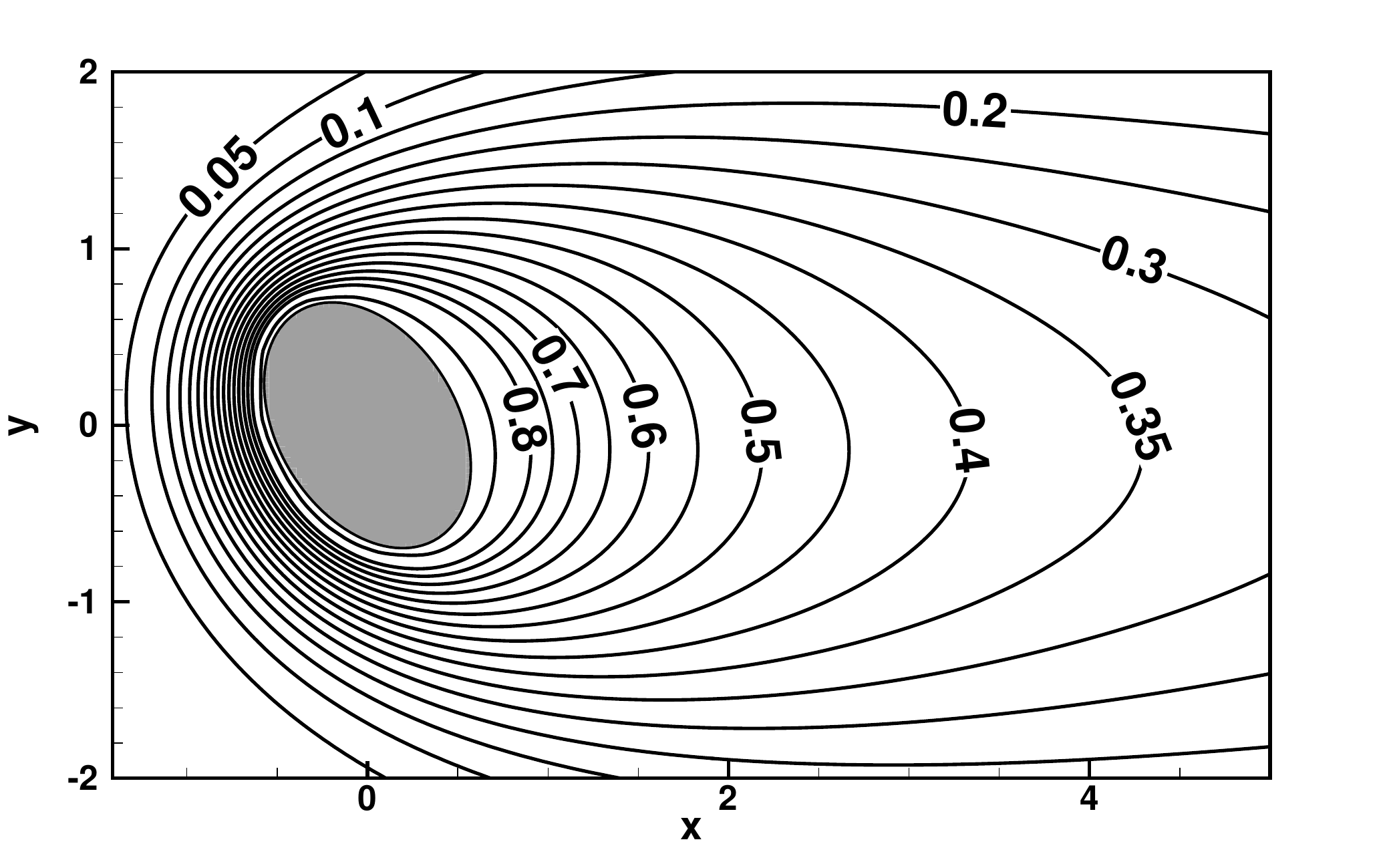} 
		\caption{$Re=10$}
	\end{subfigure}\hfil 
	\begin{subfigure}{0.25\textwidth}
		\includegraphics[width=\linewidth]{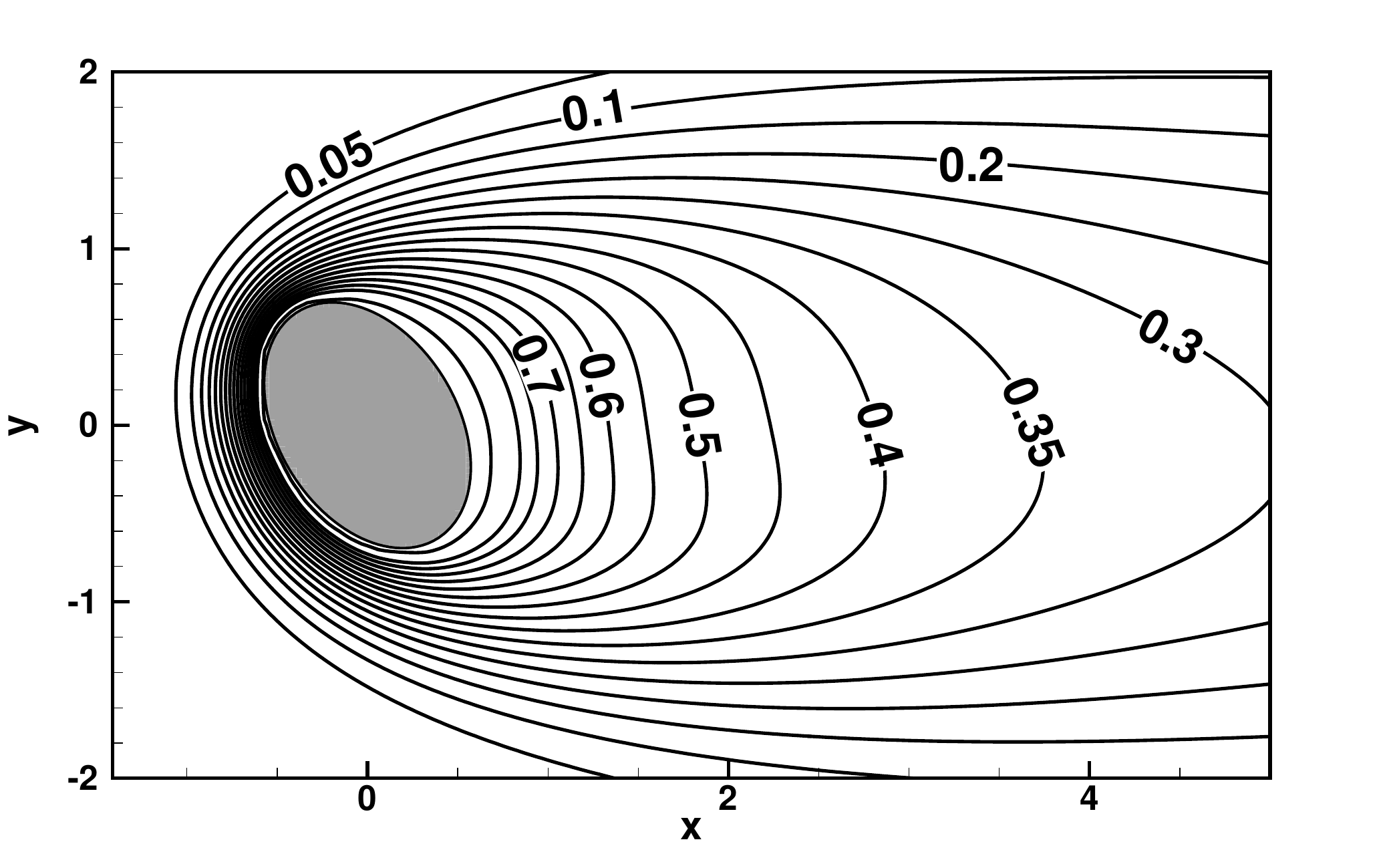} 
		\caption{$Re=20$}
	\end{subfigure}\hfil 
	\begin{subfigure}{0.25\textwidth}
		\includegraphics[width=\linewidth]{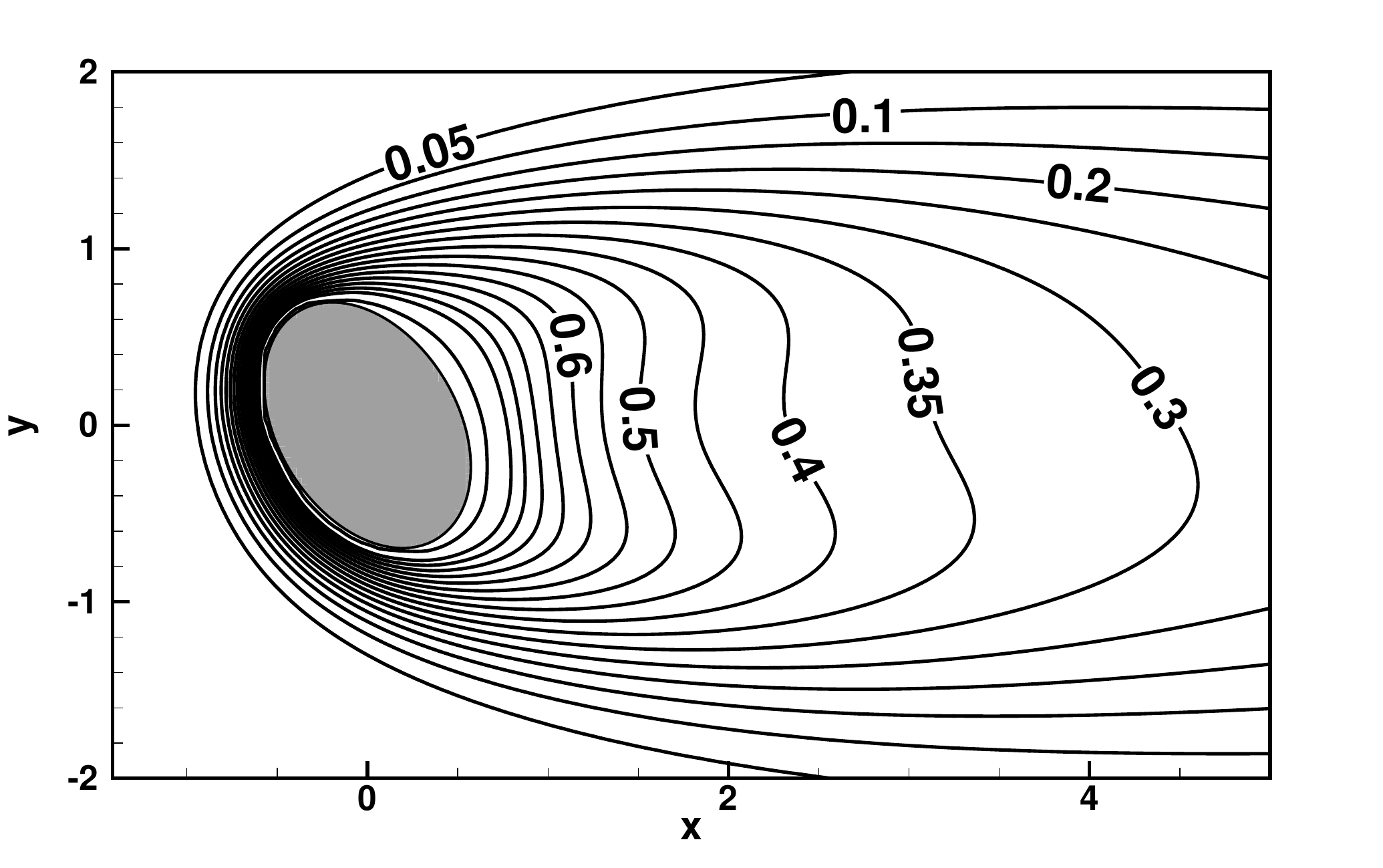} 
		\caption{$Re=30$}
	\end{subfigure}\hfil 
	\begin{subfigure}{0.25\textwidth}
		\includegraphics[width=\linewidth]{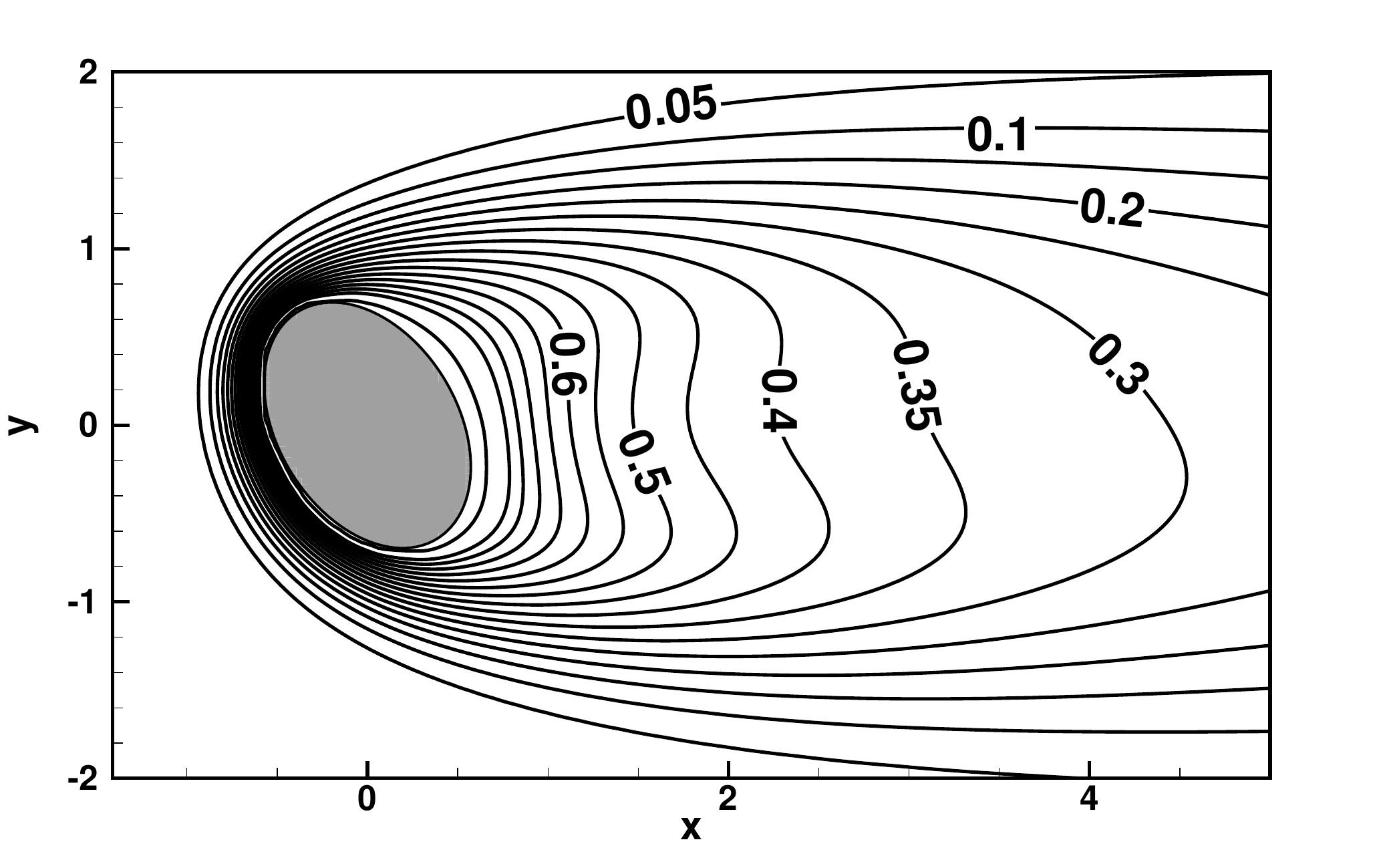} 
		\caption{$Re=31$}
	\end{subfigure}\hfil 
	\caption{\small{Steady state streamlines for $\theta=60^{\degree}$ and (a)$Re=10$, (b)$Re=20$, (c)$Re=30$, and (d)$Re=31$.}}
	\label{Fig:T-steady-60deg}
\end{figure}

Steady state streamlines and isotherms for $\theta = 75^{\degree}$ are shown in figures \ref{Fig:psi-steady-75deg} and \ref{Fig:T-steady-75deg} respectively. The $Re_c$ in this case is in the range $28 \leq Re < 29$. At $Re = 10$, we observe the formation of two recirculation regions on the surface of the cylinder as opposed to only one for $\theta = 45^{\degree}$, $60^{\degree}$ and none for $\theta = 15^{\degree}$, $30^{\degree}$. The wake region appears nearly symmetric as $\theta$ is increased. This tendency of the flow to approach symmetry is observed in the isotherms as well. Distortions in the isotherms in case appears at $Re = 20$ (figure \ref{Fig:T-steady-75deg} (b)), which is the same as for $\theta = 60^{\degree}$, but a closer look reveals that the distortion seen at $\theta = 75^{\degree}$ is more pronounced than that observed at $\theta = 60^{\degree}$. 
\begin{figure}[H]
	\centering
	\begin{subfigure}{0.3\textwidth}
		\includegraphics[width=\linewidth]{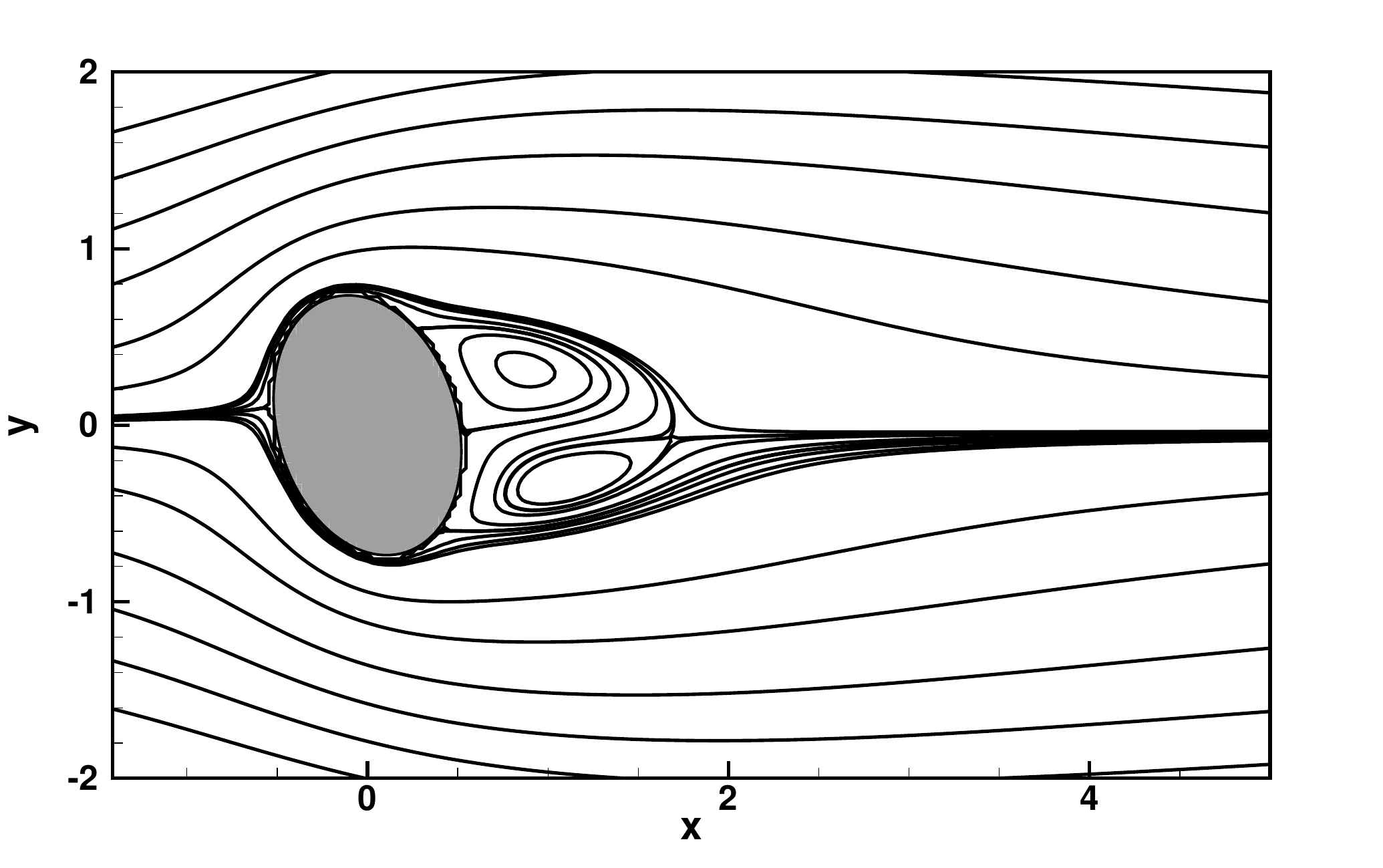} 
		\caption{$Re=10$}
	\end{subfigure}
	\begin{subfigure}{0.3\textwidth}
		\includegraphics[width=\linewidth]{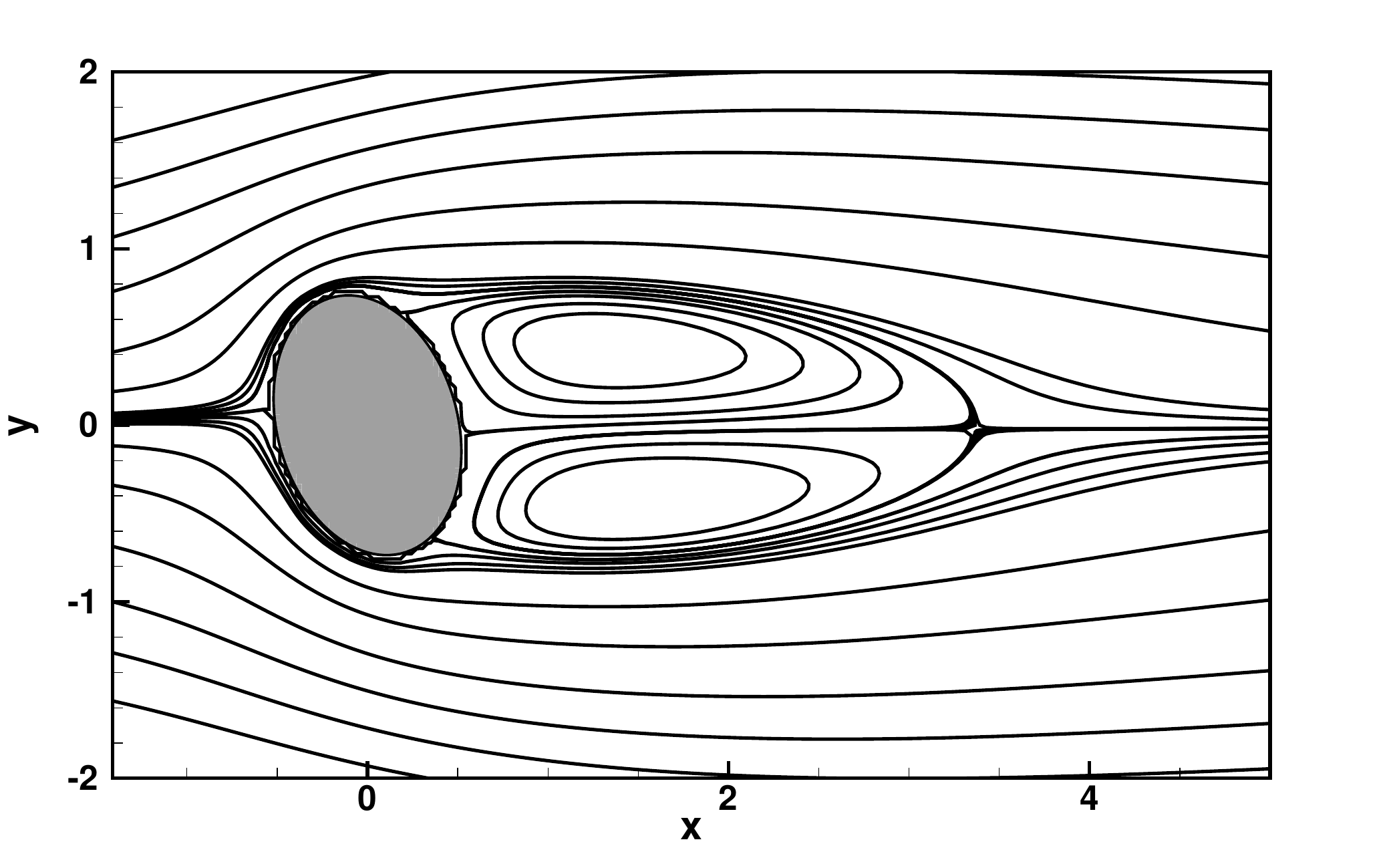} 
		\caption{$Re=20$}
	\end{subfigure}
		\begin{subfigure}{0.3\textwidth}
			\includegraphics[width=\linewidth]{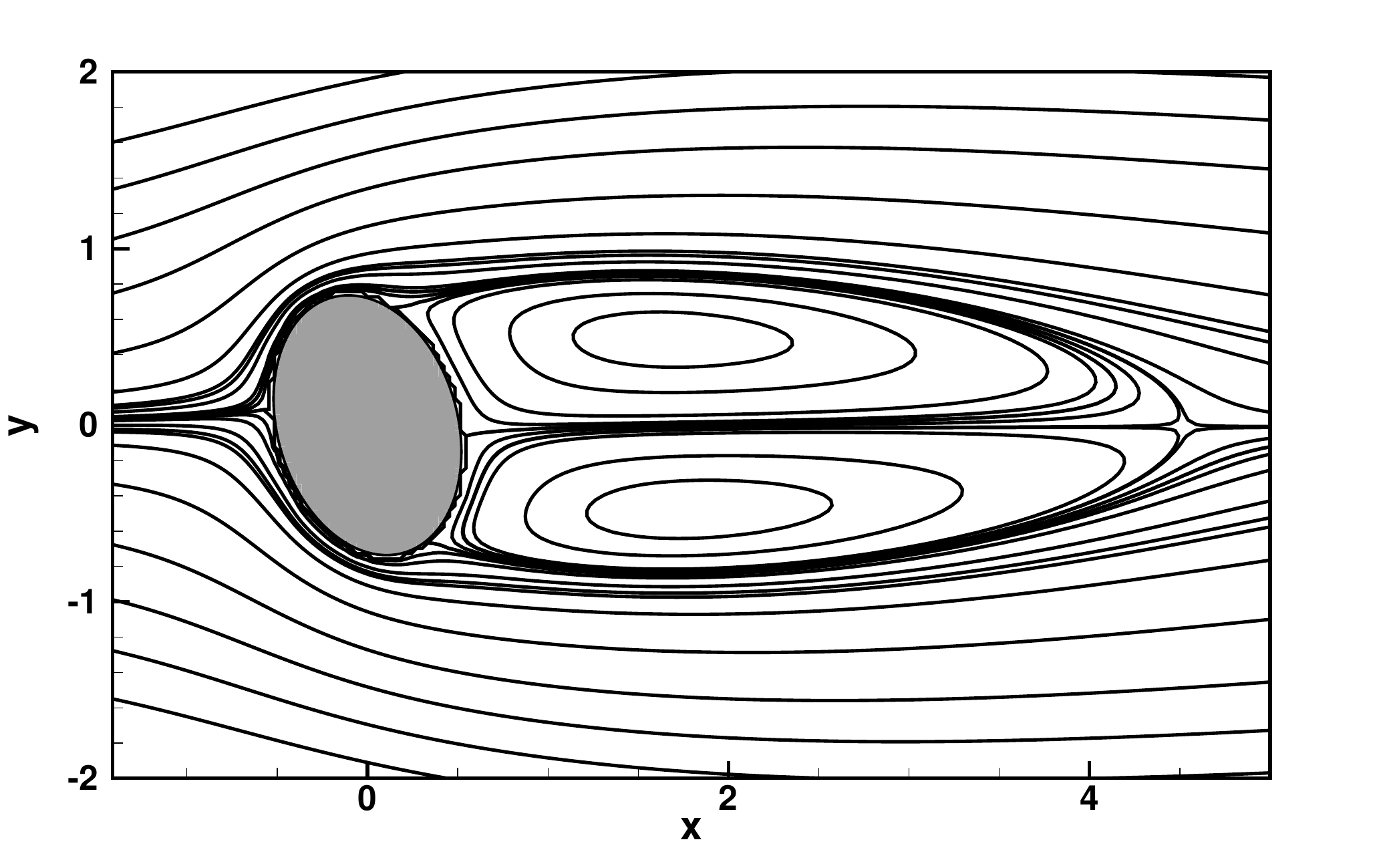} 
			\caption{$Re=28$}
		\end{subfigure} 
	\caption{\small{Steady state streamlines for $\theta=75^{\degree}$ and (a)$Re=10$, (b)$Re=20$, (c) $Re=28$.}}
	\label{Fig:psi-steady-75deg}
\end{figure}

\begin{figure}[H]
	\centering
	\begin{subfigure}{0.3\textwidth}
		\includegraphics[width=\linewidth]{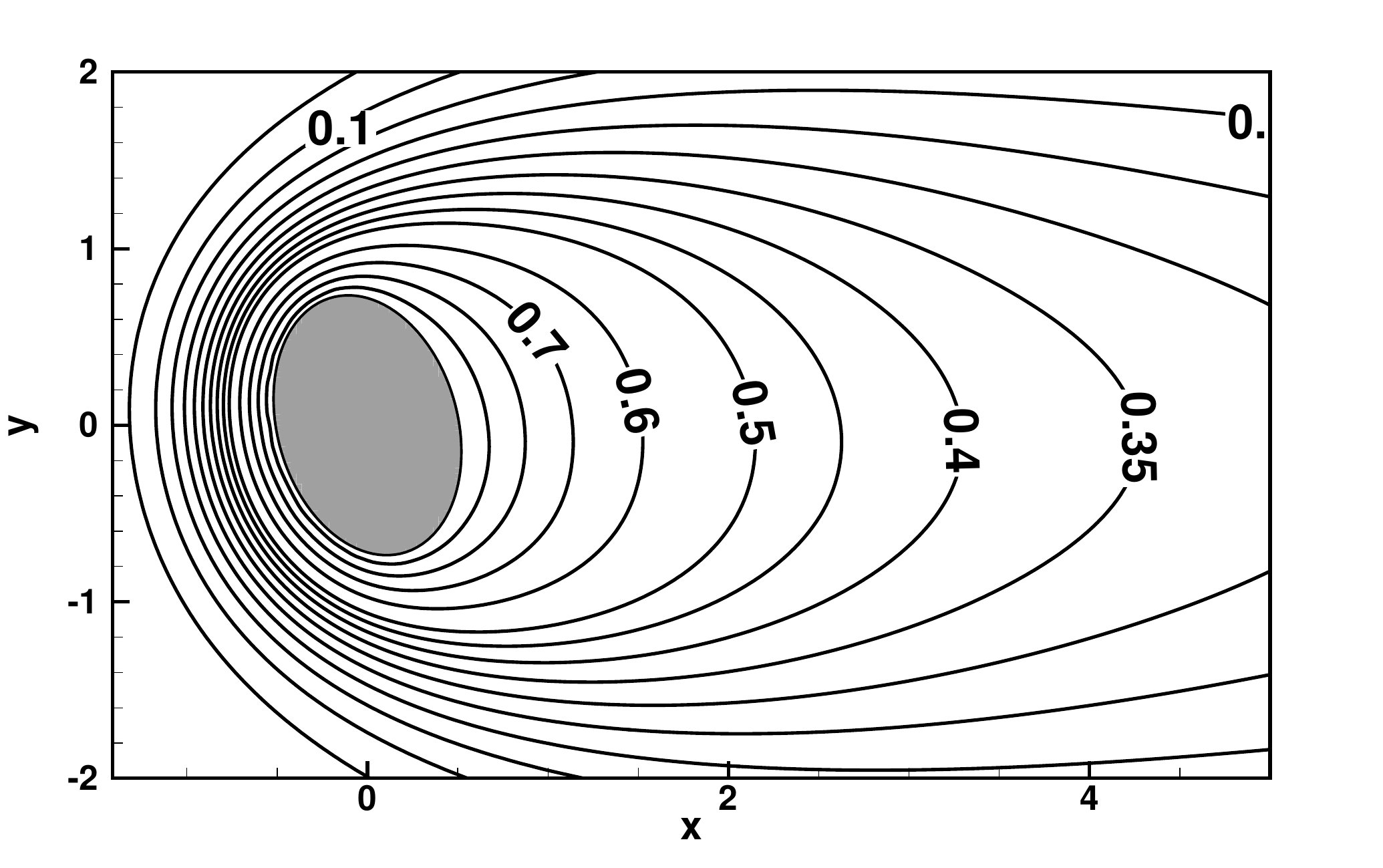} 
		\caption{$Re=10$}
	\end{subfigure}
	\begin{subfigure}{0.3\textwidth}
		\includegraphics[width=\linewidth]{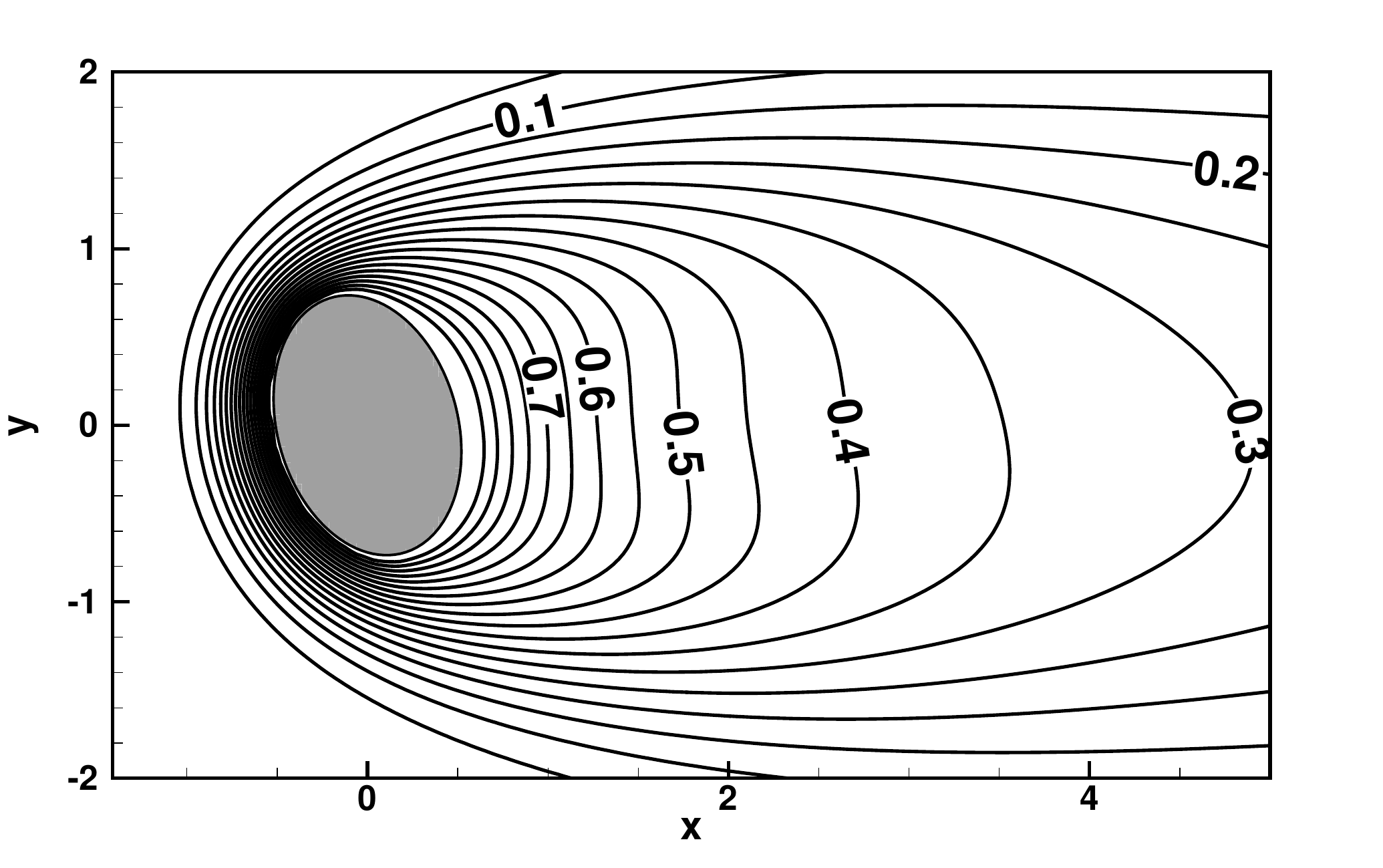} 
		\caption{$Re=20$}
	\end{subfigure}
\begin{subfigure}{0.3\textwidth}
	\includegraphics[width=\linewidth]{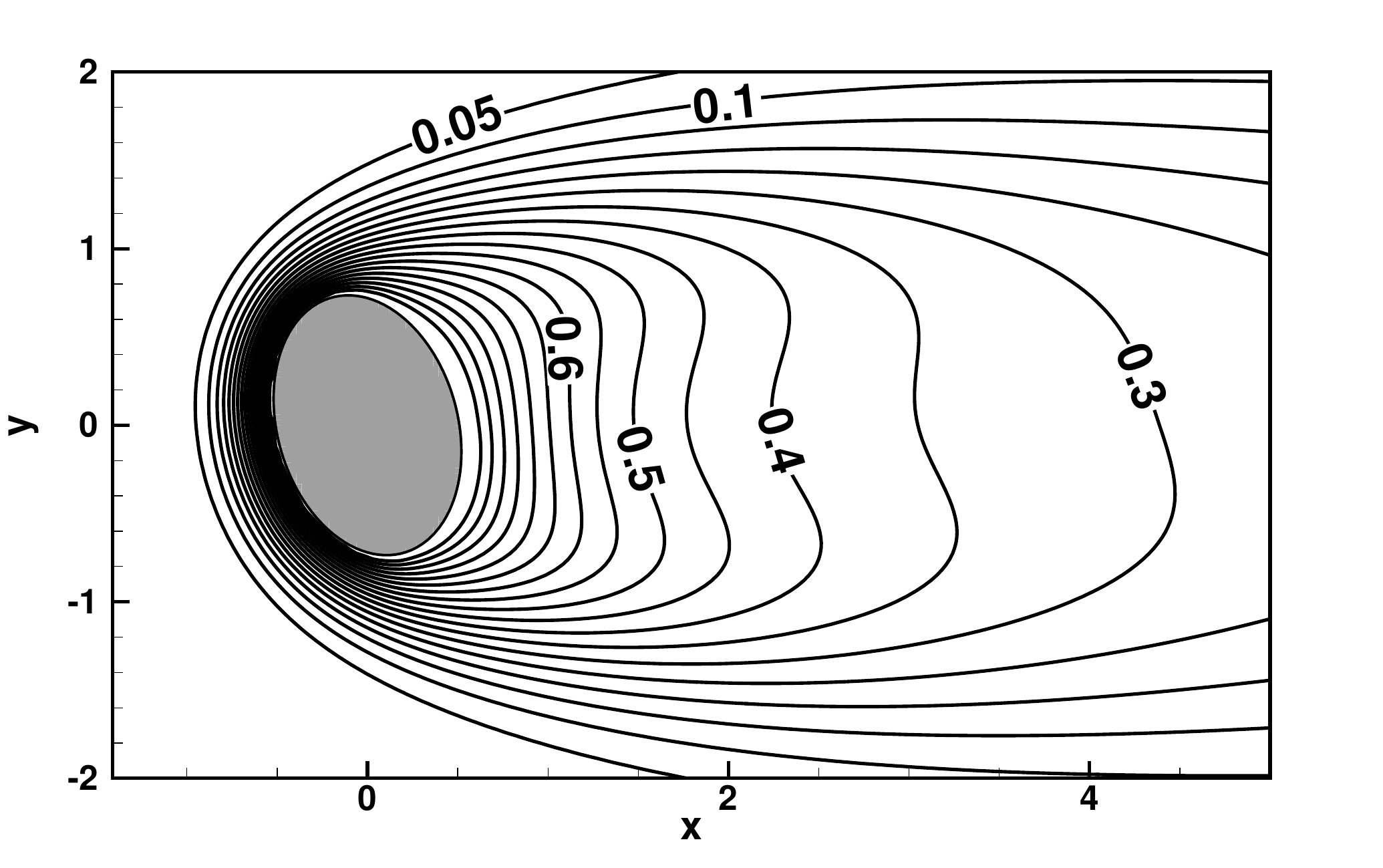} 
	\caption{$Re=28$}
\end{subfigure} 
	\caption{\small{Steady state isotherms for $\theta=75^{\degree}$ and (a)$Re=10$, (b)$Re=20$, (c) $Re=28$.}}
	\label{Fig:T-steady-75deg}
\end{figure}

As $\theta$ is increased to $90^{\degree}$ and incoming flow is symmetric w.r.t to the cylinder, we observe the flow becoming symmetric again in the cylinder wake, as evident from the streamlines and isotherms in  figures \ref{Fig:psi-steady-90deg} and \ref{Fig:isotherms-steady-90deg} respectively. The $Re_c$ is in the range $25 \leq Re < 26$. Again, distortion in the isotherms at $Re = 20$ is more pronounced than that observed at $\theta = 75^{\degree}$ (figure \ref{Fig:isotherms-steady-90deg}(b)).

\begin{figure}[H]
	\centering
	\begin{subfigure}{0.3\textwidth}
		\includegraphics[width=\linewidth]{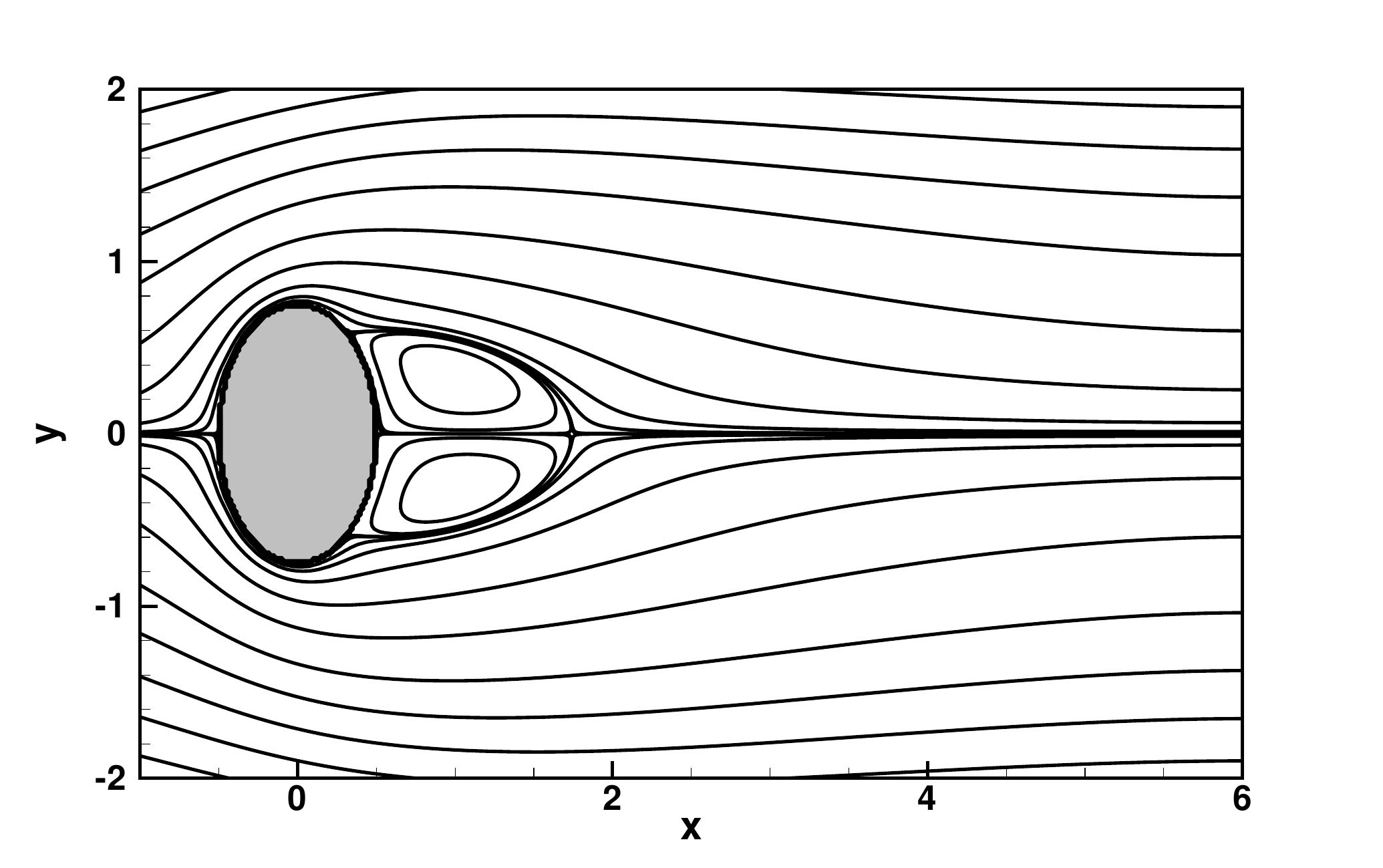} 
		\caption{$Re=10$}
	\end{subfigure}
	\begin{subfigure}{0.3\textwidth}
		\includegraphics[width=\linewidth]{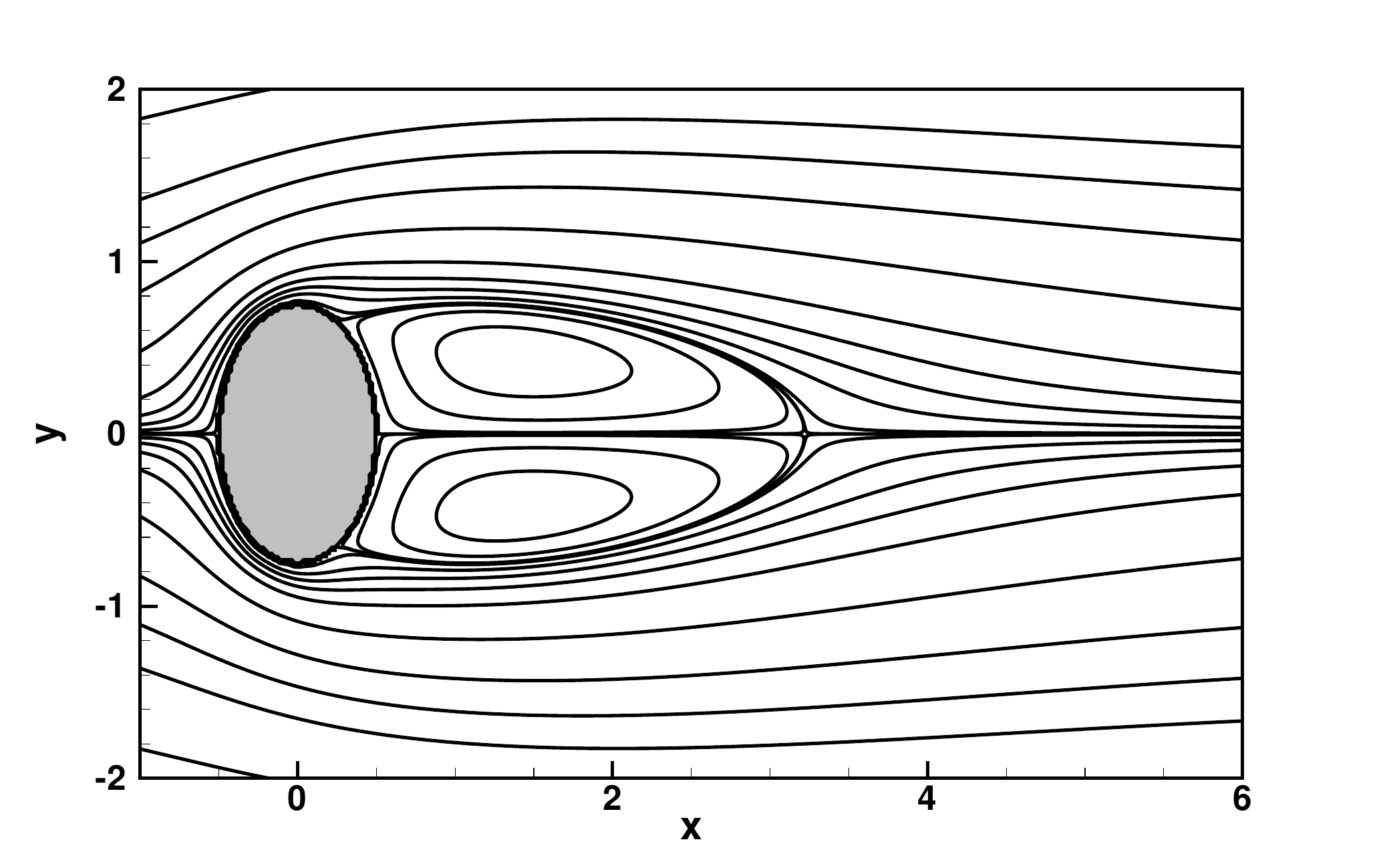} 
		\caption{$Re=20$}
	\end{subfigure}
	\begin{subfigure}{0.3\textwidth}
		\includegraphics[width=\linewidth]{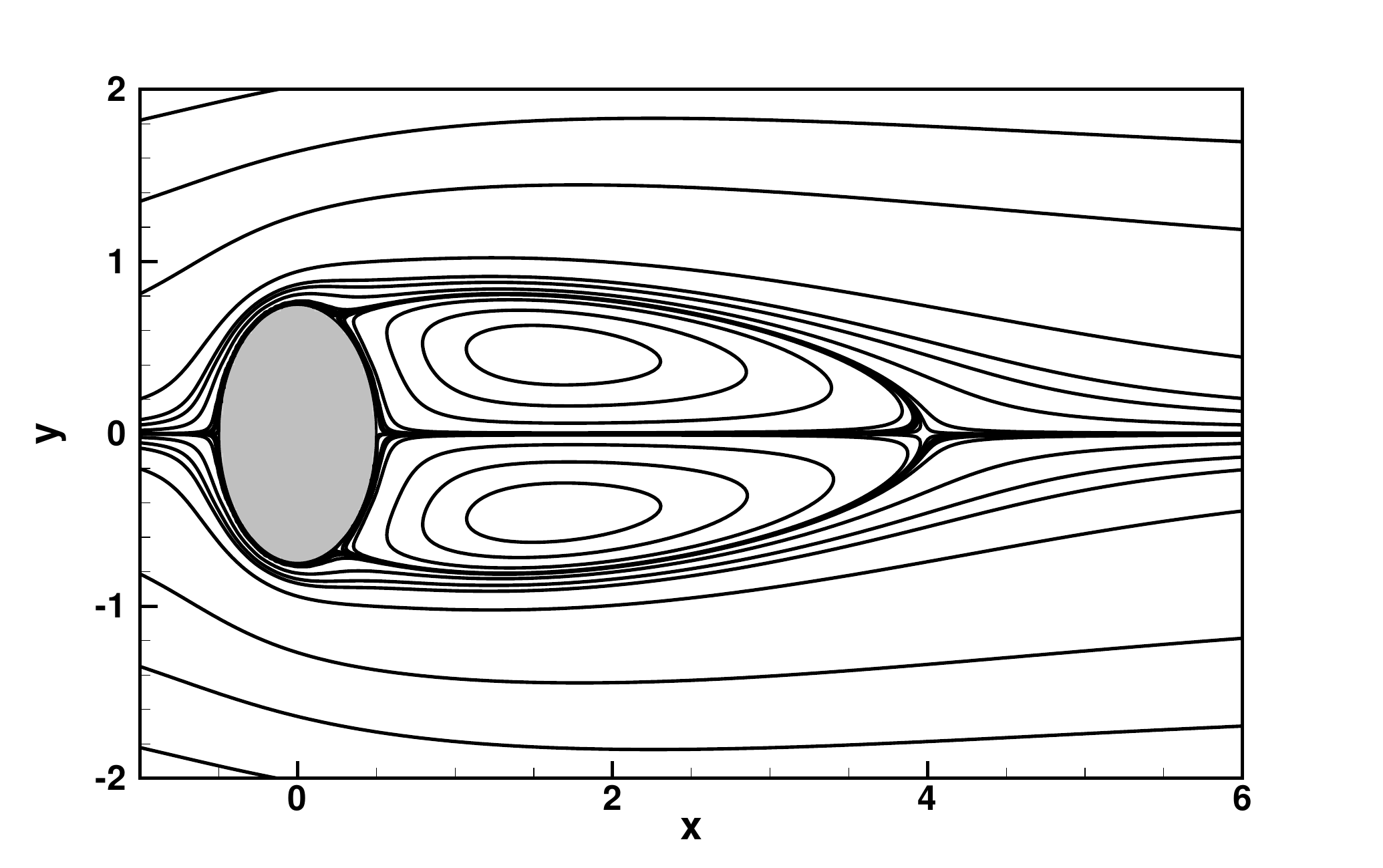} 
		\caption{$Re=25$}
	\end{subfigure} 
	\caption{\small{Steady state streamlines for $\theta=90^{\degree}$ and (a)$Re=10$, (b)$Re=20$, (c) $Re=25$.}}
	\label{Fig:psi-steady-90deg}
\end{figure}

\begin{figure}[H]
	\centering
	\begin{subfigure}{0.3\textwidth}
		\includegraphics[width=\linewidth]{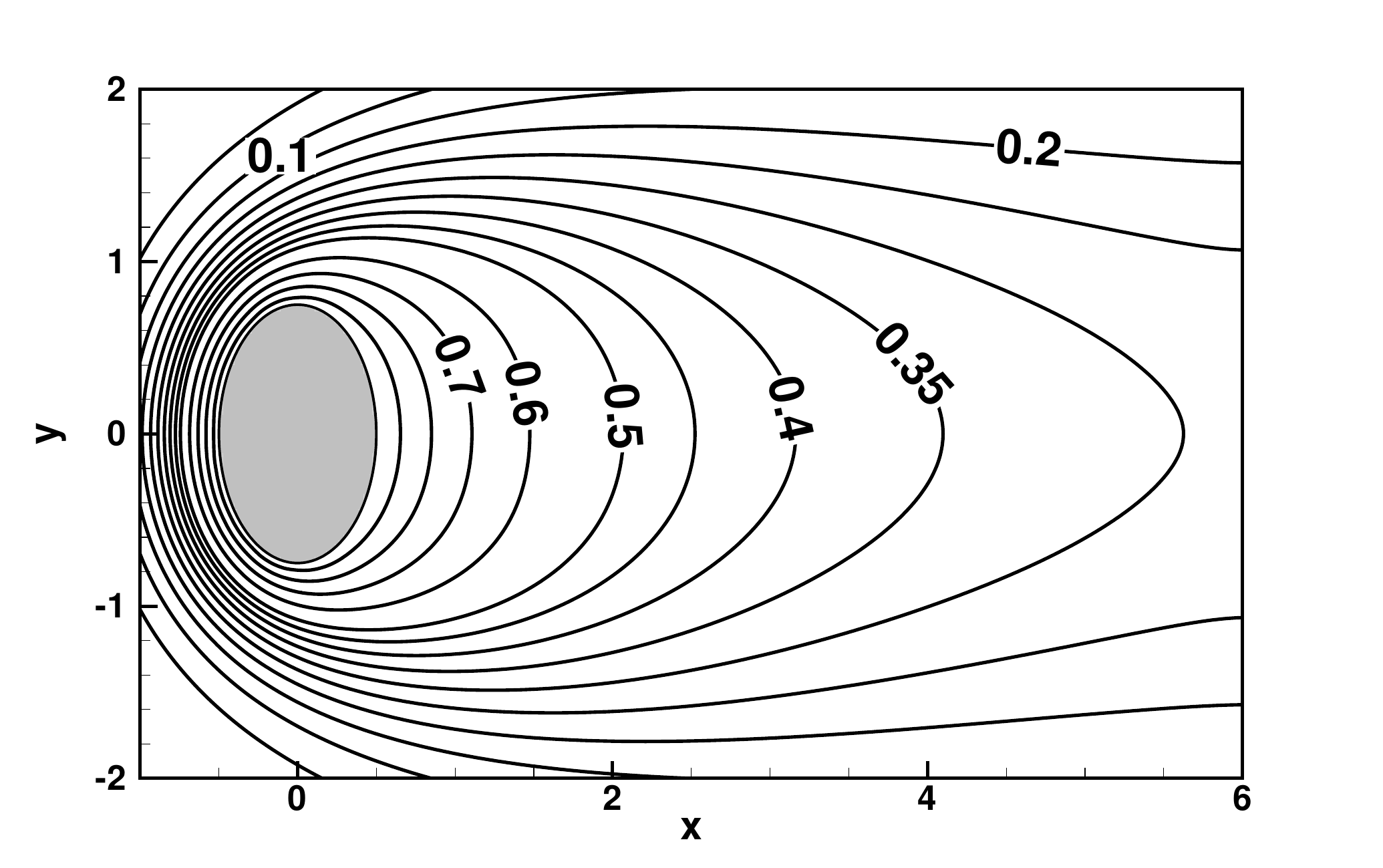} 
		\caption{$Re=10$}
	\end{subfigure}
	\begin{subfigure}{0.3\textwidth}
		\includegraphics[width=\linewidth]{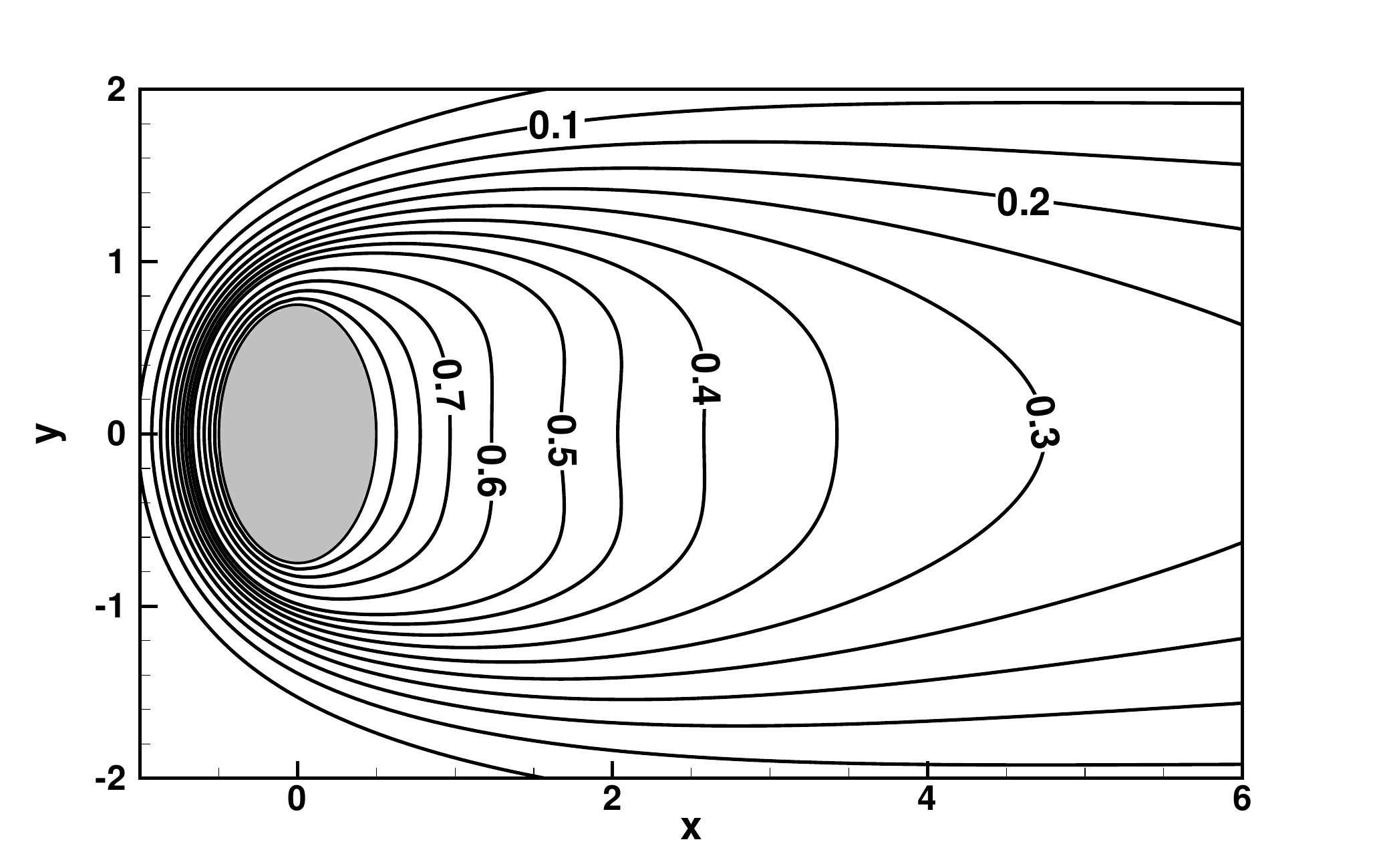} 
		\caption{$Re=20$}
	\end{subfigure}
	\begin{subfigure}{0.3\textwidth}
		\includegraphics[width=\linewidth]{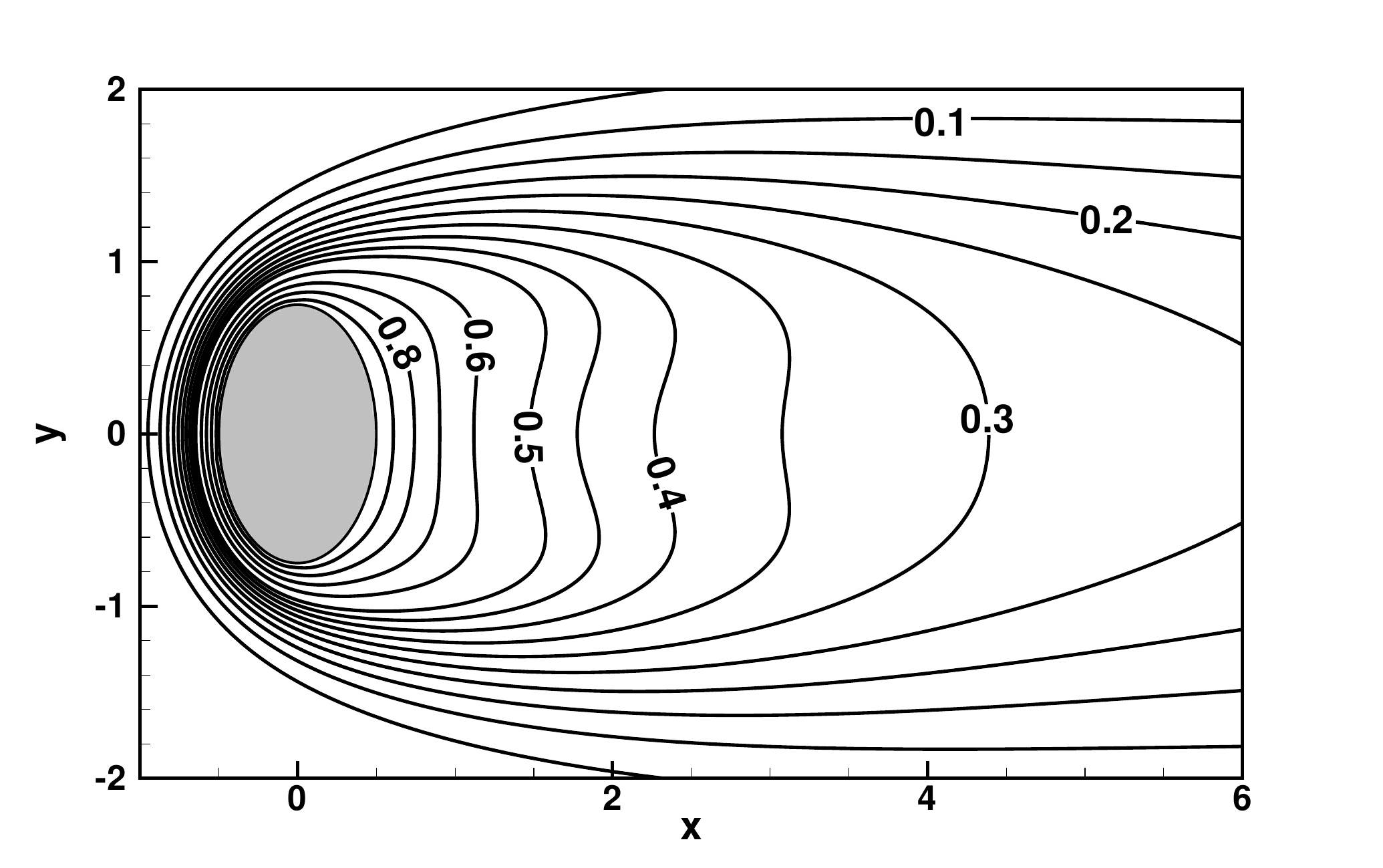} 
		\caption{$Re=25$}
	\end{subfigure} 
	\caption{\small{Steady state isotherms for $\theta=90^{\degree}$ and (a)$Re=10$, (b)$Re=20$, (c) $Re=25$.}}
	\label{Fig:isotherms-steady-90deg}
\end{figure}

The wake lengths for $\theta = 0^{\degree}$, $90^{\degree}$ are tabulated in \ref{tab:wake-length} for reference. As we can see from the table also the wake lengths at $\theta = 90^{\degree}$ are markedly higher than the wake length at $\theta = 0^{\degree}$.
\begin{table}[H]
	\caption{Wake length for $\theta = 0^{\degree}$, $90^{\degree}$}
	\centering
	\begin{tabular}{|l|l|l|} 
		\hline
		\multirow{2}{*}{$Re$} & \multicolumn{2}{c|}{$\theta$}  \\ 
		\cline{2-3}
		& $0^{\degree}$     & $90^{\degree}$                  \\ 
		\hline
		10                  & 0.085 & 1.307               \\ 
		\hline
		20                  & 0.643 & 2.725               \\ 
		\hline
		25                  & 0.922 & 3.457               \\ 
		\hline
		30                  & 1.141 & --                    \\ 
		\hline
		40                  & 1.693 & --                    \\ 
		\hline
		50                  & 2.262 & --                    \\ 
		\hline
		59                  & 2.848 & --                    \\
		\hline
	\end{tabular}\label{tab:wake-length}
\end{table}

%

\subsubsection{Average Nusselt number and Drag coefficient}\label{sec:Nusselt-steady}
The local and surface averaged Nusselt numbers are calculated from equations \eqref{eq:local_Nusselt} and \eqref{eq:Nusselt} respectively. We then plot the variation of the local $Nu$ along the surface of the cylinder. Figure \ref{Fig:nusselt-schematic} shows the schematic for measuring the perimeter of the ellipse. When $\theta=0^{\degree}$, we start at point $P$ and then move clockwise along the points $Q$, $R$, $S$, $W$. Note that $W$ coincides with $P$. Let $l_E$ denote the perimeter of the cylinder measured along $PQRSW$. When $\theta \neq 0^{\degree}$, the perimeter is measured along $P^'Q^'R^'S^'W^'$.

\begin{figure}[H] 
	\centering
	\includegraphics[width = 0.5\textwidth]{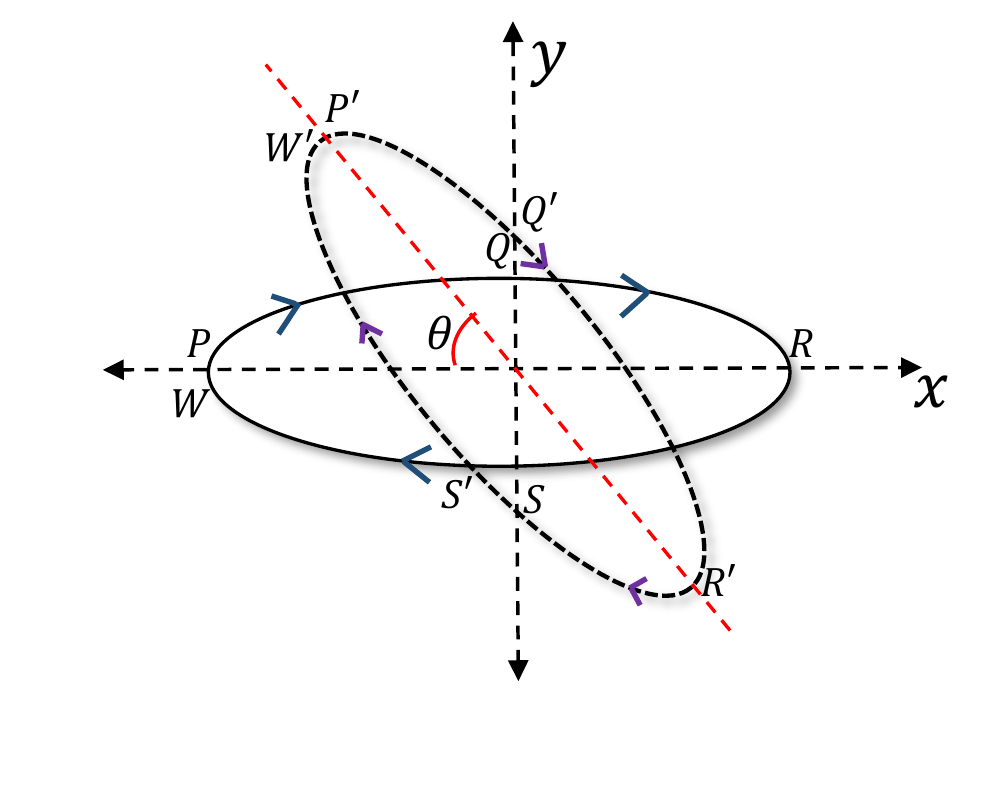}
	\caption{\small{Schematic showing the cylinder orientation for Nusselt number computation.}}\label{Fig:nusselt-schematic}
\end{figure}

Figure \ref{Fig:local-Nu} shows the variation of local $Nu$ along the surface of the cylinder for $\theta = 0^{\degree} - 90^{\degree}$. For every value of $\theta$ we observe that $Nu$ increases with $Re$. For $\theta = 0^{\degree}$ (figure \ref{Fig:local-Nu} (a)), the variation in $Nu$ is observed to be symmetric, with the maximum $Nu$ at the leading edge of the cylinder, i.e., at point $P$ shown in the schematic (figure \ref{Fig:nusselt-schematic}). At $\theta = 15^{\degree}$ (figure \ref{Fig:local-Nu} (b)), $Nu$ decreases first along the surface $P'Q'R'$. In the vicinity of the point $R'$ we observe a global minima and a local maxima of $Nu$. As one moves from the point $R'$, i.e., the trailing edge of the cylinder, to the point $W'$ along the surface $R'S'W'$, an increase in $Nu$ is observed. The variation of $Nu$ for $\theta = 30^{\degree}$ (figure \ref{Fig:local-Nu} (c)) follows a similar pattern as that of $\theta = 15^{\degree}$. However, two important differences stand out. First, the maximum value of $Nu$ for any $Re$ at $\theta = 30^{\degree}$ is greater than the maximum value of $Nu$ for the same $Re$ at $\theta = 15^{\degree}$. This is due to the fact that flow separation happens at a lower $Re$ for $\theta = 30^{\degree}$, which leads to greater mixing of the fluid thus increasing the rate of heat transfer. Thus, max($Nu$) at $Re = 40$ for $\theta = 30^{\degree}$ $>$ max($Nu$) at $Re = 40$ for $\theta = 15^{\degree}$, and so on. Secondly, there is slight shift in the locations of the local maxima and minima of $Nu$ in the clockwise direction. For $\theta=45^{\degree}$ (figure \ref{Fig:local-Nu} (d)) also, the variation in $Nu$ follows the pattern we observed for $\theta = 30^{\degree}$. We also observe that the variation of $Nu$ along the surface $P'Q'R'$ assumes an almost parabolic shape. At $\theta = 60^{\degree}$ (figure \ref{Fig:local-Nu} (e)), similar to the previous two cases, the locations of the local maxima and minima shift in the clockwise direction along the surface of the cylinder.  At $\theta=75^{\degree}$ (figure \ref{Fig:local-Nu} (f)), there is a significant reversal in one of the patterns observed in the previous four cases. Here, one can observe that the maximum value of $Nu$ for a particular $Re$ is less than the maximum value of $Nu$ for the same $Re$ at $\theta = 60^{\degree}$, i.e., max  max($Nu$) at $Re = 10$ for $\theta = 75^{\degree}$ $<$ max($Nu$) at $Re = 10$ for $\theta = 60^{\degree}$. Note that, at $\theta = 90^{\degree}$ (figure \ref{Fig:local-Nu} (g)), the maximum value of $Nu$ for a particular $Re$ is again less than the maximum value of $Nu$ for the same $Re$ at $\theta = 75^{\degree}$. Interestingly, the minimum value of $Nu$ keeps on decreasing from $\theta = 15^{\degree} - 90^{\degree}$, and it occurs on the surface $P'Q'R'$. Note that for the variation of $Nu$ is smoother along the surface  on the part $P'Q'R'$ for all values of $\theta$. 

The variation of surface averaged Nusselt number, $Nu_{\text{av}}$, with the Reynolds number for different values of $\theta$ is shown in figure \ref{Fig:Av-Nu} (a). Apart from $\theta = 0^{\degree}$, the variation of $Nu_{\text{av}}$ follows a similar pattern for all values of $\theta$. For $\theta = 0^{\degree}$, we observe that the value of $Nu_{\text{av}}$ for a particular $Re$ is markedly higher than the corresponding $Nu_{\text{av}}$ values at other values of $\theta$. For the rest of $\theta$ values considered, the average Nusselt number increases with $Re$ due an increase in flow strength as $Re$ is increased. Note that the value of $Nu_{\text{av}}$ also increases as $\theta$ is increased. Thus, the value of $Nu_{\text{av}}$ at $Re = 10$ for $\theta = 30^{\degree}$ is greater than the value of $Nu_{\text{av}}$ at $Re = 10$ for $\theta = 15^{\degree}$ and so on. Also, the value of $Nu_{\text{av}}$ is minimum at $\theta = 15^{\degree}$.

Figure \ref{Fig:Av-Nu} (b) shows the variation of drag coefficient  $C_D$ with $Re$ for different values of $\theta$, which is computed by using \eqref{eq:drag}. We can see that for a particular $\theta$, $C_D$ decreases with $Re$, which is on the expected line, as with increase in $Re$, inertial forces start dominating the viscous ones.  Two cases, however, stand out viz. $\theta = 0^{\degree}$ and $\theta = 90^{\degree}$. For a given $Re$, the values of $C_D$ at $\theta = 0^{\degree}$, $90^{\degree}$ are greater than the value of $C_D$ at the rest of $\theta$ values. Further, the drag forced experienced by the body at $\theta = 90^{\degree}$ is the highest of all for a given $Re$. Also, as $\theta$ is increased for a particular $Re$, flow separation occurs, which leads to an increase in the pressure difference between the front and rear half of the cylinder, thereby causing an increase in the pressure drag force. Thus for a fixed Reynolds number, $C_D$ increases as $\theta$ is increased.

\begin{figure}[H]
	\centering
	\begin{subfigure}{0.3\textwidth}
		\includegraphics[width=\linewidth]{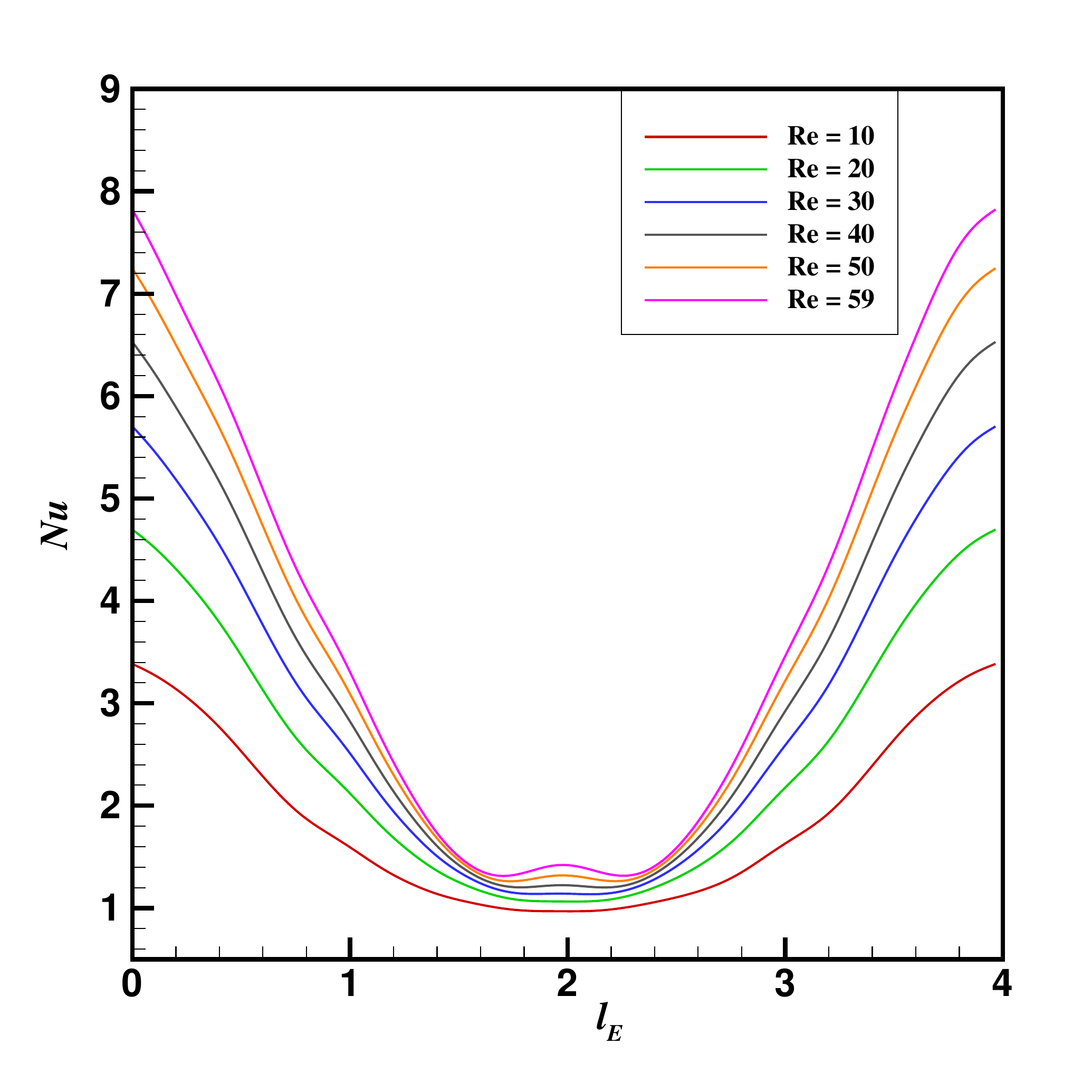} 
		\caption{$\theta=0^{\degree}$}
	\end{subfigure}\hfil 
	\begin{subfigure}{0.3\textwidth}
		\includegraphics[width=\linewidth]{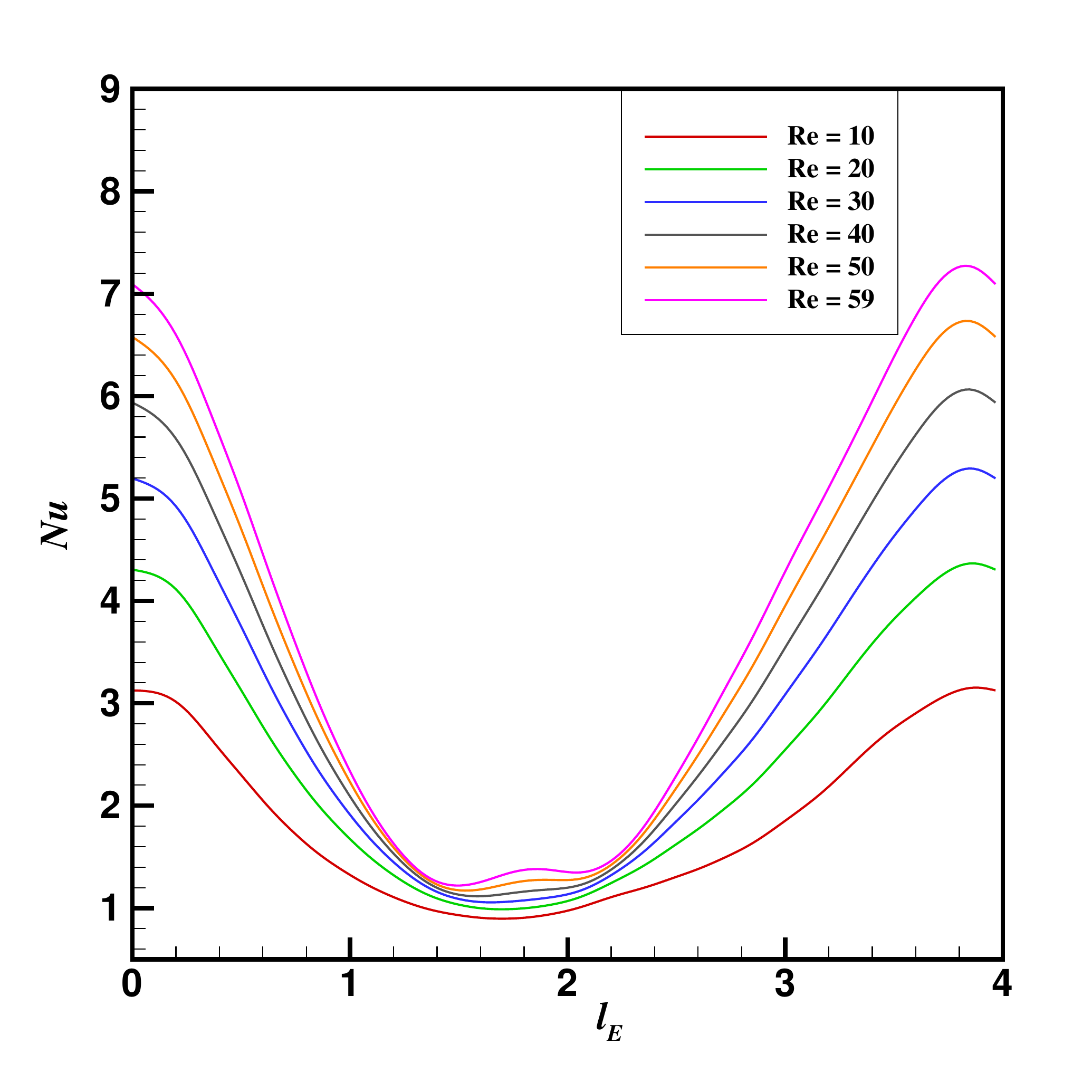} 
		\caption{$\theta=15^{\degree}$}
	\end{subfigure}\hfil 
	\begin{subfigure}{0.3\textwidth}
		\includegraphics[width=\linewidth]{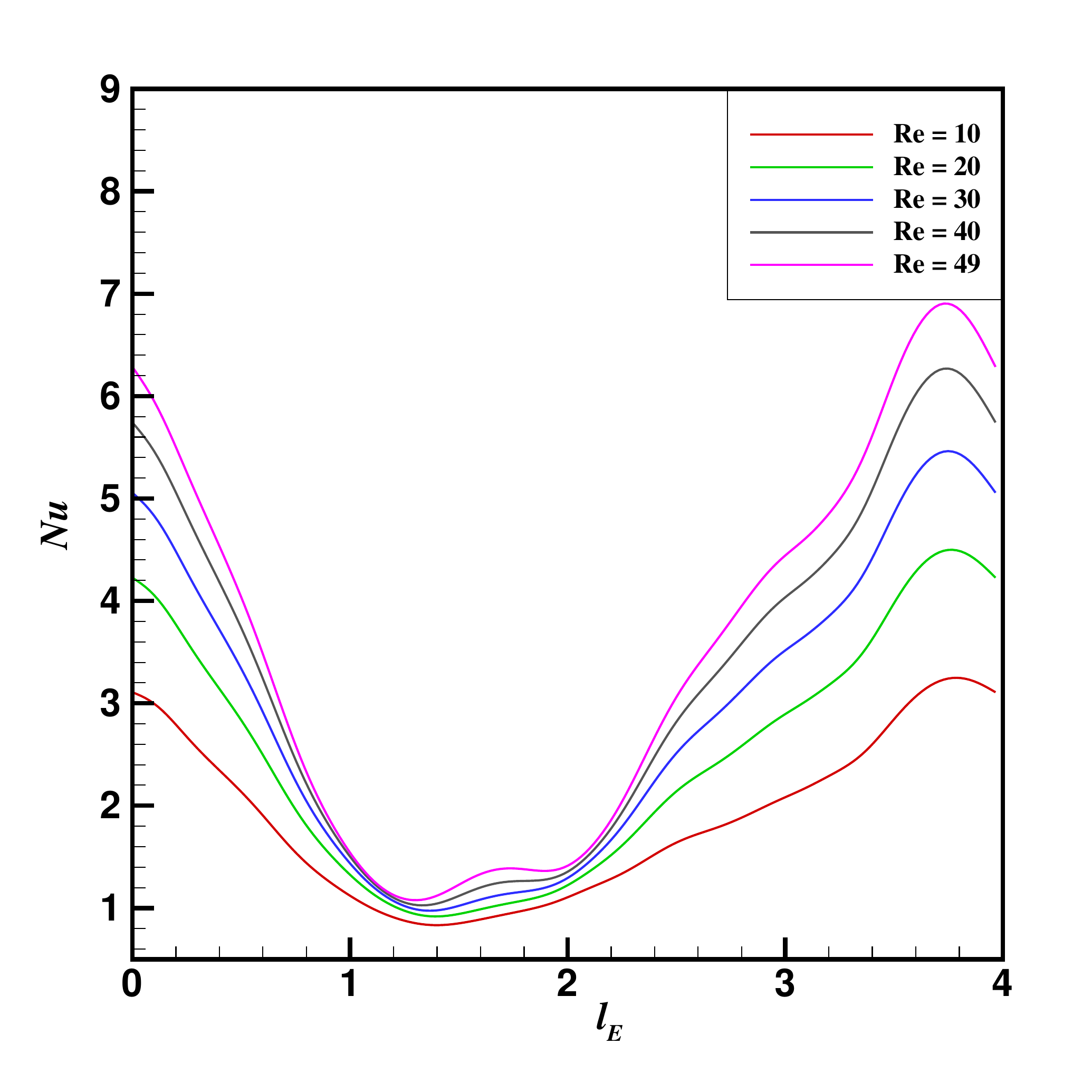} 
		\caption{$\theta=30^{\degree}$}
	\end{subfigure}\hfil 
	\begin{subfigure}{0.3\textwidth}
		\includegraphics[width=\linewidth]{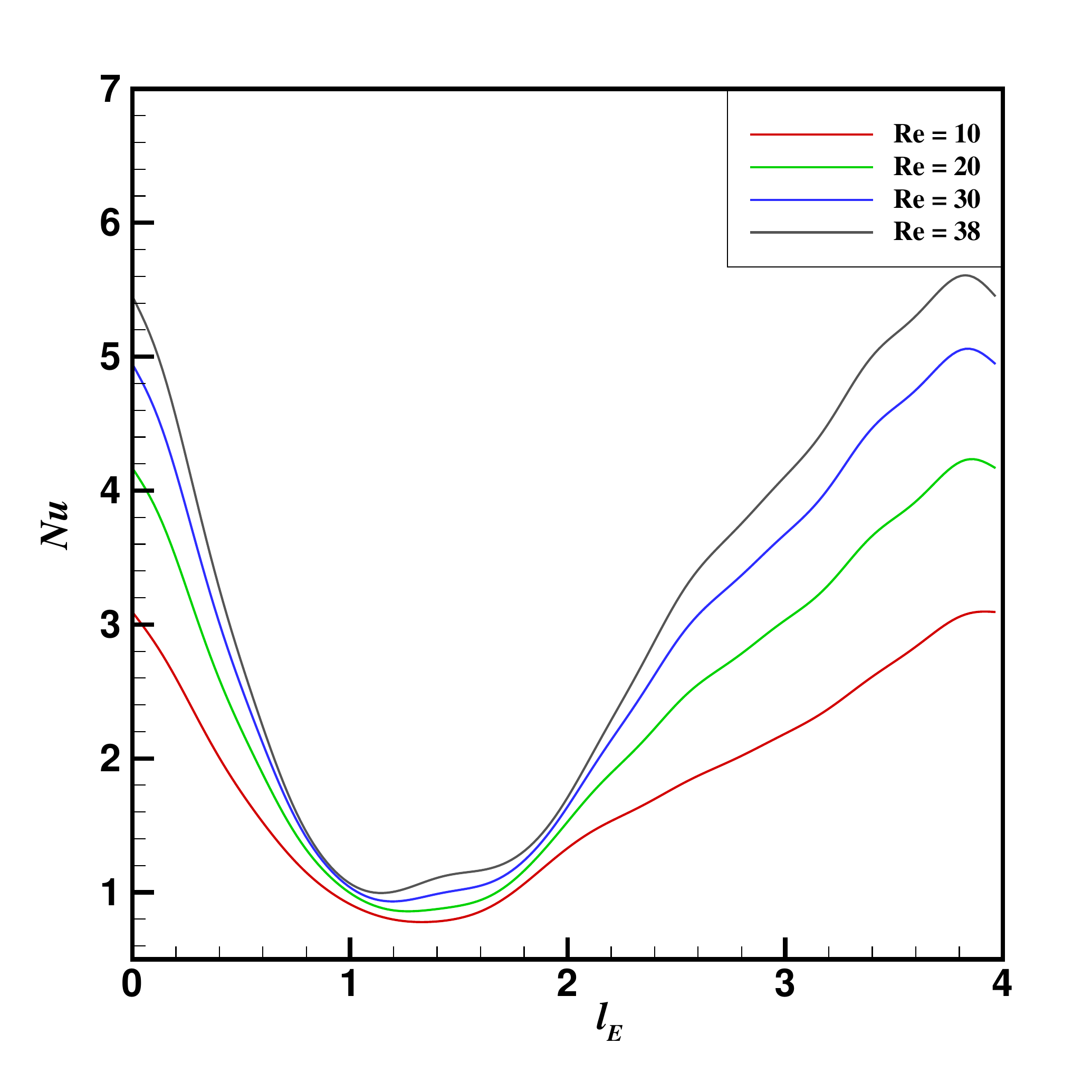} 
		\caption{$\theta=45^{\degree}$}
	\end{subfigure}\hfil 
	\begin{subfigure}{0.3\textwidth}
		\includegraphics[width=\linewidth]{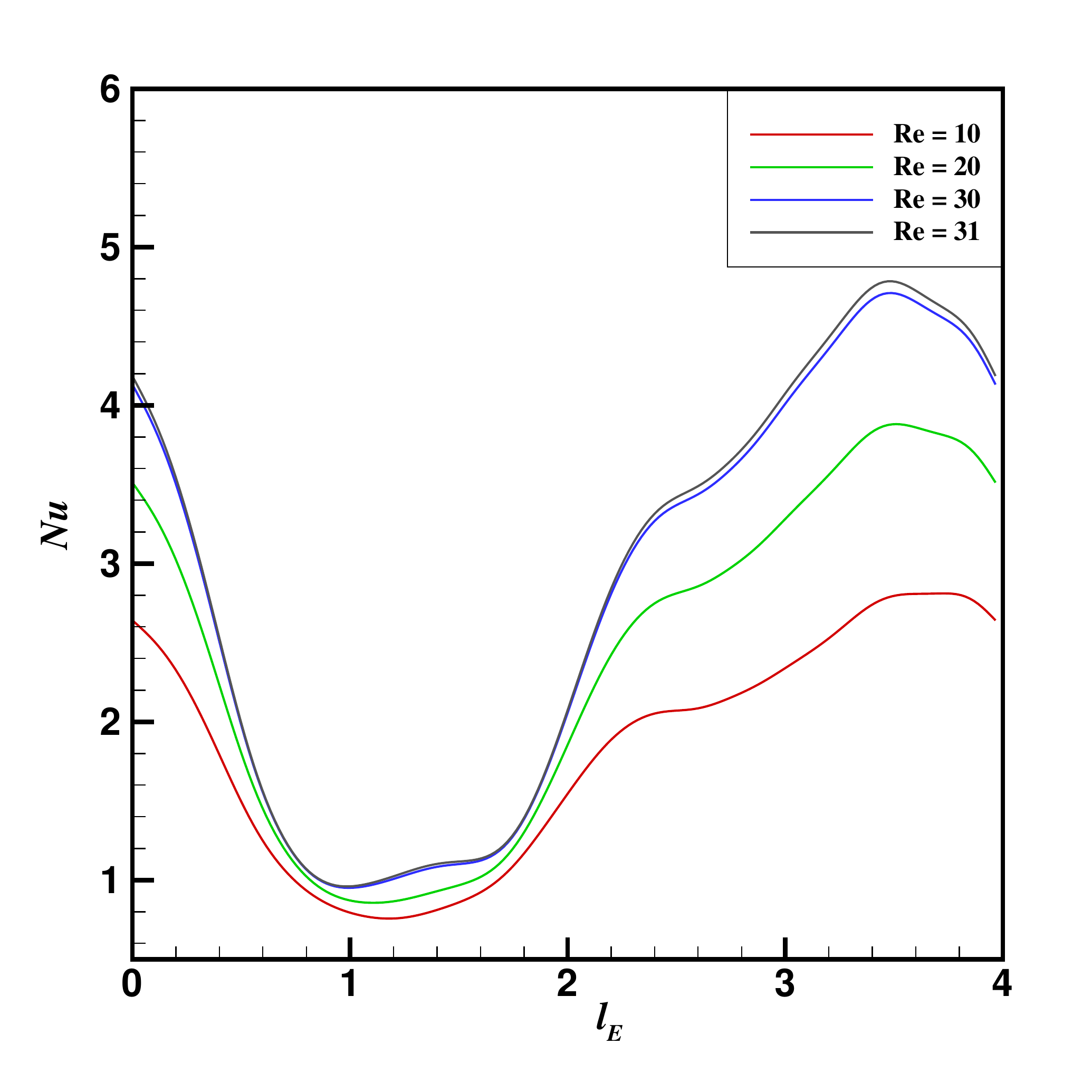} 
		\caption{$\theta=60^{\degree}$}
	\end{subfigure}\hfil 
	\begin{subfigure}{0.3\textwidth}
		\includegraphics[width=\linewidth]{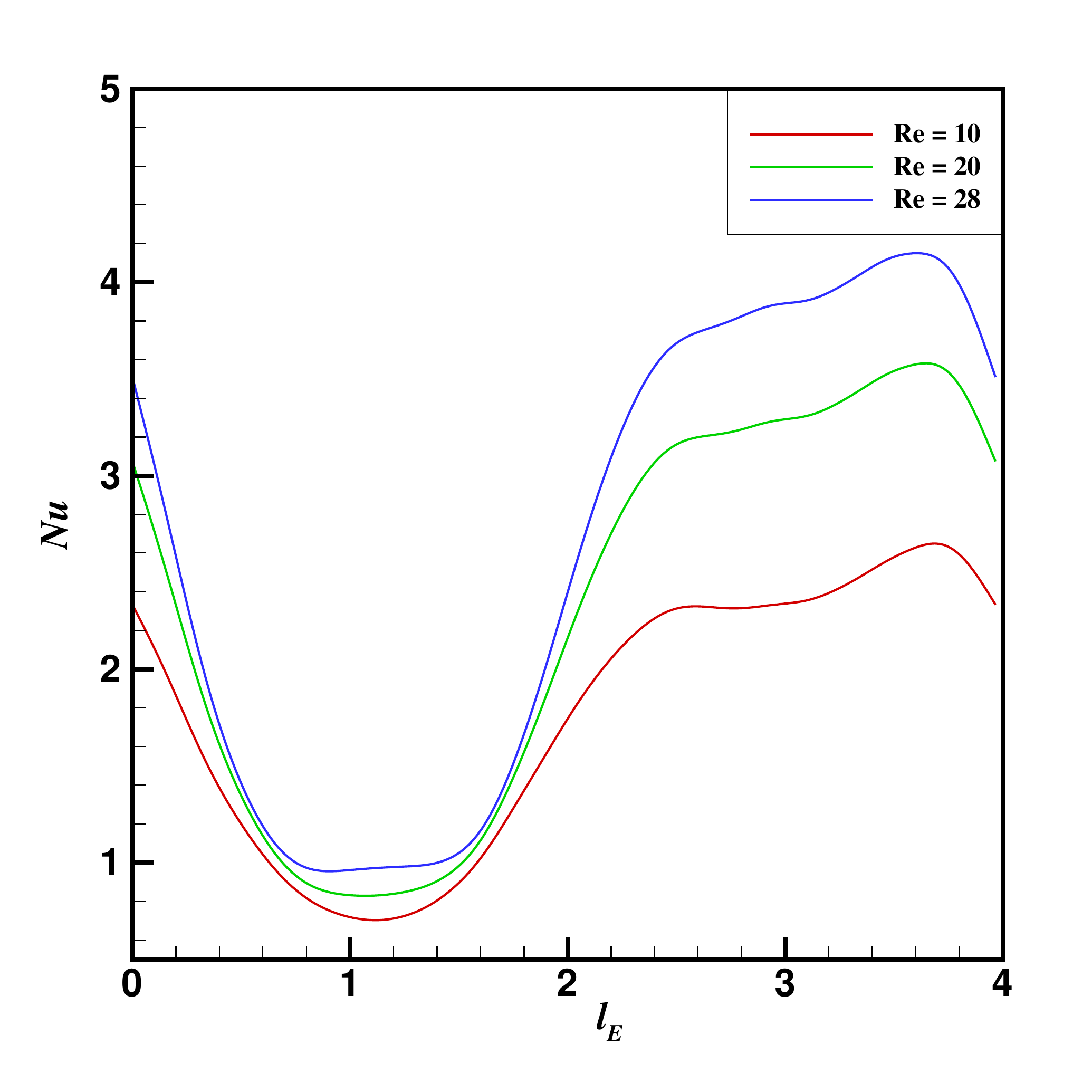} 
		\caption{$\theta=75^{\degree}$}
	\end{subfigure}\hfil 
	\begin{subfigure}{0.3\textwidth}
		\includegraphics[width=\linewidth]{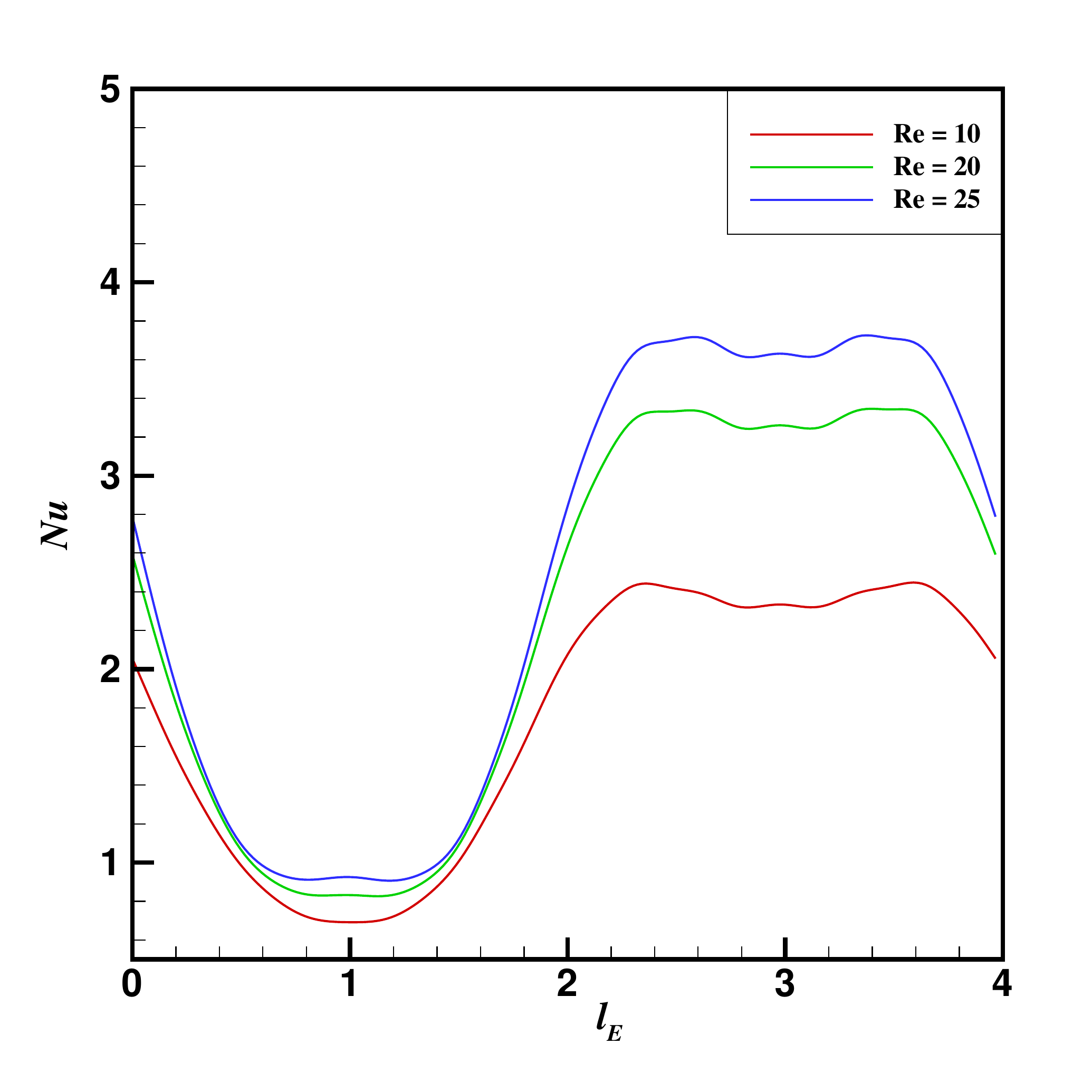} 
		\caption{$\theta=90^{\degree}$}
	\end{subfigure}\hfil 
	\caption{\small{Variation of local Nusselt number along the surface of the cylinder for (a) $\theta=0^{\degree}$, (b) $\theta=15^{\degree}$, (c) $\theta=30^{\degree}$, (d) $\theta=45^{\degree}$, (e) $\theta=60^{\degree}$, (f) $\theta=75^{\degree}$, (g) $\theta=90^{\degree}$.}}
	\label{Fig:local-Nu}
\end{figure}

\begin{figure}[H]
	\centering
	\begin{subfigure}{0.5\textwidth}
		\includegraphics[width=\linewidth]{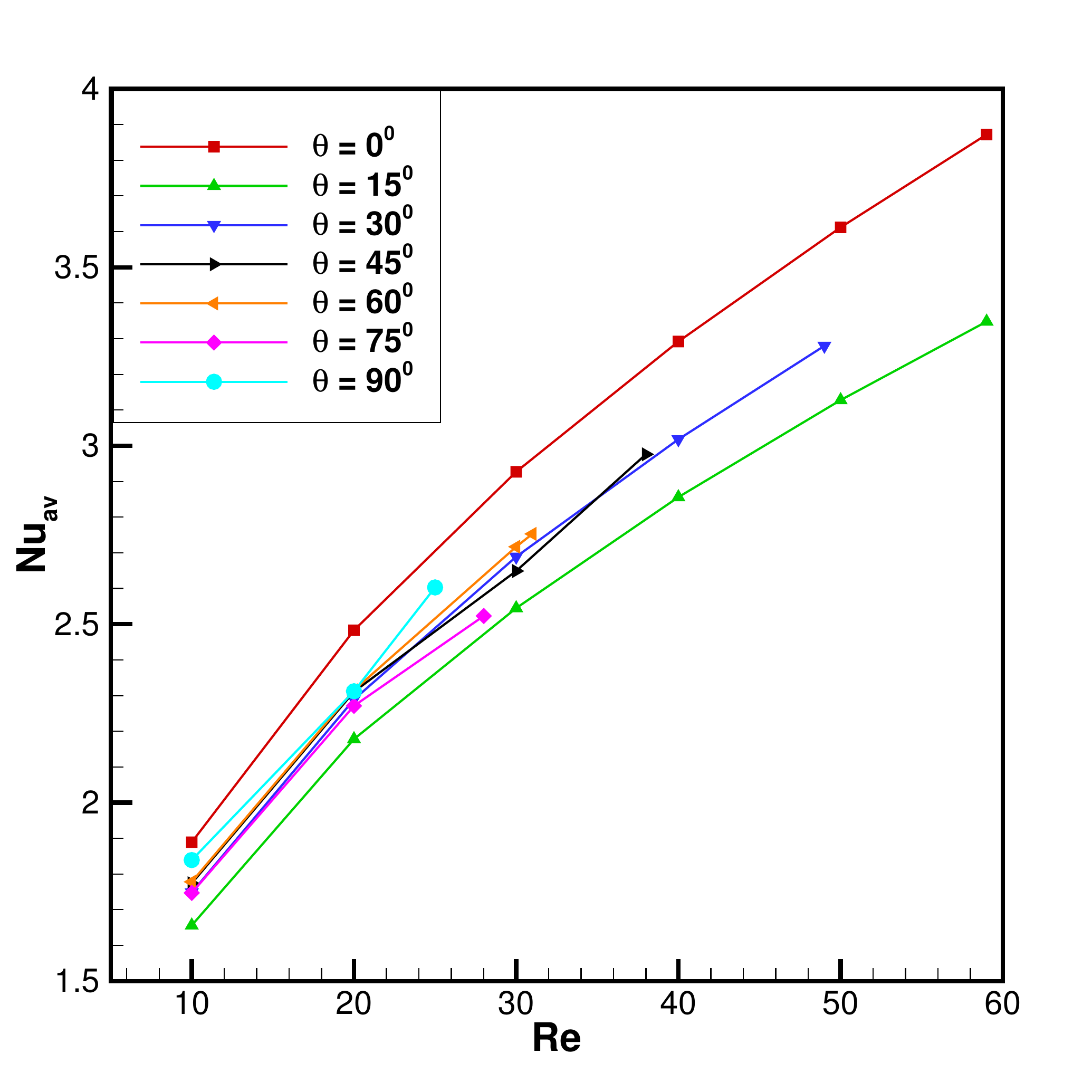} 
		\caption{$Nu_{\text{av}}$ vs. $Re$}
	\end{subfigure}\hfil 
	\begin{subfigure}{0.5\textwidth}
		\includegraphics[width=\linewidth]{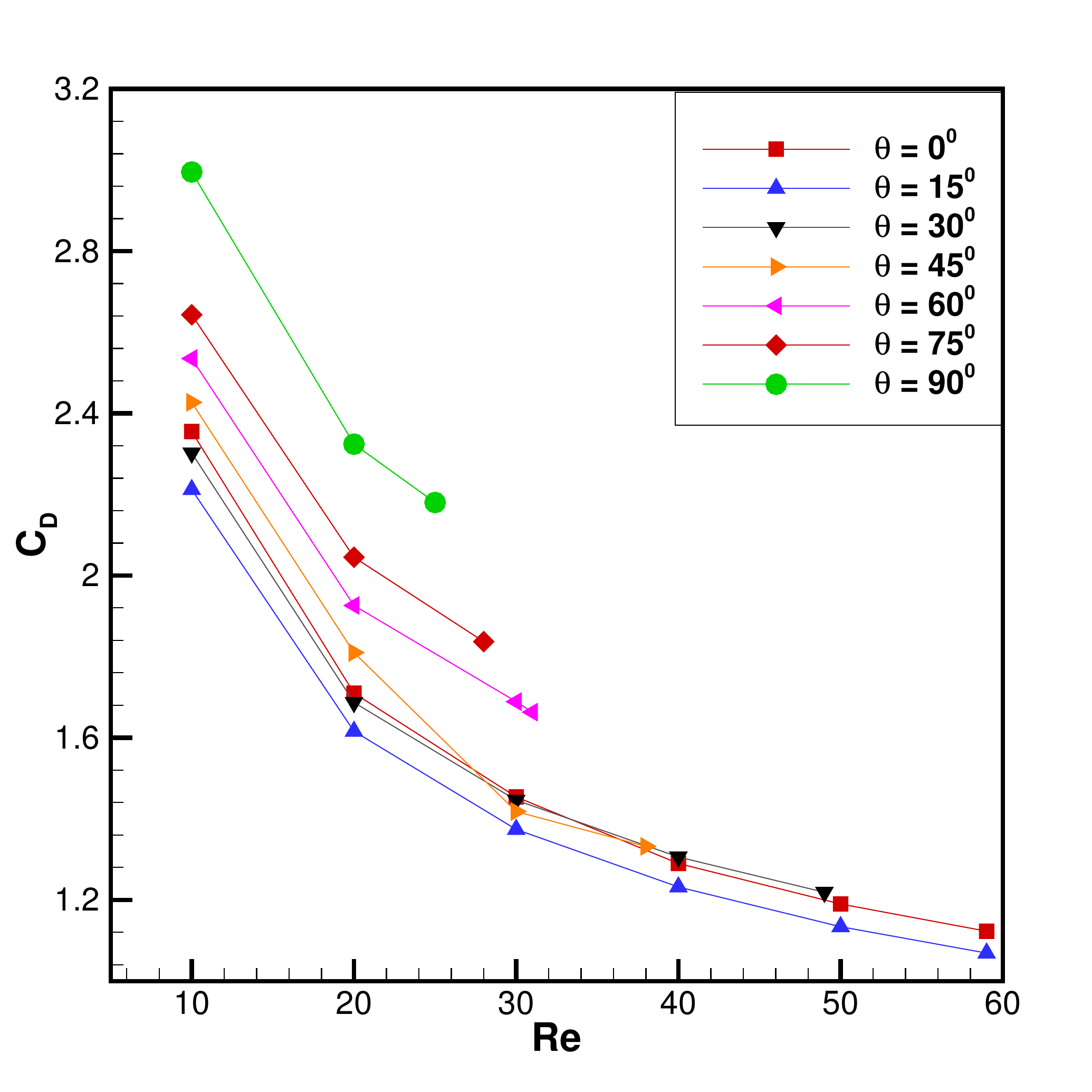} 
		\caption{$C_D$ vs. $Re$}
	\end{subfigure}\hfil 	
	\caption{\small{Variation of (a) Surface Averaged Nusselt number ($Nu_{\text{av}}$) and (b) Average drag $C_D$ with $Re$ for different values of $\theta$.}}
	\label{Fig:Av-Nu}
\end{figure}

\subsubsection{Heat and fluid flow beyond $\theta =90^\degree$}
\begin{figure}[H] 
	\centering
	\includegraphics[width = 0.5\textwidth]{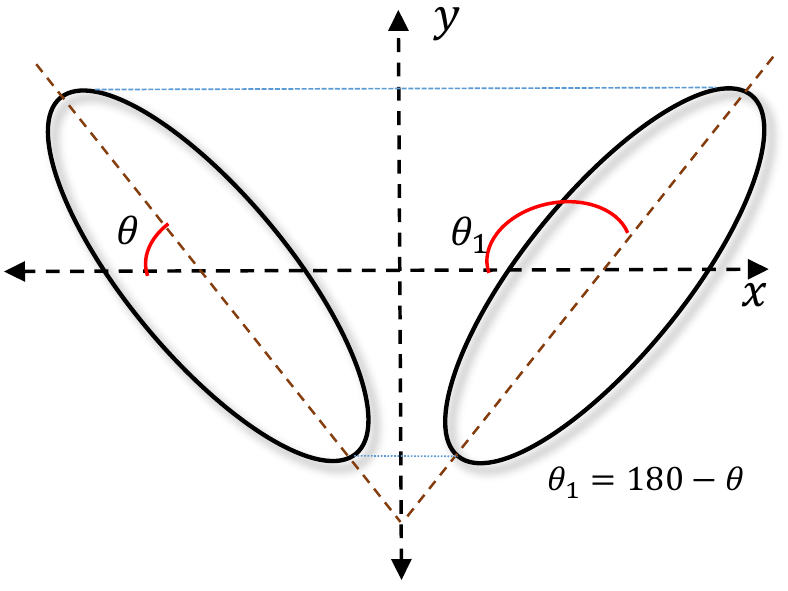}
	\caption{\small{Schematic showing mirror image of the two configurations corresponding to $\theta$ and $\theta_1=180^\degree -\theta$ for $0^\degree \leq \theta \leq  90^\degree $.}}\label{Fig:mirror-image-schematic}
\end{figure}

We carried out continued our computations for $90^{\degree} < \theta < 180^{\degree}$ as well. However, in this range of $\theta$, we observed that, about an $x= constant$ line, the flow in the wake of the cylinder is a mirror image of the flow in range of $0^{\degree} \leq \theta \leq 90^{\degree}$, i.e., flow for $\theta$ ($90^{\degree} < \theta < 180^{\degree}$) is a mirror image of flow for $180^{\degree}-\theta$ ($0^{\degree} \leq \theta \leq 90^{\degree}$). In other words, flow pattern for $\theta = 105^{\degree}$ is a mirror image of flow pattern for $\theta = 75^{\degree}$, that of $\theta = 120^{\degree}$ is a mirror image of $\theta = 60^{\degree}$ and so on. In figure \ref{Fig:mirror-image-schematic} the schematic on the left represents the configurations in the range $0^{\degree} < \theta \leq 90^{\degree}$. On the right, the configuration for the range $90^{\degree} < \theta < 180^{\degree}$ is represented. As shown in the figure, the schematic on the right with an angle of attack $\theta_1 (=180^{\degree}-\theta)$ is the mirror image about the $y$- axis of the schematic on the right. To demonstrate this interesting phenomena, we have chosen three flow configurations at different $\theta$ and $Re$, and compared them with their $180^{\degree} - \theta$ counterparts (see figure \ref{Fig:comparison-mirror-image}). This particular symmetry results from the geometry of the cylinder, as well as the particular assumption of negligible gravity on the flow (see section \ref{sec:problem-statement}).As a result of this particular symmetry there is no marked difference in the quantitative parameters as well. Thus there is no difference in the values of the average Nusselt number ($Nu_{\text{av}}$) as well as the average drag ($C_D$) for the configuration $\theta = 45^{\degree}$ and its mirror image $\theta = 135^{\degree}$, and so on as can be seen from tables \ref{T_comparison_3pi12_9pi12}-\ref{T_comparison_5pi12_7pi12}.

\begin{figure}[H]
	\centering
	\begin{subfigure}{0.5\textwidth}
		\includegraphics[width=\linewidth]{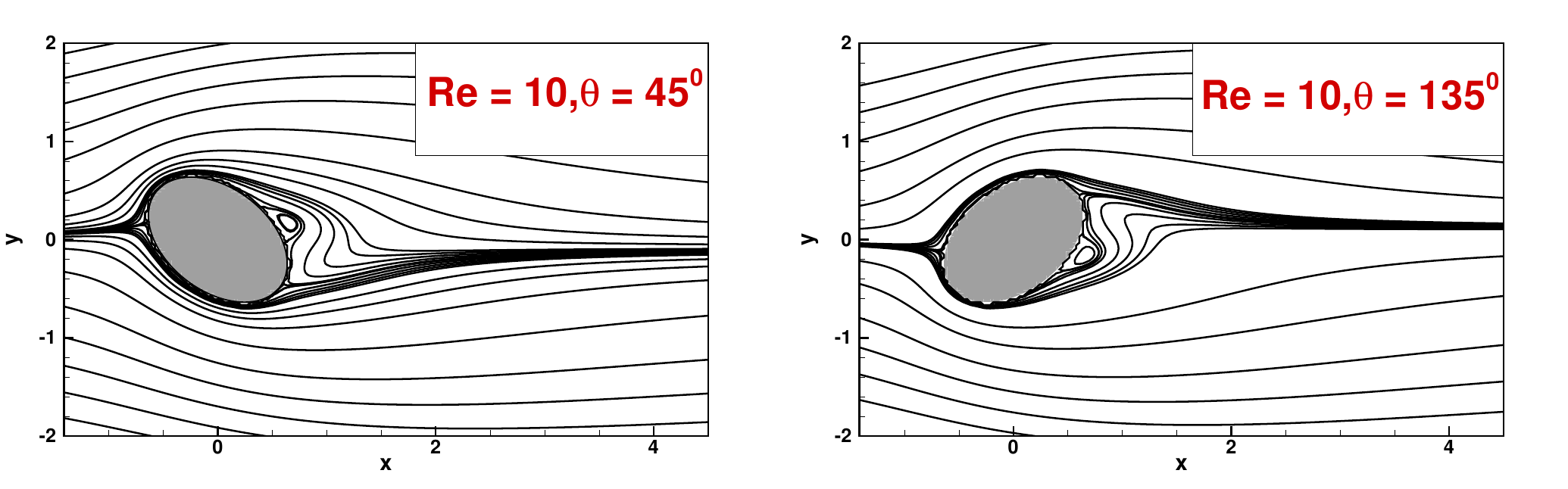} 
	\end{subfigure}\hfil 
	\begin{subfigure}{0.5\textwidth}
	\includegraphics[width=\linewidth]
{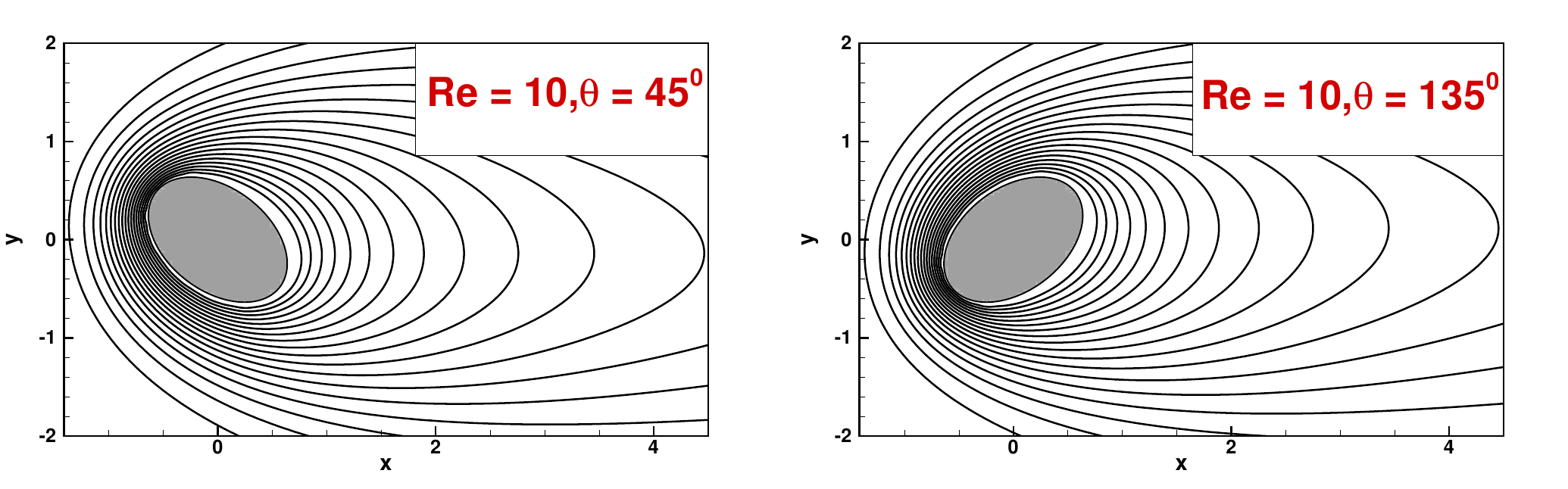}
	\end{subfigure}\hfil
\caption{Comparison of streamlines and isotherms for $Re=10$ with $\theta = 45^{\degree}$ (left) and $\theta = 135^{\degree}$ (right)}
	\begin{subfigure}{0.5\textwidth}
		\includegraphics[width=\linewidth]{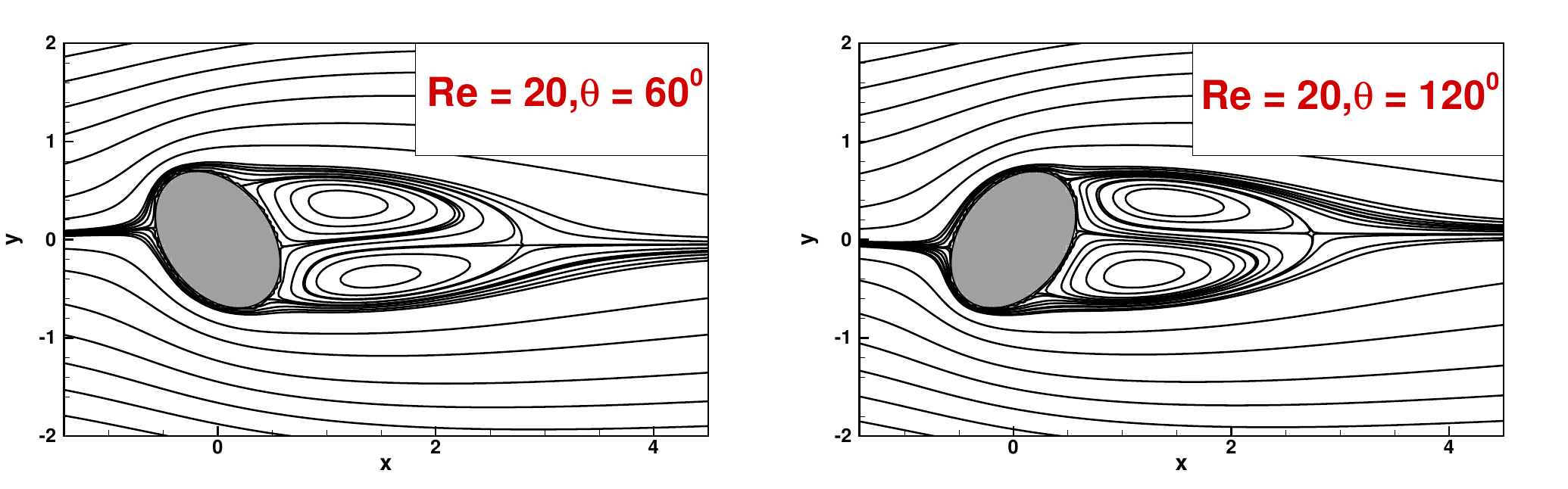} 
	\end{subfigure}\hfil 
	\begin{subfigure}{0.5\textwidth}
	\includegraphics[width=\linewidth]
{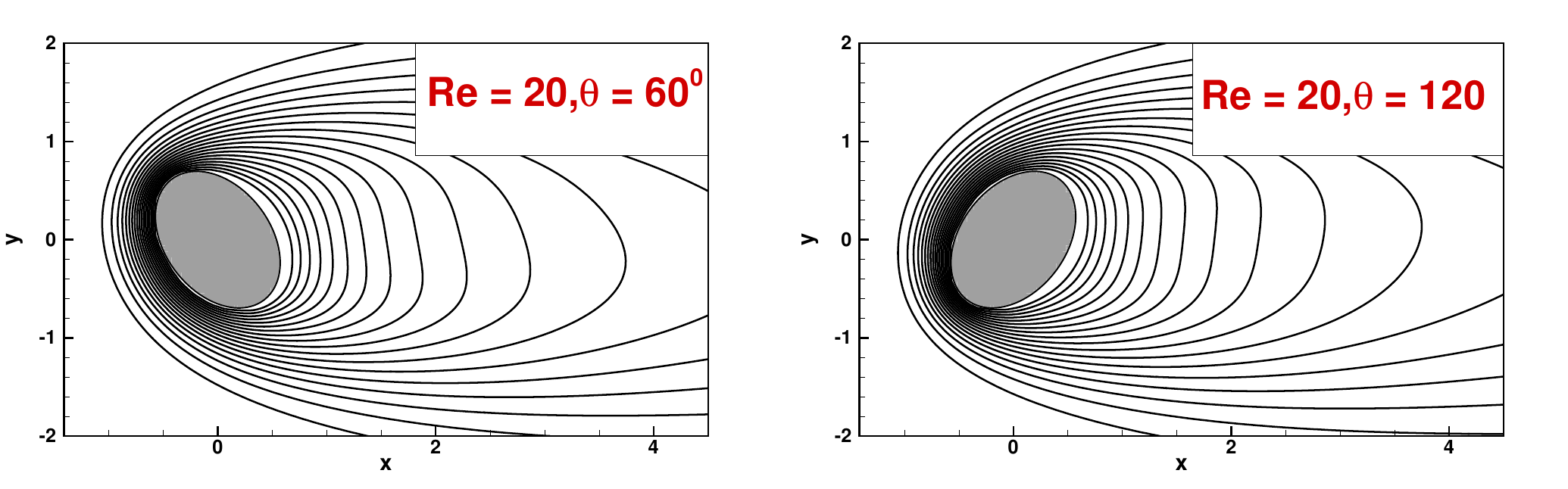}
\end{subfigure}\hfil
\caption{Comparison of streamlines and isotherms for $Re=20$ with $\theta = 60^{\degree}$ (left) and $\theta = 120^{\degree}$ (right)}
	\begin{subfigure}{0.5\textwidth}
		\includegraphics[width=\linewidth]{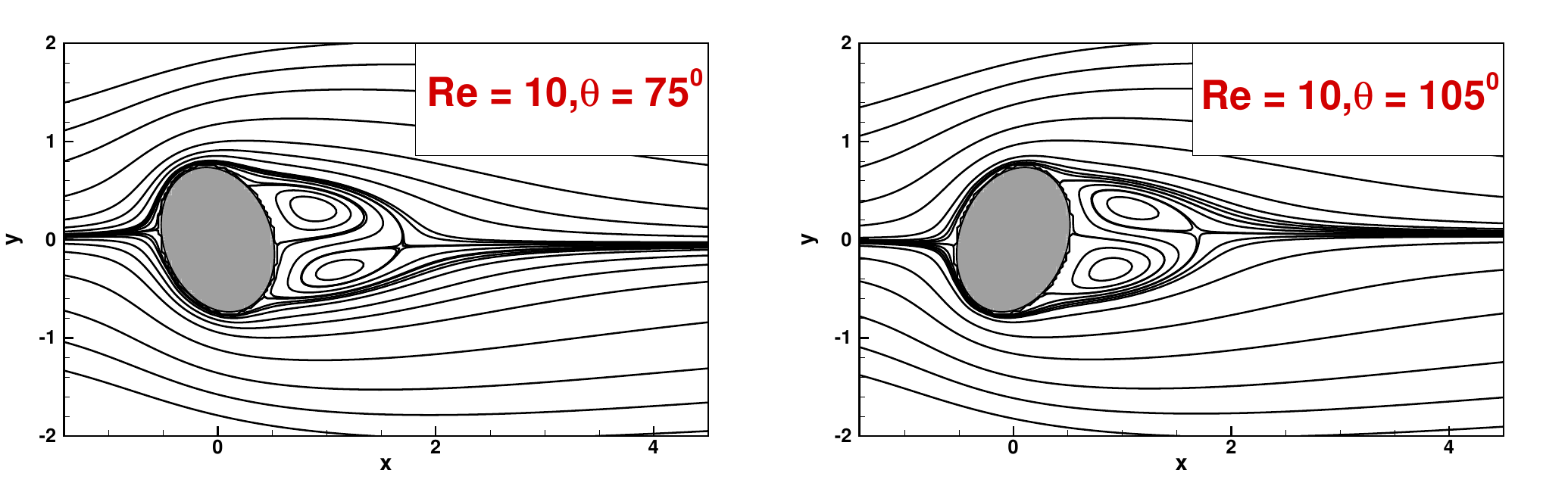} 
	\end{subfigure}\hfil 
	\begin{subfigure}{0.5\textwidth}
	\includegraphics[width=\linewidth]
{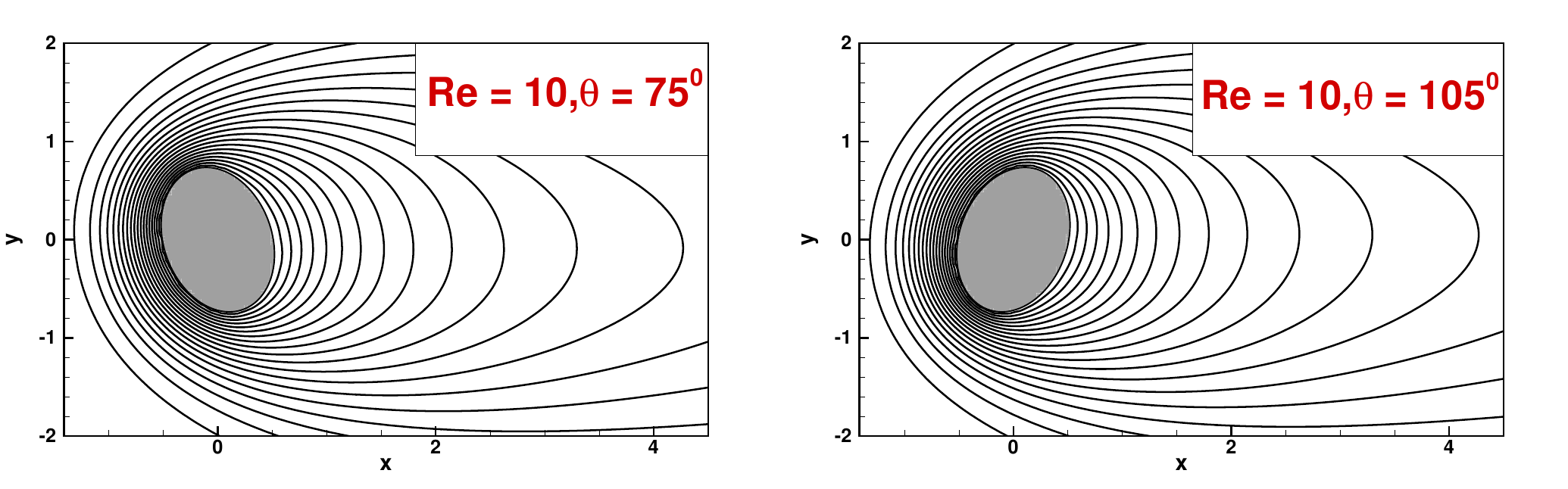}
\end{subfigure}\hfil
	\caption{Comparison of streamlines and isotherms for $Re=10$ with $\theta = 75^{\degree}$ (left) and $\theta = 105^{\degree}$ (right).}
	\label{Fig:comparison-mirror-image}
\end{figure}
\begin{table}[]
\centering
\caption{ Comparison of surface averaged Nusselt number $Nu_{av}$ and drag coefficients $C_D$ for $\theta = 45^{\degree}$ and $ \theta = 135^{\degree}$.}
\begin{tabular}{|c|cc|cc|}
\hline
$\quad $                               & \multicolumn{2}{c|}{$Nu_{av}$}                & \multicolumn{2}{c|}{$C_D$}                               \\ \hline
$Re$ &
$\theta = 45^{\degree}$ &
$ \theta = 135^{\degree}$ &
\multicolumn{1}{c}{$\theta = 45^{\degree}$} &
$\theta = 135^{\degree}$ \\ \hline
$10$  & $1.798$    & $1.802$  &$2.913$    & $2.907$  \\
$20$  & $2.349$    & $2.356$  &$2.173$    & $2.083$  \\
$30$  & $2.759$    & $2.768$  &$1.702$    & $1.692$  \\
$38$  & $3.032$    & $3.041$  &$1.598$    & $1.591$  \\
 \hline
\end{tabular}
\label{T_comparison_3pi12_9pi12}
\end{table}
\begin{table}[]
\centering
\caption{ Comparison of surface averaged Nusselt number $Nu_{av}$ and drag coefficients $C_D$ for $\theta = 60^{\degree}$ and $ \theta = 120^{\degree}$.}
\begin{tabular}{|c|cc|cc|}
\hline
$\quad $                               & \multicolumn{2}{c|}{$Nu_{av}$}                & \multicolumn{2}{c|}{$C_D$}                               \\ \hline
$Re$ &
$\theta = 60^{\degree}$ &
$ \theta = 120^{\degree}$ &
\multicolumn{1}{c}{$\theta = 60^{\degree}$} &
$\theta = 120^{\degree}$ \\ \hline
$10$  & $1.815$    & $1.821$  &$3.017$    & $2.958$  \\
$20$  & $2.366$    & $2.373$  &$2.311$    & $2.317$  \\
$30$  & $2.779$    & $2.794$  &$2.027$    & $2.034$  \\
$31$  & $2.816$    & $2.898$  &$1.996$    & $2.013$  \\
 \hline
\end{tabular}
\label{T_comparison_4pi12_8pi12}
\end{table}
\begin{table}[]
\centering
\caption{ Comparison of surface averaged Nusselt number $Nu_{av}$ and drag coefficients $C_D$ for $\theta = 75^{\degree}$ and $ \theta = 105^{\degree}$.}
\begin{tabular}{|c|cc|cc|}
\hline
$\quad $                               & \multicolumn{2}{c|}{$Nu_{av}$}                & \multicolumn{2}{c|}{$C_D$}                               \\ \hline
$Re$ &
$\theta = 75^{\degree}$ &
$ \theta = 105^{\degree}$ &
\multicolumn{1}{c}{$\theta = 75^{\degree}$} &
$\theta = 105^{\degree}$ \\ \hline
$10$  & $1.756$    & $1.761$  &$2.961$    & $2.916$  \\
$20$  & $2.283$    & $2.911$  &$2.454$    & $2.433$  \\
$28$  & $2.610$    & $2.619$  &$1.965$    & $1.946$  \\
 \hline
\end{tabular}
\label{T_comparison_5pi12_7pi12}
\end{table}
\subsection{Transient state} \label{sec:transient}
\begin{figure}[!h] 
	\centering
	\begin{subfigure}{0.3\textwidth}
		\includegraphics[width=\linewidth]{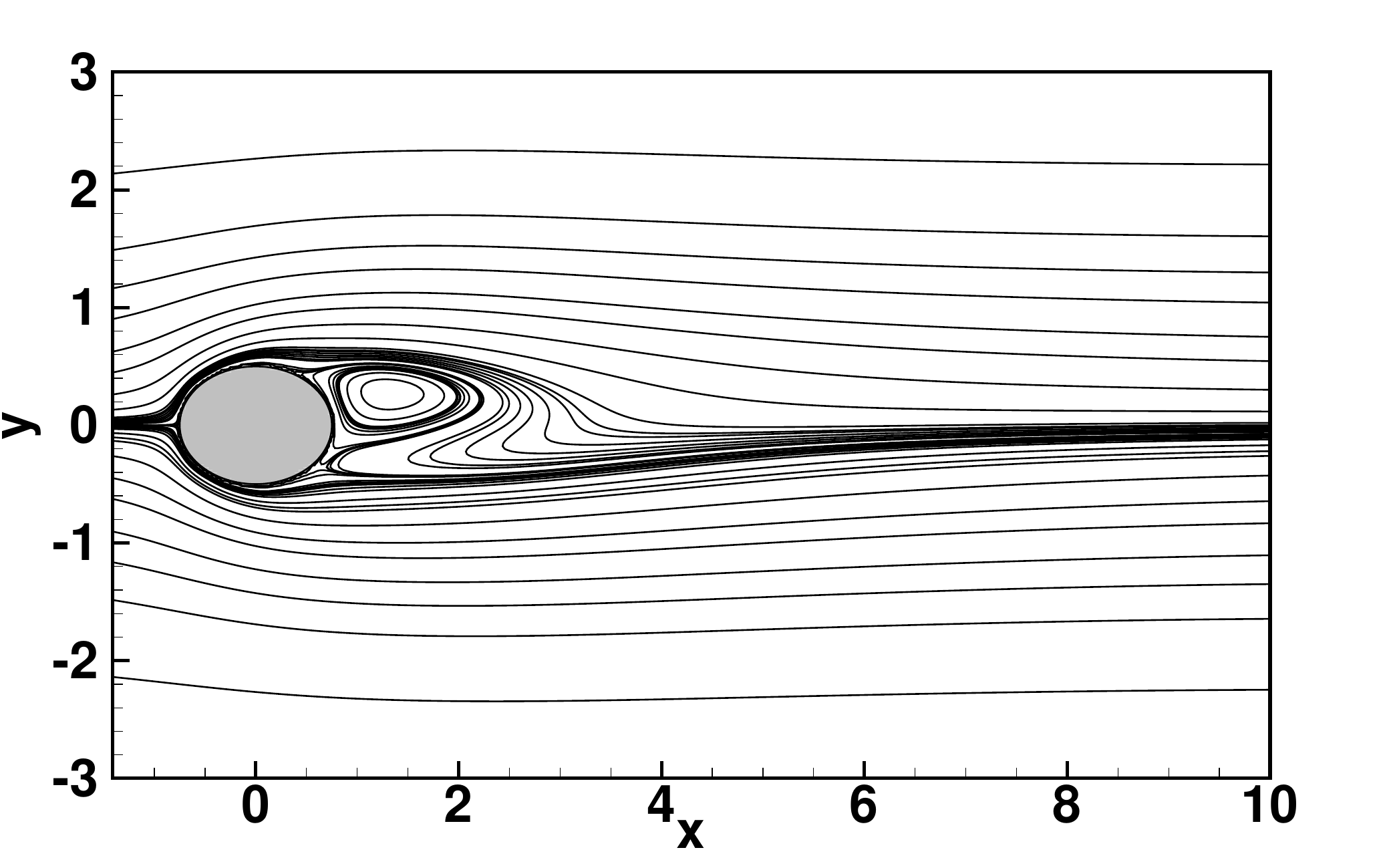} 
	\end{subfigure}\hfil
	\begin{subfigure}{0.3\textwidth}
		\includegraphics[width=\linewidth]{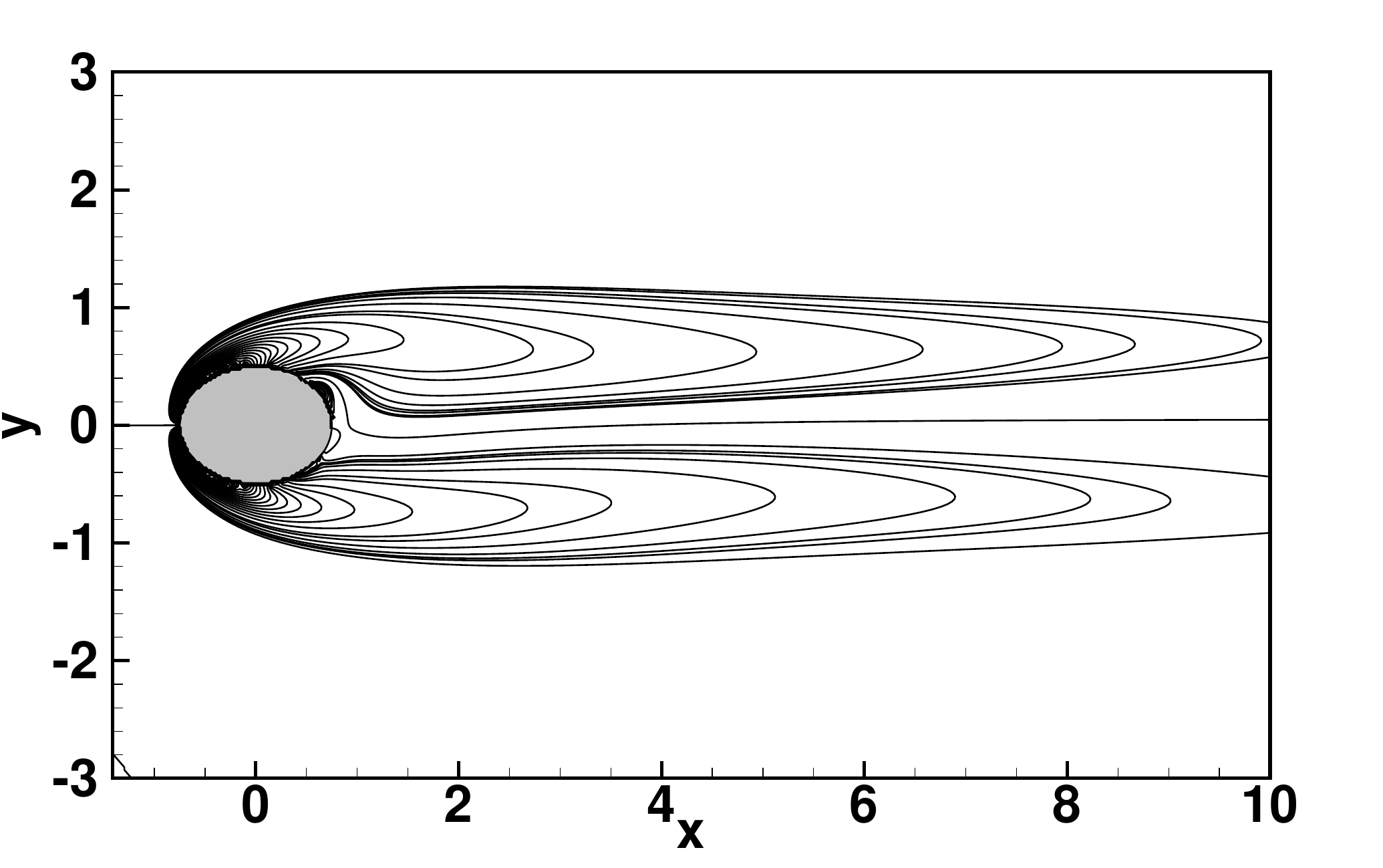} 
	\end{subfigure}\hfil
	\begin{subfigure}{0.3\textwidth}
		\includegraphics[width=\linewidth]{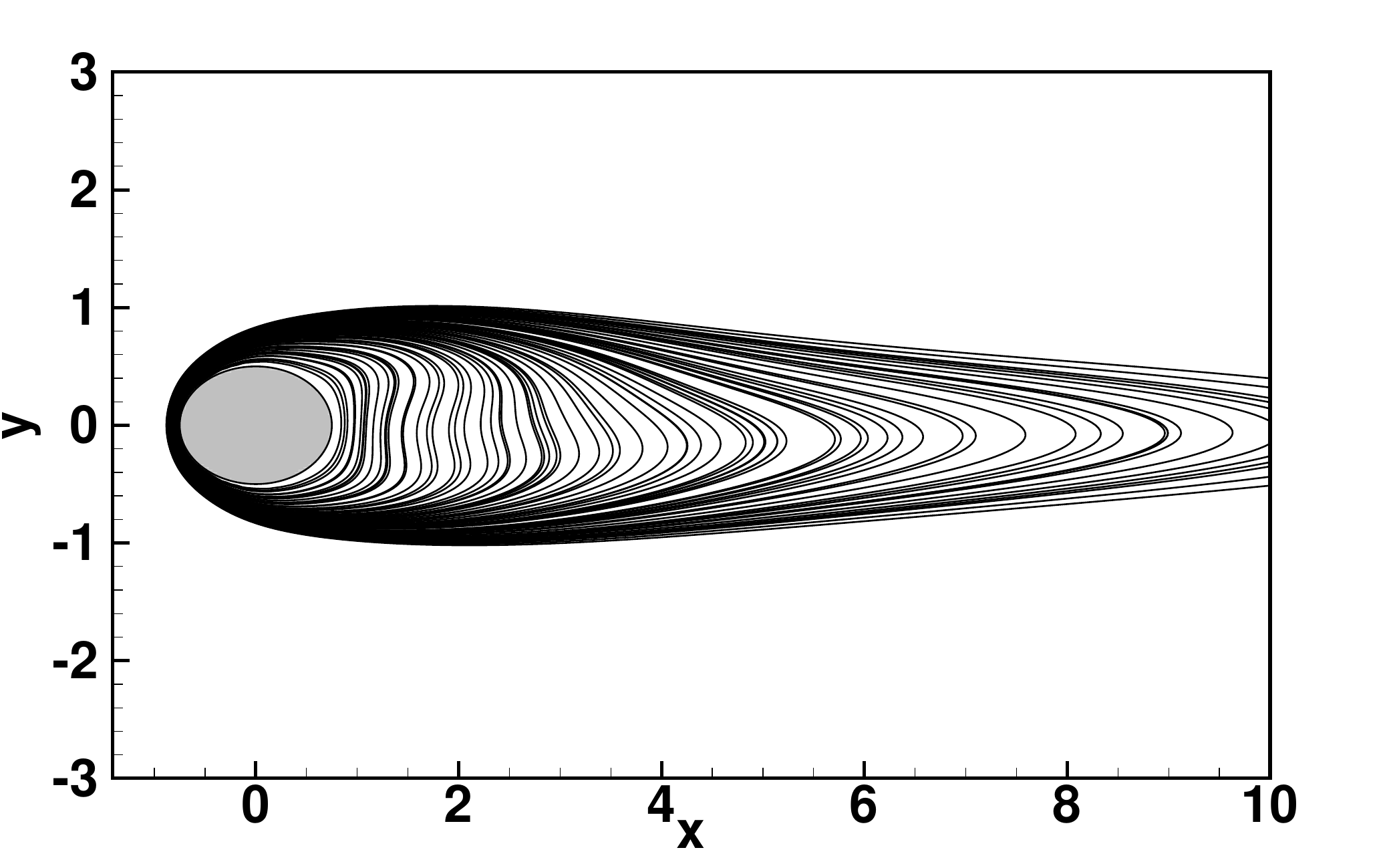} 
	\end{subfigure}\hfill
	\begin{subfigure}{0.3\textwidth}
		\includegraphics[width=\linewidth]{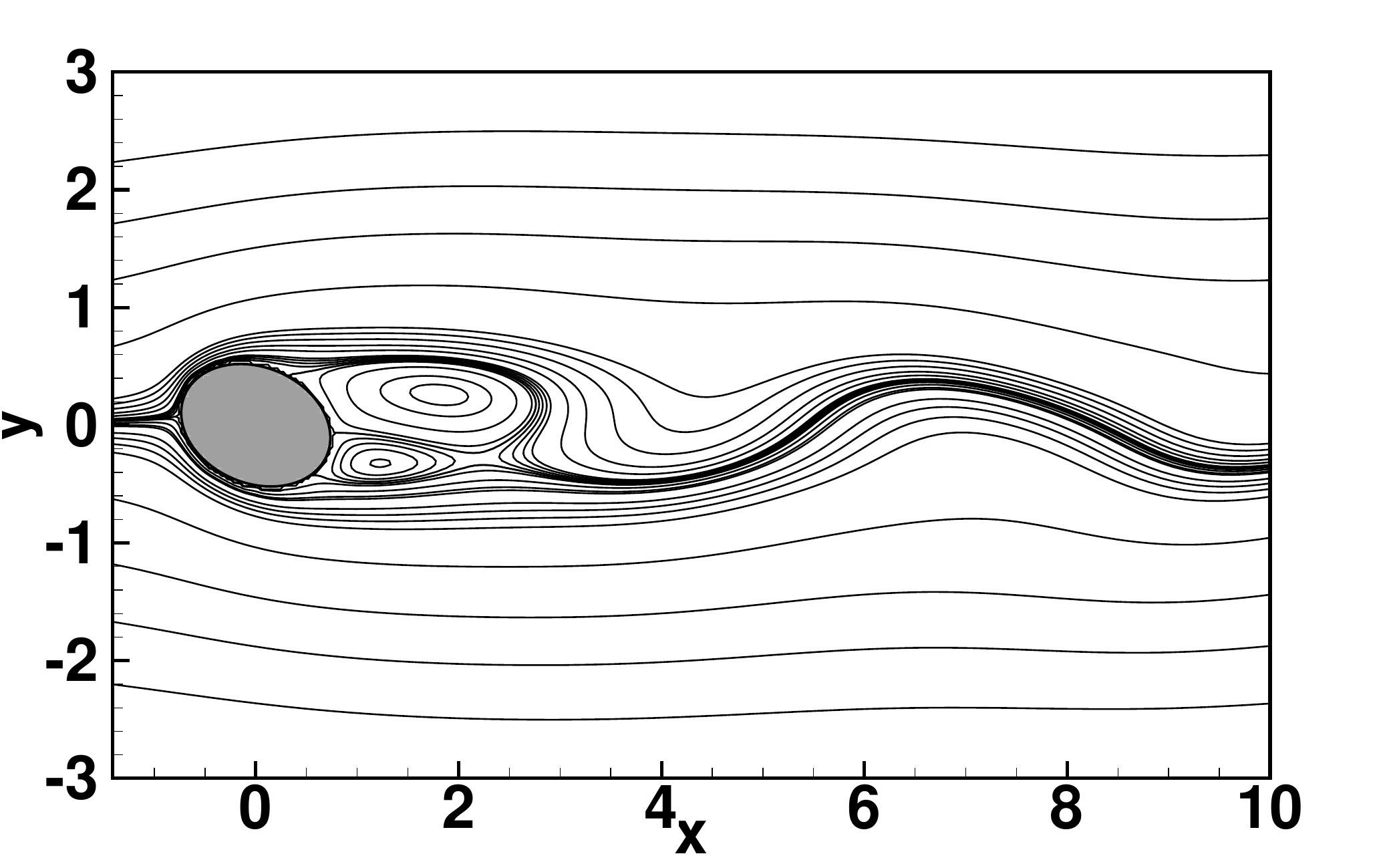} 
	\end{subfigure}\hfil
	\begin{subfigure}{0.3\textwidth}
		\includegraphics[width=\linewidth]{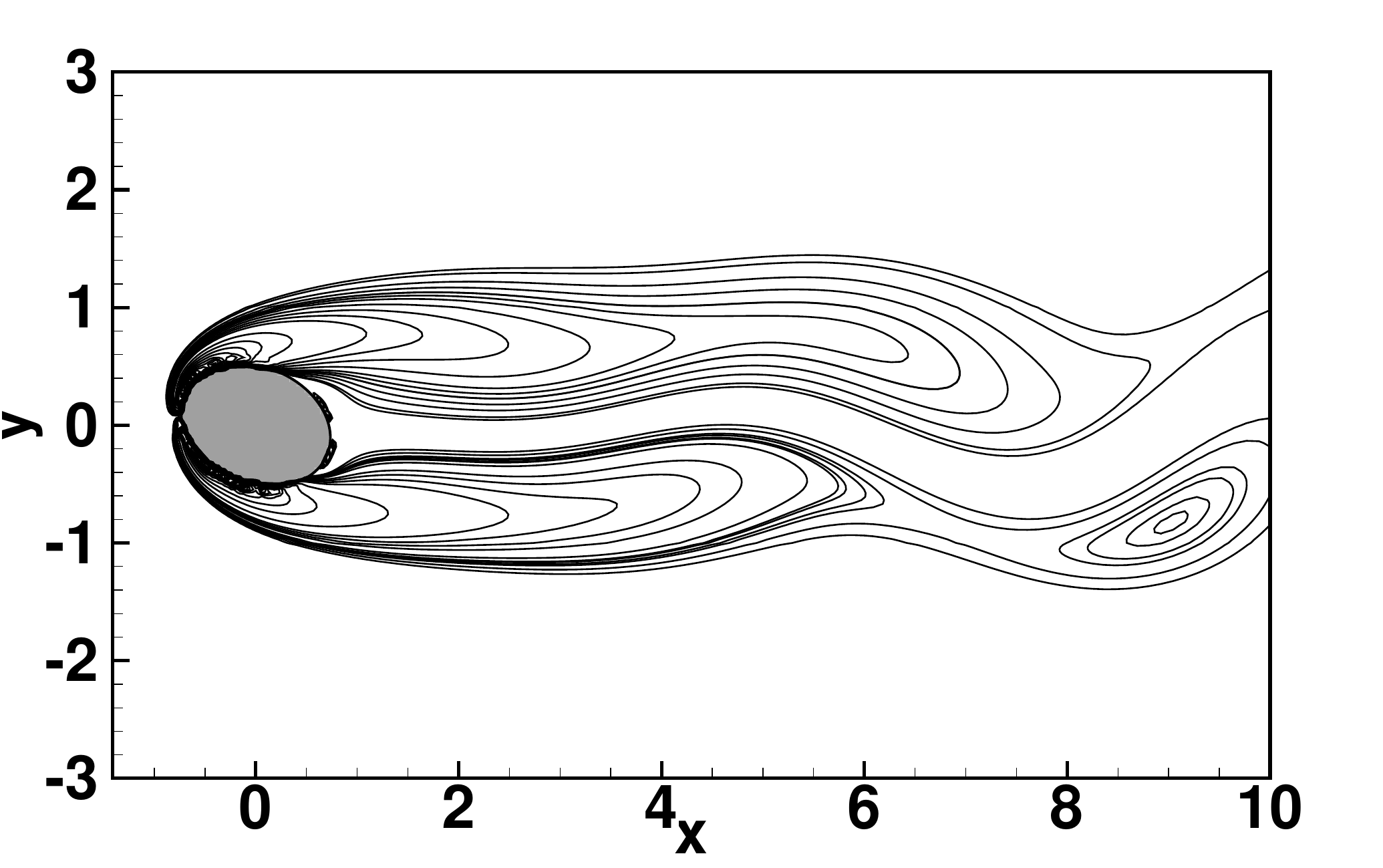} 
	\end{subfigure}\hfil
	\begin{subfigure}{0.3\textwidth}
		\includegraphics[width=\linewidth]{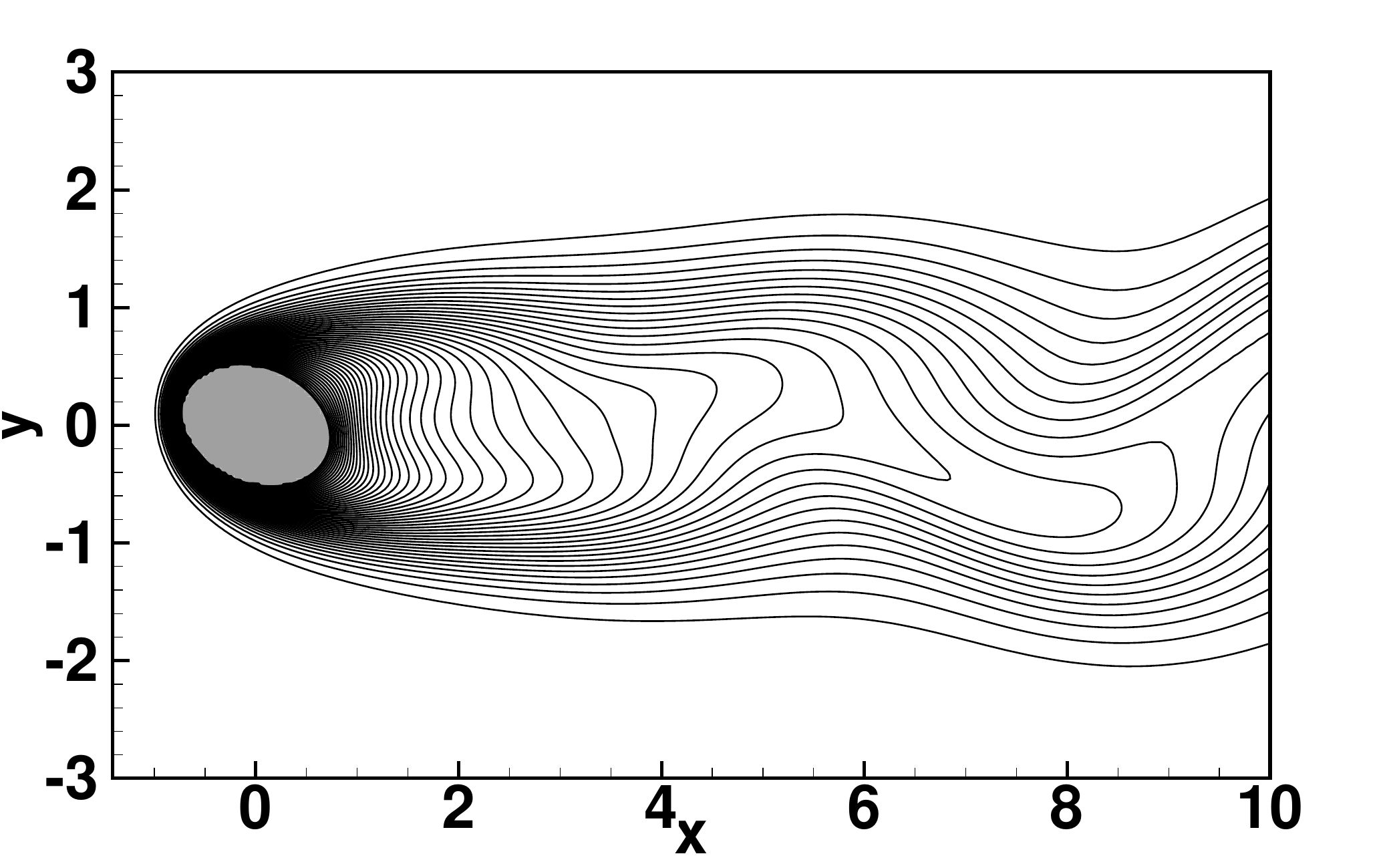} 
	\end{subfigure}\hfil
		\begin{subfigure}{0.3\textwidth}
		\includegraphics[width=\linewidth]{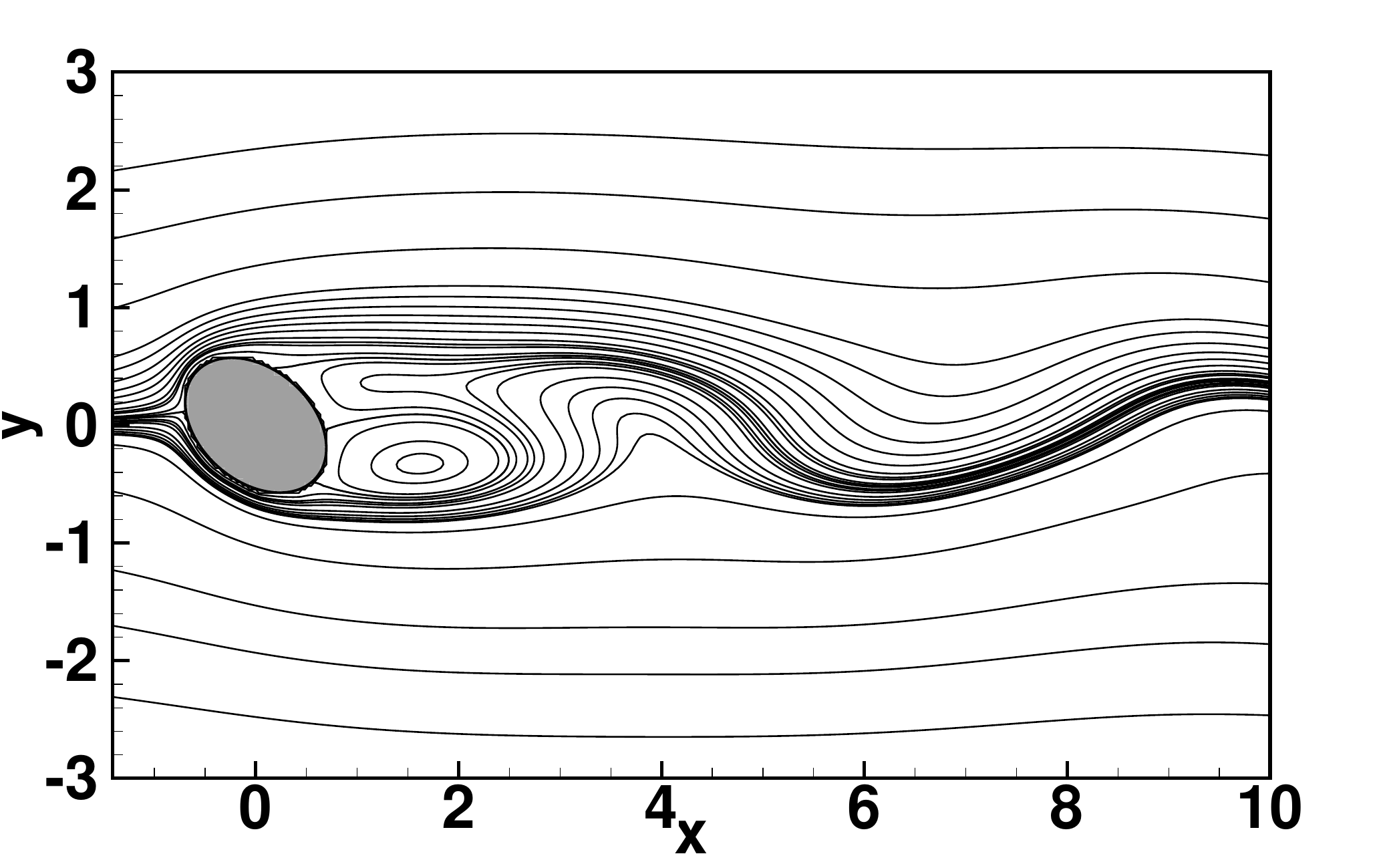} 
	\end{subfigure}\hfil
	\begin{subfigure}{0.3\textwidth}
		\includegraphics[width=\linewidth]{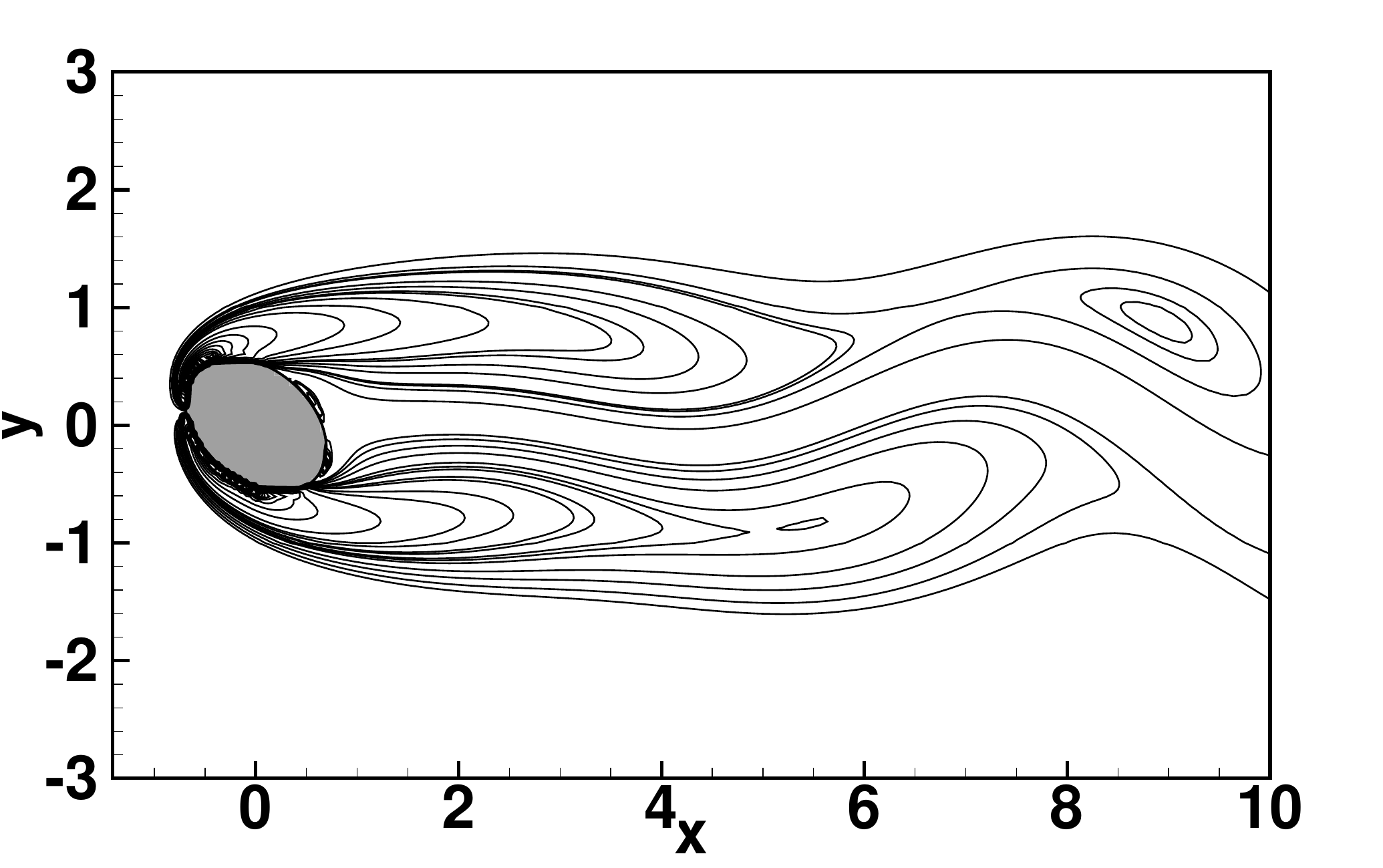} 
	\end{subfigure}\hfil
	\begin{subfigure}{0.3\textwidth}
		\includegraphics[width=\linewidth]{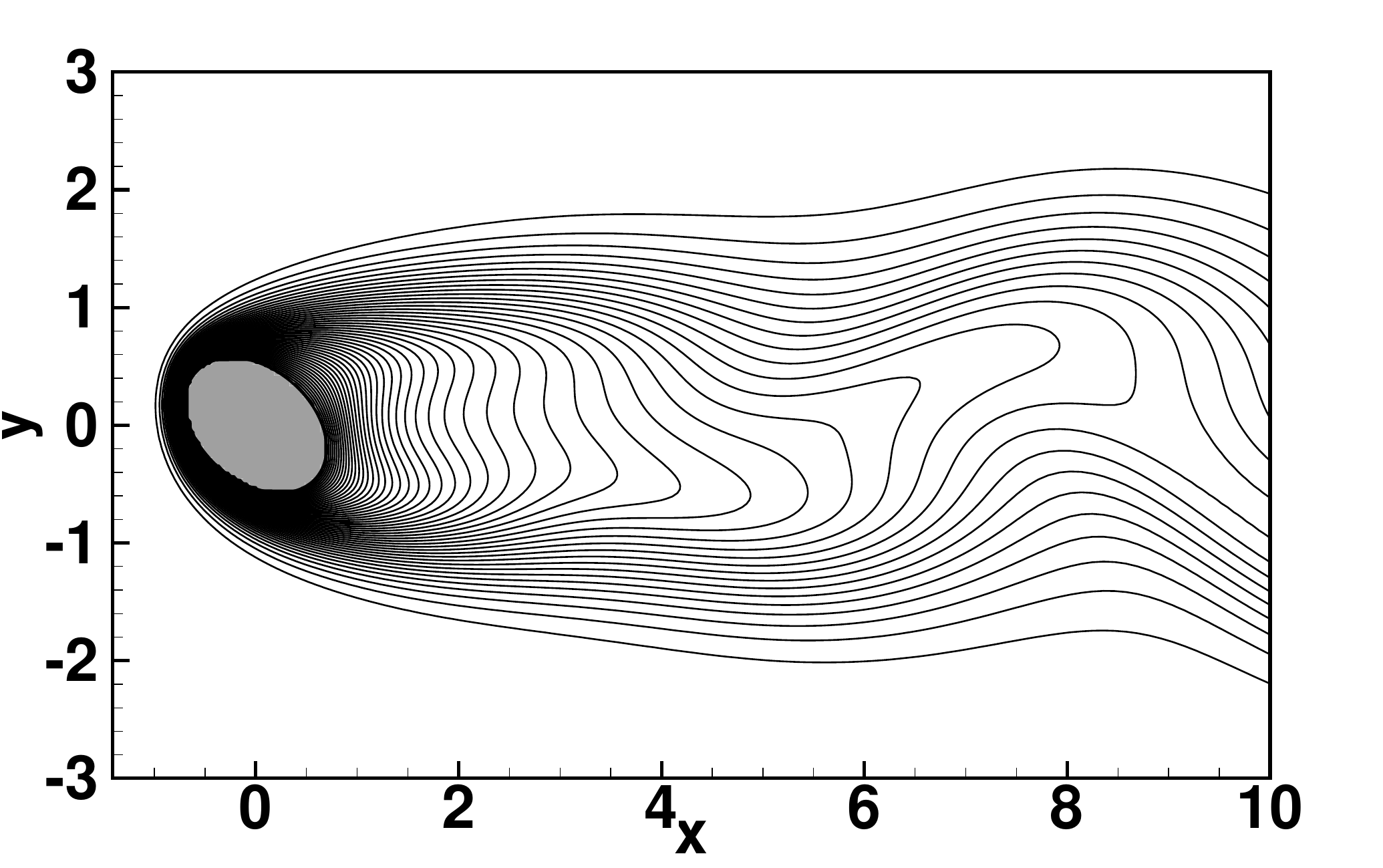} 
	\end{subfigure}\hfil
		\begin{subfigure}{0.3\textwidth}
		\includegraphics[width=\linewidth]{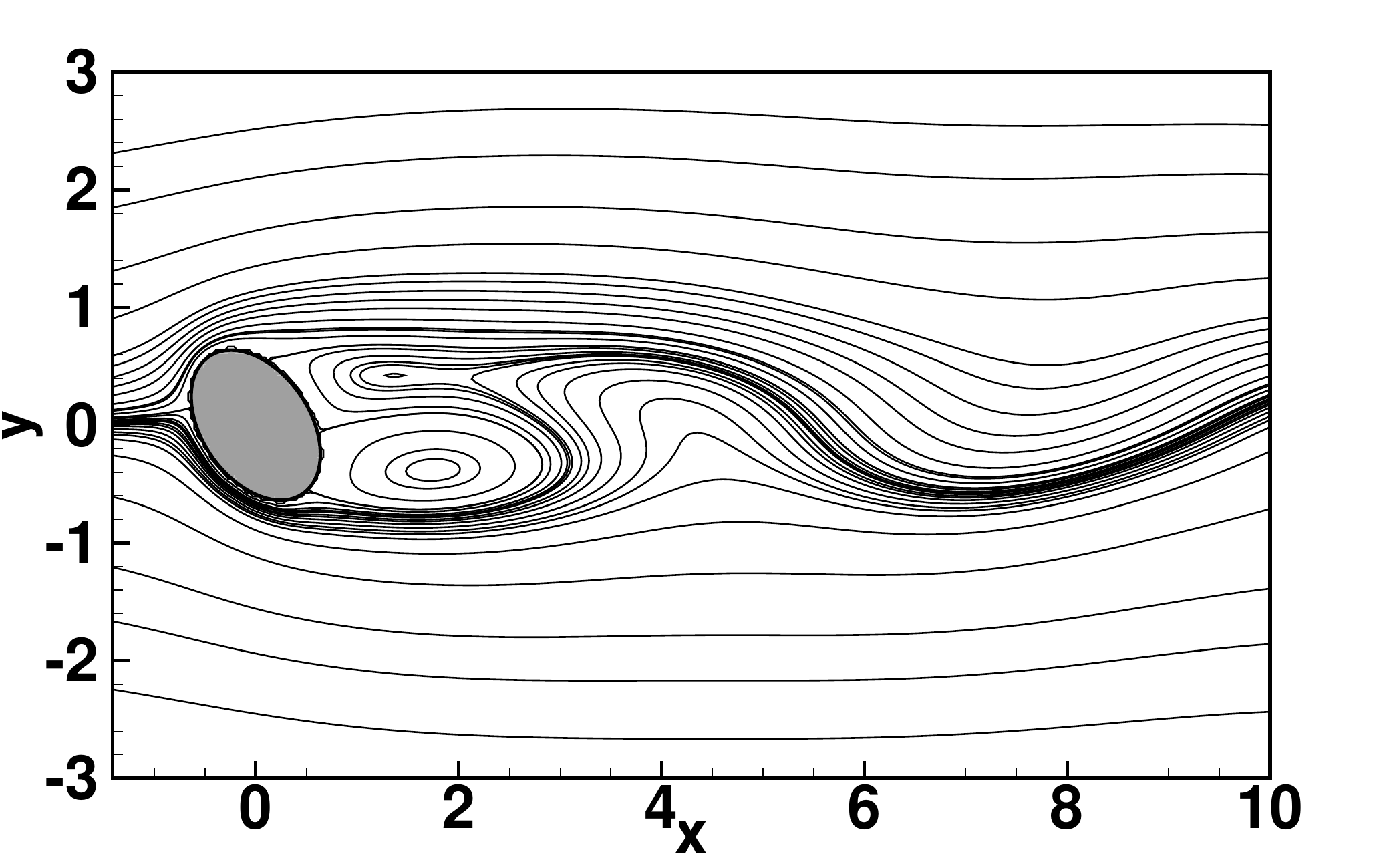} 
	\end{subfigure}\hfil
	\begin{subfigure}{0.3\textwidth}
		\includegraphics[width=\linewidth]{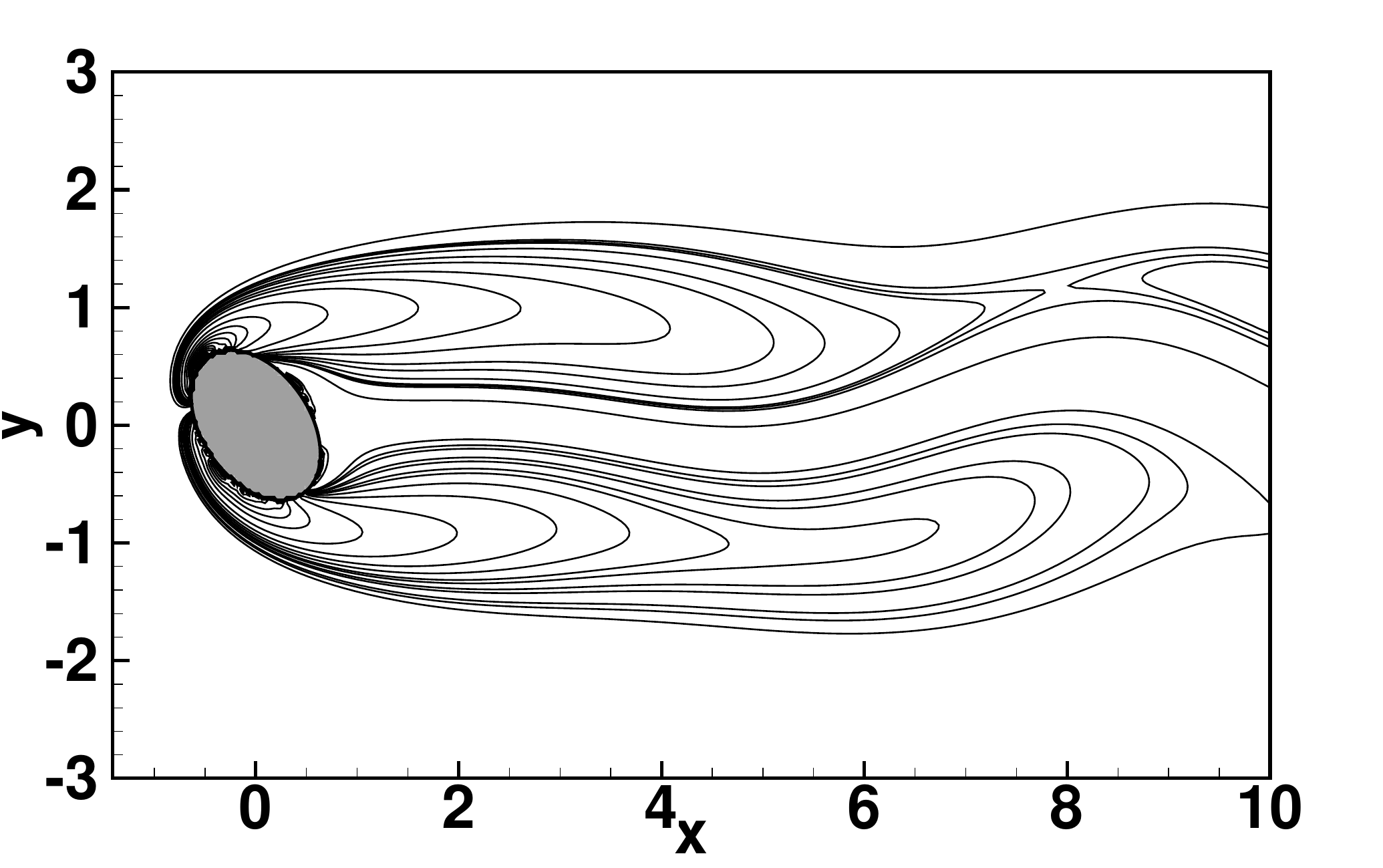} 
	\end{subfigure}\hfil
	\begin{subfigure}{0.3\textwidth}
		\includegraphics[width=\linewidth]{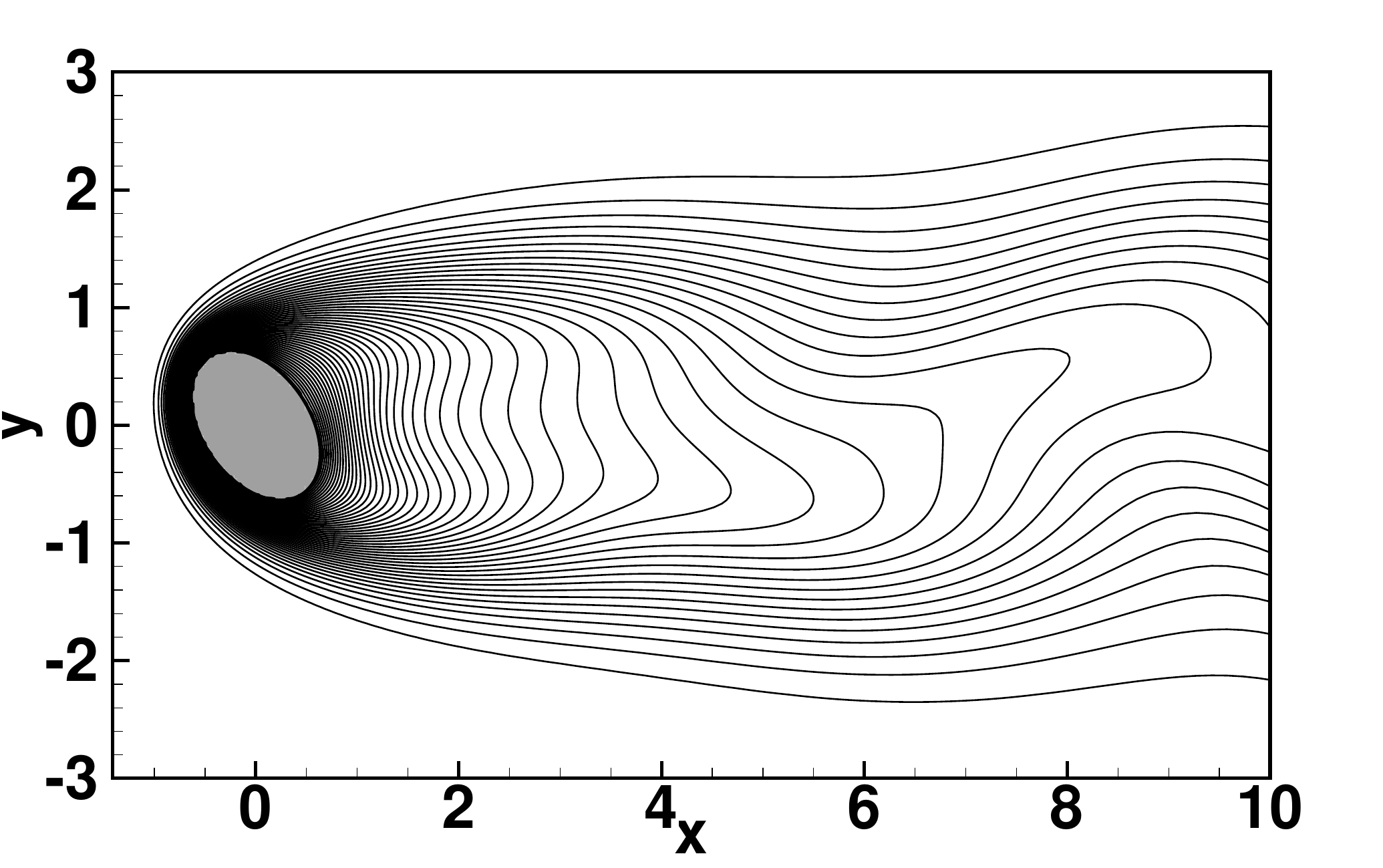} 
	\end{subfigure}\hfil
		\begin{subfigure}{0.3\textwidth}
		\includegraphics[width=\linewidth]{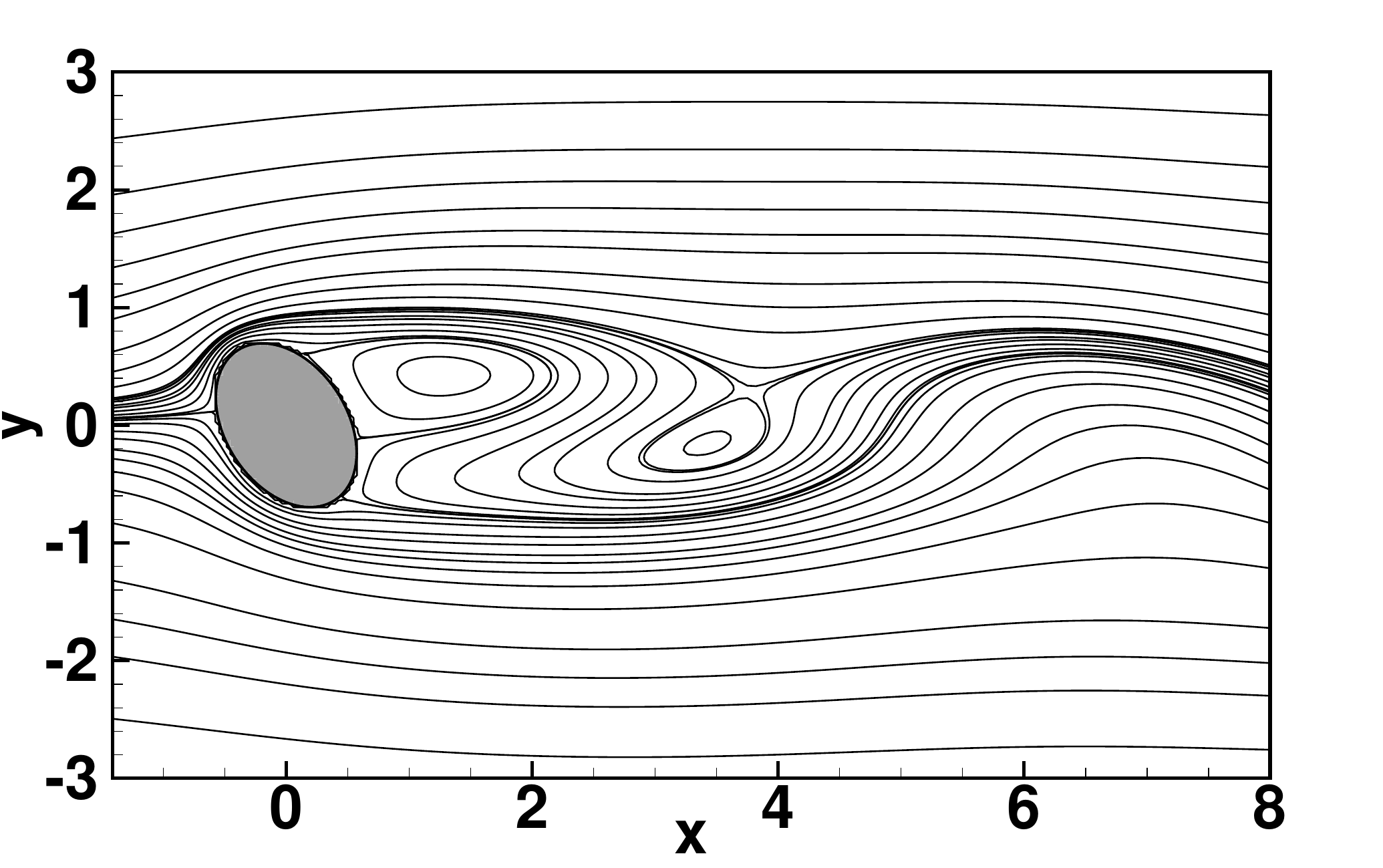} 
	\end{subfigure}\hfil
	\begin{subfigure}{0.3\textwidth}
		\includegraphics[width=\linewidth]{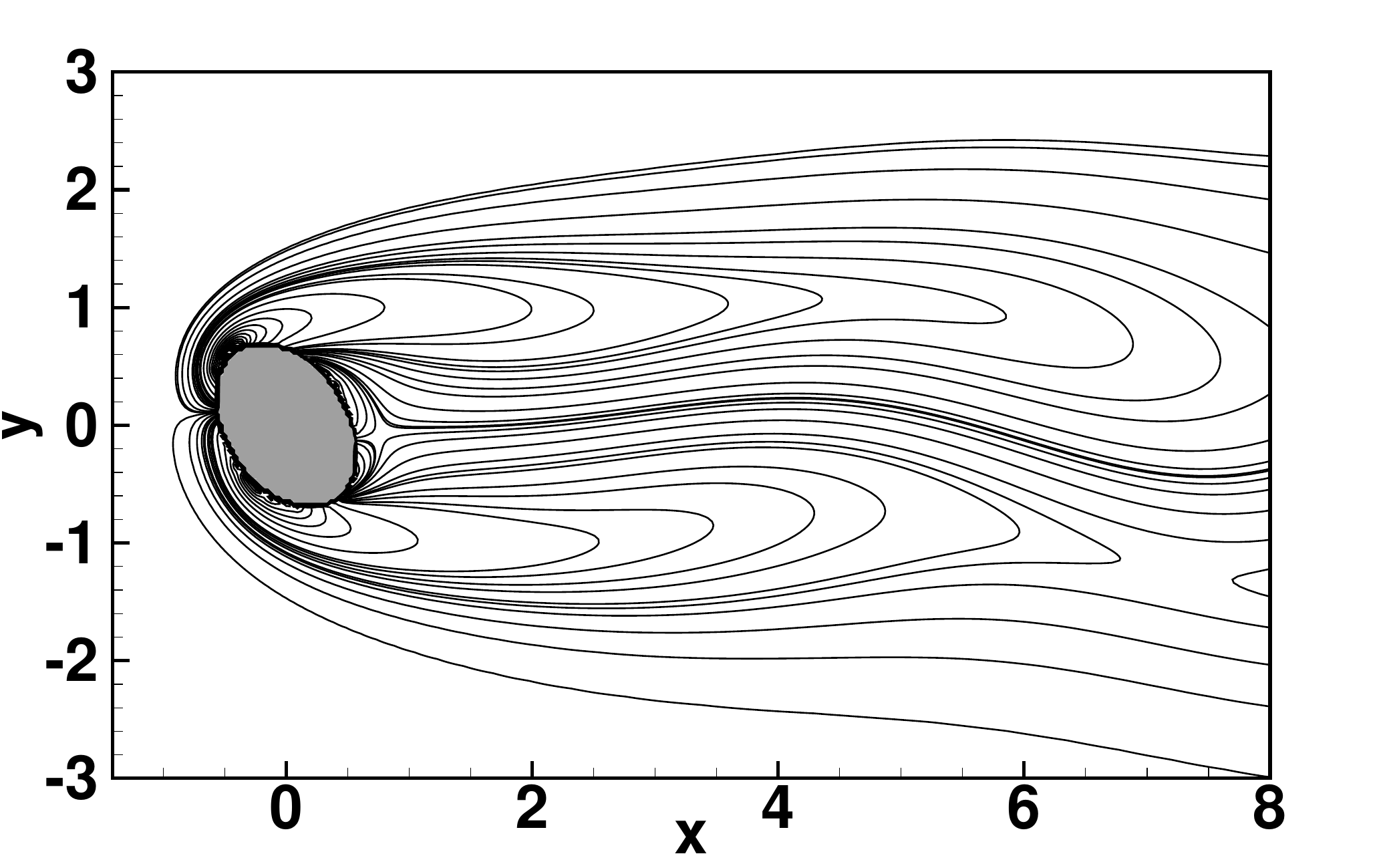} 
	\end{subfigure}\hfil
	\begin{subfigure}{0.3\textwidth}
		\includegraphics[width=\linewidth]{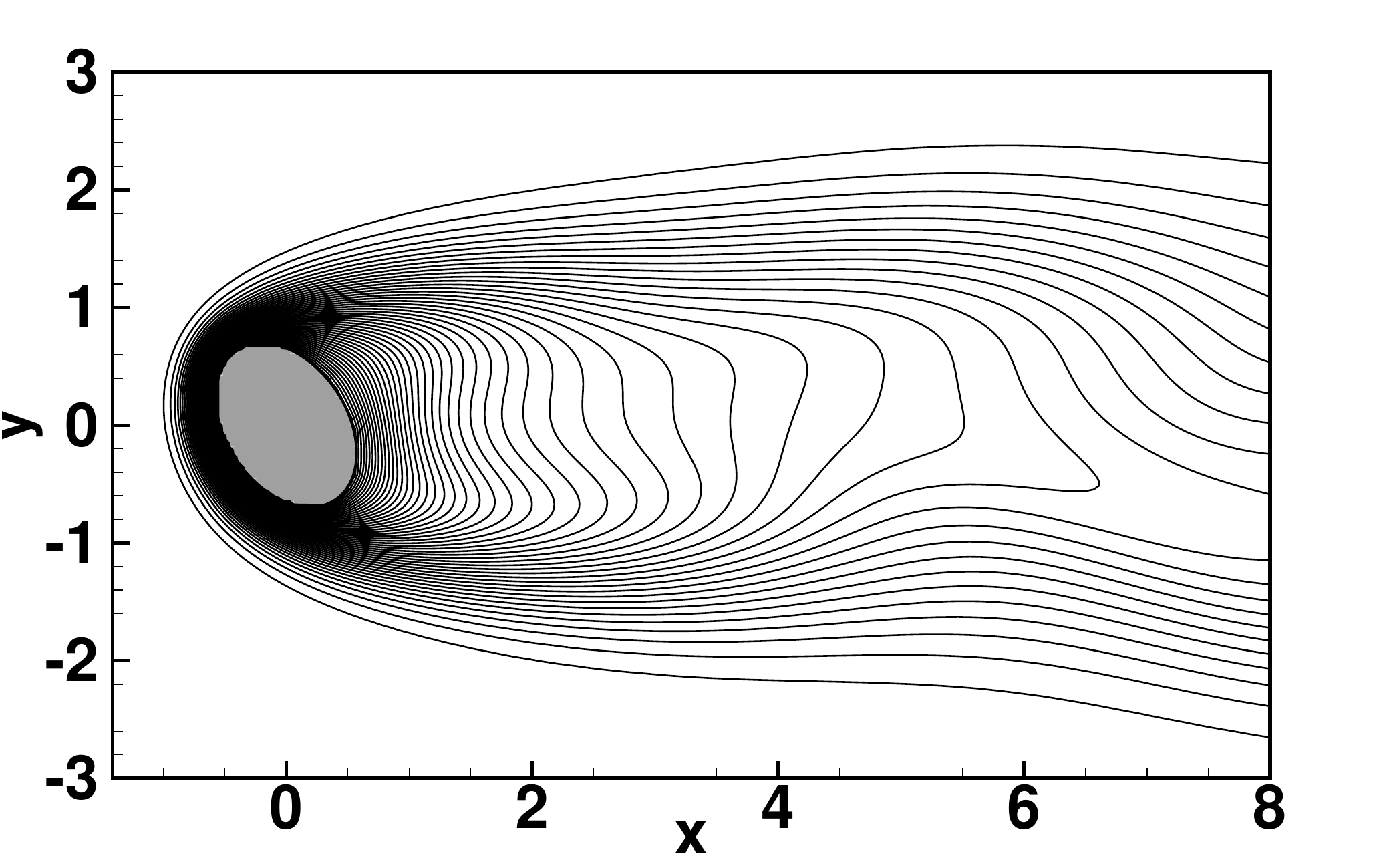} 
	\end{subfigure}\hfil
		\begin{subfigure}{0.3\textwidth}
		\includegraphics[width=\linewidth]{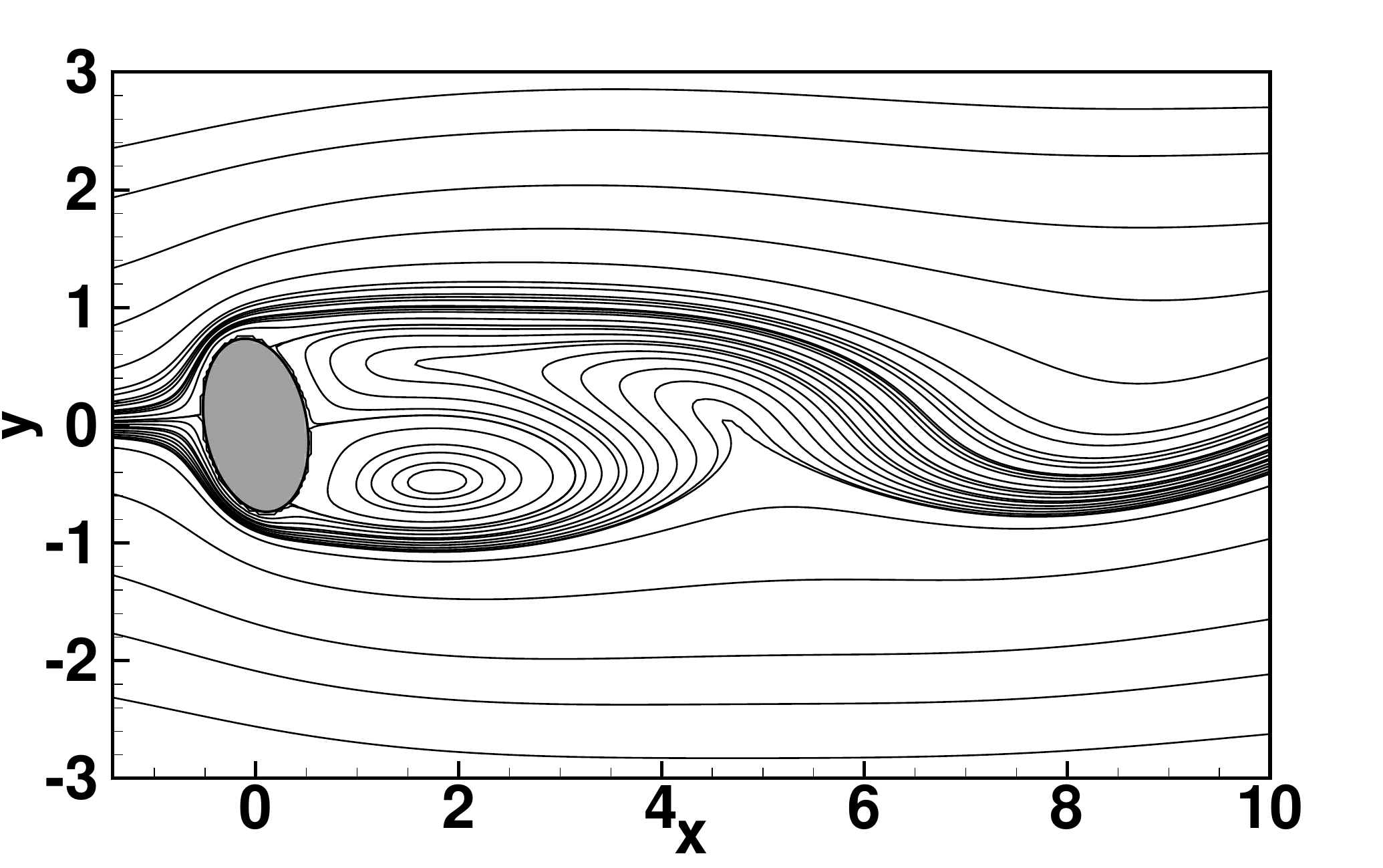} 
	\end{subfigure}\hfil
	\begin{subfigure}{0.3\textwidth}
		\includegraphics[width=\linewidth]{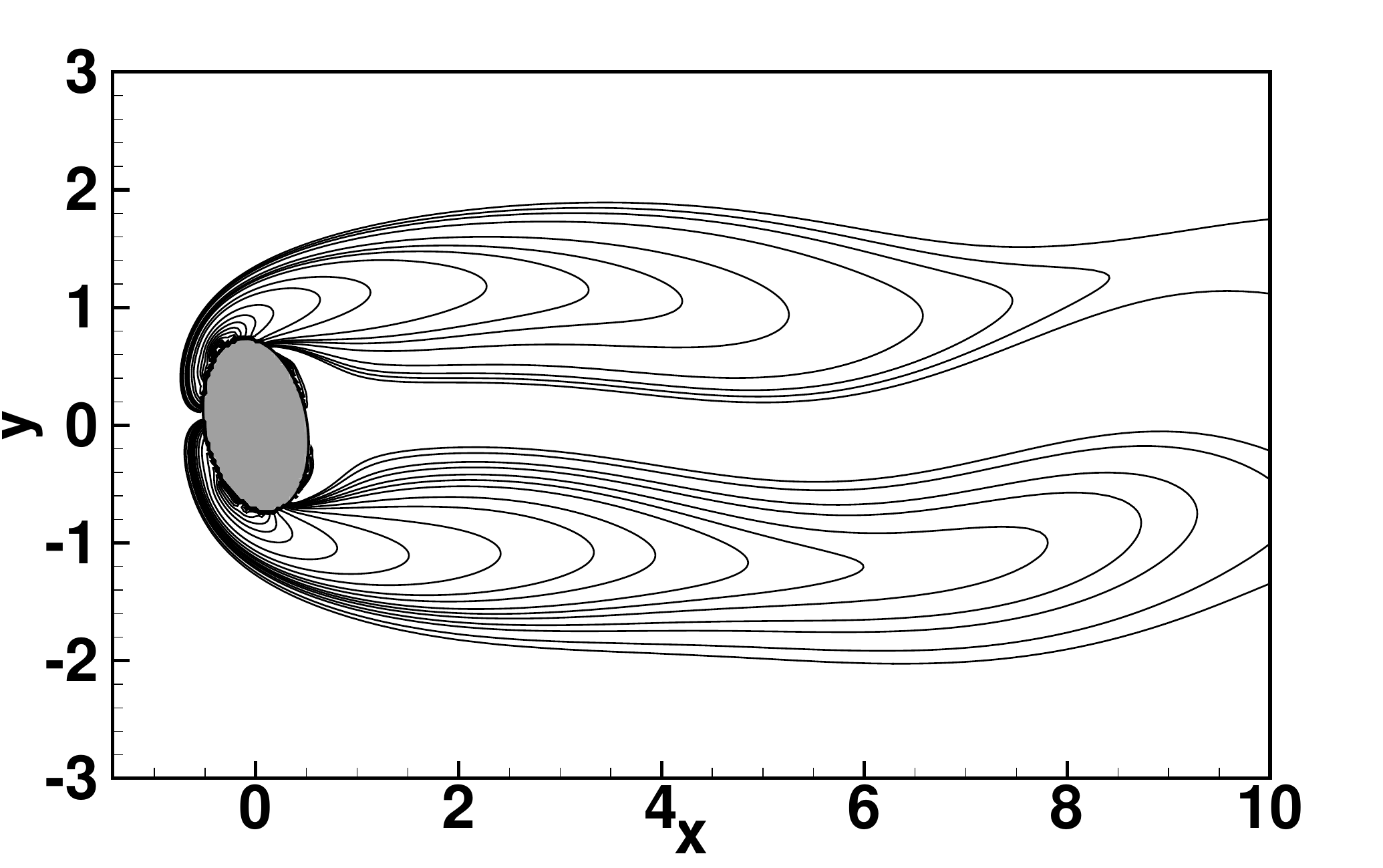} 
	\end{subfigure}\hfil
	\begin{subfigure}{0.3\textwidth}
		\includegraphics[width=\linewidth]{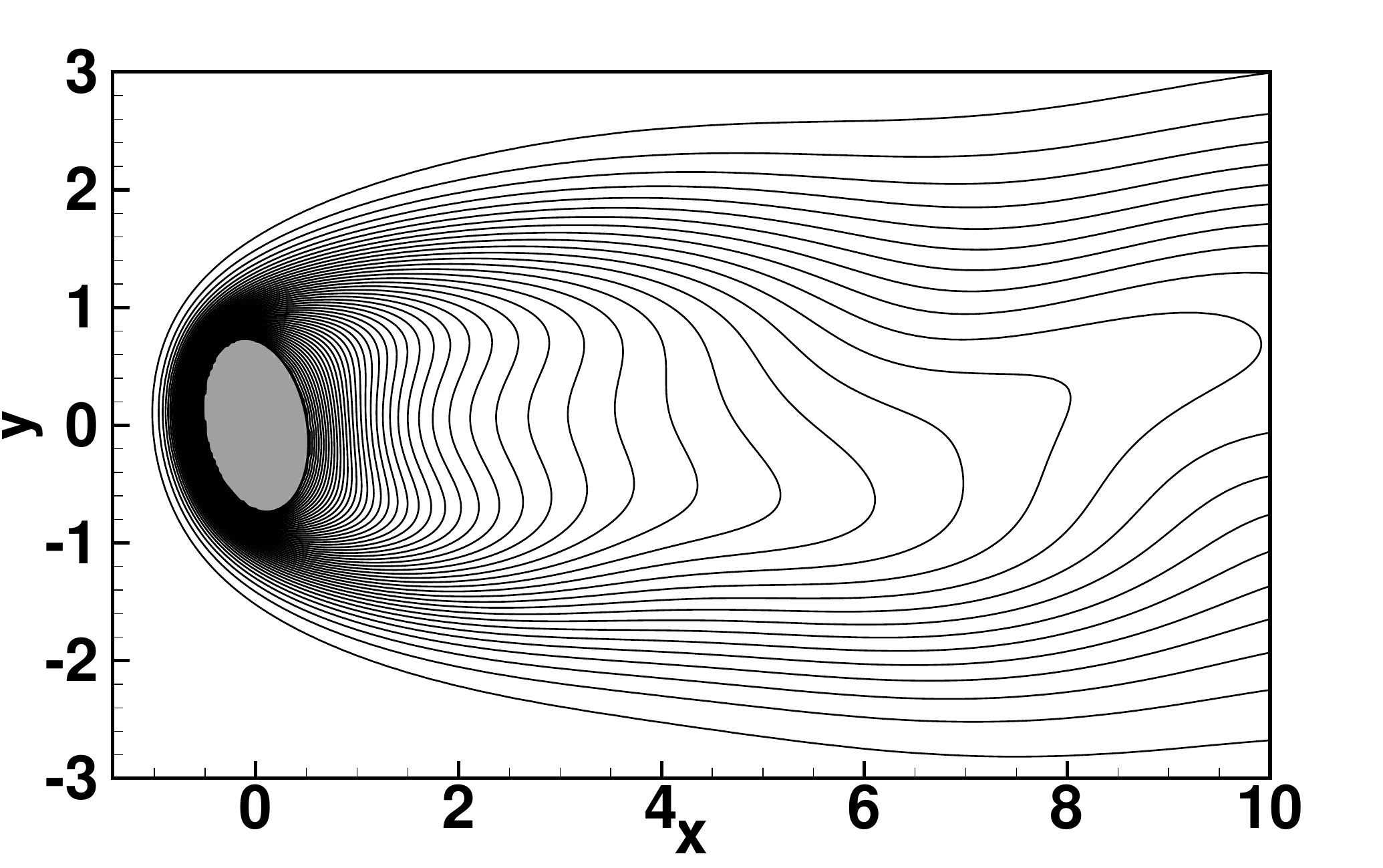} 
	\end{subfigure}\hfil
		\begin{subfigure}{0.3\textwidth}
		\includegraphics[width=\linewidth]{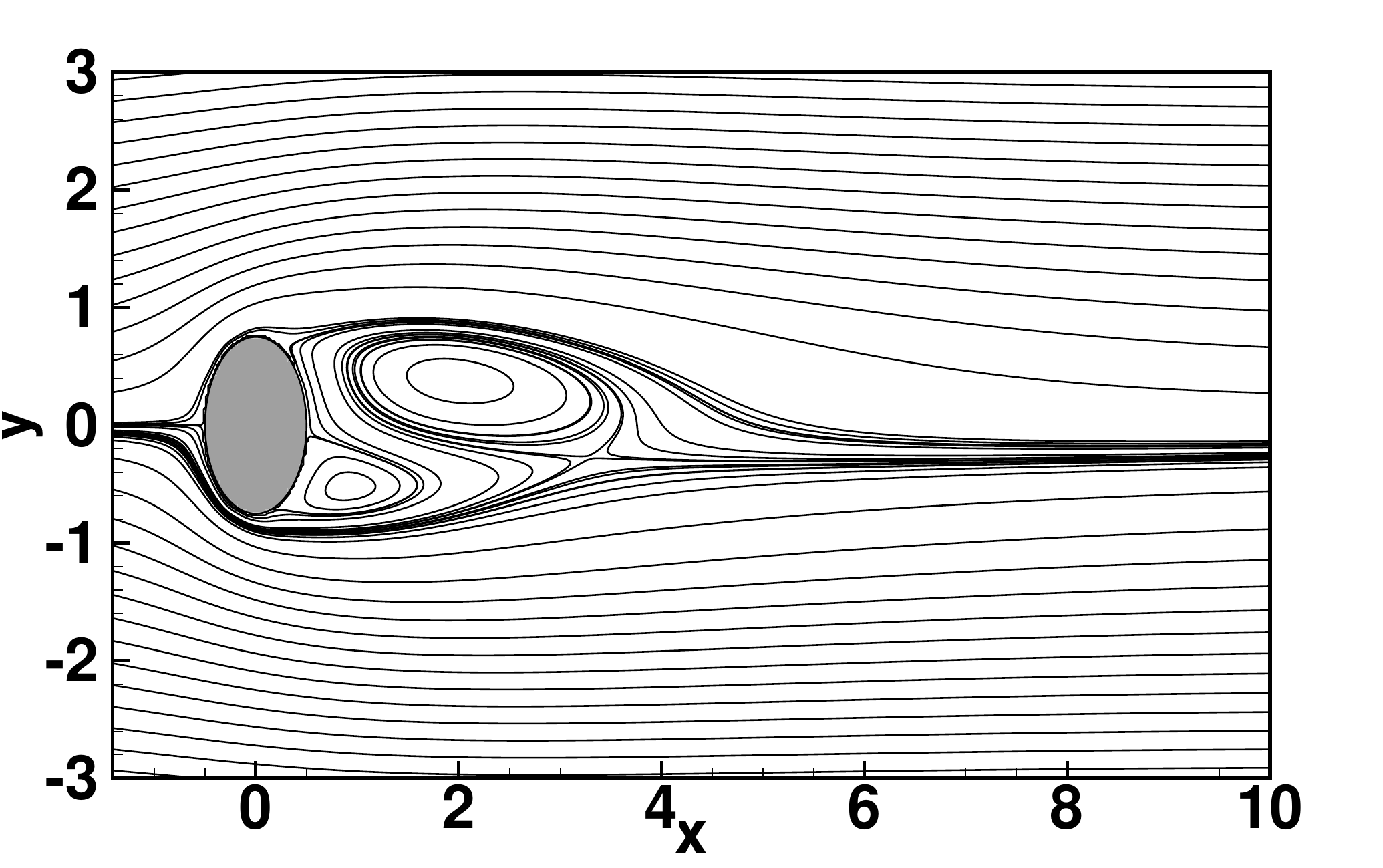} 
	\end{subfigure}\hfil
	\begin{subfigure}{0.3\textwidth}
		\includegraphics[width=\linewidth]{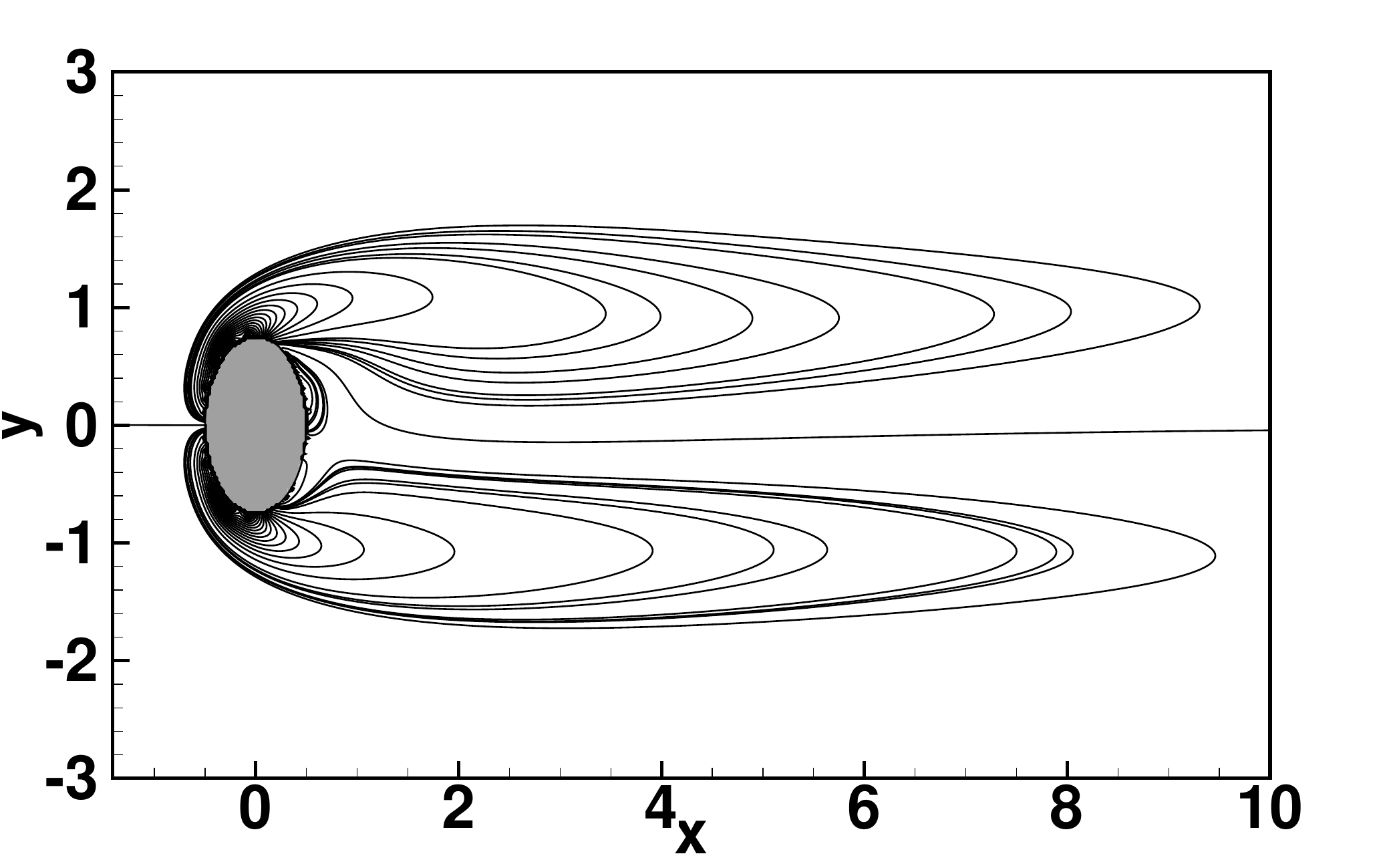} 
	\end{subfigure}\hfil
	\begin{subfigure}{0.3\textwidth}
		\includegraphics[width=\linewidth]{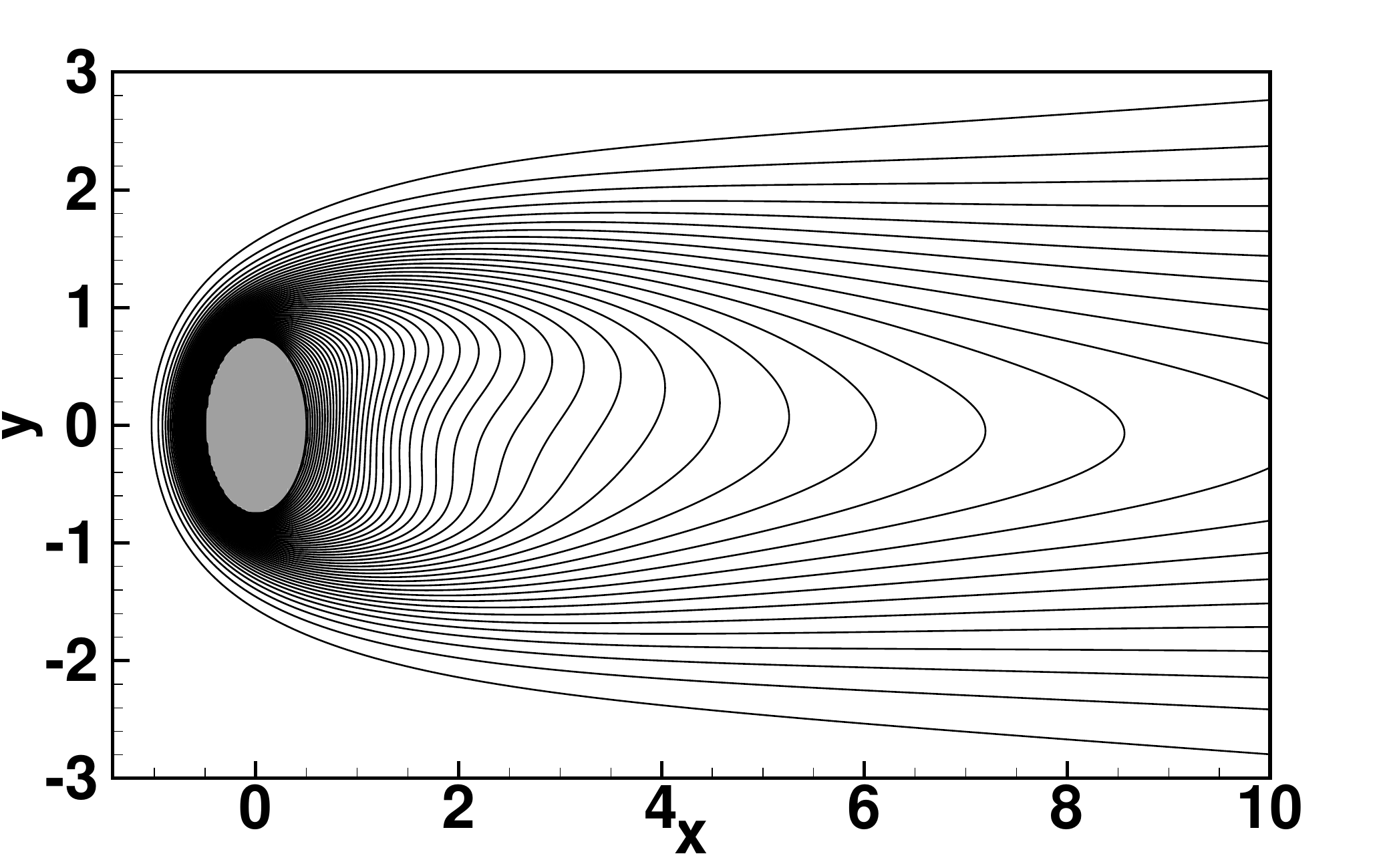} 
	\end{subfigure}\hfil
\caption{{\sl Transition to unsteadiness, flow in the vicinity of Critical Reynolds number: Streamlines (left), vorticity (middle) and isotherms (right)  for the combinations $(\theta, Re)$, from top to bottom, $ (0^{\degree},60)$, $(15^{\degree}, 60)$, $(30^{\degree}, 50)$,  $(45^{\degree}, 39)$, $(60^{\degree}, 32)$, $(75^{\degree}, 29)$, and $(90^{\degree}, 26)$, respectively.} }
\label{fig:unsteady1-critical-Re}
\end{figure}
\clearpage
In this section, we present the results for transient state flow phenomena in terms of streamlines, vorticity contours, isotherms, as well as the force coefficients ($C_D$, $C_L$), surface averaged Nusselt numbers, and Strouhal numbers. Computations were carried for $Re_c \leq Re \leq 120$, and $0 \leq \theta \leq 180^{\degree}$. However, as we noted in section \ref{sec:steady}, the flow phenomena for $ 90^{\degree} < \theta < 180^{\degree}$ is a mirror image of the flow phenomena for $ 0^{\degree} < \theta < 90^{\degree}$. Hence, we present results only for $0^{\degree} \leq \theta \leq 90^{\degree}$.

In general, in the unsteady regime, two rows of well defined vortices are formed with clockwise vortices being shed from the upper side of the cylinder and counterclockwise vortices from the lower side. This is the well known von Karman vortex street that stretches over the entire downstream region in the wake of the cylinder. Since the mechanism of vortex shedding remains same for all values of $\theta$ considered, we take $Re = 100$ as the representative Reynolds number for which we present our analysis. Quantitative parameters like Strouhal number, Drag and Lift forces, and Nusselt number will be discussed at length subsequently. Note that the flow becomes unsteady beyond the critical Reynolds number, $Re_c$. However, it is not necessary that vortex shedding commences immediately after $Re_c$. Thus, for some cases even though the flow becomes unsteady at $Re_c$, vortex shedding is seen to commence for Reynolds numbers slightly higher than $Re_c$. To exactly pinpoint the critical Reynolds number at which vortex shedding commences would require a separate study. Therefore, in order to have a fair enough idea about $Re_c$ for different inclinations of the elliptic cylinder, we plot the streamlines, vorticity contours and the isotherm contours for the $(\theta, Re)$ combination in Figure \ref{fig:unsteady1-critical-Re} such that the flow for $(\theta, Re-1)$  is always steady. These figures clearly demonstrate the unsteady nature of the flow and as such  $Re_c \in (Re-1, Re]$,  for the Reynolds numbers considered in these figures. Interestingly, the mirror phenomena described above holds true for the critical Reynolds number as well, that is, $Re_c$ is same for $\theta$ and 
$180^\degree-\theta$ for all $0^\degree \leq \theta \leq 90^\degree$.
\subsubsection{Flow field and isotherms}
Figure \ref{Fig:unsteady-psi-vort-temp-0} shows the instantaneous streamlines, vorticity contours, and isotherms for $Re = 100$ and $\theta = 0^{\degree}$ at different instants of time in a complete vortex shedding cycle. Here $T$ represents the time period of vortex shedding, and the flow patterns are shown at equal intervals of $T/4$ within a vortex shedding cycle. We can see that the growth of the upper vortex is accompanied by the formation of a lower vortex in the flow field (figure \ref{Fig:unsteady-psi-vort-temp-0} (a)). While the upper vortex begins to decay, the lower vortex grows and attaches itself to the trailing edge (figure \ref{Fig:unsteady-psi-vort-temp-0} (b)). Subsequently, the upper vortex reappears around the leading edge and grows in such a way that it suppresses the lower vortex, which starts to get
smaller (figure \ref{Fig:unsteady-psi-vort-temp-0} (c), \ref{Fig:unsteady-psi-vort-temp-0} (d)). This process is repeated for the shedding cycle.

\begin{figure}[H]
	\centering
	\begin{subfigure}{\textwidth}
		\includegraphics[width=\linewidth]{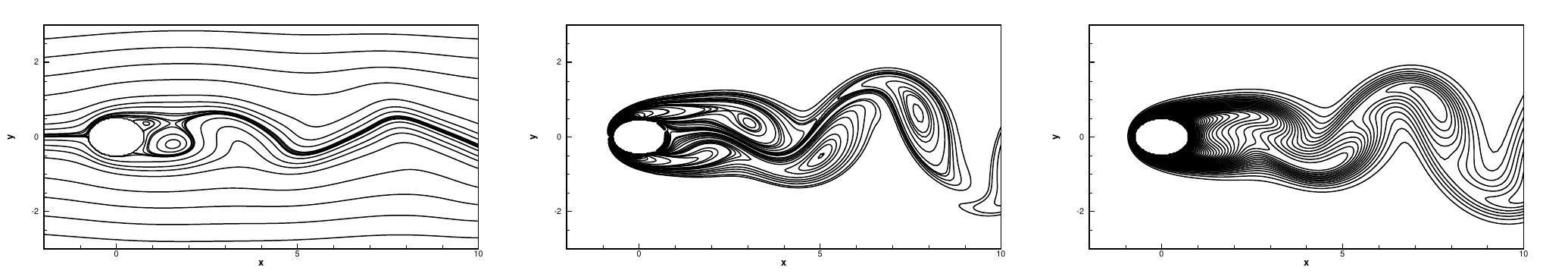} 
		\caption{$t+T/4$}
	\end{subfigure}\hfil 
	\begin{subfigure}{\textwidth}
		\includegraphics[width=\linewidth]{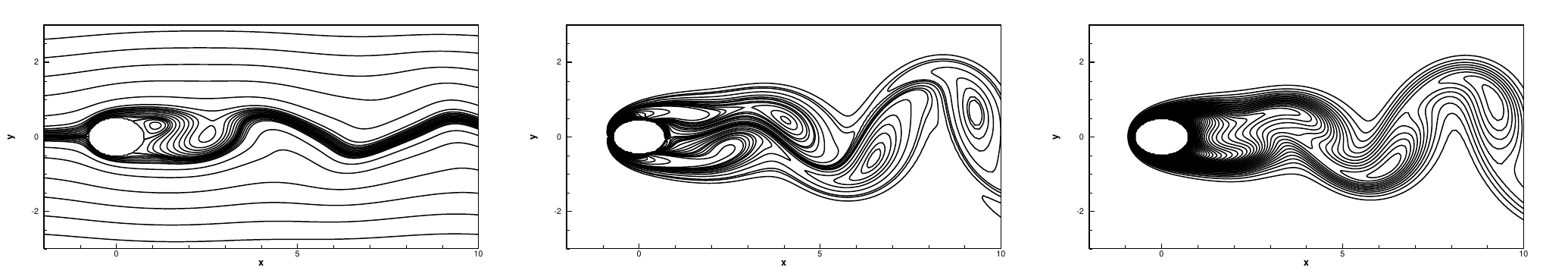} 
		\caption{$t+T/2$}
	\end{subfigure}\hfil 
	\begin{subfigure}{\textwidth}
		\includegraphics[width=\linewidth]{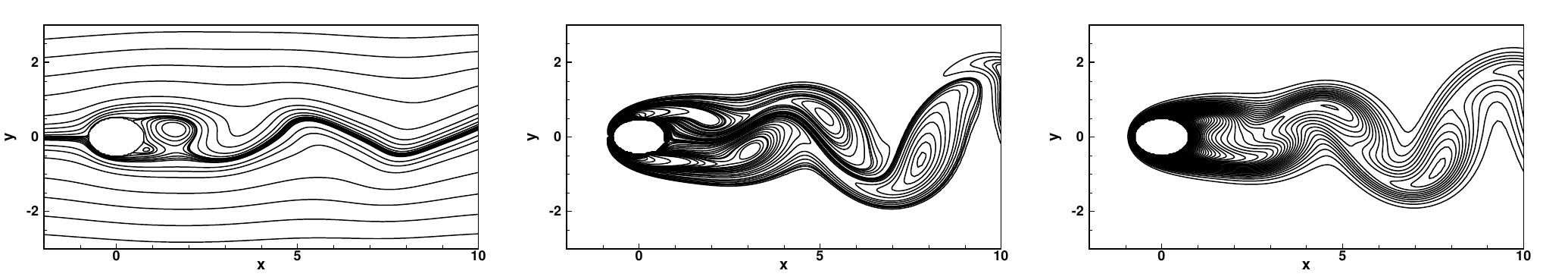} 
		\caption{$t+3T/4$}
	\end{subfigure}\hfil 
	\begin{subfigure}{\textwidth}
		\includegraphics[width=\linewidth]{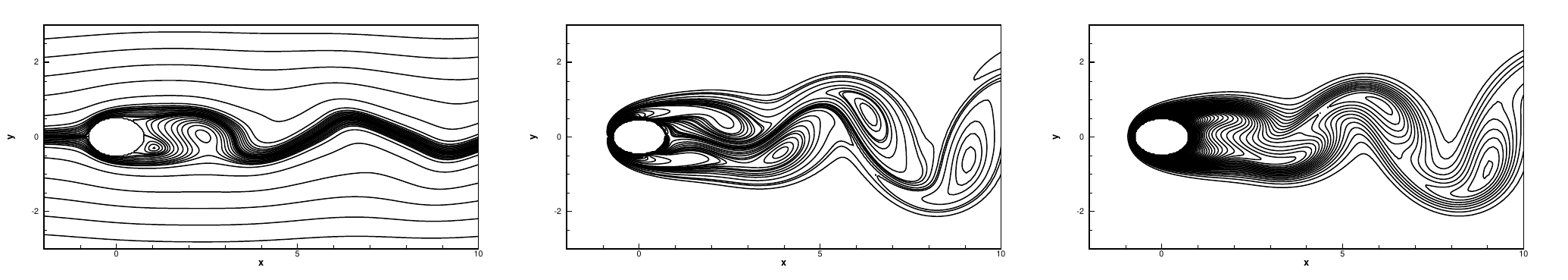} 
		\caption{$t+T$}
	\end{subfigure}\hfil 
	\caption{\small{Instantaneous streamlines (left), vorticity contours (middle) and isotherms (right) within a vortex shedding period for $Re = 100$ and $\theta=0^{\degree}$.}}
	\label{Fig:unsteady-psi-vort-temp-0}
\end{figure}
Figures \ref{Fig:unsteady-psi-vort-temp-45} and \ref{Fig:unsteady-psi-vort-temp-75} show the instantaneous streamlines, vorticity contours and isotherms for $Re = 100$, and $\theta = 45^{\degree}$, $75^{\degree}$ respectively. We see that as the angle of incidence is increased, the undulations in the streamlines become progressively complex. Also, vortex shedding occurs at a shorter distance from the trailing edge of the cylinder, becoming considerably wider as $\theta$ is increased.

\begin{figure}[H]
	\centering
	\begin{subfigure}{\textwidth}
		\includegraphics[width=\linewidth]{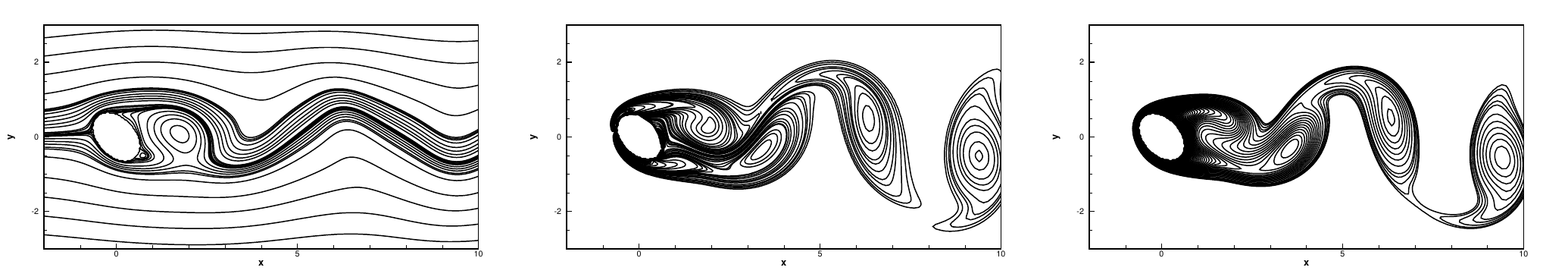} 
		\caption{$t+T/4$}
	\end{subfigure}\hfil 
	\begin{subfigure}{\textwidth}
		\includegraphics[width=\linewidth]{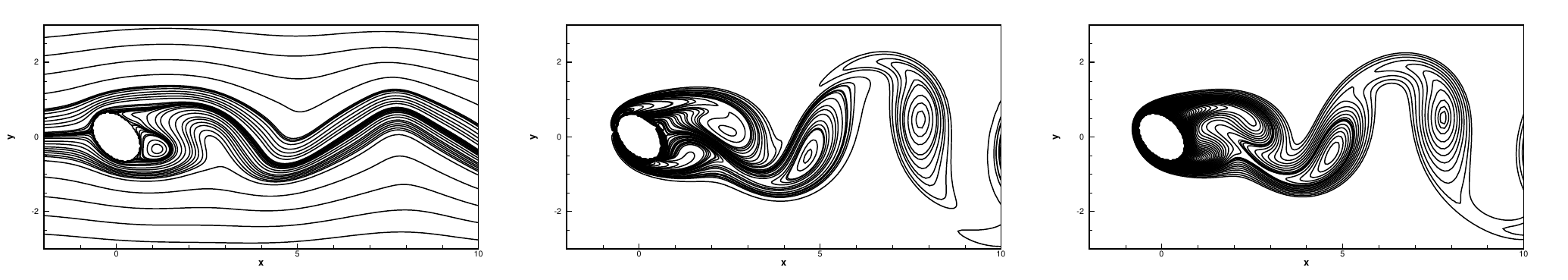} 
		\caption{$t+T/2$}
	\end{subfigure}\hfil 
	\begin{subfigure}{\textwidth}
		\includegraphics[width=\linewidth]{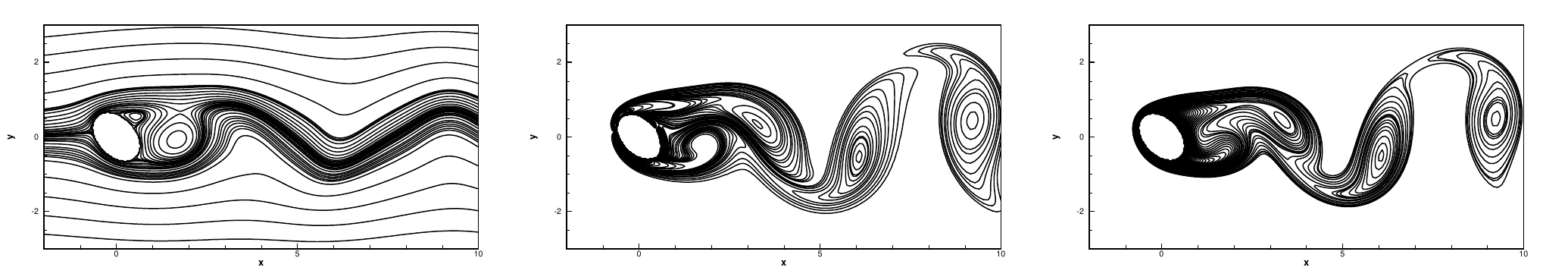} 
		\caption{$t+3T/4$}
	\end{subfigure}\hfil 
	\begin{subfigure}{\textwidth}
		\includegraphics[width=\linewidth]{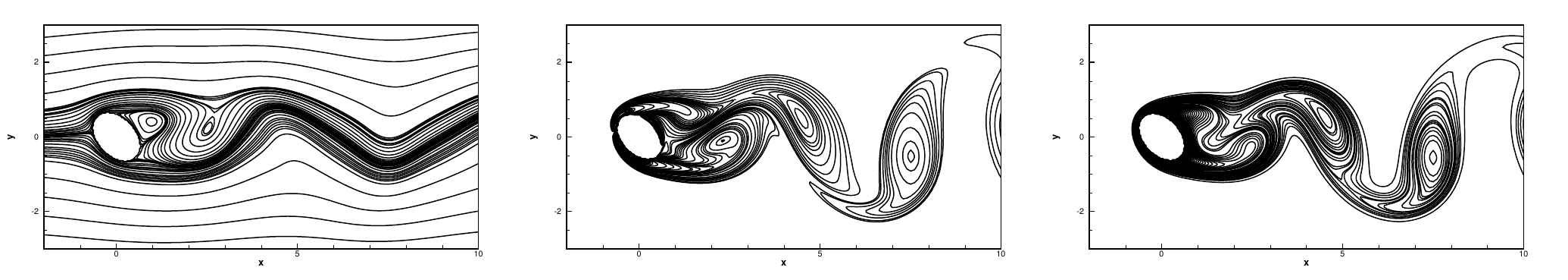} 
		\caption{$t+T$}
	\end{subfigure}\hfil 
	\caption{\small{Instantaneous streamlines (left), vorticity contours (middle) and isotherms (right) within a vortex shedding period for $Re = 100$ and $\theta=45^{\degree}$.}}
	\label{Fig:unsteady-psi-vort-temp-45}
\end{figure}

\begin{figure}[H]
	\centering
	\begin{subfigure}{\textwidth}
		\includegraphics[width=\linewidth]{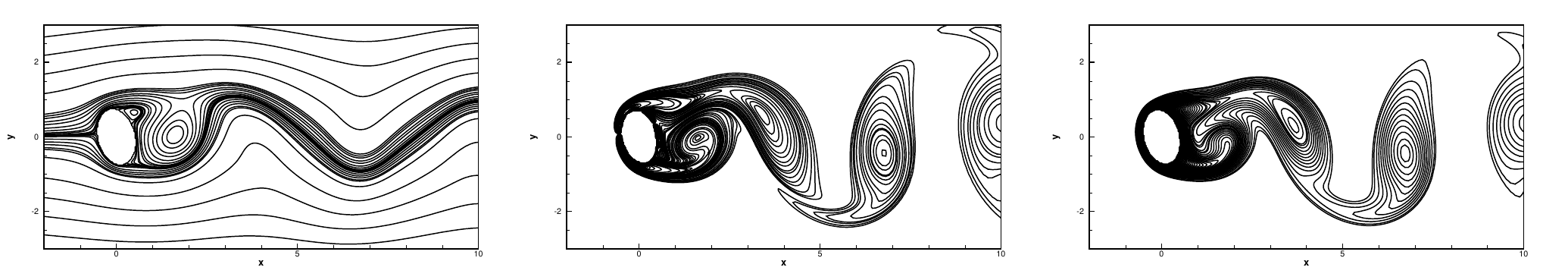} 
		\caption{$t+T/4$}
	\end{subfigure}\hfil 
	\begin{subfigure}{\textwidth}
		\includegraphics[width=\linewidth]{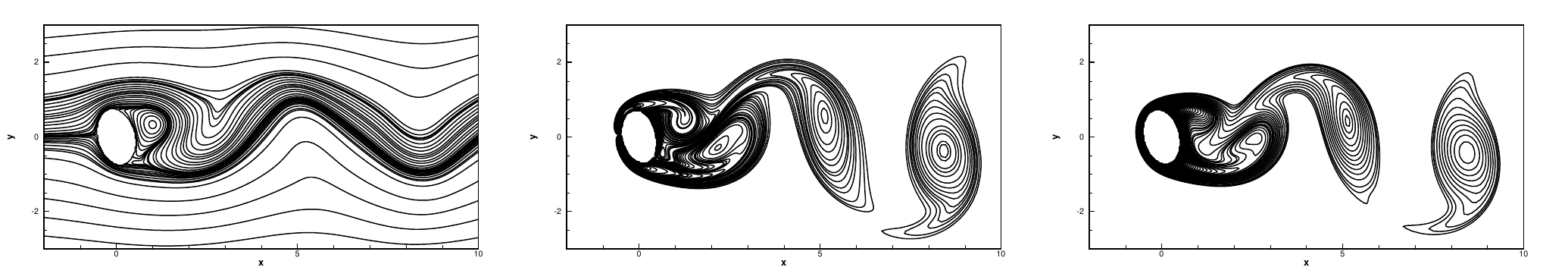} 
		\caption{$t+T/2$}
	\end{subfigure}\hfil 
	\begin{subfigure}{\textwidth}
		\includegraphics[width=\linewidth]{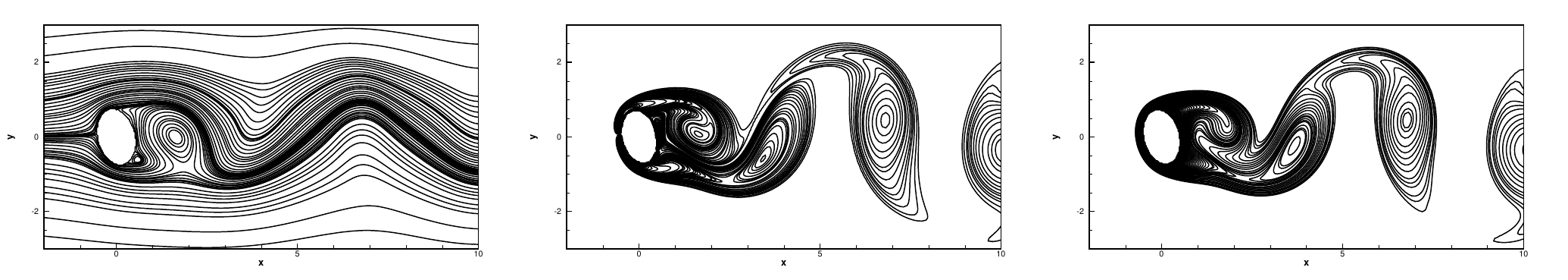} 
		\caption{$t+3T/4$}
	\end{subfigure}\hfil 
	\begin{subfigure}{\textwidth}
		\includegraphics[width=\linewidth]{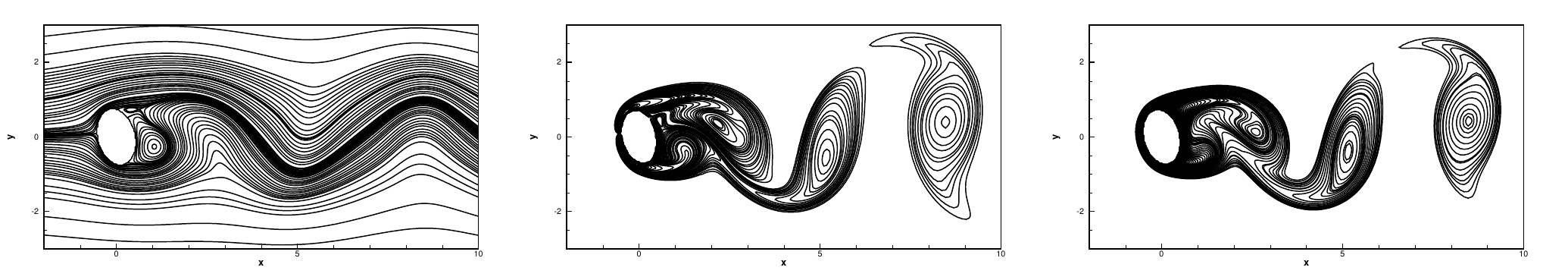} 
		\caption{$t+T$}
	\end{subfigure}\hfil 
	\caption{\small{Instantaneous streamlines (left), vorticity contours (middle) and isotherms (right) within a vortex shedding period for $Re = 100$ and $\theta=75^{\degree}$.}}
	\label{Fig:unsteady-psi-vort-temp-75}
\end{figure}

The instantaneous isotherms also depict vortex shedding (figures \ref{Fig:unsteady-psi-vort-temp-0} - \ref{Fig:unsteady-psi-vort-temp-75} ). Figure \ref{Fig:re100-vort-T} shows the instantaneous vorticity contours and isotherms for $Re=100$ and $0 \leq \theta \leq 90^{\degree}$. Note that the vorticity contours are structurally similar to the corresponding isotherms, which implies that the shedding vortices carry the heat away with them from the heated cylinder. The core of the vortex contains most of the heat, and the heat gets diffused into the free stream as the vortices are convected away from the cylinder. One can observe that the hot fluid is captured in the core of the shed vortices, as can be seen from the existence of local maxima of the contour values at the vortex centers. Also, one can see the heat being diffused into the free stream in the far wake. One of the other ways to demonstrate the diffusion of heat into the free stream is to carry out a Fast Fourier Transform (FFT) of the transverse component of velocity and temperature at different locations downstream of the cylinder. Figure \ref{Fig:fft-re100-pi12} shows the FFT of the $y$-velocity $v$ at six different locations viz. $x=10$, $x=20$, $x=30$, $x=40$, $x=50$, and $x=60$ for $Re=100$ and $\theta = 15^{\degree}$. The primary frequency $f_P$ is the vortex shedding frequency. One can observe that the value of $f_P$ remains same in all the locations. However the amplitude decreases as one moves from $x=10$ to $x=60$. This shows the diffusion of energy downstream of the cylinder.

\begin{figure}[H]
	\centering
	\begin{subfigure}{1.1\textwidth}
		\includegraphics[width=\linewidth]{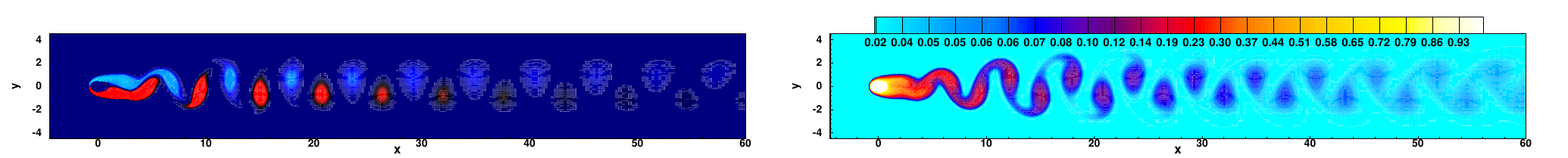} 
		\caption{$\theta = 0^{\degree}$}
	\end{subfigure}\hfil 
\begin{subfigure}{1.1\textwidth}
	\includegraphics[width=\linewidth]{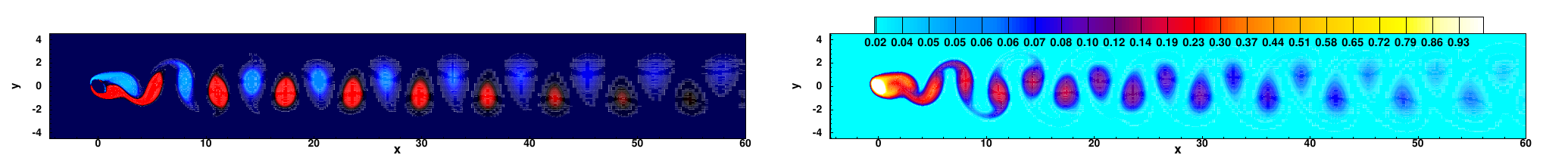} 
	\caption{$\theta = 45^{\degree}$}
\end{subfigure}\hfil
\begin{subfigure}{1.1\textwidth}
	\includegraphics[width=\linewidth]{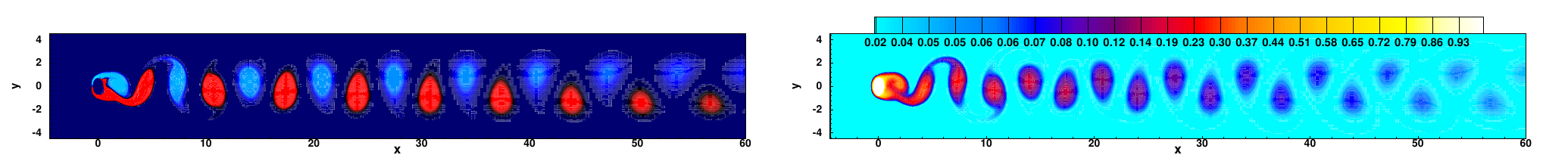} 
	\caption{$\theta = 90^{\degree}$}
\end{subfigure}\hfil
	\caption{\small{Instantaneous vorticity contours (left) and isotherms (right) for $Re=100$ and (a) $\theta = 0^{\degree}$, (b)$\theta = 15^{\degree}$, (c)$\theta = 30^{\degree}$, (d)$\theta = 45^{\degree}$, (e)$\theta = 60^{\degree}$, (f)$\theta =75^{\degree}$, and (g)$\theta = 90^{\degree}$}}
	\label{Fig:re100-vort-T}
\end{figure}
\begin{figure}[H]
	\centering
	\includegraphics[width = \textwidth]{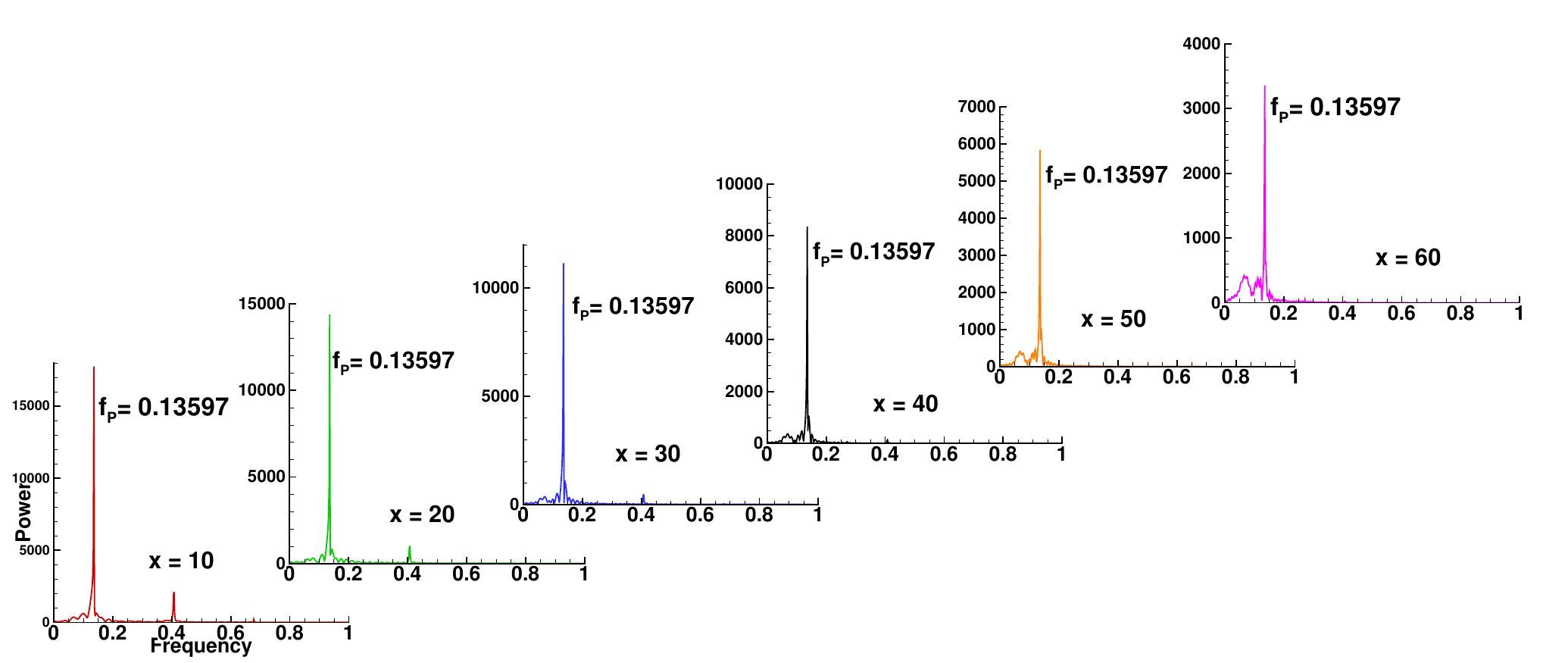}
	\caption{\small{Power spectra of the time history of $v$-velocity at six spatial locations.}}\label{Fig:fft-re100-pi12}
\end{figure}

Another interesting characteristic of the flow field is that the vortices shed from the cylinder are not of equal and opposite strength (and size) as the angle of incidence is increased. At $\theta = 0^{\degree}$ (figure \ref{Fig:re100-vort-T} (a)), counter-rotating vortices of equal and opposite strengths are shed from the cylinder. As $\theta$ increases, the upper vortex is stronger and more dominant than the lower vortex. For $0^{\degree} < \theta \leq 45^{\degree}$ (figures \ref{Fig:re100-vort-T} (a)-(d)) the lower vortex is not strong enough to overcome the suppression induced by the upper one. Thus, the growth of the lower vortex is suppressed by the upper vortex, which pushes the lower one to move slightly downstream of the flow. For $\theta \geq 60^{\degree}$ (figures \ref{Fig:re100-vort-T} (e)-(f)), the lower vortex gradually gains enough strength to balance the upper vortex, until at $\theta = 90^{\degree}$ (figure \ref{Fig:re100-vort-T} (g)) when the lower vortex balances the upper one completely, and vortices of equal sizes are shed from the cylinder.

\subsubsection{Drag and Lift coefficients}
The drag and lift coefficients, $C_D$ and $C_L$, are calculated from equations \eqref{eq:drag} and \eqref{eq:lift} respectively. Figure \ref{Fig:Cd-Cl-re100} shows the time history of $C_D$, $C_L$ for $Re=100$ and $0^{\degree} \leq \theta \leq 90^{\degree}$. Since the flow field is oscillatory in nature at this value of $Re$, the force coefficients also exhibit an oscillatory behaviour. $C_D$ and $C_L$ can written as $C_D = \overline{C_D} + C_D^'(t)$, $C_L = \overline{C_L} + C_L^'(t)$, where $\overline{C_D}$ and $\overline{C_L}$ are mean values that remain constant with time, and $C_D^'(t)$, $C_L^'(t)$ are the fluctuating components.  It can be observed from figures \ref{Fig:Cd-Cl-re100} (a), \ref{Fig:Cd-Cl-re100} (b) that the drag force first decreases as $\theta$ changes from $0^{\degree}$ to $15^{\degree}$. Note that a magnified view of $C_D$ is provided in the inset of figure \ref{Fig:Cd-Cl-re100} (a) for clarity. The value of $C_D$ then increases for $\theta = 15^{\degree} - 45^{\degree}$ (figures \ref{Fig:Cd-Cl-re100} (b) - \ref{Fig:Cd-Cl-re100} (d)). It again drops as $\theta$ is increased to $60^{\degree}$, after which it increases till $\theta = 90^{\degree}$. The value of $C_L$ is positive only for $\theta = 0^{\degree}$, $45^{\degree}$. For the rest of the values of $\theta$, we witness negative lift values. 

\begin{figure}[H]
	\centering
	\begin{subfigure}{0.3\textwidth}
		\includegraphics[width=\linewidth]{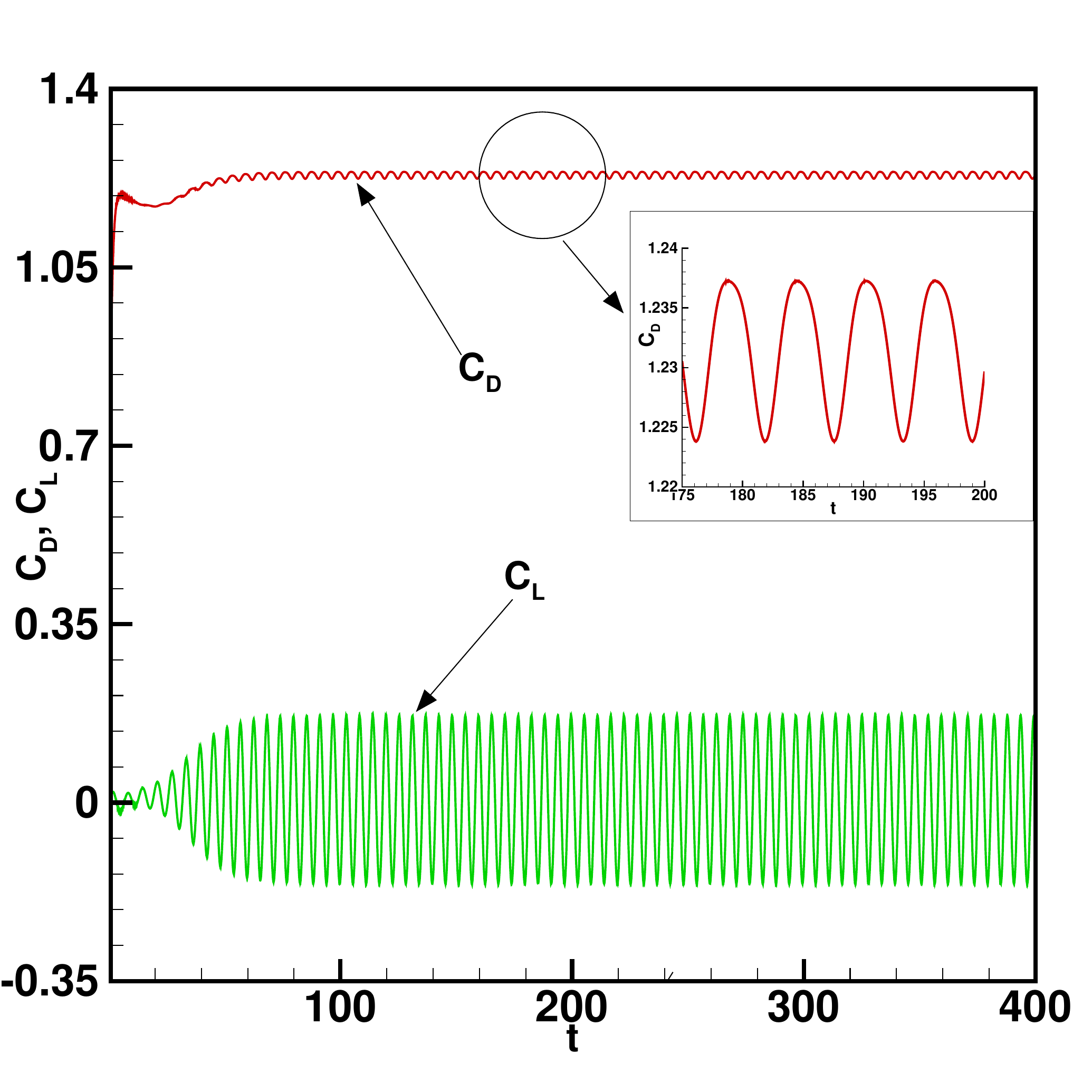} 
		\caption{$\theta=0^{\degree}$}
	\end{subfigure}\hfil 
	\begin{subfigure}{0.3\textwidth}
		\includegraphics[width=\linewidth]{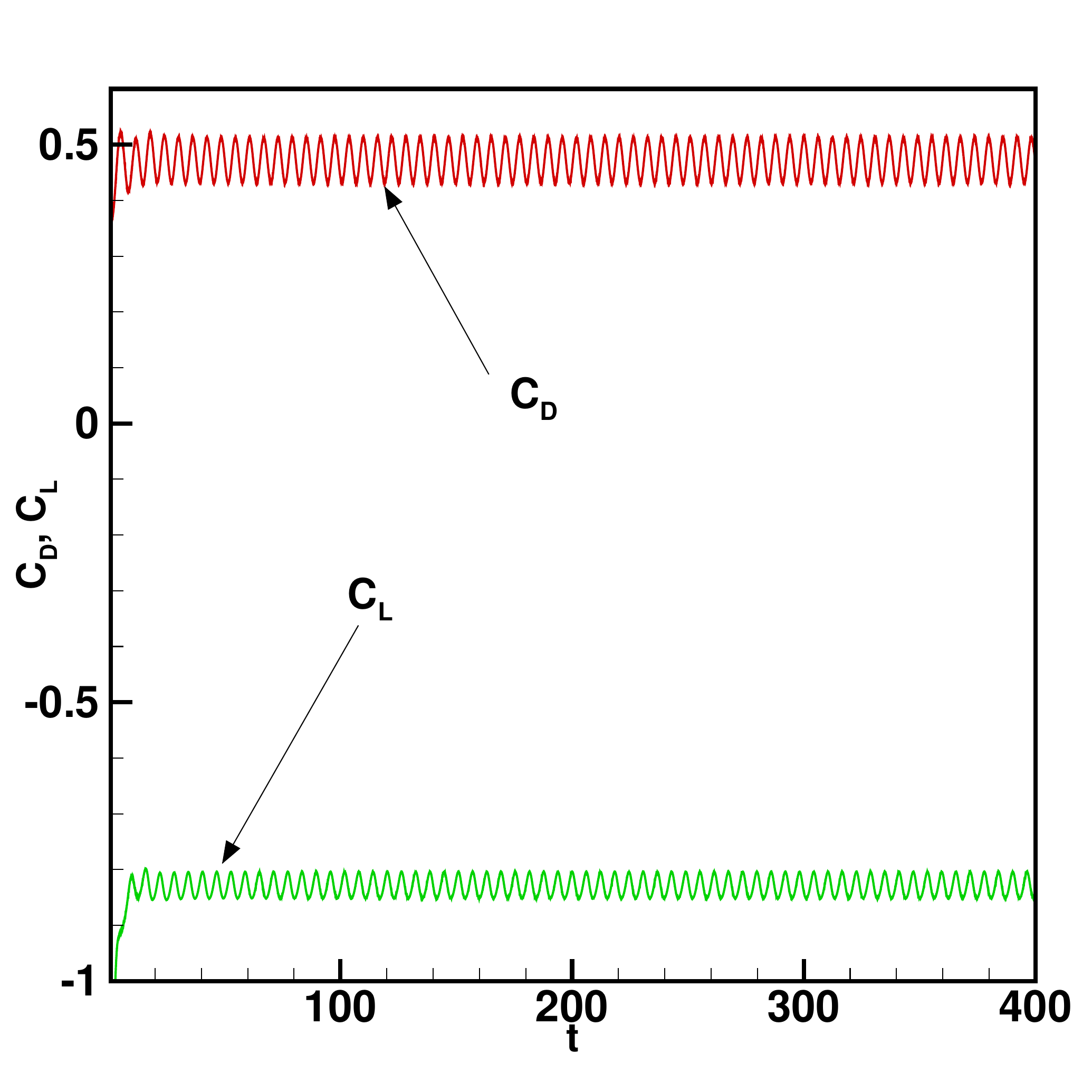} 
		\caption{$\theta=15^{\degree}$}
	\end{subfigure}\hfil 
	\begin{subfigure}{0.3\textwidth}
		\includegraphics[width=\linewidth]{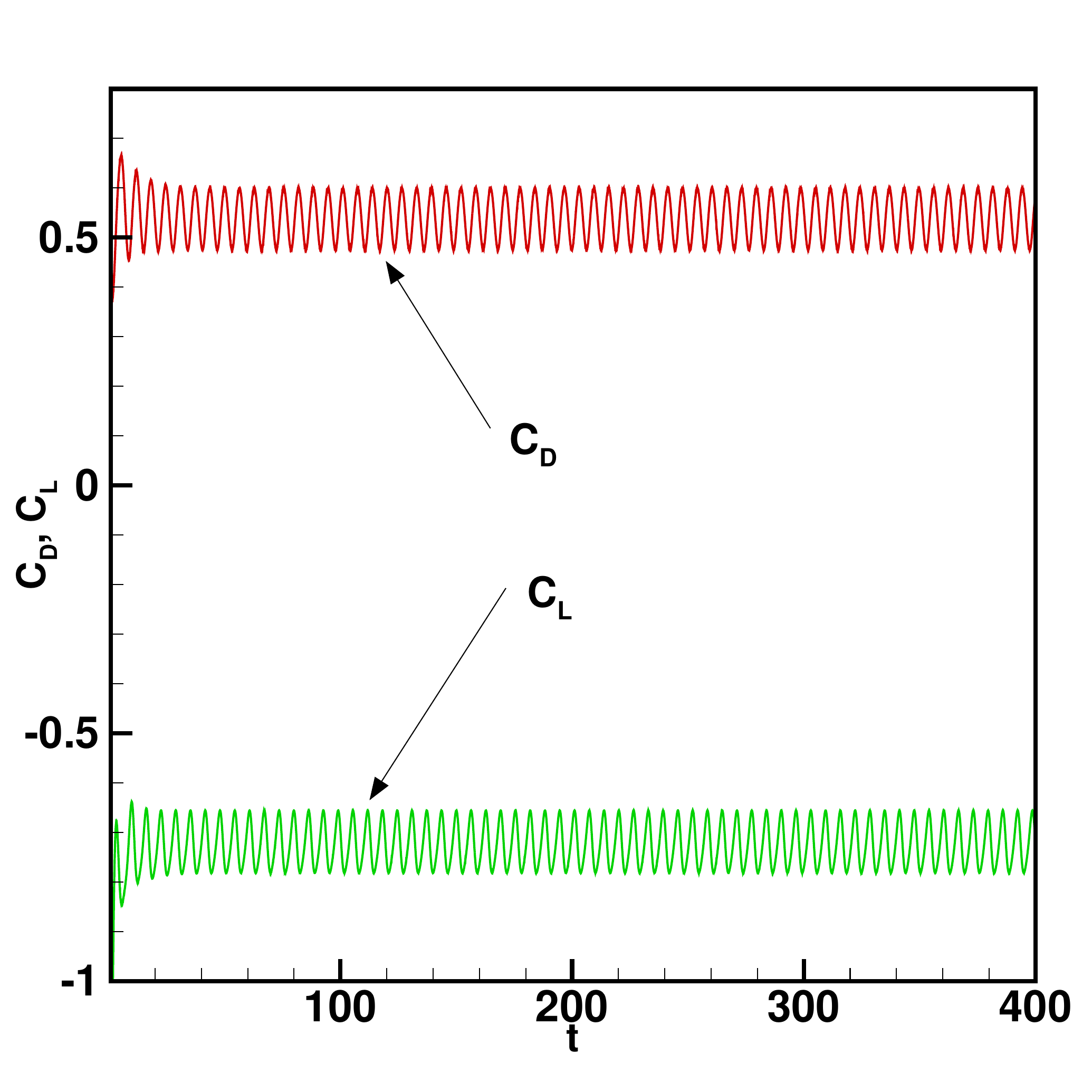} 
		\caption{$\theta=30^{\degree}$}
	\end{subfigure}\hfil 
	\begin{subfigure}{0.3\textwidth}
		\includegraphics[width=\linewidth]{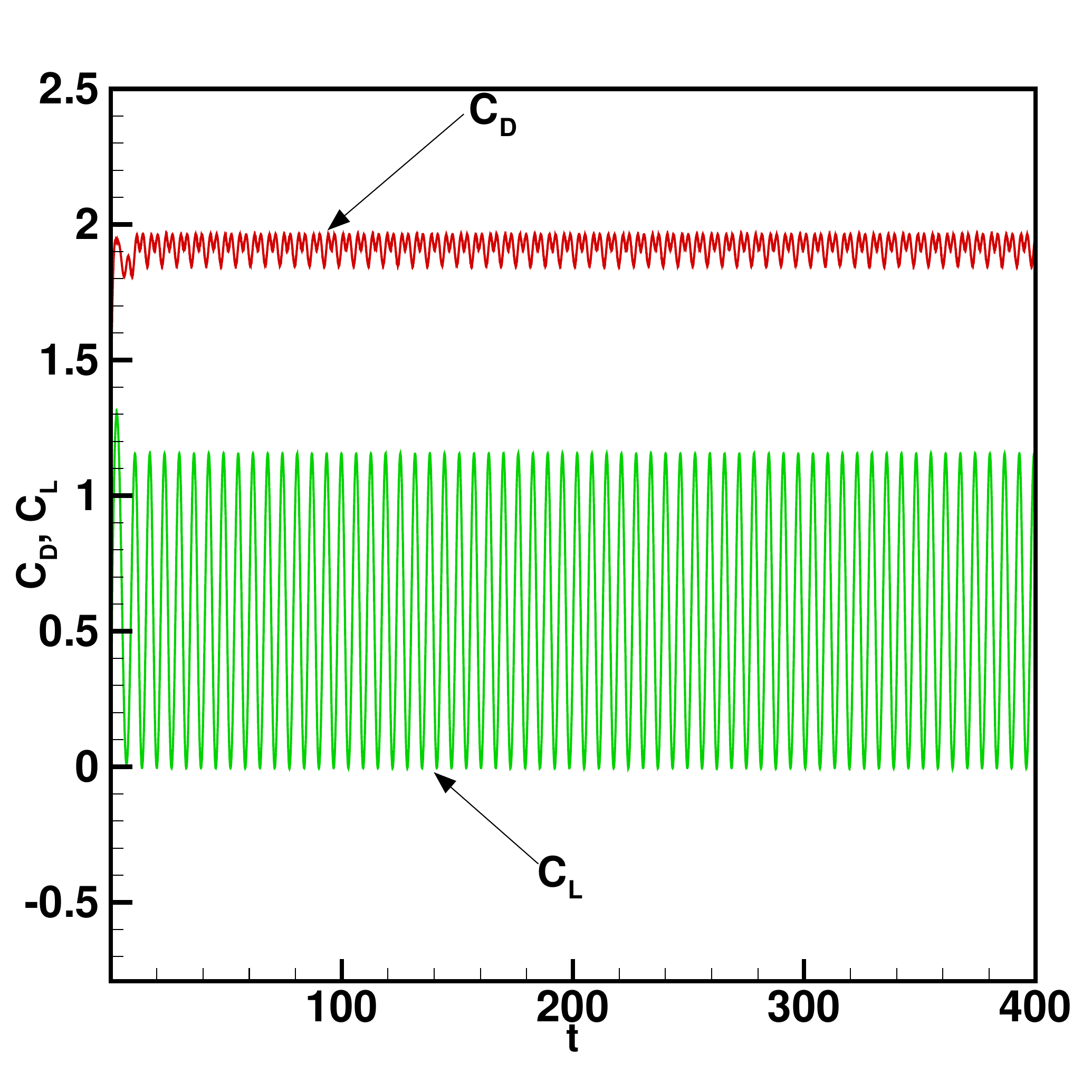} 
		\caption{$\theta=45^{\degree}$}
	\end{subfigure}\hfil 
	\begin{subfigure}{0.3\textwidth}
		\includegraphics[width=\linewidth]{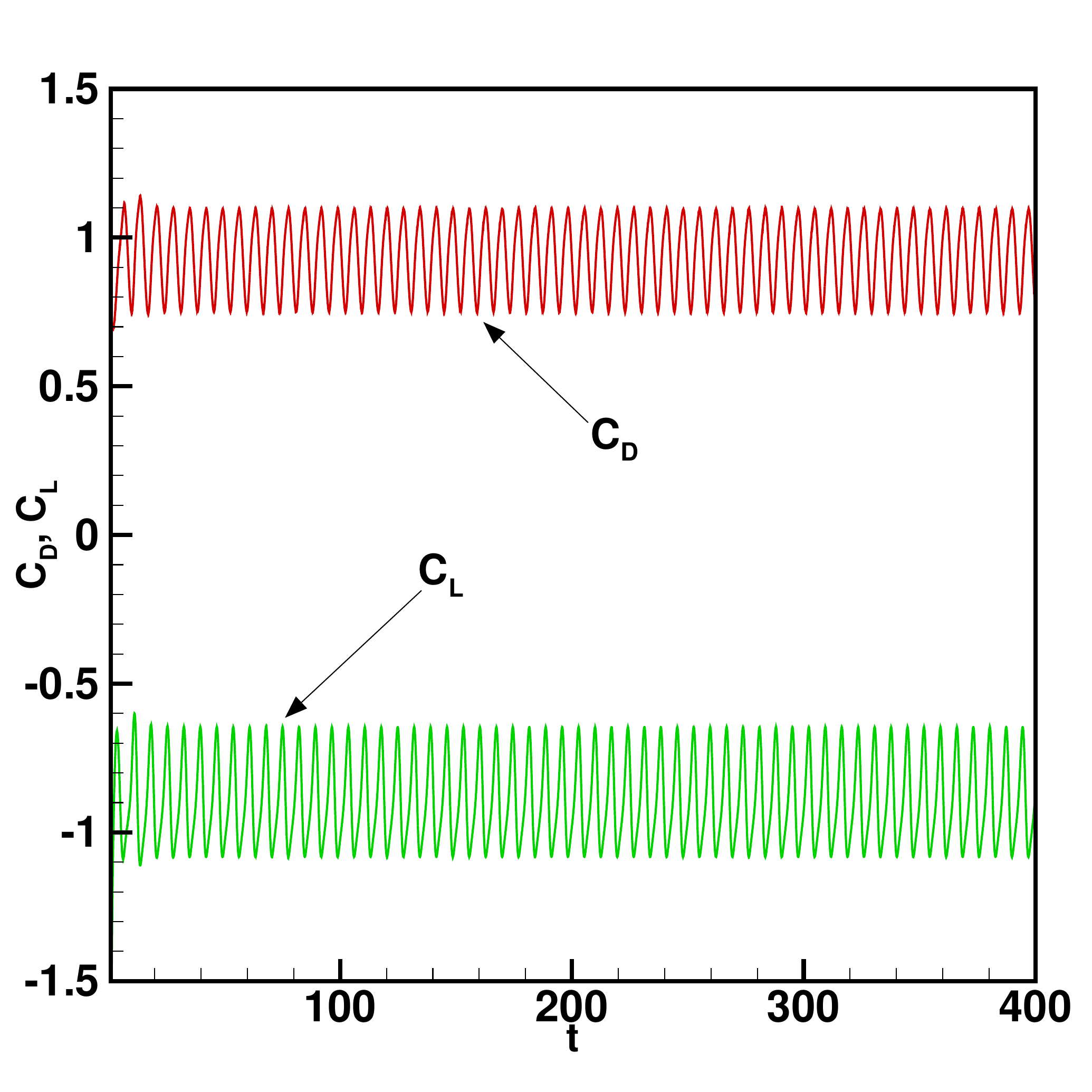} 
		\caption{$\theta=60^{\degree}$}
	\end{subfigure}\hfil 
	\begin{subfigure}{0.3\textwidth}
		\includegraphics[width=\linewidth]{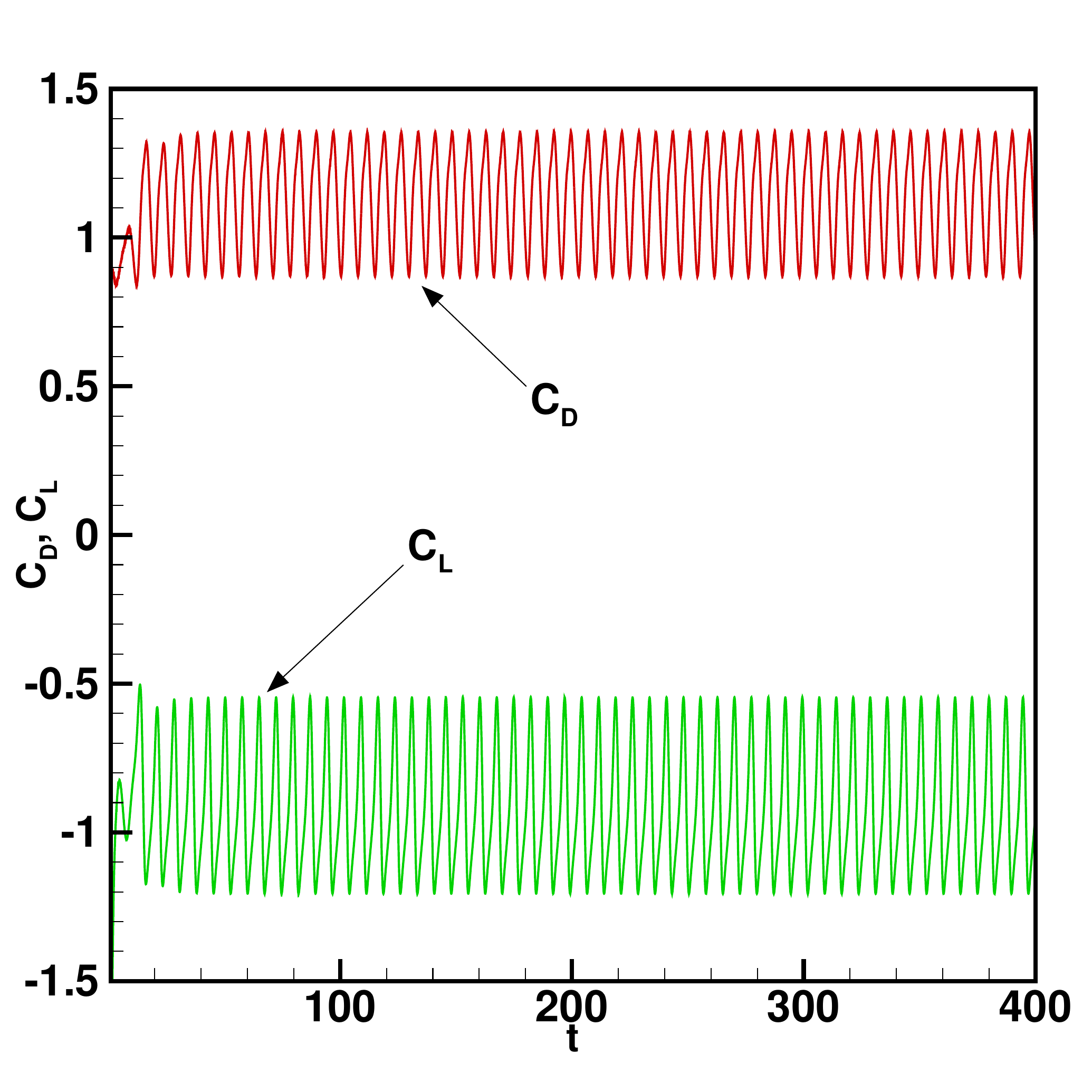} 
		\caption{$\theta=75^{\degree}$}
	\end{subfigure}\hfil 
	\begin{subfigure}{0.3\textwidth}
		\includegraphics[width=\linewidth]{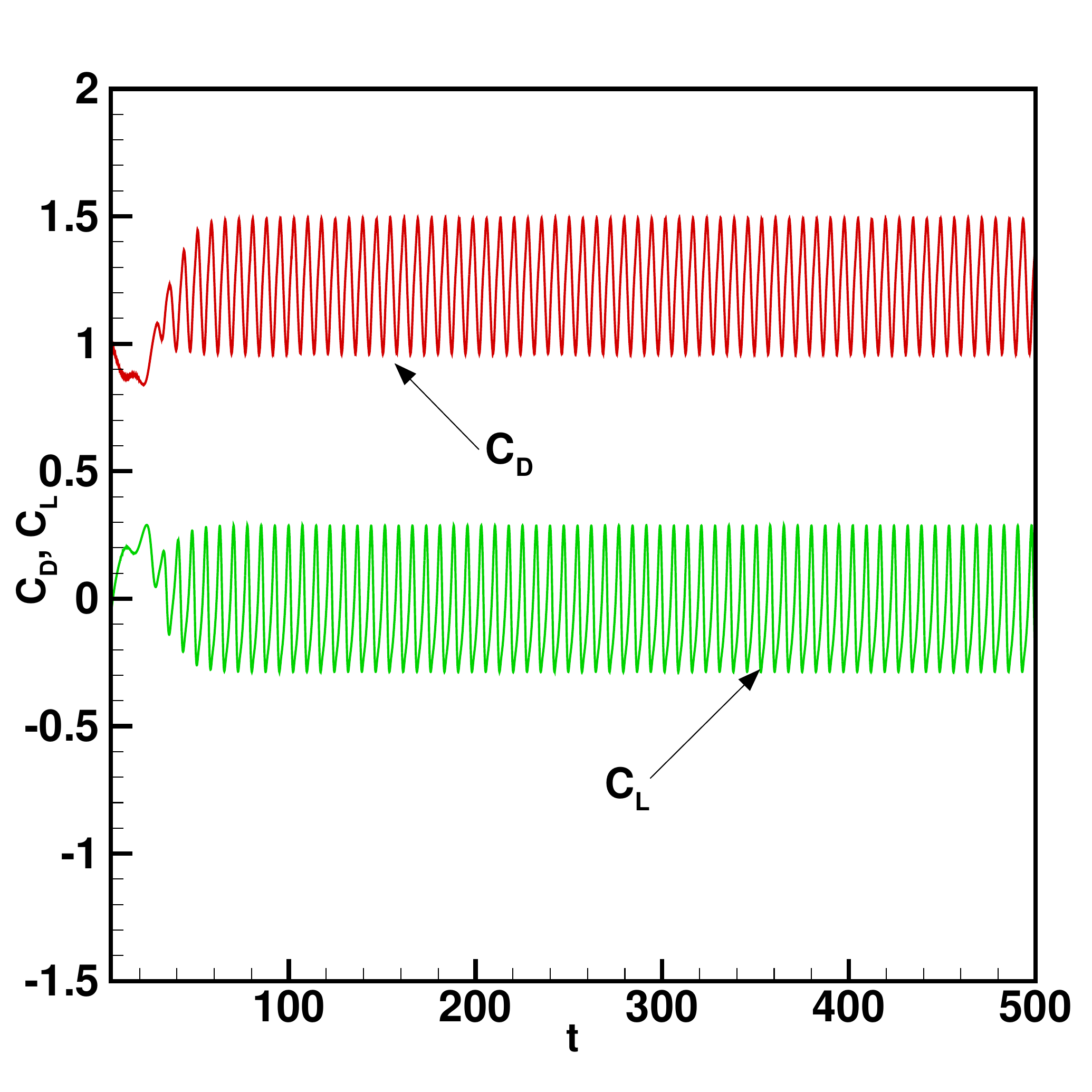} 
		\caption{$\theta=90^{\degree}$}
	\end{subfigure}\hfil 
	\caption{\small{Time variation of $C_D$ and $C_L$ for $Re=100$ and (a) $\theta=0^{\degree}$, (b) $\theta=15^{\degree}$, (c) $\theta=30^{\degree}$, (d) $\theta=45^{\degree}$, (e) $\theta=60^{\degree}$, (f) $\theta=75^{\degree}$, (g) $\theta=90^{\degree}$.}}
	\label{Fig:Cd-Cl-re100}
\end{figure}

%
%
%
%
%
%
%
%
%
%
%
%
%
%

\subsubsection{Strouhal Number}
The Strouhal number ($St$) is a measure of the vortex shedding phenomenon, which is defined as 
\begin{equation}
	St = \dfrac{fa}{U_0}
\end{equation} 

where $f$ is the vortex shedding frequency which is determined as the peak frequency
derived from the FFT of the time history of $C_L$. Note that the FFT is taken after discarding an initial period of at least 300 non-dimensional time units. $U_0$ is the free stream velocity, and $a$ is the semi-major axis of the ellipse. Figure \ref{Fig:St-theta} shows the variation of $St$ with $\theta$ for two values of $Re$. We can see that the frequency of vortex shedding decreases as $\theta$ is increased. Also, for a particular value of $\theta$, the vortex shedding frequency increases with $Re$.

\begin{figure}[H]
	\centering
	\includegraphics[width = 0.75\textwidth]{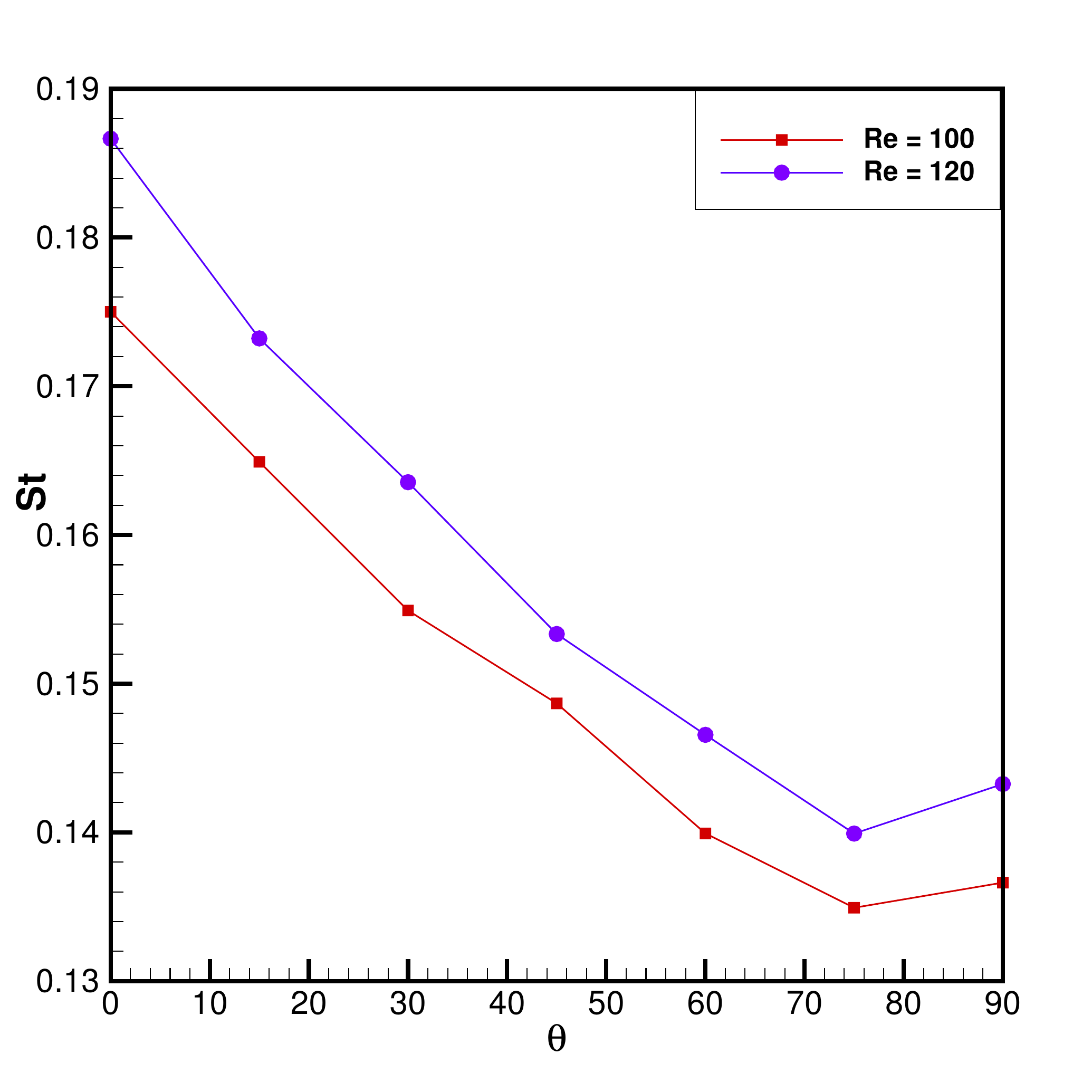}
	\caption{\small{variation of the Stouhal number $St$  against the angle of inclination $\theta$ for $Re = 100$ and $120$.}}\label{Fig:St-theta}
\end{figure}

\subsubsection{Nusselt number}
The surface averaged Nusselt number $Nu_{\text{av}}$ is calculated from the equation given by \eqref{eq:Nusselt}. Figure \ref{Fig:av-Nu-re100} shows the time history of $Nu_{\text{av}}$ for $Re=100$ and $0^{\degree} \leq \theta \leq 90^{\degree}$. For clarity, the time history is shown only for $t=350$ to $t=400$. Similar to $C_D$ and $C_L$, the surface averaged Nusselt number also exhibits a periodic behaviour w.r.t. time. In figure \ref{Fig:av-Nu-re100} (a) - (g), we have shown the time period $T_{\text{Nu}}$ for each of the angles of incidence considered. It is clear that as $\theta$ increases, $T_{\text{Nu}}$ also increases. Note that this periodicity in the variation of $Nu_{\text{av}}$ commences concurrent to vortex shedding, since the vortex shedding phenomena is invariably linked to the heat being convected away from the cylinder. One can also observe a curious co-relation between the vortex shedding phenomena and variation of $Nu_{\text{av}}$. Consider the two angles of incidence viz. $\theta = 0^{\degree}$ and $\theta = 90^{\degree}$. The Strouhal number for these two configurations are $0.175008$ and $0.136621$ respectively. Now, from figure \ref{Fig:av-Nu-re100} (a) and \ref{Fig:av-Nu-re100} (g), we see that $T_{\text{Nu}}$ for $\theta = 0^{\degree}$ and $\theta = 90^{\degree}$ are $2.857$ and $3.7049$ respectively. Thus, the frequency of oscillation of $Nu_{\text{av}}$, ($f_{\text{Nu}} = 1/T_{\text{Nu}}$) are $0.35001$ and $0.269912$ respectively. Thus we see that $f_{\text{Nu}} \approx 2\, St$. This relationship is true for all values of $\theta$ and all values of $Re$. As mentioned previously, the isotherms and vorticity contours are struturally similar owing to the fact that the shed vortices convect the heat from the cylinder downstream. Vorticity values alternate between positive and negative, whereas the temperature always remains positive. Thus, it can be expected that the frequency of isotherms being shed would be twice the vortex shedding frequency. The above exercise simply demonstrates this.
\begin{figure}[H]
	\centering
	\begin{subfigure}{0.3\textwidth}
		\includegraphics[width=\linewidth]{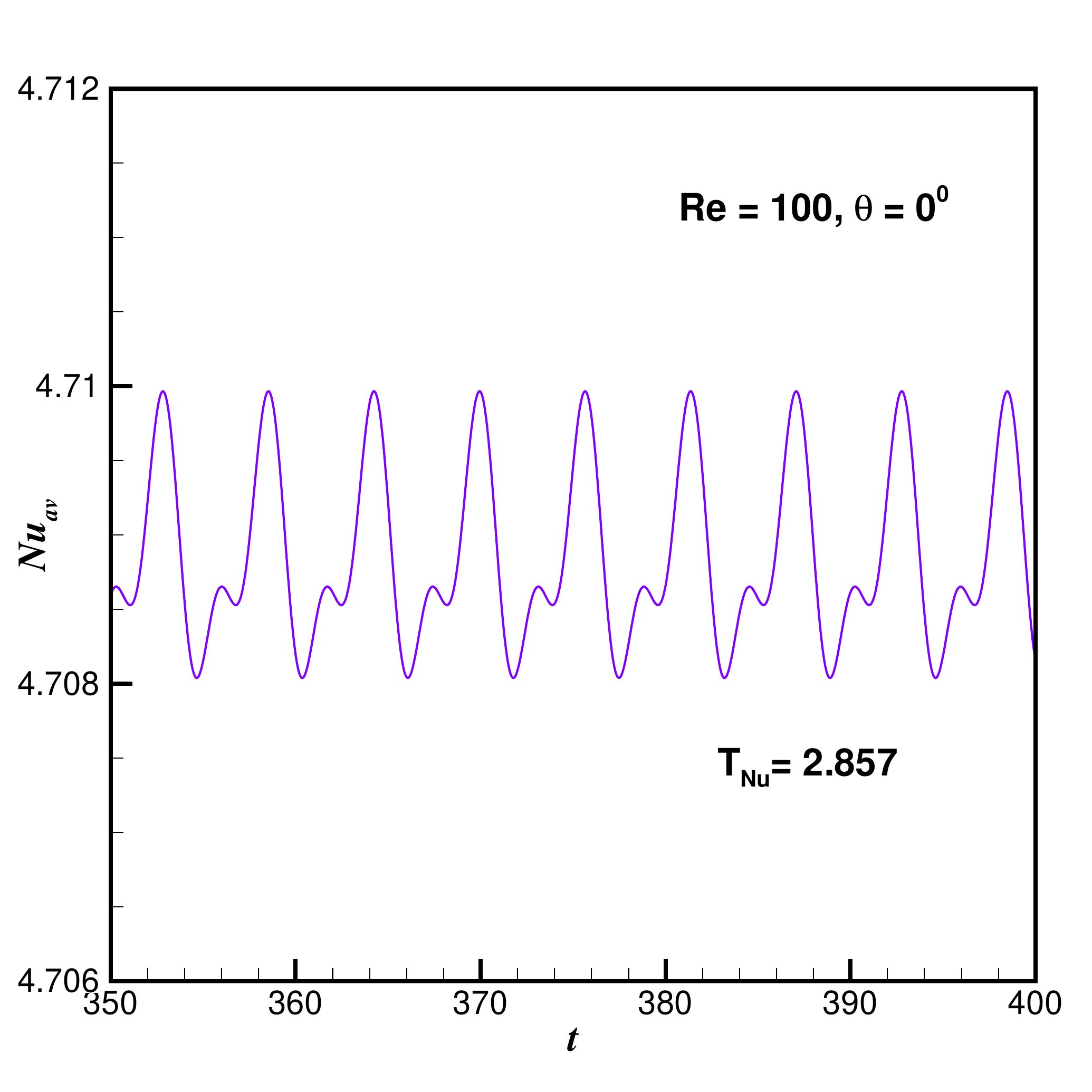} 
		\caption{$\theta=0^{\degree}$}
	\end{subfigure}\hfil 
	\begin{subfigure}{0.3\textwidth}
		\includegraphics[width=\linewidth]{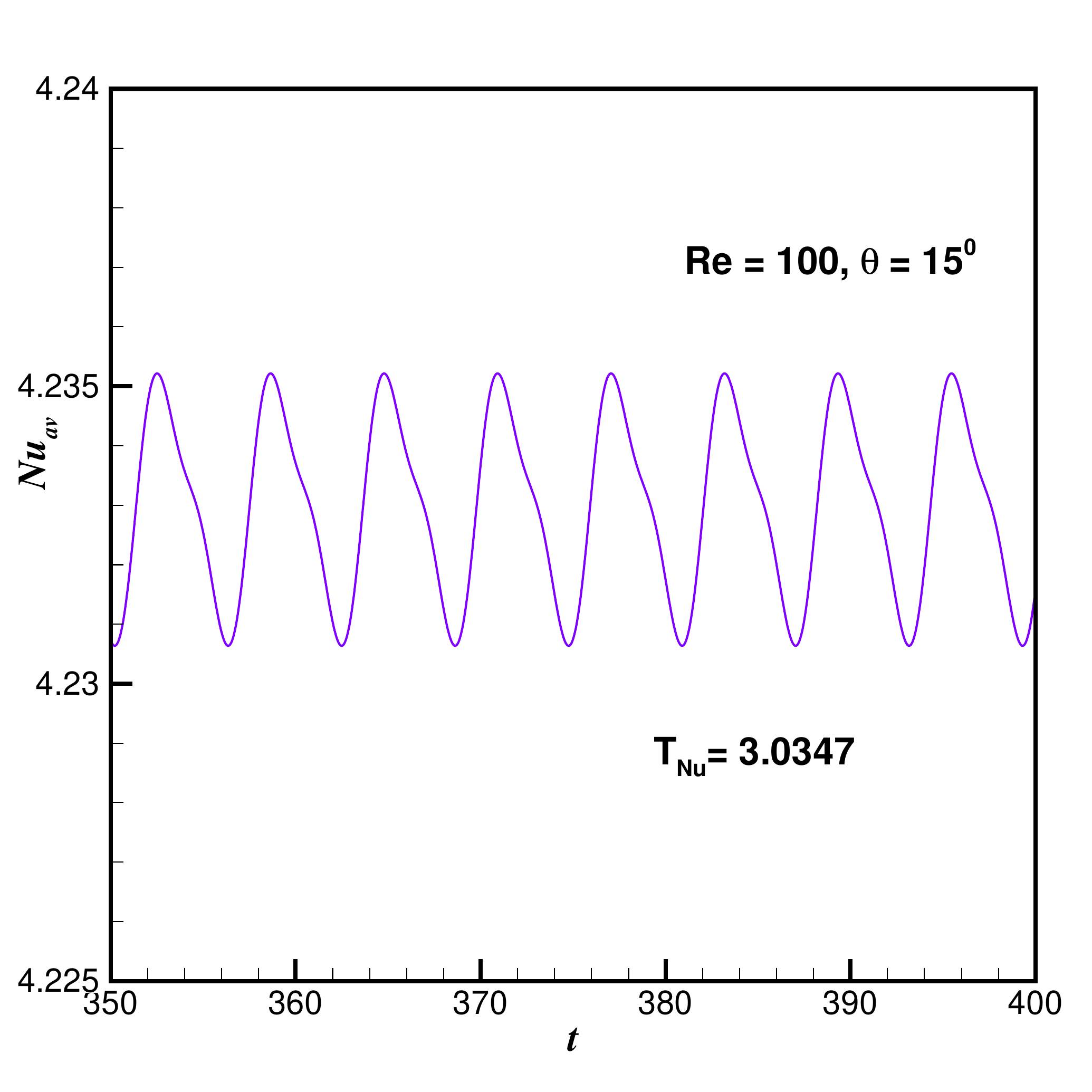} 
		\caption{$\theta=15^{\degree}$}
	\end{subfigure}\hfil 
	\begin{subfigure}{0.3\textwidth}
		\includegraphics[width=\linewidth]{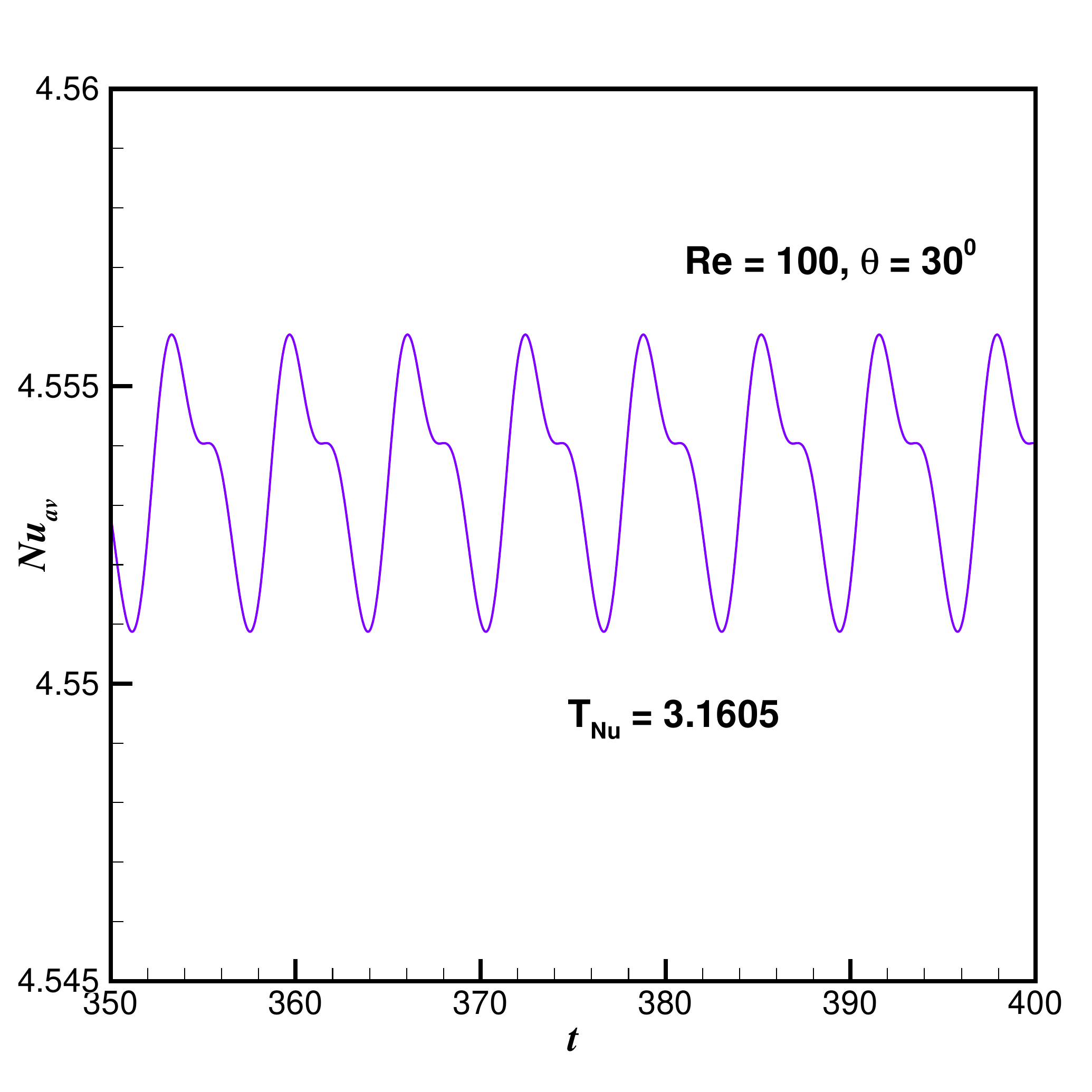} 
		\caption{$\theta=30^{\degree}$}
	\end{subfigure}\hfil 
	\begin{subfigure}{0.3\textwidth}
		\includegraphics[width=\linewidth]{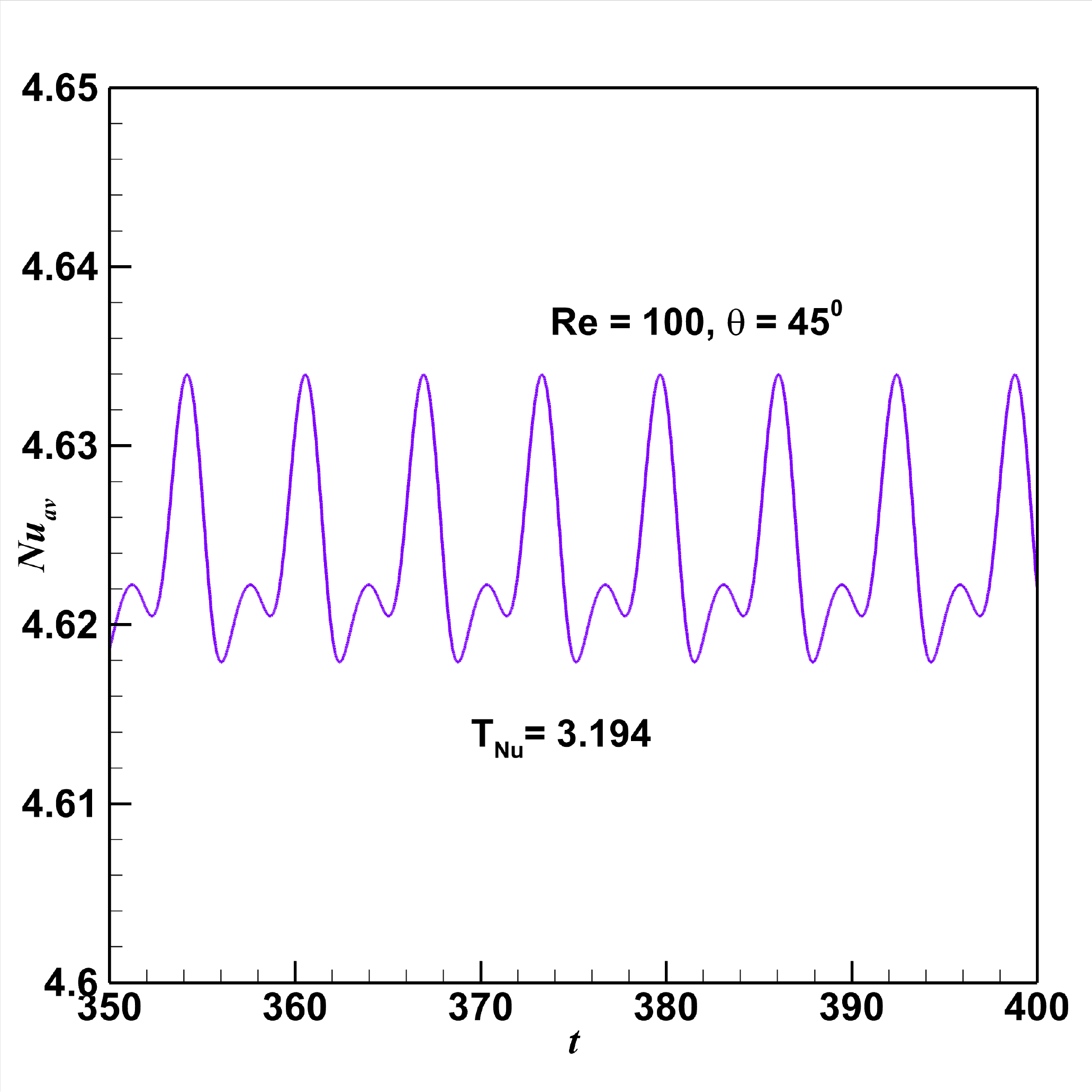} 
		\caption{$\theta=45^{\degree}$}
	\end{subfigure}\hfil 
	\begin{subfigure}{0.3\textwidth}
		\includegraphics[width=\linewidth]{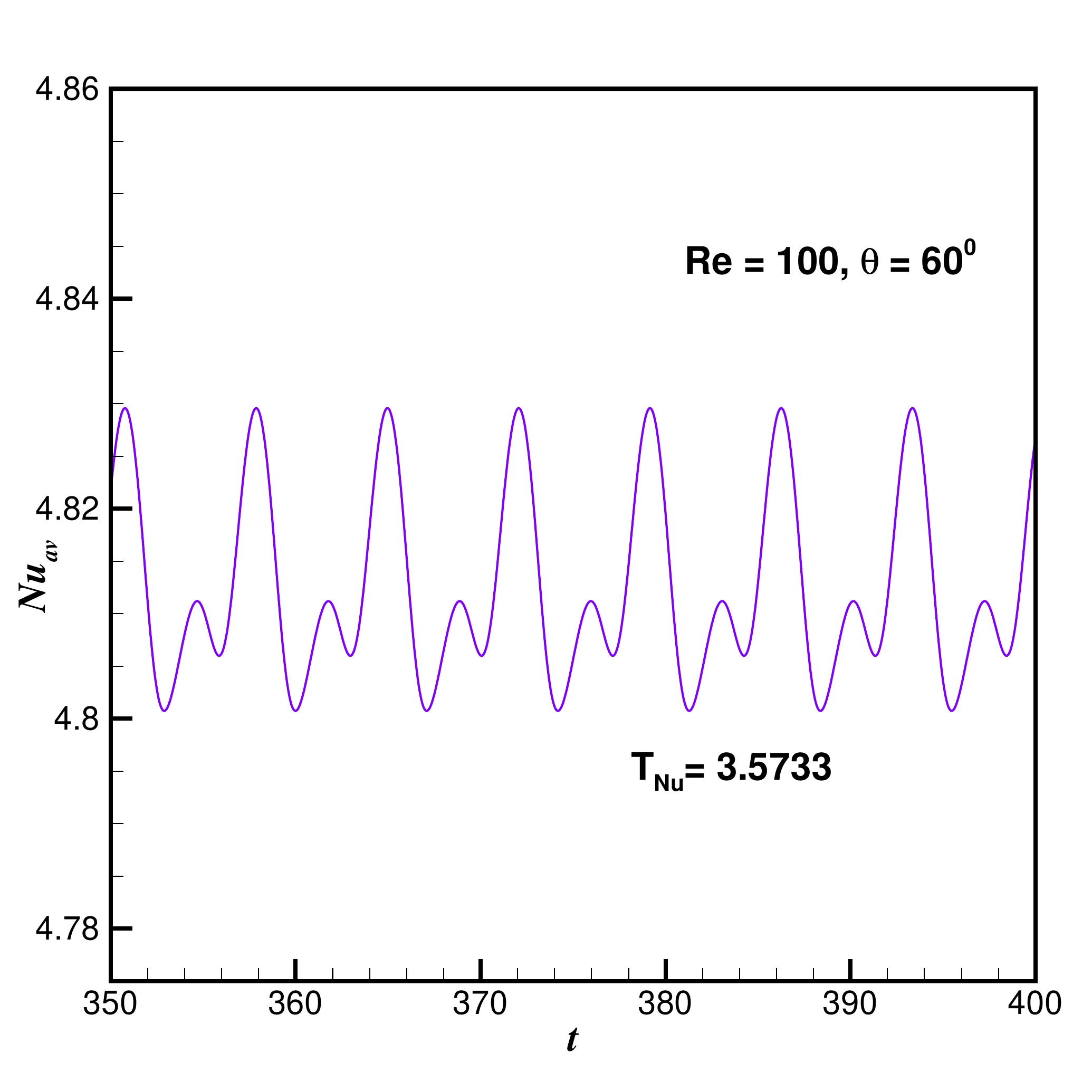} 
		\caption{$\theta=60^{\degree}$}
	\end{subfigure}\hfil 
	\begin{subfigure}{0.3\textwidth}
		\includegraphics[width=\linewidth]{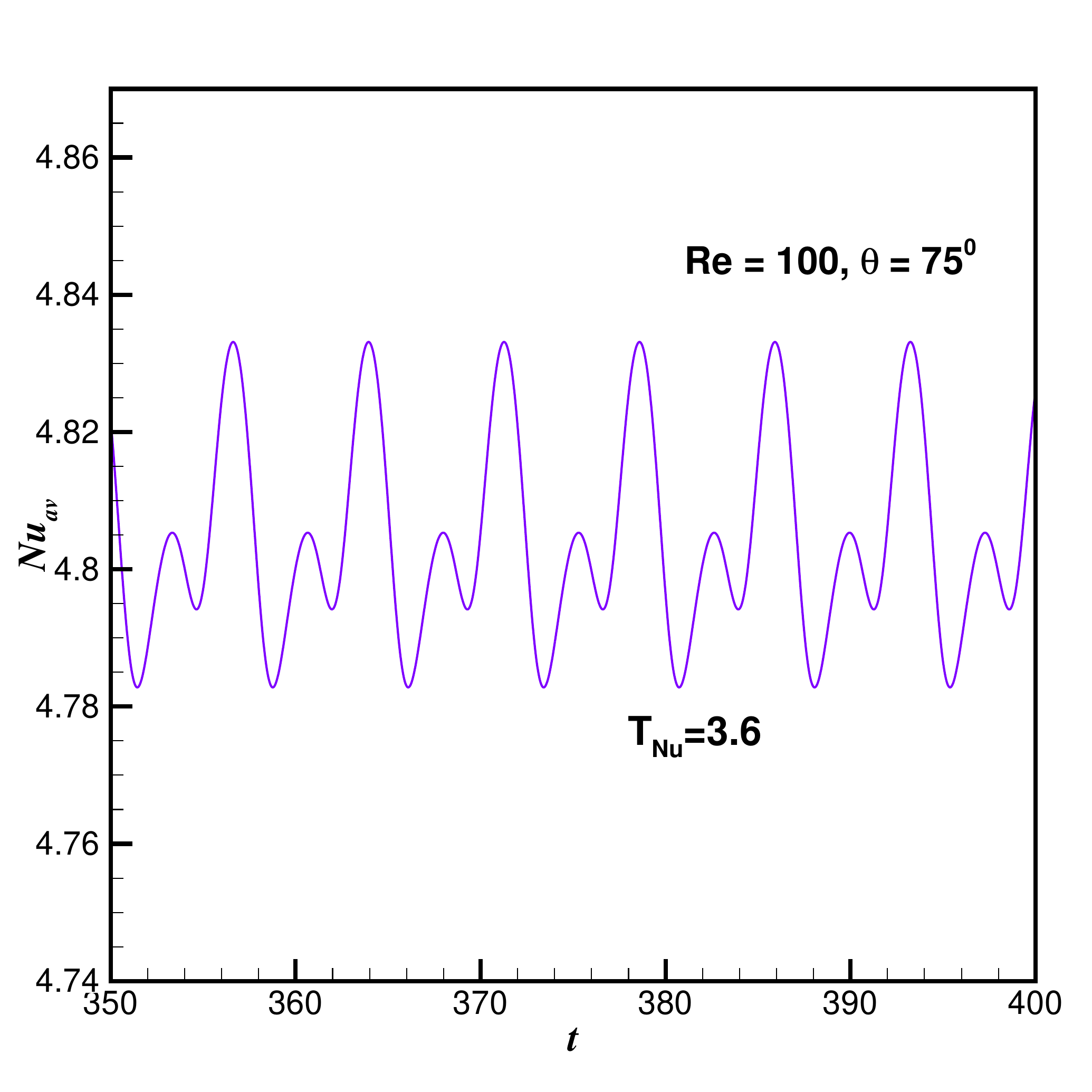} 
		\caption{$\theta=75^{\degree}$}
	\end{subfigure}\hfil 
	\begin{subfigure}{0.3\textwidth}
		\includegraphics[width=\linewidth]{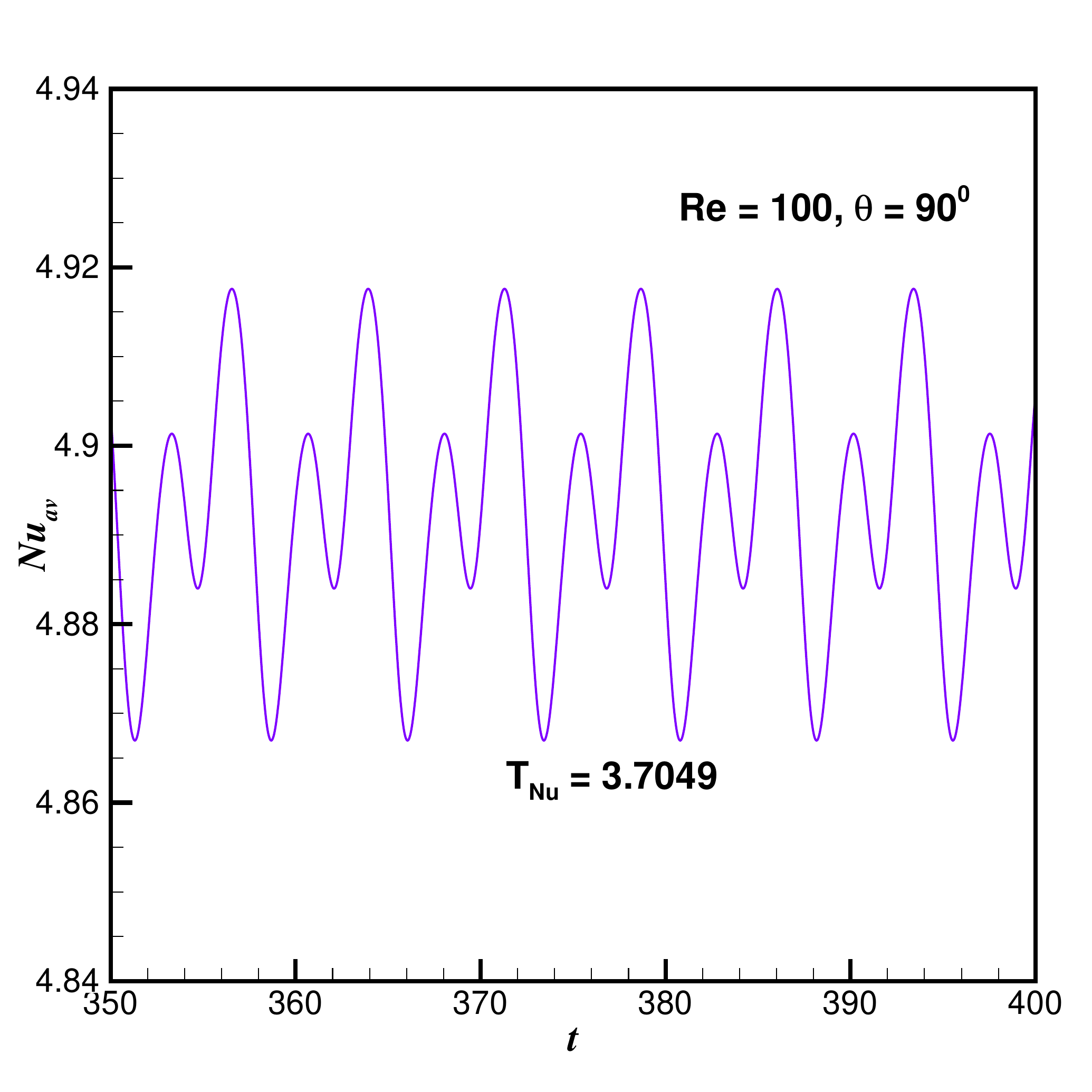} 
		\caption{$\theta=90^{\degree}$}
	\end{subfigure}\hfil 
	\caption{\small{Time variation of surface averaged Nusselt number ($Nu_{\text{av}}$) for $Re=100$ and (a) $\theta=0^{\degree}$, (b) $\theta=15^{\degree}$, (c) $\theta=30^{\degree}$, (d) $\theta=45^{\degree}$, (e) $\theta=60^{\degree}$, (f) $\theta=75^{\degree}$, (g) $\theta=90^{\degree}$.}}
	\label{Fig:av-Nu-re100}
\end{figure}

The time-varying average Nusselt number can be written as the sum of a mean value $\overline{Nu_{\text{av}}}$ and a fluctuating component $Nu_{\text{av}}^'$, i.e., $Nu_{\text{av}} = \overline{Nu_{\text{av}}} + Nu_{\text{av}}^'$. The value of $Nu_{\text{av}}^'$ is nothing but the amplitude of $Nu_{\text{av}}$ w.r.t the $\overline{Nu_{\text{av}}}$ value. Table \ref{tab:average-Nu-re100} shows the breakup of $Nu_{\text{av}}$ for $Re = 100$ at different values of $\theta$. We see that for $\theta > 0^{\degree}$ there is an increase in the value of $\overline{Nu_{\text{av}}}$ with $\theta$. At $\theta = 0^{\degree}$, value of the fluctuating component $Nu_{\text{av}}^'$ is the lowest. It increases gradually with an increase in the angle of incidence. Thus, in general the value of $Nu_{\text{av}}$ increases with $\theta$. Also, the amplitude of oscillation increases as $\theta$ increases.
\begin{table}[H]
\caption{ Surface averaged Nusselt number for different values of $\theta$ at $Re=100$}
\begin{center}
\begin{tabular}{|c|c|}  \hline
$\theta$    & $Nu_{\text{av}} = \overline{Nu_{\text{av}}} + Nu_{\text{av}}^'$  \\ \hline
 $0^{\degree}$    & $4.709005 \pm 0.000965$ \\
 $15^{\degree}$   & $4.232925 \pm 0.002285$ \\
 $30^{\degree}$   & $4.553334 \pm 0.002459$ \\
 $45^{\degree}$   & $4.625938 \pm 0.008026$ \\ 
 $60^{\degree}$   & $4.813349 \pm 0.012608$ \\
 $75^{\degree}$   & $4.807956 \pm 0.025163$\\ 
 $90^{\degree}$   & $4.892274 \pm 0.025311$\\ \hline
 \end{tabular}\label{tab:average-Nu-re100}
\end{center}
\end{table}

\section{Conclusions} \label{sec:conclusions}
In this paper, we have made a comprehensive investigation of the phenomena of forced convection heat transfer over a heated elliptical cylinder inclined to a uniform free stream of incompressible viscous flows . A recently developed HOC finite difference Immersed Interface Method for 2D transient problems involving bluff bodies immersed in fluid flows on Cartesian mesh has been employed to simulate the flow. Numerical simulations were carried out for the range of Reynolds number  $10 \leq Re \leq 120$,  inclination angle $0^{\degree} \leq \theta \leq 180^{\degree}$, with air as the working fluid ($Pr = 0.71$) and the aspect ratio is taken $2/3$. In the process we also proposed a novel way to calculate the Nusselt number. To the best of our knowledge, no other comprehensive study exists for forced convection heat transfer over an elliptical cylinder where such wide variation of angles of inclination and Reynolds numbers are considered. Hence, code validation is carried out by simulating forced convection over a horizontal circular cylinder at low Reynolds numbers, and excellent match is obtained with well established results in the literature.

Results  for both steady and unsteady regimes have been reported in terms of streamlines, vorticity contours, isotherms, drag and lift coefficients, Strouhal number, and Nusselt number. In the process, we have also proposed a novel method of estimating the Nusselt number by showing how the flow variables could be computed along the normal at a point to the ellipse boundary.  The flow field for $180^{\degree}-\theta$ was found out to be a mirror image of flow for  $\theta$ ($0^{\degree} \leq \theta \leq 90^{\degree}$). 

For the steady regime, flow in the wake of the cylinder exhibited a symmetry about the $x$-axis for $\theta = 0^{\degree}$, $90^{\degree}$. Thus the streamlines as well as isotherms are symmetric for these two angles of incidence. As the angle of incidence increases, flow separation and formation of recirculation bubble were found to occur at a lower value of $Re$. Also, the value of the critical Reynolds number $Re_c$ decreases with $\theta$. For $0^{\degree} < \theta < 90^{\degree}$, it was observed that the size and strength of the upper vortex was greater than the lower one. This difference in size and strength was pronounced for lower values of $\theta$, and it was seen to decrease gradually as $\theta \rightarrow 90^{\degree}$. As $\theta$ increased further, this trend was reversed. Heat transfer phenomena was demonstrated via the local and surface averaged Nusselt number. The variation in the local Nusselt number was plotted along the surface of the cylinder, and the trends observed could be satisfactorily correlated to the flow field. The surface averaged Nusselt number was observed to increase with $Re$ for a given $\theta$. Further, for a particular $Re$, thee maximum value was seen to attain at $\theta = 0^{\degree}$. On the other hand, the drag force acting on the cylinder decreased with the increase in $Re$ , which however, was seen to increase with $\theta$ for a given $Re$.

Since the unsteady laminar regime is characterized by periodic vortex shedding, results for only a single $Re$ was demonstrated as a representative case. Streamlines, vorticity contours, and isotherms were shown for a vortex shedding cycle at different values of $\theta$. In a shedding cycle, it was seen that the growth of the upper vortex is accompanied by the formation of a lower vortex in the flow field. While the upper vortex begins to decay, the lower vortex grows and attaches itself to the trailing edge. Subsequently, the upper vortex reappears around the leading edge and grows in such a way that it suppresses the lower vortex, which starts to get smaller. As $\theta$ increases, the undulations in the streamlines were seen to grow more complicated and vortex shedding occurring at a shorter distance from the trailing edge of the cylinder, becoming much wider as $\theta$ is increased. On account of the shed vortices carrying away the heat from the cylinder, the isotherms were also seen to depict vortex shedding as they are structurally similar. The core of the vortex contained most of the heat and it got diffused into the free stream. This diffusion process is demonstrated by the contour plots of temperature and vorticity, as well as a FFT of the $y$- component of the velocity at different locations in the domain. A plot of the Strouhal number showed that vortex shedding frequency increases with $Re$, and decreases with $\theta$ for a given $Re$. The surface averaged Nusselt number showed a periodic variation with time, its time period being half the time period of vortex shedding. The mean value of $Nu_{\text{av}}$ as well as the amplitude of oscillations were also observed to increase with $\theta$.

\bibliographystyle{apalike}
\bibliography{ellipse}
\end{document}